\documentclass[11pt,a4paper]{article}
\usepackage[utf8]{inputenc}
\usepackage{styling}
 
\hypersetup{
  pdftitle={Balanced Allocations: Caching and Packing, Twinning and Thinning},
  pdfauthor={Dimitrios Los, Thomas Sauerwald, John Sylvester}
  }

\title{Balanced Allocations:~\\ Caching and Packing, Twinning and Thinning\footnote{Full version of a paper appearing in SODA 2022}}

\author[1]{Dimitrios Los}  
\author[1]{Thomas Sauerwald\thanks{T.S. was supported by the ERC Starting Grant 679660 (DYNAMIC MARCH).}}
\author[2]{John Sylvester\thanks{J.S. was supported by EPSRC project EP/T004878/1: Multilayer Algorithmics to Leverage Graph Structure }}
\affil[1]{Department of Computer Science \& Technology, University of Cambridge, UK\\ \texttt{firstname.lastname@cl.cam.ac.uk}} 
\affil[2]{School of Computing Science, University of Glasgow, UK\\ \texttt{john.sylvester@glasgow.ac.uk}}

\newcommand{\WOne}{{\hyperlink{w1}{\ensuremath{\mathcal{W}_1}}}\xspace}
\newcommand{\WTwo}{{\hyperlink{w2}{\ensuremath{\mathcal{W}_2}}}\xspace}
\newcommand{\WThree}{{\hyperlink{w3}{\ensuremath{\mathcal{W}_3}}}\xspace}

\newcommand{\POne}{{\hyperlink{p1}{\ensuremath{\mathcal{P}_1}}}\xspace}
\newcommand{\PTwo}{{\hyperlink{p2}{\ensuremath{\mathcal{P}_2}}}\xspace}
\newcommand{\PThree}{{\hyperlink{p3}{\ensuremath{\mathcal{P}_3}}}\xspace}
\newcommand{\PFour}{{\hyperlink{p4}{\ensuremath{\mathcal{P}_4}}}\xspace}

\begin{document}
 
\maketitle

\begin{abstract}
We consider the sequential allocation of $m$ balls (jobs) into $n$ bins (servers) by allowing each ball to choose from some bins sampled uniformly at random. The goal is to maintain a small gap between the maximum load and the average load.

In this paper, we present a general framework that allows us to analyze  various allocation processes that slightly prefer allocating into underloaded, as opposed to overloaded bins.
Our analysis covers several natural instances of processes, including: 
\begin{itemize}
 \item The \textsc{Caching process} (a.k.a. memory protocol) as studied by Mitzenmacher, Prabhakar and Shah (2002): 
At each round we only take one bin sample, but we also have access to a cache in which the most recently used bin is stored. We place the ball into the least loaded of the two.
 \item The \textsc{Packing process}: At each round we only take one bin sample. If the load is below some threshold (e.g., the average load), then we place as many balls until the threshold is reached; otherwise, we place only one ball.
 \item The \textsc{Twinning process}: At each round, we only take one bin sample. If the load is below some threshold, then we place two balls; otherwise, we place only one ball.
 \item The \textsc{Thinning process} as recently studied by Feldheim and Gurel-Gurevich (2021): At each round, we first take one bin sample. If its load is below some threshold, we place one ball; otherwise, we place one ball into a {\em second} bin sample.
\end{itemize}

As we demonstrate, our general framework implies for all these processes a gap of $\Oh(\log n)$ between the maximum load and average load, even when an arbitrary number of balls $m \geq n$ are allocated (heavily loaded case). Our analysis is inspired by a previous work of Peres, Talwar and Wieder (2010) for the $(1+\beta)$-process, however here we rely on the interplay between different potential functions to prove stabilization.

	\medskip
	
	
\end{abstract}

\clearpage

\clearpage
\tableofcontents
~
\clearpage

\section{Introduction}

We consider the sequential allocation of $m$ balls (jobs or data items) to $n$ bins (servers or memory cells), by allowing each ball to choose from a set of randomly chosen bins. The goal is to allocate balls efficiently, while also keeping the load distribution balanced. The balls-into-bins framework has found numerous applications in 
hashing, load balancing, routing, but has also shown to be nicely related to more theoretical topics such as randomized rounding or pseudorandomness (we refer to the surveys~\cite{W17} and~\cite{MR01} for more details).

A classical algorithm is the \DChoice process introduced by Azar et al.~\cite{ABKU99} and Karp et al.~\cite{KLM96}, where for each ball to be allocated, we sample $d \geq 1$ bins uniformly and then place the ball in the least loaded of the $d$ sampled bins. It is well-known that for the \OneChoice process ($d=1$), the gap between the maximum and average load is $\Theta\bigr( \sqrt{ m/n \cdot \log n } \bigr)$. In particular, this gap grows significantly as $m/n \rightarrow \infty$, which is called the {\em heavily loaded case}. For $d=2$,~\cite{ABKU99} proved that the gap is only $\log_2 \log n+\Oh(1)$ for $m=n$. This result was generalized by Berenbrink et al.~\cite{BCSV06} who proved that the same guarantee also holds for $m \geq n$, in other words, even as $m/n \rightarrow \infty$, the difference between the maximum and average load remains a slowly growing function in $n$ that is independent of $m$.
This dramatic improvement of \TwoChoice over \OneChoice has been widely known as the ``power of two choices''.

It is natural to investigate allocation processes that are less powerful than \TwoChoice in their ability to sample two bins, sample uniformly or distinguish between the load of the sampled bins. Such processes make fewer assumptions than \TwoChoice and can thus be regarded as more sampling-efficient and robust.

One example of such an allocation process is the $(1+\beta)$-process introduced by Peres et al.~\cite{PTW15}, where each ball is allocated using \OneChoice with probability $1-\beta$ and otherwise is allocated using \TwoChoice. The authors proved that for any $\beta \in (0,1]$, the gap is only $\Oh(\frac{\log n}{\beta})$. Hence, only a ``small'' fraction of \TwoChoice rounds are enough to inherit the property of \TwoChoice that the gap is independent of $m$. A similar class of adaptive sampling schemes (where, depending on the loads of the samples so far, the ball may or may not sample another bin) was analyzed by Czumaj and Stemann~\cite{CS01}, but their results hold only for $m=n$.

Another important example is {\em $2$-\Thinning} (which we later simply refer to as $\Thinning$), which was studied in~\cite{IK04,FG18}. In this process, each ball first samples a bin uniformly and 
if its load is below some threshold, the ball is placed into that sample. Otherwise, the ball is placed into a {\em second} bin sample, without inspecting its load. In~\cite{FG18}, the authors proved that for $m=n$, there is a fixed threshold which achieves a gap of $\Oh(\sqrt{ \frac{\log n}{\log \log n}})$. This is a significant improvement over \OneChoice, but also the total number of samples is $(1+o(1)) \cdot m$, which is an improvement over \TwoChoice. Similar threshold processes have been studied in queuing~\cite{ELZ86}, \cite[Section 5]{M96} and discrepancy theory~\cite{DFG19}.
For values of $m$ sufficiently larger than $n$, \cite{FGL21} and \cite{LS21} prove some lower and upper bounds for a more general class of \emph{adaptive} thinning protocols (here, adaptive means that the choice of the threshold \emph{may} depend on the load configuration). Related to this line of research, \cite{LS21} also analyze a so-called \Quantile-process, which is a version of \Thinning where the ball is placed into a second sample only if the first bin has a load which is at least the median load.

Finally, we mention the \Caching process (also known as  memory-protocol) analyzed by Mitzenmacher et al.~\cite{MPS02}, which is essentially a version of the \TwoChoice process with cache. At each round, we take a uniform sample but we also have access to a cache. Then the ball is placed in the least loaded of the sampled and the bin in the cache, and after that the cache is updated if needed. It was shown in~\cite{MPS02} that for $m=n$, the process achieves a better gap than \TwoChoice.

From a more technical perspective, apart from analyzing a large class of natural allocation processes, an important question is to understand how sensitive the gap is to changes in the allocation distribution. To this end,~\cite{PTW15} formulated general conditions on the distribution vector, which, when satisfied in all rounds, imply a small gap bound. These were then later applied not only to the $(1+\beta$)-process, but also to analyze a ``graphical'' allocation model where a pair of bins is sampled by picking an edge of a graph uniformly at random. Other works which study perturbations on the allocation distribution are: allocations on hypergraphs~\cite{G08}, balls-into-bins with correlated choices~\cite{W07}; or balls-into-bins with hash functions~\cite{CRSW11}.

\textbf{Our Results.} We introduce a general framework that allows us to deduce a small gap independent of $m$ for many of the processes above, but also opens the avenue for the study of novel allocation processes such as \Twinning or \Packing. To ensure that an allocation process produces a balanced load distribution, we could either bias the allocation towards underloaded bins by skewing the probability by which a bin is chosen for an allocation (probability bias), or alternatively, we could add more balls to a bin if it is underloaded (weight bias).

Note that a small bias in the probability distribution can be achieved in various ways, for example, by taking a {\em second} bin sample: $(i)$ in each round (\TwoChoice), or $(ii)$ in each round with some probability $\beta$ ($(1+\beta)$-process), or $(iii)$ in each round dependent on the load of the first sample (\Thinning). Regarding the weight bias, as pointed out in~\cite{MPS02}, allocating consecutive balls to the same bin is very natural in applications like sticky routing~\cite{GKK88} or load balancing, where tasks (=balls) may often arrive in bursts. We note that the filling bias can be also seen as a ``popularity-bias''; once a lightly loaded bin has been found, it will be a preferred location for the next balls until it becomes overloaded.

The first type of allocation processes matching the above definition are called \textbf{Filling Processes}, where it suffices to sample one uniform bin for each allocation, but then we are able to ``fill'' the bin with balls until it becomes overloaded: 
\begin{enumerate}\itemsep-4pt
    \item[] \textbf{Condition \POne}: At each round, the probability of allocating to any fixed bin is majorized by the uniform distribution (i.e., \OneChoice).
   \item[] \textbf{Condition \WOne}: At each round, if an underloaded bin is chosen, allocate just enough balls to it so that the bin becomes overloaded. Otherwise, place a single ball.
\end{enumerate}
As our first main result, we prove that if \POne \emph{and} a (more technical, but relaxed) version of \WOne both hold, then a gap bound of $\Oh(\log n)$ follows. While it is easy to see that \Packing meets the two conditions, some care is needed to apply the framework to \Caching (a.k.a. memory-protocol) due to its use of the cache, which creates correlations between the allocation of any two consecutive balls. Hence \Packing and \Caching are at least as sampling-efficient as \OneChoice while achieving a gap independent of $m$, which matches the gap bound of $(1+\beta)$-process which requires strictly more samples per ball.

The second type of allocation processes matching the above description are called \textbf{Non-Filling Processes}. In these processes, we can only allocate a constant number of balls to a bin, regardless of its load. Neglecting some technical details (see \cref{sec:framework_thinning}), the two conditions are as follows: 
\begin{enumerate}\itemsep-4pt
    \item[] \textbf{Condition \PThree}: At each round, the probability of allocating to any fixed underloaded bin is larger than $\frac{1}{n}$, while the probability for any fixed overloaded bin is smaller than $\frac{1}{n}$.
    \item[] \textbf{Condition \WThree}: At each round, if an underloaded bin is chosen for allocation, place more (but a constant number) of balls than if an overloaded bin is chosen.
\end{enumerate}
As our second main result, we prove a gap bound of $\Oh(\log n)$ whenever a process satisfies \emph{at least one} of these two conditions. It turns out that \PThree is satisfied by \MeanThinning (i.e., the $2$-Thinning process with threshold being the average load), while \WThree is satisfied by \Twinning; thus for both of these processes a gap bound of $\Oh(\log n)$ follows immediately. The result for \Twinning gives an example of a process with gap $\Oh(\log n)$, which $(i)$ samples one bin in each round and $(ii)$ allocates at most a constant number of balls in any round. The result for \MeanThinning extends via an inductive\arxive{I don't think we should call this majorization here; otherwise it is confusing with the next sentence} argument to any \Thinning process with a threshold between $0$ and $\Oh(\log n)$ above the mean load.
Finally, using the idea of majorizing allocation probabilities~\cite{PTW15}, our framework also applies to $(1+\beta)$-processes for constant $\beta \in (0,1]$.

To the best of our knowledge, most of the related work on balls-into-bins focuses on a small number (usually one or two) allocation processes, and analyzes their gap (with~\cite{PTW15,V99} being exceptions). Here we develop a framework that captures a variety of existing and new processes. This generalization comes at a price, e.g., the $\Oh(\log n)$ bound for \Caching is probably not tight and a more specialized analysis may be needed for an improved bound. Nonetheless, for many processes our upper bounds are tight up to constant factors (or $\Oh(\log \log n)$), see \cref{tab:overview}.

Our analysis, especially of \Twinning, bears some resemblance to the weighted balls-into-bins setting. There, the weights of the balls are usually drawn from some distribution {\em before} the bin is sampled (see~\cite{TW07,PTW15, BFHM08} for some results). The major difference is that in our model weights depend on the chosen bin, whereas in previous works this is not the case, but instead there is a time-invariant probability bias towards lightly loaded bins (e.g., as in $(1+\beta)$-process and \TwoChoice).

\begin{table}[h]	\resizebox{\textwidth}{!}{
\renewcommand{\arraystretch}{1.5}
		\centering
\begin{tabular}{|c|c|c|c|c|}
\hline 
\multirow{2}{*}{Process} & \multicolumn{2}{c|}{Lightly Loaded Case $m=\Oh(n)$} & \multicolumn{2}{c|}{Heavily Loaded Case $m = \omega(n)$} \\ \cline{2-5}
& Lower Bound & Upper Bound & Lower Bound & Upper Bound
\\ \hline 
$(1+\beta)$, fixed $\beta \in (0,1)$  & \multicolumn{2}{c|}{\phantom{[14]}$\qquad \quad\;\,\frac{\log n}{\log \log n}\qquad\quad\;\,$\cite{PTW15} } & \cellcolor{Gray} $\log n$ & \cellcolor{Gray}  $ \log n$ \\ 	\hhline{|-|-|-|-|-|}
\Caching & \multicolumn{2}{c|}{\phantom{[12]}$\qquad \quad\log \log n\quad \qquad $ \cite{MPS02}} & -- & \cellcolor{Greenish} $\log n$ \\ \hhline{|-|-|-|-|-|}
\Packing & \multicolumn{2}{c|}{\cellcolor{Greenish}$ \frac{\log n}{\log \log n}$} &\cellcolor{Greenish} $\frac{\log n}{\log \log n}$ & \cellcolor{Greenish}$ \log n$ \\
\hhline{|-|-|-|-|-|}
\Twinning & \multicolumn{2}{c|}{ \cellcolor{Greenish}$ \frac{\log n}{\log \log n}$} & \multicolumn{2}{c|}{ \cellcolor{Greenish} $ \log n$} \\ \hhline{|-|-|-|-|-|}
\MeanThinning & \multicolumn{2}{c|}{\cellcolor{Greenish} $ \frac{\log n}{\log \log n}$} & \multicolumn{2}{c|}{\cellcolor{Greenish} $ \log n$} \\ \hhline{|-|-|-|-|-|} 
$\Thinning\left(\Theta(\sqrt{ \frac{\log n}{\log \log n}})\right)$ 
& \multicolumn{2}{c|}{\phantom{[7]}$\qquad\quad\sqrt{ \frac{\log n}{\log \log n}}\qquad\quad $ \cite{FL20}}   &   \phantom{\;\; [11]}$\frac{\log n}{\log \log n}$  \;\; \cite{LS21} & \cellcolor{Greenish} $\log n$ 
\\ \hhline{|-|-|-|-|-|}
\textsc{Adaptive-}\Thinning & \multicolumn{2}{c|}{ \phantom{[7]}$ \qquad\quad\sqrt{\frac{\log n}{\log \log n}}\qquad\quad$ \cite{FL20}} & \phantom{\;\; [11]} $\frac{\log n}{\log \log n}$ \;\; \cite{LS21} & $ \frac{\log n}{\log \log n}$ \cite{FGL21} \\ 
\hline
\end{tabular}}
 
\caption{Overview of the Gap achieved (with probability at least $1-n^{-1}$), by different allocation processes considered in this work. All stated bounds hold asymptotically; upper bounds hold for all values of $m$, while lower bounds may only hold for certain values of $m$. Cells shaded in \colorbox{Greenish}{Green} are new results, cells shaded \colorbox{Gray}{Gray} are known results we re-prove. The upper bounds for \Packing and \Twinning when $m=\Oh(n)$ follow immediately from \OneChoice.
}\label{tab:overview}

\end{table}

\textbf{Proof Overview.}
To analyze filling processes, we consider an exponential potential function $\Phi^{t}$ with parameter $\alpha > 0$ (a variant of the one used in~\cite{PTW15, S77}). This potential function considers only bins that are overloaded by at least two balls; thus it is blind to the load configuration across underloaded bins. Our first lemma (\cref{lem:filling}) upper bounds $\Ex{ \Phi^{t+1} \, \mid \, \Phi^{t}}$, and essentially proves that this expectation is maximized if the process uses a uniform distribution for picking a bin. We then use this upper bound to deduce that: $(i)$ in many rounds, the potential drops by a factor of at least $1-\Theta(\frac{\alpha}{n})$, and $(ii)$ in other rounds, the potential increases by a factor of at most $1+\Theta(\frac{\alpha^2}{n})$ (see \cref{lem:filling_good_quantile}). Taking $\alpha$ sufficiently small, we conclude that $\Ex{\Phi^{t}}=\Oh(n)$, and by using Markov's inequality the desired gap bound is implied.

The analysis of non-filling processes is considerably more involved. The main technical challenge is that for large $m$, these processes may undergo long phases in which either most of the bins are overloaded or most of the bins are underloaded, and consequently most balls are allocated using \OneChoice (see \cref{fig:experiments} for some experiments). Note that this is not the case for other processes such as the $(1+\beta)$-process (for constant $\beta >0$) let alone \TwoChoice.

We overcome this challenge by establishing a sufficient ``self-stabilisation'' of the quantile of the mean load. This is done by first proving that the quadratic potential drops as long as the absolute potential is $\Omega(n)$ (\cref{lem:quadratic_absolute_relation_for_w_plus_w_minus}). Hence, the absolute potential has to be $\Oh(n)$ at some point. Then we prove that once this happens, the quantile of the mean load will be stable, i.e., bounded away from $0$ and $1$, for sufficiently many rounds (\cref{lem:good_quantile}). The final ingredient is to prove that in rounds were the mean quantile is stable, the exponential potential drops by a sufficiently large factor in expectation, while in other rounds it increases by a smaller factor. This step is similar to the analysis of filling processes, but more complicated; one reason being that the exponential potential function here includes both overloaded and underloaded bins.

\textbf{Organization.} 
In \cref{sec:processes} we describe the aforementioned allocation processes  more formally (in addition, an illustration of the four processes mentioned in the abstract can be found in \cref{fig:processes} on page~\pageref{fig:processes}). In \cref{sec:notation} we define some basic mathematical notation needed for our analysis. In \cref{sec:framework} we present our general framework for filling and non-filling processes. In this part, we also verify that the specific processes mentioned in the abstract fall into (or can be reduced to) the framework. In \cref{sec:analysis_filling}, we prove an $\Oh(\log n)$ bound on the gap for filling processes. In \cref{sec:analysis_thinning} we sketch the proof of the same gap bound for non-filling processes. 
In \cref{sec:non_filling_analysis_tools} we introduce a crucial concept for the analysis of non-filling processes, which are tools that imply that the mean quantile stabilizes. In \cref{sec:non_filling_potential_functions}, we derive several inequalities involving different kinds of potential functions. In \cref{sec:non_filling_analysis}, we complete the non-filling analysis.
The derivation of the lower bound is given in \cref{sec:lower}. We also present experiments (\cref{fig:experiments} on page~\pageref{fig:experiments})
illustrating the interplay between the different potential functions used in the analysis. In \cref{sec:experiemental_results}, we empirically compare the gaps of the different processes.

\section{Balanced Allocation Processes}\label{sec:processes}

We consider sequential allocation processes that allocate $m$ balls into $n$ bins. By $x^{t}$ we denote the load vector at round $t$ (the state after $t$ rounds have been completed). Further, $W^{t}:=|x^{t}|_1$ denotes the total number of balls allocated.
We begin with the classical \DChoice process~\cite{ABKU99,KLM96}.  

\begin{framed}
\vspace{-.45em} \noindent
\underline{\DChoice Process:}\\
\textsf{Iteration:} For each $t \geq 0$, sample $d$ uniform bins $i_1,\ldots, i_d$ independently. Place a ball in a bin $i_{min}$ satisfying $x_{i_{\min}}^t = \min_{1\leq j\leq d} x_{i_{j}}^t $, breaking ties arbitrarily.\vspace{-.5em}
\end{framed}
\noindent Mixing \OneChoice with \TwoChoice rounds at a rate $\beta$, one obtains the $(1+\beta)$ process~\cite{PTW15}:
\begin{framed}
\vspace{-.45em} \noindent
\underline{$(1+\beta)$ Process:}\\
\textsf{Parameter:} A probability $\beta \in (0,1]$.\\
\textsf{Iteration:} For each $t \geq 0$, with probability $\beta$ allocate one ball via the \TwoChoice process, otherwise allocate one ball via the \OneChoice process. \vspace{-.5em}
\end{framed}

The memory protocol was introduced by Mitzenmacher et al.~\cite{MPS02} and works under the assumption that the address $b$ of a single bin can be stored or `cached'. The process is essentially a \TwoChoice process where the second sample is replaced by the bin in the cache. 

\begin{framed}
\vspace{-.45em} \noindent
\underline{\Caching Process (a.k.a. Memory Protocol):}\\
\textsf{Iteration:} For each $t \geq 0$, sample a uniform bin $i$, and update its load (or of cached bin $b$):
\begin{equation*}
    \begin{cases}
       x_{i}^{t+1} = x_{i}^{t} + 1  & \mbox{if $x_{i}^{t} <  x_{b}^{t} $} \qquad \mbox{(also update cache $b=i$)}, \\
        x_{i}^{t+1} = x_{i}^{t} + 1  & \mbox{if $x_{i}^{t} =  x_{b}^{t} $}, \\
      x_{b}^{t+1} = x_{b}^{t} + 1 & \mbox{if $x_{i}^{t} >  x_{b}^{t}$}.
   \end{cases}
 \end{equation*}\vspace{-1.5em}
\end{framed}

\begin{figure}[h!]
\begin{center}
\includegraphics[scale=0.45]{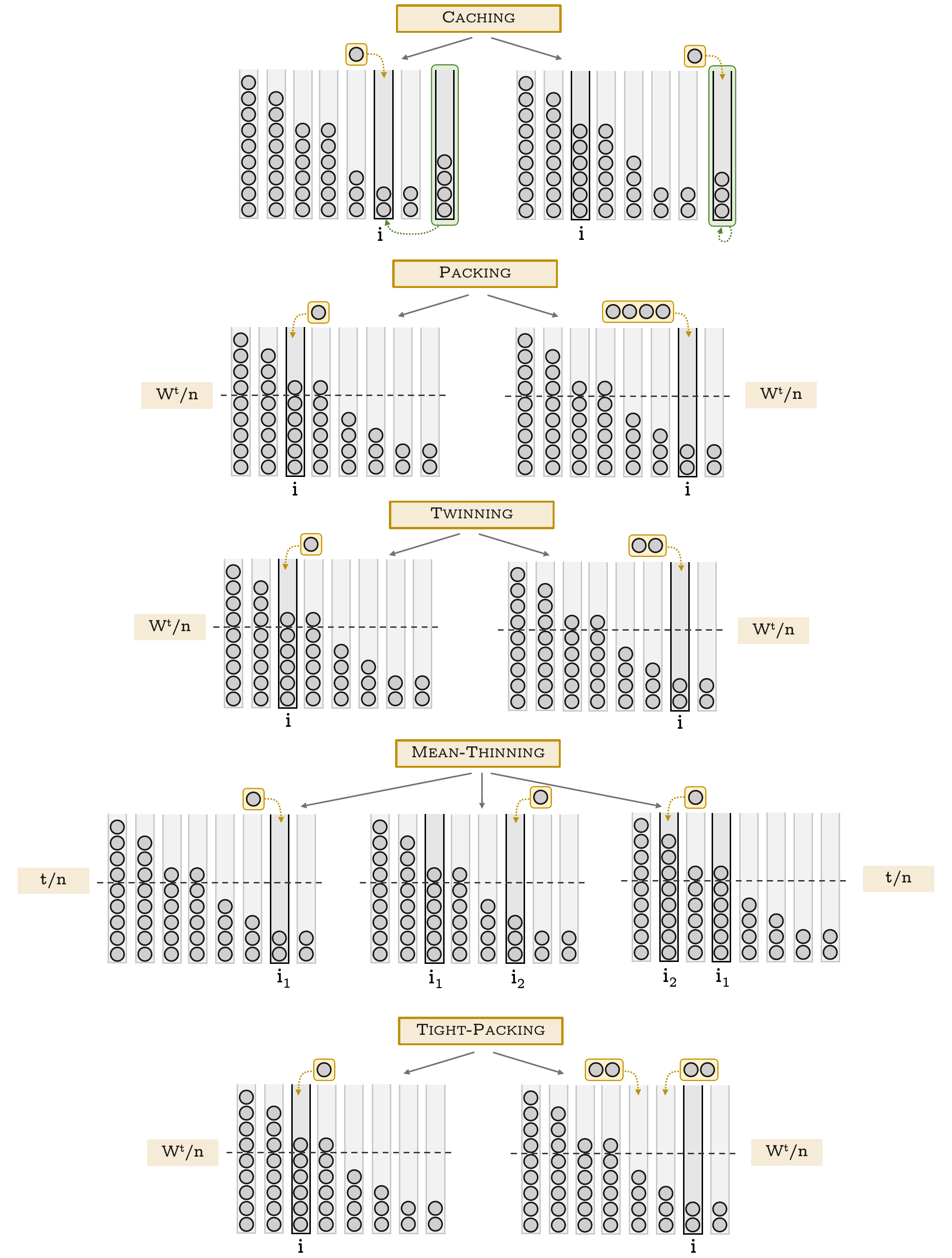}
\end{center}
\caption{Illustration of the different possibilities of allocating one (or more) balls in each of the four processes: \Caching, \Packing, \Twinning, \MeanThinning, \OverPacking.}\label{fig:processes}
\end{figure}

As mentioned earlier, \cite{FGL21,LS21} studied a more powerful version of \Thinning, where the threshold function may be adapted to the current load vector, i.e., $f=f(n,x^{t})$. However, in this work we focus on non-adaptive threshold functions $f(n)$.

We further introduce two new processes: $(a)$ the \Packing process, which places ``greedily'' as many balls as possible into an underloaded bin (i.e., a bin with load below average) and $(b)$ the \Twinning process, which is similar to \Packing, but in a sense it is less extreme as it only  places two balls into an underloaded bin. This process is possibly the most similar to \OneChoice among our processes, however, we still prove that it has a $\Oh(\log n)$ gap. %
\begin{multicols}{2}
\begin{framed}
\vspace{-.45em} \noindent
\underline{\Packing~Process:}\\
\textsf{Iteration:} For each $t \geq 0$, sample a uniform bin $i$, and update its load:
    \begin{equation*} x_{i}^{t+1} =
    \begin{cases}
       \lceil \frac{W^{t}}{n}  \rceil + 1  & \mbox{if $x_{i}^{t} <  \frac{W^{t}}{n} $}, \\
      x_{i}^{t} + 1 & \mbox{if $x_{i}^{t} \geq  \frac{W^{t}}{n} $}.
   \end{cases}
 \end{equation*}\vspace{-1.5em}
\end{framed} 
\begin{framed}
\vspace{-.45em} \noindent
\underline{\Twinning Process:}\\
\textsf{Iteration:} For each $t \geq 0$, sample a uniform bin $i$, and update its load:
    \begin{equation*} x_{i}^{t+1} =
    \begin{cases}
       x_{i}^{t} + 2 & \mbox{if $x_{i}^{t} <  \frac{W^{t}}{n}$}, \\
      x_{i}^{t} + 1 & \mbox{if $x_{i}^{t} \geq  \frac{W^{t}}{n}$}.
   \end{cases}
 \end{equation*}\vspace{-1.5em}
\end{framed}

\end{multicols}

The following process has been studied by several authors~\cite{IK04,FG18,FL20}, where the threshold is usually set in advance. A special case is $\Thinning(0)$ which we also call \MeanThinning. 
 
\begin{framed}
\vspace{-.45em} \noindent
\underline{$\Thinning(f(n))$ Process:}\\
\textsf{Parameter:} A threshold function $f(n) \geq 0$\\
\textsf{Iteration:} For $t \geq 0$, sample two uniform bins $i_1$ and $i_2$ independently, and update:  
    \begin{equation*}
    \begin{cases}
     x_{i_1}^{t+1} = x_{i_1}^{t} + 1 & \mbox{if $x_{i_1}^{t} <  \frac{t}{n}+f(n)$}, \\
      x_{i_2}^{t+1} = x_{i_2}^{t} + 1 & \mbox{if $x_{i_1}^{t} \geq  \frac{t}{n}+f(n)$}.
   \end{cases}
 \end{equation*}\vspace{-1.5em}
\end{framed}

The above processes arise naturally when one is trying to achieve a small gap from few bin queries or random samples. In contrast, the following process is rather contrived, as
whenever an underloaded bin is chosen, balls can be placed into (possibly different) underloaded bins that have the highest load. However, it is interesting that even this process achieves a small gap, as we will prove later.
\begin{framed}
	\vspace{-.45em} \noindent
	\underline{\OverPacking~Process:}\\
	\textsf{Iteration:} For each $t \geq 0$, sample a uniform bin $i$, and update:
	\begin{equation*}
 \left\{ \!\!\!	\begin{tabular}{ l l   }
	\text{Place $\lceil - x_i^t + \frac{W^t}{n}\rceil + 1$ balls one by one into the bins with highest loads} &  \multirow{2}{*}{\mbox{if $x_{i}^{t} <  \frac{W^{t}}{n} $,}}   \\ 
\text{such that $x_j^{t+1} < \frac{W^t}{n}$, except for one bin $j$ with $\frac{W^t}{n}\leq  x_{j}^{t+1} < \frac{W^t}{n} + 1$,} &    \vspace{.15cm } \\  
	$x_{i}^{t+1} = x_{i}^{t} + 1$ &  \mbox{if $x_{i}^{t} \geq  \frac{W^{t}}{n} $}.    
	\end{tabular}\right. 
	\end{equation*}\vspace{-1.5em}
\end{framed}

An equivalent description of \OverPacking with decreasingly sorted bin loads $x_1^{t} \geq x_2^{t} \geq \cdots \geq x_{n}^t$ goes as follows. After picking the random bin $i$, we update the load of the maximally loaded underloaded bin $j \in [n]$ at round $t$ to $x_j^{t+1} = x_j^{t} + \lceil - x_j^{t} + \frac{W^t}{n} \rceil$, so its new load satisfies $x_j^{t+1} \in [ \frac{W^t}{n},\frac{W^t}{n}+1)$. 
Then the remaining $\lceil - x_i^t + \frac{W^t}{n} \rceil + 1 - \lceil - x_j^{t} + \frac{W^t}{n} \rceil \geq 0$ balls (if there are any), are allocated to bins $j+1,j+2,\ldots,j+\ell$ for some integer $\ell \geq 0$, such that all bins $i \in [j+1,j+\ell-1]$ have the same load value in $[\frac{W^t}{n}-1,\frac{W^t}{n})$.

For illustrations of \OverPacking and other processes defined in this section see \cref{fig:processes}.

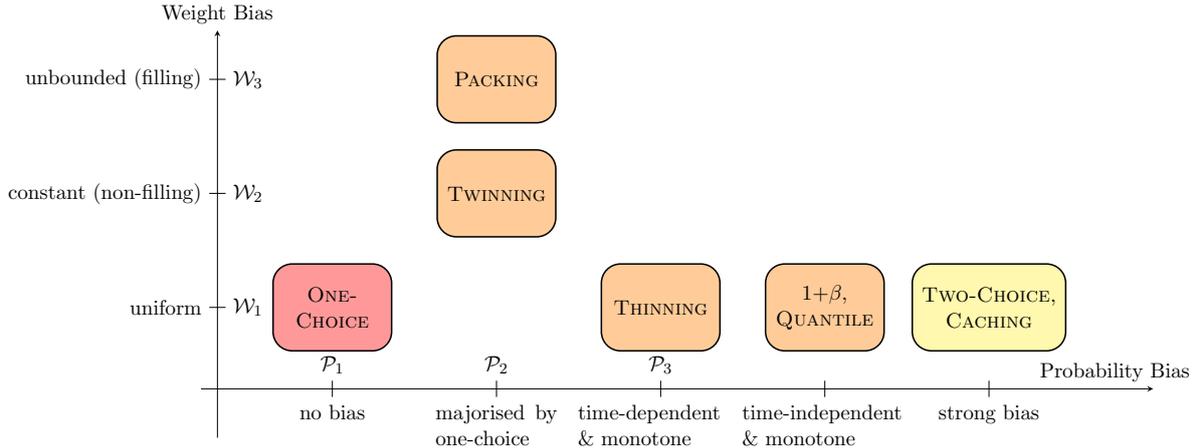
\begin{figure}
\scalebox{0.72}{
\begin{tikzpicture}[scale=3,knoten/.style={rectangle,rounded corners=10pt,opacity=1,minimum size=1.6cm,draw=black,thick,fill=orange!40}]

\draw[-stealth] (-0.8,0) to node[pos=0.96,above]{Probability Bias} (5,0);
\draw[-stealth] (-0.7,-0.3) to node[pos=1.0,above]{Weight Bias}(-0.7,2.2);

\draw (-0.75,0.5) to node[pos=0,left]{uniform} node[pos=1,right]{\WOne} (-0.65,0.5);

\draw (-0.75,1.2) to node[pos=0,left]{constant (non-filling)} node[pos=1,right]{\WTwo} (-0.65,1.2);

\draw (-0.75,1.9) to node[pos=0,left]{unbounded (filling)} node[pos=1,right]{\WThree} (-0.65,1.9);

\draw (0,0.05) to node[pos=1,below]{no bias} node[pos=0,above]{\POne} (0,-0.05);

\draw (1,0.05) to node[pos=1,below]
{\begin{minipage}{0.14\textwidth}
majorised by one-choice
\end{minipage}} 
node[pos=0,above]{\PTwo} (1,-0.05);

\draw (2,0.05) to node[pos=1,below]{
\begin{minipage}{0.19\textwidth}
time-dependent \\ \& monotone
\end{minipage}
} node[pos=0,above]{\PThree} (2,-0.05);

\draw (3,0.05) to node[pos=1,below]{
\begin{minipage}{0.19\textwidth}
 time-independent  \\ \& monotone 
\end{minipage}
} node[pos=0,above]{} (3,-0.05);

\draw (4,0.05) to node[pos=1,below]{strong bias} node[pos=0,above]{} (4,-0.05);

\node[knoten,fill=red!40] (0) at (0,0.5) [label=north east:{}]{
\begin{minipage}{0.12\textwidth}
\begin{center}
\OneChoice
\end{center}
\end{minipage}
};

\node[knoten] (5) at (2,0.5) [label=north east:{}]{
\begin{minipage}{0.12\textwidth}
\begin{center}
\Thinning
\end{center}
\end{minipage}
};
\node[knoten] (6) at (3,0.5) [label=north east:{}]{
\begin{minipage}{0.12\textwidth}
\begin{center}
1+$\beta$, \\
\Quantile{}
\end{center}
\end{minipage}
};
\node[knoten,fill=yellow!40] (7) at (4,0.5) [label=north east:{}]{
\begin{minipage}{0.16\textwidth}
\begin{center}
\TwoChoice,
\Caching
\end{center}
\end{minipage}
};
\node[knoten] (7) at (1,1.2) [label=north east:{}]{
\begin{minipage}{0.12\textwidth}
\begin{center}
\Twinning
\end{center}
\end{minipage}
};
\node[knoten] (8) at (1,1.9) [label=north east:{}]{
\begin{minipage}{0.12\textwidth}
\begin{center}
\Packing
\end{center}
\end{minipage}
};

\end{tikzpicture}
}

\caption{A schematic overview of different allocation processes. For a formal definition of \Quantile, we refer the reader to \cite{LS21}.}

\end{figure}

\section{Notation}\label{sec:notation}

As mentioned before, $x^t$ is the load vector of the $n$ bins at round $t=0,1,2,\ldots$ (the state after $t$ rounds have been completed). So in particular, $x^0:=(0,\ldots,0)$. In many parts of the analysis, we will assume an arbitrary labeling of the $n$ bins so that their loads at round $t$ are ordered decreasingly, i.e., 
\[
 x_1^t \geq x_2^t \geq \cdots \geq x_n^t.
\]
Unlike many of the standard balls-into-bins processes, some of our processes may allocate more than one ball in a single round. To this end we define $W^t := \sum_{i\in[n]} x_i^t$ as the total number of balls allocated that are allocated by round $t$ (for \Thinning and \Caching we allocate one ball per round, so $W^t=t$; for \Packing and \Twinning, $W^t \geq t$). We will also use $w^{t}:=W^{t+1}-W^{t}$ to denote the number of balls allocated in round $t$. 

We define the gap as $\Gap(t):=\max_{i \in [n]} x_i^{t} - \frac{W^t}{n}$, which is the difference between the maximum load and average load\footnote{It is common in the literature to focus on this difference, rather than the difference between maximum and minimum load; however, our results for non-filling extend to this stronger notion.}
at round $t$. When referring to the gap of a specific process $P$, we write $\Gap_{P}(t)$ but may simply write $\Gap(t)$ if the process under consideration is clear from the context.
Finally, we define the normalized load of a bin $i$ as:
\[
 y_i^{t} :=  x_i^t - \frac{W^t}{n}.
\]
Further, let $B_{+}^{t}:=\left\{ i \in [n] \colon y_i^{t} \geq 0 \right\}$ be the set of overloaded bins, and $B_{-}^{t} := [1,n] \setminus B_{+}^{t}$ be the set of underloaded bins. Let $\delta^t := |B_+^t|/n \in [\frac{1}{n},1]$ which is the quantile corresponding to the average load. Following \cite{PTW15}, let $p_{1}^t,p_{2}^t,\ldots,p_{n}^t$ be the distribution vector of an allocation process, where for every $i \in [n]$, $p_i^{t}$ is the probability that the process
picks bin the $i$-th most heavily loaded bin round $t$ for the allocation. We denote by $i \in_{p^t} [n]$ a sample of $[n]$ according to this vector $p^t$. Note that in most processes after taking the sample $i$, all balls will be allocated to $i$; however, when we define filling processes in the next section, this requirement will be relaxed.  A special case of a distribution vector is the uniform distribution vector of \OneChoice, which is $p_i^t = p_i = \frac{1}{n}$ for all $i$ and $t$. For two distribution vectors $p$ and $q$ (or analogously, for two sorted load vectors), we say that $p$ majorizes $q$ if for all $1 \leq k \leq n$,
$
 \sum_{i=1}^{k} p_i \geq \sum_{i=1}^k q_i.
$
Let $P_+^t := \sum_{i \in B_+^t} p_i^t$ and $P_-^t := \sum_{i \in B_-^t} p_i^t$.

In the following, $\mathfrak{F}^{t}$ is the filtration corresponding to the first $t$ allocations of the process (so in particular, $\mathfrak{F}^{t}$ reveals $x^{t}$). For random variables $Y,Z$ we say that $Y$ is stochastically smaller than $Z$ (or equivalently, $Y$ is stochastically dominated by $Z$)  if $ \Pro{ Y \geq x } \leq \Pro{ Z \geq x }$ for all real $x$. 

Throughout the paper, we often make use of statements and inequalities which hold only for sufficiently large $n$. For simplicity, we do not state this explicitly.

\NewConstant{quad_delta_drop}{c}
\NewConstant{quad_const_add}{c}
\NewConstant{good_quantile_mult}{c}
\NewConstant{bad_quantile_mult}{c}
\NewConstant{v_mult_factor}{c}
\NewConstant{poly_n_gap}{c}

\NewConstantWithName{stab_time}{\ensuremath{c_s}}
\NewConstantWithName{rec_time}{\ensuremath{c_r}}
\NewConstantWithName{small_delta}{C}

\NewConstantWithName{lambda_bound}{c}

\section{General Framework and Upper Bounds} \label{sec:framework}

We present two general frameworks that cover allocation processes with a certain bias towards underloaded bins. The first framework is called ``filling''-processes, as such processes will be allowed to fill underloaded bins to their maximum capacity (which is defined to be the average load plus one). The second framework covers so-called ``non-filling''-processes, where the number of balls that can be allocated in one single round is bounded by a constant.

\subsection{Framework for Filling-Processes (Caching and Packing)}%

We make the following assumptions about the processes.
For each round $t \geq 0$,
we sample bin $i=i^t$ and then place a certain number of balls to $i$ (or other bins).

More formally:
\begin{enumerate}\itemsep0pt
 \item[] \textbf{Condition \hypertarget{p1}{$\mathcal{P}_1$}}: For each round $t=0,1,\ldots$, pick an arbitrary labeling of the $n$ bins such that $y_1^{t} \geq y_2^{t} \geq \cdots \geq y_n^{t}$. Then let $i = i^t \in_{p^t} [n]$, where the distribution vector $p^t$ is majorized by the uniform distribution (\OneChoice). 
 \item[] \textbf{Condition \hypertarget{w1}{$\mathcal{W}_1$}}: For each round $t =0,1,\ldots$, with $i$ being the bin chosen above:
 \begin{itemize}
  \item If $y_i^{t} < 0$ then allocate exactly $\lceil - y_i^{t} \rceil + 1 \geq 2$ balls such that there can be at most two bins $k_1, k_2 \in [n]$ where:
  \begin{enumerate}
      \item $k_1 \in [n]$ receives $\lceil -y_{k_1}^t \rceil + 1$ balls,
      \item $k_2 \in [n]$ receives $\lceil -y_{k_2}^t \rceil$ balls,
      \item all other bins $j \in [n]$ receive at most $\lceil -y_{j}^t \rceil - 1$ balls.
  \end{enumerate}
   \item If $y_i^{t} \geq 0$ then allocate a single ball to bin $i$.
 \end{itemize} 
\end{enumerate}

\begin{figure}[ht]
    \centering
    \includegraphics[scale=0.63]{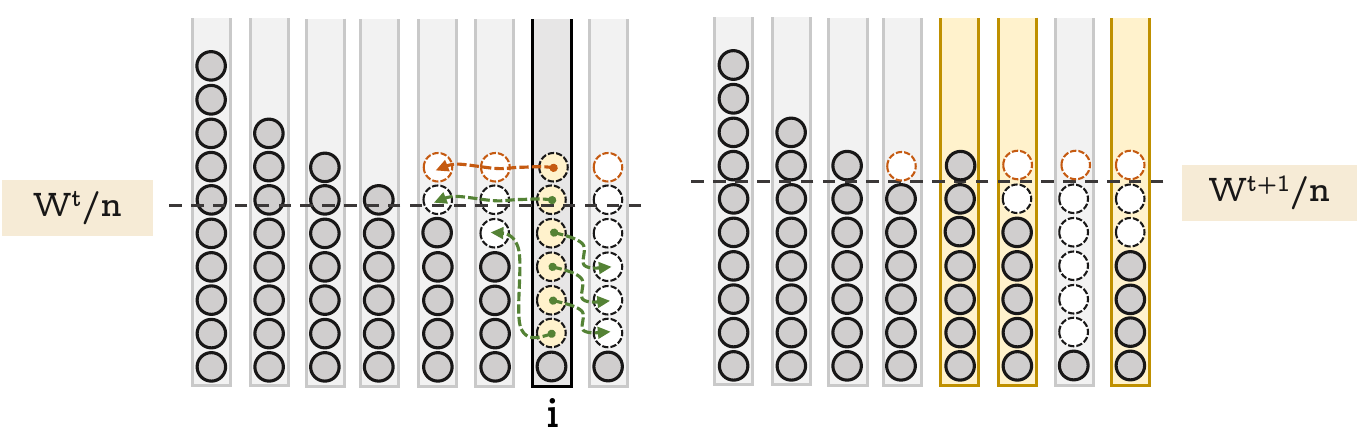}
    \caption{Illustration of an allocation for a filling process. After the underloaded bin $i$ is picked, $\lceil -y_i^{t} \rceil + 1 = 6$ balls are allocated; $5$ of these balls are placed in the underloaded regions of other bins and one ball is placed in the overloaded region (shown in orange).}
    \label{fig:general_filling_process}
\end{figure}

See \cref{fig:general_filling_process} for an example of a process satisfying \WOne. A natural example of a process satisfying \POne and \WOne is \Packing, which only allocates to the (initially) chosen bin $i$.
\begin{lem}
The \Packing process satisfies conditions \POne and \WOne.
\end{lem}
\begin{proof}
Note that the \Packing process picks a uniform bin $i$ at each round $t$, thus it satisfies \POne. Furthermore, if $y_i^{t} < 0$, it allocates exactly $\lceil -y_i^{t} \rceil + 1$ balls to bin $i$; otherwise, it allocates one ball to $i$, and thus \WOne is also satisfied.
\end{proof}

\begin{lem}
The \OverPacking process satisfies conditions \POne and \WOne.
\end{lem}
\begin{proof}
Note that the \OverPacking process picks a uniform bin $i$ at each round $t$, thus it satisfies \POne.
Then the allocation satisfies the following properties.

First, in case of $y_i^{t} < 0$, then:
$(i)$ we allocate
exactly $\lceil -y_i^{t} \rceil + 1$ balls, $(ii)$ one bin receives $\lceil - y_j^{t} \rceil$ balls, 
and $(iii)$ every other bin $j$ receives a number of balls between $[0, \lceil -y_j^{t} \rceil - 1]$.

Secondly, in case of $y_i^{t} \geq 0$ we place one ball to $i$.
\end{proof}

We emphasize that in condition \POne, the distribution $p^{t}$ may not only be time-dependent, but may also depend on the load distribution at round $t$, that is, the filtration $\mathfrak{F}^{t}$. Also, the allocation used in \WOne may depend on the filtration $\mathfrak{F}^{t}$. Thus the framework also applies in the presence of an adaptive adversary, which directs all the $\lceil -y_i^t \rceil + 1$ balls to be allocated in round $t$ towards the ``most loaded'' underloaded bins. At the other end of the spectrum, there are natural processes which have a propensity to place these balls into ``less loaded'' underloaded bins. One specific, more complex example is the \Caching process, where due to the update of the cache after each single ball, the allocation is more skewed towards ``less loaded'' underloaded bins. However, as we shall discuss shortly, \Caching only satisfies \POne and \WOne after a suitable compression of rounds.  %

Next we present our result for the gap of such filling processes:

\def\fillingkey{Consider any allocation process which  satisfies the conditions \POne and \WOne at each round. Then, there is a constant $C>0$ such that for any $m \geq 1$, 
 \[
 \Pro{ \Gap(m) \leq C \log n} \geq 1-n^{-2}.
 \]}

\begin{thm}\label{thm:filling_key}
    \fillingkey
\end{thm}

This result is shown in \cref{sec:analysis_filling} by analyzing an exponential potential function $\Phi^m$ over the overloaded bins and eventually establishing that for any $m \geq 1$,
$
 \Ex{ \Phi^m } = \Oh(n).
$
Then a simple application of Markov's inequality yields the desired result for the gap.

As mentioned above \Caching does not satisfy conditions \POne and \WOne directly. The issue is that \Caching, as defined in \cref{sec:notation}, allocates only {\em one} ball at each round whereas, due to condition \WOne, a filling process may place several balls into underloaded bins in a single round. 
To overcome this issue, we define a so-called ``\textbf{unfolding}'' of a filling process. The filling process proceeds in rounds $0,1,\ldots$, whereas the unfolded version is a coupled process, which is encoded as a sequence of \emph{atomic allocations} (each allocation places exactly one ball into a bin).
Using this unfolding along with the gap bound for a filling process satisfying \WOne and \POne, we can then bound the number of atomic allocations with a large gap.
  
   First, we define the process of unfolding formally:\label{folding} Let $x^{t}$ be the load vector of a filling process $P$, and $\widehat{x}^{t}$ be the load vector of the unfolding of $P$ called $U=U(P)$. Initialize $x^{1}:=\widehat{x}^{1}$ as the all-zero vector, corresponding to the empty load configuration. For every round $t \geq 0$, where $P$ has allocated $W^{t-1}$ balls before, we create atomic allocations $S(t):=\{W^{t-1}+1,\ldots,W^{t}\}$, where $W^{t} = W^{t-1}+\min\left\{\lceil -y_i^t \rceil + 1, \; 1\right\}$ for process $U$, such that for each $s \in S(t)$ a single ball is allocated, and  
we have $x^{t+1}:=\widehat{x}^{W^{t+1}}$, i.e., the load distribution of $P$ after round $t+1$ corresponds to load distribution of $U$ after atomic allocation $W^{t+1} \geq t+1$. Note that the unfolding of a process is not (completely) unique, as the allocations in $[W^{t-1}+1,W^{t}]$ of $U(P)$ can be permuted arbitrarily.

\def\memorycompressed{
There is a process $P$ satisfying conditions \POne and \WOne, such that for a suitable unfolding $U=U(P)$, the process $U$ is an instance of \Caching.}

\begin{lem} \label{lem:memory_compressed}
\memorycompressed
\end{lem}

Now, applying \cref{thm:filling_key} to the unfolded version of a filling process, and exploiting that in a balanced load configuration not too many atomic allocations can be created through unfolding, we obtain the following result:

\def\sloweddown{
Fix any constant $c> 0$. Then for any filling process $P$ satisfying \WOne and \POne, there is a constant $C=C(c)>0$ such that for any number of atomic allocations $m \geq 1$, with probability at least $1-n^{-2}$, any unfolding $U=U(P)$ satisfies
 \[
 \bigl| \left\{ t \in [m] \colon \Gap_{U}(t) \leq C \cdot \log n \right\} \bigr| \geq n^{-c} \cdot m \log m.
 \]}

\begin{lem}\label{lem:slowed-down} 
\sloweddown
\end{lem}

Hence for any $m$, all but a polynomially small fraction of the first $m$ atomic allocations in any unfolding of a filling process (e.g.\ \Caching) have a logarithmic gap. This behavior matches the one of the original filling process; the only limitation is that we cannot prove a small gap that holds for an arbitrarily large {\em fixed} atomic allocation $m$. However, if $m$ is polynomial in $n$, that is, $m \leq  n^{c}$ for some constant $c>0$, then \cref{lem:slowed-down} implies directly that with high probability, the gap at atomic allocation $m$ (and at all atomic allocations before) is logarithmic. 

Returning to the special case of \Caching, we will now perform a more tailored analysis to obtain the gap bound analogous to \cref{thm:filling_key} so that it holds for an arbitrarily large number of atomic allocations of the unfolding:

\def\cachingfixedtimestep{
For the \Caching process (which allocates exactly one ball per round) there is a constant $C > 0$ such that for any number of rounds $m \geq 1$,
\[
\Pro{\Gap(m) 
\leq C \cdot \log n} \geq 1 - n^{-2}.
\]}

\begin{lem}\label{lem:cachingfixedtimestep}
\cachingfixedtimestep
\end{lem}

\subsection{Framework for Non-Filling-Processes (Thinning and Twinning)}
\label{sec:framework_thinning}
Recall that for any $i \in [n]$,  $p_{i}^t$ is the probability of picking the $i$-th most heavily loaded bin for an allocation at round $t$. %
As before, $p_i^t$ may be chosen {\em adaptively}, i.e., dependent on the load vector $x^{t}$ (equivalently, based on the filtration $\mathfrak{F}^{t}$).

Our allocation process for non-filling work as follows. 
First, we define $p_{+}^t:=\max_{i \in B_{+}^t} p_i^t$ and $p_{-}^t:=\min_{i \in B_{-}^t} p_i^t$, as the largest (smallest) probability for allocating to a fixed overloaded (underloaded) bin at round $t$, respectively. As in filling processes, we then first sample a bin $i=i^t \in [n]$ and then place a certain number of balls to $i$.
More formally,
\begin{itemize}\itemsep0pt
\item[] \textbf{Condition \hypertarget{p2}{$\mathcal{P}_2$}}:
For each round $t =0,1,\ldots$, we consider an arbitrary labeling such that $y_1^t \geq y_2^t \geq \cdots \geq y_n^t$. Then let $i = i^t \in_{p^t} \in [n]$, where the distribution vector $p^t$ must satisfy $p_{+}^t \leq \frac{1}{n} \leq p_{-}^t$. 

\item[] \textbf{Condition \hypertarget{w2}{$\mathcal{W}_2$}}: For each round $t=0,1,\ldots$, when $i$ is chosen for allocation,
\begin{itemize}
    \item If  $y_i^{t} < 0$, then place $w_{-}$ balls into bin $i$,
    \item If  $y_i^{t} \geq 0$, then place $w_{+}$ balls into bin $i$,
    \end{itemize}
   where $1 \leq w_{+} \leq w_{-}$ are constant integers independent of $t$ and $n$.
\end{itemize}

Both conditions are natural, but on their own they are not sufficient to establish a good bound on the gap, as the \OneChoice process satisfies both conditions with equality. Thus,
we will require processes to satisfy at least one of two stronger versions of \PTwo and \WTwo:

\NewConstant{p2k1}{k}
\NewConstant{p2k2}{k}
\begin{itemize}\itemsep0pt
    \item[]\textbf{Condition \hypertarget{p3}{$\mathcal{P}_3$}}:  This is as Condition \PTwo, but additionally, there are time-independent constants $\C{p2k1} \in (0,1],\C{p2k2} \in (0,1]$ such that for each round $t \geq 0$:
    \begin{align*}
    p_{+}^t &\leq  \frac{1-\C{p2k1}}{n}
    + \frac{\C{p2k1} \cdot |B_{+}^t|}{n^2} = \frac{1}{n} - \frac{\C{p2k1} \cdot (1-\delta^t)}{n}, \\
     p_{-}^t &\geq \frac{1}{n} + \frac{\C{p2k2} \cdot |B_+^t|}{n^2}
     = \frac{1}{n} + \frac{\C{p2k2} \cdot \delta^t}{n}.
    \end{align*}
    \item[] \textbf{Condition \hypertarget{w3}{$\mathcal{W}_3$}}: This is as Condition \WTwo, but additionally we have the strict inequality: $w_{+} < w_{-}$. Also,  we assume that for each $t \geq 0$, distribution vector $p_{i}^t$ is non-decreasing in $i$.
\end{itemize}
The reason we attach the non-decreasing property of $p^t$ to \WThree and not to \PThree is to make our main result slightly stronger.

The rationale behind condition \PThree is that we wish to slightly bias the distribution vector $p^t$ towards underloaded bins at each round $t$. However, it is natural to assume that this influence is limited by a process that samples, say, at most two bins uniformly and independently, and then allocates balls to the least loaded of the two. Concretely, if a process takes \emph{two} independent and uniform bin samples at each round, the probability of picking two overloaded bins equals
$
 \bigl( \frac{|B_+^t|}{n} \bigr)^2.
$
Hence by averaging, there must be a bin $i \in B_+^t$ such that
\[
 p_{+}^t \geq p_i^t \geq \left( \frac{|B_+^t|}{n} \right)^2 \cdot \frac{1}{|B_+^t|} = \frac{|B_+^t|}{n^2}.
\]
The relaxation of the first constraint in \PThree by taking a strict convex combination of $\frac{1}{n}$ and $\frac{|B_+^t|}{n^2}$
ensures some slack, for instance, it allows the framework to cover a ``noisy'' version of \MeanThinning, where at each round, we perform a \OneChoice allocation with some constant probability $<1$, and otherwise we perform an allocating following the \MeanThinning process (see \cref{lem:noisy_threshold} for details). 
Similarly, for any process which takes at most two uniform samples, by averaging, there must be a bin $j \in B_{-}^t$ such that
\begin{align*}
 p_{-}^t \leq p_{j}^t \leq \frac{1 - \frac{|B_+^t|^2}{n^2}}{|B_{-}^t|}
 = \frac{(n - |B_+^t|) \cdot (n + |B_+^t|)}{n^2|B_{-}^t|}
 = \frac{1}{n} +  \frac{|B_{+}^t|}{n^2}.
\end{align*}

Finally, we remark that \PThree resembles the framework of \cite[Equation~2]{PTW15}, where
$p^t_i=p_i$ is non-decreasing in $i$ and $p_{n/3} \leq \frac{1-4\epsilon}{n}$ and $p_{2n/3} \geq \frac{1+4\epsilon}{n}$ holds for some $0 < \epsilon < 1/4$ (not necessarily constant). In contrast to that, for constants $k_1,k_2>0$, the conditions in \PThree are relaxed as they only imply such a bias if $\delta^t$ is bounded away from $0$ and $1$, which may not hold in all rounds.

\begin{lem}
The \Twinning process satisfies \PTwo and \WThree. 
\end{lem}
\begin{proof}
The statement for the \Twinning process is obvious, since $w_- = 2 > 1 = w_+$ and the bin $i$ is picked uniformly at each round.
\end{proof}

\begin{lem}
The \MeanThinning process satisfies \PThree and \WTwo.
\end{lem}
\begin{proof}
 For the \MeanThinning process, we use the following equivalent two-stage description: Each ball is allocated to the first sample, if the bin is underloaded; otherwise, the ball is allocated to the second sample. Hence the probability of allocating to any overloaded bin $i \in B_+^t$ is
$
p_i^t = \frac{\delta^t}{n} = \frac{1}{n} - \frac{1 \cdot (1 - \delta^t)}{n},
$
so we can choose $\C{p2k1} := 1$. For any underloaded bin $i \in B_-^t$,
$
p_i^t = \frac{1+\delta^t}{n} = \frac{1}{n} + \frac{1 \cdot \delta^t}{n},
$
so we can choose $\C{p2k2} := 1$, and \PThree holds.
\end{proof}

By an inductive argument involving adding ``extra balls'', we can reduce $\Thinning(f(n))$ with a non-negative threshold function $f(n)$, possibly depending on $n$ but independent of $t$, to \MeanThinning. Note that \Thinning with a suitable positive function $f(n)$ is appealing, as results from~\cite{FG18} have shown a sublogarithmic gap in the lightly loaded case where $m=\Oh(n)$ (see also \cref{tab:overview}).

\def\fextension{
\label{lem:f(n)extension}
	Let $f(n)$ be any non-negative function.
	Let $\Gap_{0}$ and $\Gap_{f(n)}$ be the gaps of \MeanThinning and $\Thinning(f(n))$. Then $\Gap_{f(n)}$ is stochastically smaller than $\Gap_{0}+f(n)$.}

\begin{lem}%
\fextension
\end{lem}

Before proving the lemma, we need the following domination result:

\begin{lem}\label{lem:coupling}
Let process $P_{A}$ be a \Thinning process with an empty load distribution at round $0$, and using threshold $t/n+f(n)$ at any round $t \geq 0$, where $f(n)$ is non-negative. Further, let process $P_{B}$ be a \Thinning process with initial load distribution $x_1^0=x_2^0=\cdots=x_n^0=f(n)$, and using threshold $t/n+f(n)$ at any round $t \geq 0$. Then, there is a coupling so that at any round $t \geq 0$ and for any bin $i \in [n]$, $x_i^{t}(A) \leq x_i^{t}(B)$.
\end{lem}
\begin{proof}[Proof of \cref{lem:coupling}]
For any round $t \geq 0$, let $i_1=i_1^t$ and $i_2=i_2^t$ the two random bin samples, which are uniform and independent over $[n]$. We consider a coupling between $P_{A}$ and $P_{B}$, where these random bin samples are identical, and prove inductively that for any $t \geq 0$ and any $i \in [n]$,
\[
x_i^{t}(A) \leq x_i^{t}(B).
\]
The base case $t=0$ holds by definition. For the induction step, we make a case distinction:

\noindent\textbf{Case 1} [$x_{i_1}^t(A) < t/n+f(n)$]. In this instance $P_{A}$ allocates a ball to $i_1$. If $x_{i_1}^t(B) < t/n+f(n)$, then $P_{B}$ also allocates a ball to $i_2$; otherwise, we have $x_{i_1}^t(B) \geq t/n+f(n)$, and hence $x_{i_1}^t(B) > x_{i_1}^t(A)$, i.e., $x_{i_1}^t(B) \geq x_{i_1}^t(A)+1$. This implies
    \[
    x_{i_1}^{t+1}(B) = x_{i_1}^{t}(B) \geq  x_{i_1}^t(A) + 1 = x_{i_1}^{t+1}(A),
    \]
    and the inductive step follows from this and the induction hypothesis.
    \medskip

\noindent\textbf{Case 2} [$x_{i_1}^t(A) \geq t/n+f(n)$]. In this instance $P_{A}$ allocates a ball to $i_2$. By induction hypothesis, $x_{i_1}^t(A) \leq x_{i_1}^t(B)$, which implies $P_{B}$ also allocates a ball to $i_2$. Thus we have   \[
    x_{i_1}^{t+1}(A) = x_{i_1}^{t}(A)+1 \geq  x_{i_1}^t(B) + 1 = x_{i_1}^{t+1}(B),
    \] and the inductive step is complete. 
    
    Since in both cases all other bins remain unchanged the proof is complete. 
\end{proof}

We can now complete the proof of \cref{lem:f(n)extension}.

\begin{proof}[Proof of \cref{lem:f(n)extension}]
Instead of starting with an empty load distribution, we modify it by adding $f(n)$ extra balls to each bin. We refer to these as the \textit{red} balls. We now run the \Thinning process with threshold $(t/n+f(n))$ on these bins, assigning one new \textit{blue} ball each round. Due to the presence of the $f(n)$ red balls in each bin the process will consider any bin with more than $t/n$ blue balls as being overloaded.
Thus from the perspective of the blue balls this is a \MeanThinning process.

Using the coupling from \cref{lem:coupling}, adding $f(n) \cdot n$ red balls in total without changing the threshold, always results in a load vector that point-wise majorizes the unmodified \Thinning process with threshold $t/n+f(n)$. 
\end{proof}
Let us now define a noisy version of \MeanThinning, called $(1+\eta)$-\MeanThinning, where $\eta \in (0,1]$ is a parameter, analogous to $\beta$ in the $(1+\beta)$-process. At each round $t \geq 0$, with probability $\eta$, we allocate a ball using \MeanThinning; otherwise we allocate a ball using \OneChoice.

\begin{lem}\label{lem:noisy_threshold}
For any constant $\eta > 0$, the $(1+\eta)$-\MeanThinning process satisfies \PThree and \WTwo.
\end{lem}
\begin{proof}
Let $p$ be the distribution vector of $(1+\eta)$-\MeanThinning with parameter $\eta \in (0,1)$; that is, with probability $\eta$ it performs \MeanThinning, otherwise \OneChoice. 
Then for any $i \in B_{+}^t$,
\begin{align*}
 p_{i}^t &= (1 - \eta) \cdot \frac{1}{n} + \eta \cdot \frac{\delta^t}{n}
 = \frac{1}{n} - \frac{\eta \cdot (1-\delta^t)}{n}. 
\end{align*}
Similarly, for any $i \in B_{-}^t$,
\begin{align*}
 p_{i}^t &= (1 - \eta) \cdot \frac{1}{n} + \eta \cdot \frac{1+\delta^t}{n}
 = \frac{1}{n} + \frac{\eta \cdot \delta^t}{n}.
\end{align*}
Thus for $k_1=k_2=\eta \in (0, 1)$, the $(1+\eta)$-\MeanThinning process satisfies \PThree and \WTwo.
\end{proof}

While the $(1+\beta)$-process does not satisfy \PThree at each round, the process can be majorized by $(1+ \eta)$ with $\eta = \beta$.%
\begin{lem}
\label{lem:beta_domination}
For any constant $\beta \in (0,1]$, let $\Gap_{\beta}$ be the gap of the $(1+\beta)$-process. Then $(1+\eta)$ process with $\eta = \beta$ majorizes $(1+\beta)$ at each round. Hence, $\Gap_{\beta}$ is stochastically smaller than $\Gap_{\eta}$, the gap of the $(1+\eta)$ process.
\end{lem}

Before proving this lemma, we state one auxiliary result, which is implicit in~\cite{PTW15}. Theorem 3.1 proves the required majorization for time-independent distribution vectors, but as mentioned in the proof of Theorem 3.2~\cite{PTW15}, the same result generalizes to time-dependent distribution vectors.
\begin{lem}[{see~\cite[Section 3]{PTW15}}]\label{lem:majorisation}
Consider two allocation processes $Q$ and $P$ with $w_{+}=w_{-}=1$. The allocation process $Q$ uses at each round a fixed allocation distribution $q$. The allocation process $P$ uses a time-dependent allocation distribution $p^{t}$, which may depend on $\mathfrak{F}^{t}$ but majorizes $q$ at each round $t \geq 0$. Let $y^{t}(Q)$ and $y^{t}(P)$ be the two load vectors, sorted decreasingly. Then there is a coupling such that for all rounds $t \geq 0$, $y^{t}(Q)$ is majorized by $y^{t}(P)$.
\end{lem}

\begin{proof}[Proof of \cref{lem:beta_domination}]
We will show that for any round $t \geq 0$ and for any load configuration, the $(1+\eta)$ for $\eta = \beta$ distribution vector majorizes the distribution vector of $(1+\beta)$-process. So, by \cref{lem:majorisation}, the claim will follow. 

Recall that the $(1+\beta)$ distribution vector $q^t$ is given by,
\[
q_i^t=q_i=(1-\beta) \cdot \frac{1}{n} + \frac{\beta (2i-1)}{n^2}.
\]
The distribution vector $p^t$ for the $(1+\eta)$ process is non-decreasing and uniform over $B_{-}^t$ and $B_{+}^t$, so majorization follows immediately once we prove that
\[
 \sum_{i=1}^{|B_{+}^t|} q_i \leq \sum_{i=1}^{|B_{+}^t|} p_i^t.
\]
For the distribution vector $q$ we have,
\begin{align*}
 \sum_{i=1}^{\delta^t \cdot n} q_i = 
 \sum_{i=1}^{\delta^t \cdot n} (1-\beta) \cdot \frac{1}{n} + \sum_{i=1}^{\delta^t \cdot n} \frac{\beta (2i-1)}{n^2}
 = (1-\beta) \cdot \delta^t + \frac{\beta (\delta^t \cdot n)^2}{n^2}
 = \delta^t - \beta \cdot ( \delta^t - (\delta^t)^2 ).
\end{align*}
Similarly, for the distribution vector $p^t$ we have,
\[
 \sum_{i=1}^{\delta^t \cdot n} p_{i}^t
 = \sum_{i=1}^{\delta^t \cdot n} \frac{1}{n}
 - \sum_{i=1}^{\delta^t \cdot n} \frac{ \beta(1-\delta^t)}{n} = \delta^t - \beta \cdot (\delta^t-(\delta^t)^2).\qedhere
\]
\end{proof}

Since $\beta=1$ yields \TwoChoice, our framework also applies to the \TwoChoice process.

\medskip

The following basic, yet crucial result follows from the preconditions in \cref{thm:main_technical}:
\begin{lem}
\label{lem:additive_drift}
Consider any process satisfying the conditions \PThree and \WTwo, or, \PTwo and \WThree. Then for the constant $\C{quad_delta_drop} := \C{quad_delta_drop}(\C{p2k1}, \C{p2k2}) > 0$, it holds that for any $t \geq 0$,
\[
p_-^t \cdot w_- - p_+^t \cdot w_+ \geq \frac{\C{quad_delta_drop}}{n}.
\]
\end{lem}

\begin{proof}
First assume \PTwo and \WThree holds. In this case, \PTwo implies $p_-^t \geq \frac{1}{n} \geq p_+^t$, and thus
\[
 p_-^t \cdot w_- - p_+^t \cdot w_+ \geq \frac{1}{n} \cdot w_- - \frac{1}{n} \cdot w_+ \geq \frac{\C{quad_delta_drop}}{n},
\]
for $\C{quad_delta_drop} := w_- - w_+ > 0$, since the weights are constants satisfying $w_- > w_+$.

Next assume \PThree and \WTwo holds. In this case, \WTwo implies $w_{-} \geq w_{+}\geq 1$. Using \PTwo,
\begin{align*}
p_-^t \cdot w_- - p_+^t \cdot w_+ 
 & \geq p_-^t \cdot w_+ - p_+^t \cdot w_+ \geq (p_-^t - p_+^t) \cdot 1 \\ 
 & \geq \Big( \frac{1}{n} + \frac{\C{p2k2} \cdot \delta^t}{n} - \frac{1}{n} + \frac{\C{p2k1} \cdot (1 - \delta^t)}{n} \Big) 
 \geq \frac{\C{quad_delta_drop}}{n},
\end{align*}
for $\C{quad_delta_drop} := \min \{ \C{p2k1}, \C{p2k2} \}$.
\end{proof}

It is important to highlight that the quantity $p_-^t \cdot w_- - p_+^t \cdot w_+$ does not involve the \emph{number} of underloaded/underloaded bins.
Hence even if we know that $p_-^t \cdot w_- - p_+^t \cdot w_+ > 0$ holds, then this does not necessarily imply that the expected weight allocated to the set of underloaded bins is larger than the expected weight allocated to the set of overloaded bins. This conclusion would only be true if the quantile $\delta^t$ is sufficiently close to $1/2$.

One large portion of the analysis is devoted to derive a {\em weaker} quantile condition, that is, we prove that for sufficiently many rounds, the quantile $\delta^t$ is in the range $(\epsilon,1-\epsilon)$ for some (small) constant $\epsilon$. For this analysis, the inequality in \cref{lem:additive_drift} will be useful when we establish a connection between the absolute and quadratic potential function (\cref{lem:quadratic_absolute_relation_for_w_plus_w_minus}).

\def\maintechnical{
For any $\PThree \cap \WTwo$-process or $\PThree \cap \WThree$-process, there exists a constant $\kappa > 0$ such that for any $m \geq 1$,
\[
 \Pro{ \max_{i \in [n]} \left| x_i^m - \frac{W^t}{n} \right| \leq \kappa \log n } \geq 1-n^{-3};
\]
so in particular, $\Pro{ \Gap(m) \leq \kappa \log n } \geq 1-n^{-3}$. }

\begin{thm}\label{thm:main_technical} 
\maintechnical
\end{thm}

Thanks to the reductions in \cref{lem:f(n)extension} and \cref{lem:beta_domination}, we also deduce:
\begin{cor}
For the $(1+\beta)$-process for any constant $\beta \in (0,1]$ and for any \Thinning$(f(n))$-process where $f(n) \in [0,\Oh(\log n)]$ there exists a constant $\kappa > 0$ such that for any $m \geq 1$, \[\Pro{ \Gap(m) \leq \kappa \log n } \geq 1-n^{-3}.\]
\end{cor}

\section{Analysis of Filling Processes}\label{sec:analysis_filling}

In this section, we present our analysis for filling processes. We begin in \cref{sec:filling_potential} by defining the exponential potential function and analyze its behavior (depending on some other constraints). Then in \cref{sec:complete_filling_proof} we use a super-martingale argument to show that the exponential potential function decreases, which eventually yields the desired gap bound in \cref{thm:filling_key}.
In \cref{sec:unfolding_general} we present our results on (general) unfoldings of filling processes. We also prove that \Caching can be seen as the unfolded version of a filling process satisfying \WOne and \POne, and thus deduce $\Oh(\log n)$ gap bound as long as the number of rounds is polynomial in $n$.
Finally, in \cref{sec:ballbyball_caching} we strengthen this by showing an $\Oh(\log n)$ upper bound on the gap for the \Caching process which holds for any $m \geq 1$.

\subsection{Potential Function Analysis}\label{sec:filling_potential}
We consider a version of the exponential potential function $\Phi^t $ which only takes bins into account whose load is at least two above the average load. This is given by
\[ 
\Phi^t := \sum_{i: y_{i}^t\geq 2} \exp\left( \alpha \cdot  y_i^t \right) =  \sum_{i=1}^n \exp\left( \alpha \cdot  y_i^t \right)\cdot  \mathbf{1}_{\{ y_{i}^{t}\geq 2\}} ,
\]
where we recall that $y_i^t = x^t - \frac{W^t}{n}$ is the normalized load of  bin $i$ at round t and $\alpha > 0$ is a sufficiently small constant to be fixed later. Let $\Phi^t_i= \exp\left( \alpha \cdot  y_i^t \right)\cdot  \mathbf{1}_{\{ y_{i}^{t}\geq 2\}}$ and thus $\Phi^t=\sum_{i=1}^n\Phi^t_i$. 
We will also use the absolute value potential: 
    \[
    \Delta^{t}:=\sum_{i=1}^{n} \left| y_i^t \right|.
   \]

The next lemma provides a useful upper bound on the expected potential. It establishes that to bound $\ex{\Phi^{t + 1} \mid \mathfrak{F}^t}$ from above, we may assume the distribution vector $p^{t}$ is uniform.

\begin{lem}
\label{lem:filling}
Consider any allocation process satisfying \POne and \WOne. Then for any round $t \geq 0$, 
\[
 \ex{\Phi^{t + 1} \mid \mathfrak{F}^t} \leq \frac{1}{n}  \sum_{i=1}^n \Phi_i^{t} \cdot\left(  \sum_{j \colon y_j^{t} < 1}  
e^{\frac{-\alpha (\lceil -y_j^t \rceil + 1)}{n}}  +  e^{-\frac{\alpha}{n}} \cdot (|B_{\geq 1}^{t}|-1) +  e^{\alpha-\frac{\alpha}{n}} \right) + e^{3\alpha},
\]
where $B_{\geq 1}^t$ denotes the set of bins with load at least $1$. 
\end{lem}

\begin{proof}Recall that the filtration $\mathfrak{F}^t$ reveals the load vector $x^t$. Throughout this proof, we consider the labeling chosen by the process in round $t$ such that $x_1^{t} \geq x_2^{t} \geq \cdots \geq x_n^{t}$ and $p^t$ being majorized by \OneChoice (according to \POne). We emphasize that for this labeling, $x_i^{t+1}$ may not be non-increasing in $i \in [n]$.

To begin, using $\Phi_i^{t+1}=e^{\alpha y_{i}^{t+1}}\mathbf{1}_{\{ y_{i}^{t+1}\geq 2\}}$, we can split $\Phi^{t+1}$ over the $n$ bins as follows,
\[
	\Ex{ \Phi^{t+1} \, \mid \, \mathfrak{F}^t}   = \sum_{j=1}^n \Ex{
	\Phi_{j}^{t+1}
 \, \mid \, \mathfrak{F}^t }.
\]
Now consider the effect of picking bin $i$ for the allocation in round $t$ to the potential $\Phi^{t+1}$. Note that bin $i$ is chosen with probability equal to $p_i^t$. Observe that if bin $i$ satisfies $y_i^t<1$, then it receives at most $\lceil -y_i^t\rceil +1$ balls in round $t$, thus $y_{i}^{t+1}\geq 2$ if and only if $y_i^{t}\geq 1$. 
Using condition \WOne, we distinguish between the following three cases based on how allocating to $i$ changes $\Phi_{j}^{t}$ for $j \neq i$ and for $j=i$:
\medskip 

\noindent\textbf{Case 1.A} [$y_i^{t} < 1$, $j\neq i$, $y_j^t < 1$].
We will allocate $\lceil - y_i^{t} \rceil +1$ many balls to bins $k$ with $y_k^t <1$ (not necessarily to $i$) subject to \WOne. This increases the average load by $(\lceil - y_i^{t} \rceil + 1)/n$. Since $y_j^t < 1$, $\Phi_j^t =0 $. Further, by condition \WOne we can increase the load of $y_j^t$ by at most $\lceil - y_j^t \rceil + 1$, hence, $\Phi_j^{t+1}=0$. Therefore,
$
\Phi_{j}^{t+1} =  \Phi_j^{t} \cdot e^{\frac{-\alpha (\lceil -y_i^t \rceil + 1)}{n}}$ for $j \neq i$. 

\medskip

\noindent\textbf{Case 1.B} [$y_i^{t} < 1$, $j\neq i$, $y_j^t \geq 1$].
As in Case 1a, we will allocate $\lceil - y_i^{t} \rceil +1$ many balls to bins $k$ with $y_k^t <1$ (not necessarily to $i$) subject to \WOne, which increases the average load by $(\lceil - y_i^{t} \rceil + 1)/n$.
Additionally, bin $j$ will receive no balls (by condition \WOne), thus $\Phi_{j}^{t+1} =  \Phi_j^{t} \cdot e^{\frac{-\alpha (\lceil -y_i^t \rceil + 1)}{n}}$ for $j \neq i$. 

\medskip 

\noindent\textbf{Case 2} [$y_i^{t} \geq 1$, $j\neq i$]. We allocate one ball to $i$, which increases the average load by $1/n$, and thus $\Phi_{j}^{t+1} =  \Phi_j^{t} \cdot e^{-\frac{\alpha}{n}}$ for $j \neq i$, which again also holds for bins $j$ that do not contribute. 

\medskip 

\noindent\textbf{Case 3} [$j=i$].
Finally, we consider the effect on $\Phi_{i}^{t+1}$. Again if $y_i^{t} < 1$, then $\Phi_{i}^t=0$ and $\Phi_{i}^{t+1}=0$. Otherwise, we have $y_i^{t} \geq 1$ and we allocate one ball to $i$, and thus
\[
\Phi_{i}^{t+1} =  e^{ \alpha \cdot (y_{i}^{t}+
1) - \frac{\alpha}{n}}\mathbf{1}_{\{ y_{i}^{t+1}\geq 2\}} = e^{ \alpha \cdot y_i^{t}}\mathbf{1}_{\{ y_{i}^{t+1}\geq 2\}} \cdot e^{ \alpha - \frac{\alpha}{n}} \leq 
(e^{ \alpha \cdot y_i^{t}}\mathbf{1}_{\{ y_{i}^{t}\geq 2\}} +e^{2\alpha} ) \cdot e^{ \alpha - \frac{\alpha}{n}},
\]
where the $+e^{2\alpha}$ is added to account for the case where $1\leq y_i^{t} <2 $ and so $\Phi_{i}^t =0$ but $0<\Phi_{i}^{t+1} \leq e^{3\alpha -\alpha/n}$. Thus in this case, $\Phi_{i}^{t+1}\leq \Phi_{i}^{t}e^{ \alpha - \frac{\alpha}{n}} +  e^{3 \alpha - \frac{\alpha}{n}}\leq \Phi_{i}^{t}e^{ \alpha - \frac{\alpha}{n}} +  e^{3 \alpha }.$
\medskip 

By aggregating the three cases above, and observing that $\sum_{i=1}^n p_i^t\cdot e^{3 \alpha }=e^{3 \alpha }$ , we see that 
\begin{align}
 \sum_{j=1}^n \ex{\Phi_{j}^{t + 1} \mid \mathfrak{F}^t}    &=
 \sum_{i=1}^n p_i^t \sum_{j=1}^n \ex{\Phi_{j}^{t + 1} \mid \mathfrak{F}^t, \;\text{Bin i is selected at round t}}   \notag  \\  &\leq \sum_{i=1}^n  p_i^t \cdot \mathbf{1}_{\{y_i^t < 1\}}\sum_{j=1}^n \Phi_j^{t}  \cdot e^{\frac{-\alpha (\lceil -y_i^t \rceil + 1)}{n}}  + \sum_{i =1}^n  p_i^t \cdot \mathbf{1}_{\{y_i^t \geq 1\}}\sum_{j \neq i} \Phi_j^{t} \cdot e^{-\frac{\alpha}{n}} \notag \\ &\qquad + \sum_{i=1}^n p_i^t
 \cdot  \mathbf{1}_{\{y_i^t \geq 1\}}\cdot  \Phi_i^{t} \cdot e^{\alpha-\frac{\alpha}{n}}+ e^{3 \alpha } \label{eq:firstbd} 
\end{align}

 We will now rewrite \eqref{eq:firstbd} in order to establish that it is maximized if $p$ is the uniform distribution.
 Adding $\sum_{i=1}^n p_i^{t} \cdot \mathbf{1}_{\{y_i^t \geq 1\}}\cdot \Phi_i^{t} \cdot e^{-\frac{\alpha}{n}}$ to the middle sum (corresponding to Case 2) in \eqref{eq:firstbd} and subtracting it from the last sum (corresponding to Case 3) transforms \eqref{eq:firstbd} into
\begin{align*}
& \!\!\! \!\sum_{i=1}^n  p_i^t  \mathbf{1}_{\{y_i^t < 1\}}\sum_{j=1}^n \Phi_j^{t}   e^{\frac{-\alpha (\lceil -y_i^t \rceil + 1)}{n}}  + \sum_{i =1}^n  p_i^t  \mathbf{1}_{\{y_i^t \geq 1\}}\sum_{j=1}^n \Phi_j^{t}  e^{-\frac{\alpha}{n}}  + \sum_{i=1}^n p_i^t
   \Phi_i^{t}  \mathbf{1}_{\{y_i^t \geq 1\}}e^{ -\frac{\alpha}{n}}\left(e^{\alpha}-1\right) + e^{3 \alpha}\notag  \\
&\!\!\! \! = \sum_{i=1}^n  p_i^t \left( \! \! \Bigg(\underbrace{ \mathbf{1}_{\{y_i^t < 1\}}   e^{\frac{-\alpha (\lceil -y_i^t \rceil + 1)}{n}} + \mathbf{1}_{\{y_i^t \geq 1\}}\cdot e^{-\frac{\alpha}{n}} }_{g(i)} \Bigg) \sum_{j=1}^n  \Phi_j^{t} +   \underbrace{ \Phi_i^{t} \cdot \mathbf{1}_{\{y_i^t \geq 1\}}e^{ -\frac{\alpha}{n}}\left(e^{\alpha}-1\right) }_{f(i)} \right) + e^{3 \alpha } 
\end{align*}

Recall that $y_1^{t} \geq y_2^{t} \geq \cdots \geq y_n^{t}$, which implies $\Phi_1^{t} \geq \Phi_2^{t} \geq \cdots \geq \Phi_n^{t}\geq 0$. Thus $f(i)$ and $g(i)$ are non-negative and non-increasing in $i$ and $\sum_{j=1}^{n} \Phi_j^{t}\geq 0 $. Consequently, the function $h(i) = f(i)+ g(i)\cdot \sum_{j=1}^n \Phi_j^{t}$ is non-negative and non-increasing in $i$. Note that by condition \POne, for any $k \in [n]$ it holds that $\sum_{i=1}^{k} p_i^t \leq \frac{k}{n}$. Thus we can apply \cref{lem:quasilem} which implies $ \sum_{i=1}^n p_i^t \cdot h(i) \leq \sum_{i=1}^n \frac{1}{n} \cdot h(i)$. Applying this to the above, rearranging, and splitting $f(i)$ gives 
\begin{equation}\begin{aligned}\label{eq:nearlythere!}
\ex{\Phi^{t + 1} \mid \mathfrak{F}^t}&\leq \frac{1}{n} \sum_{i=1}^n   \mathbf{1}_{\{y_i^t < 1\}}   e^{\frac{-\alpha (\lceil -y_i^t \rceil + 1)}{n}} \sum_{j=1}^n  \Phi_j^{t}  +  \frac{1}{n} \sum_{i=1}^n  \mathbf{1}_{\{y_i^t \geq 1\}}\cdot e^{-\frac{\alpha}{n}}  \sum_{j=1}^n  \Phi_j^{t} \\&\qquad -\frac{1}{n}\sum_{i=1}^n  \mathbf{1}_{\{y_i^t \geq 1\}} \Phi_i^{t} \cdot e^{ -\frac{\alpha}{n}} +    \frac{1}{n}\sum_{i=1}^n   \mathbf{1}_{\{y_i^t \geq 1\}}\Phi_i^{t} \cdot e^{\alpha -\frac{\alpha}{n}}   + e^{3 \alpha }
\end{aligned}\end{equation} Now observe combining the second and third terms above gives
 \begin{align}\frac{1}{n} \sum_{i=1}^n  \mathbf{1}_{\{y_i^t \geq 1\}}\cdot e^{-\frac{\alpha}{n}}  \sum_{j=1}^n  \Phi_j^{t} -\frac{1}{n}\sum_{i=1}^n  \mathbf{1}_{\{y_i^t \geq 1\}} \Phi_i^{t} \cdot e^{ -\frac{\alpha}{n}} &= \frac{1}{n}\cdot e^{-\frac{\alpha}{n}}\cdot \sum_{i=1}^n  \sum_{j=1}^n  \mathbf{1}_{\{y_i^t \geq 1\}}\cdot  \Phi_j^{t} \cdot \mathbf{1}_{\{j\neq i\}} \notag \\
 &=   \frac{1}{n}\cdot e^{-\frac{\alpha}{n}}\cdot  \sum_{i=1}^n \Phi_i^{t} \sum_{j=1}^n  \mathbf{1}_{\{y_j^t \geq 1\}}\cdot \mathbf{1}_{\{j\neq i\}} \notag \\
 &=  \frac{1}{n}\cdot e^{-\frac{\alpha}{n}}\cdot  \sum_{i=1}^n \Phi_i^{t} \cdot (|B_{\geq 1}^{t}|-1), \label{eq:reshuffle} \end{align} where the last line follows since $ \Phi_i^{t}  = e^{\alpha y_i^t}\mathbf{1}_{\{y_{i}^t \geq 2\}}$. 
 
 Now, substituting \eqref{eq:reshuffle} into \eqref{eq:nearlythere!}, exchanging the first double summation and using the bound $\mathbf{1}_{\{y_i^t\geq 1 \}}\leq 1 $ on the last sum, and finally grouping terms gives the following  
\begin{equation*} 
\ex{\Phi^{t + 1} \mid \mathfrak{F}^t} \leq \frac{1}{n}  \sum_{i=1}^n  \Phi_i^{t} \left( \sum_{j=1}^n\mathbf{1}_{\{y_j^t < 1\}}  e^{\frac{-\alpha (\lceil -y_j^t \rceil + 1)}{n}} +  (|B_{\geq 1}^{t}|-1)e^{ -\frac{\alpha}{n}} +     e^{\alpha -\frac{\alpha}{n}} \right)   + e^{3 \alpha },\end{equation*}  
as claimed.  
\end{proof}

Let $\mathcal{G}^t$ be the event that at round $t\geq 0$ either there are at least $n/20$ underloaded bins or there is a an absolute value potential of at least $n/10$. In symbols this is given by
\begin{equation}\label{eq:eventGt}
    \mathcal{G}^t := \left\{B_{-}^{t} \geq n/20 \right\}\cup \left\{ \Delta^t \geq n/10\right\}. 
\end{equation}
The next lemma provides two estimates on the expected exponential potential $\Phi^{t+1}$ in terms of $\Phi^{t}$. The first estimate holds for any round and it establishes that the process does not perform worse than \OneChoice, meaning that the potential increases by a factor of at most $(1+\Oh(\alpha^2/n))$. The second estimate is stronger for rounds where we have a lot of underloaded bins or a large value of the absolute value potential. This stronger estimate states that the potential decreases by a factor of $(1-\Omega(\alpha/n))$ in those rounds. Note that as we show in~\cref{clm:counterexample}, the potential may increase in expectation for certain load configurations, so it seems hard to prove a decrease without additional constraints.

\begin{lem}
\label{lem:filling_good_quantile}
Consider any allocation process satisfying \POne and \WOne. There exists a constant $c_1>0$ such that for any $0 < \alpha < 1$ and any $t \geq 0$,
\[
\Ex{ \Phi^{t+1} \, \mid \,\mathfrak{F}^t} \leq  \left( 1 
+ \frac{c_1 \alpha^2}{n} \right) \cdot \Phi^{t} 
+ e^{3 \alpha}.
\]
Further, there exists a constant $c_2 > 0$ such that for any $0< \alpha<1/100$ and any $t \geq 0$,
\[
 \Ex{ \Phi^{t+1} \, \mid \, \mathfrak{F}^t, \mathcal{G}^t } \leq \left(1 - \frac{c_2 \alpha}{n} \right) \cdot \Phi^{t}  + e^{3\alpha}.
\]
\end{lem}

\begin{proof}Recall that, as before, we fix the labeling chosen by the process in round $t$ such that $x_1^{t} \geq x_2^{t} \geq \cdots \geq x_n^{t}$, thus $x_i^{t+1}$ may not be non-increasing in $i \in [n]$ and $p^t$ being majorized by \OneChoice (according to \POne). 

Let $A_i$ be the bracketed term in the expression for $\ex{\Phi^{t + 1} \mid \mathfrak{F}^t} $ in \cref{lem:filling_good_quantile}, given by 
\begin{equation}\label{eq:Ai}A_i=   \sum_{j \colon   y_j^{t} < 1}
e^{\frac{-\alpha (\lceil -y_j^t \rceil + 1)}{n}}  +  (|B_{\geq 1}^{t}| -1)\cdot e^{-\frac{\alpha}{n}} +   e^{\alpha-\frac{\alpha}{n}}.\end{equation}
Observe that $-\alpha (\lceil -y_j^t \rceil + 1) \leq -\alpha$ whenever $y_j^{t} < 1$ and thus 
\[A_i \leq    \sum_{j \colon   y_j^{t} < 1}
e^{-\frac{\alpha}{n}}  +  (|B_{\geq 1}^{t}| -1)\cdot e^{-\frac{\alpha}{n}} +   e^{\alpha-\frac{\alpha}{n}} = e^{-\frac{\alpha}{n}} \cdot \left[ n-1 + e^{\alpha} \right].
\]Applying the Taylor estimate $e^{z} \leq 1+z+z^2$, which holds for any $z \leq 1$, twice gives 
\[A_i\leq  \left(1 - \frac{\alpha}{n} +  \frac{\alpha^2}{n^2}  \right) \left( n + \alpha + \alpha^2 \right) = n \cdot \left(1 - \frac{\alpha}{n} +  \frac{\alpha^2}{n^2} \right) \left( 1 + \frac{\alpha}{n} +  \frac{\alpha^2}{n} \right) \leq n \cdot \left(1 + \frac{c_1\alpha^2}{n}  \right),
\]
for some constant $c_1 >0$. The first statement in the lemma now follows as \cref{lem:filling} gives 
  \[
\ex{\Phi^{t + 1} \mid \mathfrak{F}^t  } \leq \frac{1}{n}  \sum_{i=1}^n \Phi_i^{t} \cdot A_i + e^{3\alpha} \leq \frac{1}{n} \cdot n\left( 1 + \frac{c_1 \alpha^2}{n}  \right)\cdot  \sum_{i=1}^n \Phi_i^{t} + e^{3\alpha} \leq\left( 1 +\frac{c_1 \alpha^2}{n} \right)\Phi^{t} + e^{3\alpha}. \] 
We shall now show the second statement of the lemma. By splitting sums in \eqref{eq:Ai} we have 
 \begin{align} 
   A_i&=  \sum_{j \in B_{-}^t}
 e^{\frac{-\alpha (\lceil -y_j^t \rceil + 1)}{n}} + \sum_{j \colon 0\leq y_j^{t} < 1}
 e^{\frac{-\alpha (\lceil -y_j^t \rceil + 1)}{n}}  +  (|B_{\geq 1}^{t}| -1)\cdot e^{-\frac{\alpha}{n}} +   e^{\alpha-\frac{\alpha}{n}}\notag \\
 &=  \sum_{j \in B_{-}^t}
 e^{\frac{-\alpha (\lceil -y_j^t \rceil + 1)}{n}} +(|B_{+}^{t}|-1)\cdot e^{-\frac{\alpha}{n}}+   e^{\alpha-\frac{\alpha}{n}}\label{eq:A_ibdd}\\
 &\leq |B_{-}^t|\cdot e^{-\frac{2\alpha}{n}} +  (|B_{+}^{t}|-1)\cdot e^{-\frac{\alpha}{n}} +    e^{\alpha-\frac{\alpha}{n} }\label{eq:A_ibdd2} \end{align}
 We shall now first assume that $|B_-^t| \geq n/20$. Recall the bound $e^x\leq 1 + x + 0.6 \cdot x^2 $ which holds for any $x\leq 1/2$. If $\alpha<1/2$, we can apply this bound to \eqref{eq:A_ibdd2}, giving
 \begin{align} 
 A_i&\leq e^{-\frac{\alpha}{n}}\cdot \left(|B_{-}^t|\cdot \left(1 -\frac{\alpha}{n} + \frac{6\alpha^2}{10n^2}\right) +  (|B_{+}^{t}|-1)+    1 +\alpha  + \frac{6\alpha^2}{10} \right) \notag  \\
&\leq  \left(1 -\frac{\alpha}{n} + \frac{6\alpha^2}{10n^2}\right)\cdot n\left( 1 -\frac{\alpha}{20n} + \frac{6\alpha^2}{200n^2} +   \frac{\alpha}{n}  + \frac{6\alpha^2}{10n} \right)  \notag\\
&= n\cdot\left(1 -\frac{\alpha(1-12\alpha)}{20n} + \Oh\!\left(\frac{\alpha^2}{n^2} \right)  \right). \label{eq:1stcondredux} \end{align} %
We now assume that $|\Delta^{t}| \geq n/10$. 
Observe that by Schur-convexity (see~\cref{clm:schur}) and the assumption on $|\Delta^{t}|$ we have 
\begin{align*}
\sum_{j \in B_{-}^t}
e^{\frac{-\alpha (\lceil -y_j^t \rceil + 1)}{n}} &\leq e^{-\frac{\alpha}{n}} \sum_{j \in B_{-}^t}e^{\frac{\alpha  y_j^t }{n}} \\ &\leq  e^{-\frac{\alpha}{n}} \cdot \left( (|B_{-}^{t}|-1) \cdot e^{-\frac{\alpha}{n}\cdot 0 } + 1 \cdot e^{\frac{\alpha}{n} \cdot \sum_{j \in B_{-}^t} y_i^{t} } \right) \\ &\leq (|B_{-}^{t}|-1) \cdot e^{-\frac{\alpha}{n}} +  e^{-\alpha/20}, \end{align*}
where we used the fact that $\sum_{j \in B_{-}^t} y_i^{t}= -\frac{1}{2} \Delta^t$.
Applying this and the bound $e^x\leq 1 + x + 0.6 \cdot x^2 $, for $x\leq 1/2$, to \eqref{eq:A_ibdd} gives  
 \begin{align}
A_i & \leq  (|B_{-}^{t}|-1) \cdot e^{-\frac{\alpha}{n}} +  e^{-\alpha/20}   +   (|B_{+}^{t}| -1)\cdot e^{-\frac{\alpha}{n}}+   e^{\alpha}\notag \\
& = (n-2) \cdot e^{-\frac{\alpha}{n}} +  e^{-\alpha/20}  +   e^{\alpha}\notag \\
&\leq (n-2)\cdot \left(1 -\frac{\alpha}{n} + \frac{6\alpha^2}{10n^2}\right) + \left(1 -\frac{\alpha}{20} +  \frac{6\alpha^2}{4000} \right)+ \left(1 + \alpha  +  \frac{6\alpha^2}{10}\right)\notag \\
&= n\left( 1 - \frac{\alpha(200-2406\alpha)}{4000n}  + \Oh\!\left(\frac{\alpha^2 }{n^2}\right)\right). \label{eq:2ndcondredux}
 \end{align}
 Thus we see by \eqref{eq:1stcondredux} and \eqref{eq:2ndcondredux} that if $\mathcal{G}^t$ holds and we take $\alpha<1/100$ and $n$ sufficiently large, then there exists some constant $c_2>0$ such that $A_i \leq n(1-c_2\alpha/n)$. Thus  \cref{lem:filling} gives 
  \[
 \ex{\Phi^{t + 1} \mid \mathfrak{F}^t, \mathcal{G}^t } \leq \frac{1}{n}  \sum_{i=1}^n \Phi_i^{t} \cdot A_i + e^{3\alpha} \leq \frac{1}{n} \cdot n\left(1-\frac{c_2\alpha}{n}\right)\cdot  \sum_{i=1}^n \Phi_i^{t} + e^{3\alpha} \leq\left(1-\frac{c_2\alpha}{n}\right)\Phi^{t} + e^{3\alpha}, \] as claimed. \end{proof}

The next lemma shows that the event $\mathcal{G}^t$ given by \eqref{eq:eventGt} holds for sufficiently many rounds.

\begin{lem}\label{cor:enough_good_quantiles}
Consider any allocation process satisfying \POne and \WOne. For every integer $t_0 \geq 1$, there are at least $n/40$ rounds $t \in [t_0,t_0+n]$ with $(i)$ $\Delta^{t} \geq 1/10 \cdot n$ or $(ii)$ $|B_{-}^{t}| \geq n/20$.
\end{lem}

\begin{proof} 
 We claim that if $\Delta^{s} \leq 1/10 \cdot n$ for some round $s$, then for each round $t \in [s+n/5,s+9n/40]$ we have $|B_{-}^t| \geq n/20$ (deterministically). The lemma follows from that, since then we either have $\Delta^{t} \geq 1/10 \cdot n$ for all $t \in [t_0,t_0+n/40]$, or, thanks to the claim there is a $s \in [t_0,t_0+n/40]$ such that for all $t \in [s+n/5,s+9n/40]$ (this interval has length $n/40$) we have $|B_{-}^t| \geq n/20$. 

To establish the claim, assume we are at any round $s$ where $\Delta^{s} \leq 1/10 \cdot n$. Then at most $n/2$ bins $i$ satisfy $|y_i^{s}| \geq 1/5$, and in turn at least $n/2$ bins satisfy $|y_i^{s}| < 1/5$; let us call this latter set of bins $\mathcal{B}:=\left\{ i \in [n]: |y_i^s| < 1/5 \right\}.$
In the rounds $[s,s+9n/40]$, we can choose at most $9n/40$ bins in $\mathcal{B}$ that are overloaded (at the time when chosen), and then we place exactly one ball into them. Furthermore, in each round $t \in  [s,s+9n/40]$ we can turn at most two bins in $\mathcal{B}$ which is underloaded at round $t$ and make it overloaded. Hence it follows that at least $n/2-2\cdot 9n/40 = n/20$ of the bins in $\mathcal{B}$ are not chosen in the interval $[s,s+9n/40]$. Consequently, these bins must be all underloaded in the interval $[s+n/5,s+9n/40]$.
\end{proof}

\subsection{Completing the Proof of Theorem~\ref{thm:filling_key}}\label{sec:complete_filling_proof}

We now introduce a new potential function $\tilde{\Phi}^t$ which is the product of $\Phi^t$ with two additional terms (and an additive centering term). These multiplying terms have been chosen based on the one step increments in the two statements in~\cref{lem:filling_good_quantile} and \cref{cor:enough_good_quantiles}. The purpose of this is that using Lemmas \ref{lem:filling_good_quantile} and \ref{cor:enough_good_quantiles} we can show that $\tilde{\Phi}^t$ is a super-martingale. We then use the super-martingale  property to bound the exponential potential at an arbitrary step.

Here, we re-use the definition of the event $\mathcal{G}^t$ from \eqref{eq:eventGt}.
Now fix an arbitrary round $t_0 \geq 0$. Then, for any $s > t_0$, let $G_{t_0}^s$ be the number of rounds $r \in [t_0,s)$ satisfying $\mathcal{G}^r$, and let $B_{t_0}^s:=(s-t_0)-G_{t_0}^s$. Further, let the constants $c_1>0$ and $c_2>0$ be as in \cref{lem:filling_good_quantile}, let $c_3 := 2 e^{3 \alpha} \exp( c_2 \alpha)>0$, and then 
define a sequence by $\tilde{\Phi}^{t_0}:=\Phi^{t_0}$, and for any $s > t_0 $,
	\begin{equation}\label{eq:tildephi}
	\tilde{\Phi}^{s}:= \Phi^{s} \cdot \exp\left( +\frac{c_2 \alpha}{n} \cdot G_{t_0}^{s-1} \right) \cdot \exp\left( -\frac{c_1 \alpha^2}{n} \cdot B_{t_0}^{s-1} \right) - c_3  \cdot (s-t_0).
	\end{equation}
The next lemma proves that the sequence $\tilde{\Phi}^{s}$, $s \geq t_0$, forms a super-martingale:

\begin{lem}
\label{lem:filling_supermartingale}
	 Let $0<\alpha<1/100$ be an arbitrary but fixed constant, and $t_0 \geq 0$ be an arbitrary integer. Then, for any $s \in [t_0,t_0+n]$ we have	\[
	\ex{ \tilde{\Phi}^{s+1}  \mid \mathfrak{F}^s} \leq \tilde{\Phi}^{s}.
	\]
\end{lem}
\begin{proof}
	First, recalling the definition $\tilde{\Phi}^{s}$ from \eqref{eq:tildephi}, we rewrite  $\ex{\tilde{\Phi}^{s+1} \mid \mathfrak{F}^s}$ to give 
	\begin{align*}
\lefteqn{ \ex{\tilde{\Phi}^{s+1} \mid \mathfrak{F}^s} } \\
		&=  \ex{ \Phi^{s+1}  \mid \mathfrak{F}^s}  \cdot \exp\left( \frac{c_2 \alpha}{n} \cdot G_{t_0}^{s} \right) \cdot \exp\left( - \frac{c_1 \alpha^2}{n} \cdot B_{t_0}^{s} \right)  - c_3 \cdot (s+1-t_0) \\ 
		& = \ex{ \Phi^{s+1} \mid \mathfrak{F}^s} \cdot \exp\left(\frac{\alpha}{n} \cdot ( c_2 \cdot \mathbf{1}_{\mathcal{G}^s} -c_1 \alpha\cdot(1 - \mathbf{1}_{\mathcal{G}^s}) ) \right)  \cdot \exp\left( \frac{c_2 \alpha}{n} \cdot G_{t_0}^{s - 1} \right) \cdot \exp\left( - \frac{c_1 \alpha^2}{n} \cdot B_{t_0}^{s - 1} \right) \\
		&\qquad \mbox{} - c_3 - c_3 \cdot (s-t_0).
	\end{align*}

	We claim that it suffices to prove 
	\begin{align}\label{eq:sufficient}
		\ex{ \Phi^{s+1} \mid \mathfrak{F}^s} \cdot \exp\left(\frac{\alpha}{n} \cdot ( c_2 \cdot \mathbf{1}_{\mathcal{G}^s} -c_1 \alpha\cdot  (1 - \mathbf{1}_{\mathcal{G}^s}) \right) &\leq \Phi^{s} + c_3 \cdot \exp\left(- c_2 \alpha  \right).
	\end{align}

	Indeed, observe that $G_{t_0}^{s-1} \leq s-t_0 \leq n$, and so assuming \eqref{eq:sufficient} holds we have  \NOTE{T:}{This part is new}\NOTE{J}{I like the red bit, however I didn't like how we state the sufficent eq, show it is sufficent, then argue a simplier stronger eq does the job. I now start with the simpler eq and show that is sufficient.}
	\red{	\begin{align*}
		\lefteqn{  \ex{\tilde{\Phi}^{s+1} \mid \mathfrak{F}^s} } \\ &\leq  \left( \Phi^{s} + c_3 \cdot \exp\left(- c_2 \alpha  \right) \right) 
		\cdot \exp\left( \frac{c_2 \alpha}{n} \cdot G_{t_0}^{s - 1} \right) \cdot \exp\left( - \frac{c_1 \alpha^2}{n} \cdot B_{t_0}^{s - 1} \right) -c_3 - c_3 \cdot (s-t_0) \\
		&= \Phi^s \cdot \exp\left( \frac{c_2 \alpha}{n} \cdot G_{t_0}^{s - 1} \right) \cdot \exp\left( - \frac{c_1 \alpha^2}{n} \cdot B_{t_0}^{s - 1} \right) + c_3 \cdot \exp\left( - \frac{c_1 \alpha^2}{n} \cdot B_{t_0}^{s - 1} \right) -c_3 - c_3 \cdot (s-t_0) \\
		&\leq \tilde{\Phi}^s.
		\end{align*}	}
	To show \eqref{eq:sufficient}, we consider two cases based on whether $\mathcal{G}^s$ holds. 
	
	\smallskip 
	
\noindent\textbf{Case 1} [$\mathcal{G}^s$ holds]. By \cref{lem:filling_good_quantile} (first statement) we have
	\[
	\ex{\Phi^{s+1} \mid \mathfrak{F}^s, \mathcal{G}^s} \leq \Phi^s \cdot \left(1 - \frac{c_2 \alpha}{n} \right) + e^{3\alpha} \leq \Phi^s \cdot \exp\left(- \frac{c_2 \alpha}{n} \right) +  e^{3\alpha} .
	\]
	Hence, if $\mathcal{G}^s$ holds then the left hand side of \eqref{eq:sufficient} is equal to  
	\begin{align*}
		\ex{\Phi^{s+1} \mid \mathfrak{F}^s, \mathcal{G}^s } \cdot \exp\left( \frac{\alpha}{n} \cdot c_2 \right) &\leq \left( \Phi^s \cdot \exp\left(- \frac{c_2 \alpha}{n} \right) + e^{3 \alpha} \right) \cdot \exp\left( \frac{\alpha}{n} \cdot c_2 \right) \\
		&\leq \Phi^s + 2 e^{3 \alpha} \\
		&= \Phi^s + c_3 \cdot \exp\left( -c_2 \alpha \right), 
	\end{align*}
	where the last line holds by definition of $c_3 = 2 e^{3 \alpha} \exp( c_2 \alpha)$.
	
	\smallskip 
	
\noindent	\textbf{Case 2} [$\mathcal{G}^s$ does not hold]. \cref{lem:filling_good_quantile} (second statement) gives the  unconditional bound 
	\[
	\ex{\Phi^{s+1} \mid \mathfrak{F}^s, \neg \mathcal{G}^s} \leq \Phi^s \cdot \left(1 + \frac{c_1 \alpha^2}{n} \right) + e^{3\alpha} \leq \Phi^s \cdot \exp\left(\frac{c_1 \alpha^2}{n} \right) + e^{3\alpha}.
	\]
	 Thus, if $\mathcal{G}^s$ does not hold the left hand side of \eqref{eq:sufficient} is at most
	\[
	\ex{ \Phi^{s+1} \mid \mathfrak{F}^s, \neg \mathcal{G}^s } \cdot \exp\left( \frac{\alpha}{n} \cdot (-c_1 \alpha) \right) \leq \Phi^s + e^{3 \alpha} \leq \Phi^s + c_3 \cdot \exp\left( -c_2 \alpha \right),
	\] 
	which establishes \eqref{eq:sufficient} and the proof is complete.
\end{proof}

\comJ{\begin{lem}\label{lem:filling_supermartingale}
Fix any round $t_0$. For any $s \geq t_0$, let $G_{t_0}^s$ be the number of rounds $r \in [t_0,s)$ with $(i)$ $B_{-}^t \geq n/20$ or $(ii)$ $\Delta^t \geq n/10$, and let $B_{t_0}^s:=(s-t_0)-G_{t_0}^s$. Then for any $1 \leq s \leq n$ and for $c_3 := 2 e^{3 \alpha} \exp( c_2 \alpha)$, the sequence
\[
 \tilde{\Phi}^{s}:= \Phi^{s} \cdot \exp\left( +\frac{c_2 \alpha}{n} \cdot G_{t_0}^{s-1} \right) \cdot \exp\left( -\frac{c_1 \alpha^2}{n} \cdot B_{t_0}^{s-1} \right) - c_3  \cdot (s-t_0),
\]
where $s \in [t_0,t_0+n]$, forms a super-martingale.
\end{lem}
\begin{proof}
For any round $s\geq1$, let $\mathcal{G}^s$ be the (``good'') event that the precondition of the second statement of~\cref{lem:filling_good_quantile} holds, i.e., we have $B_{-}^{s} \geq n/20$ or $\Delta^s \geq n/10$.

Our goal is to show that $\ex{\tilde{\Phi}^{s+1} \mid \mathfrak{F}^s} \leq \tilde{\Phi}^{s}$. First, we rewrite $\ex{\tilde{\Phi}^{s+1} \mid \mathfrak{F}^s}$ as 
\begin{align*}
\lefteqn{ \ex{\tilde{\Phi}^{s+1} \mid \mathfrak{F}^s}  } \\
 & = \ex{ \Phi^{s+1}  \mid \mathfrak{F}^s}  \cdot \exp\left( \frac{c_3 \alpha}{n} \cdot G_{t_0}^{s} \right) \cdot \exp\left( - \frac{c_1 \alpha^2}{n} \cdot B_{t_0}^{s} \right)  - c_6 \cdot (s+1-t_0) \\ 
 & = \ex{ \Phi^{s+1} \mid \mathfrak{F}^s} \cdot \exp\left(\frac{\alpha}{n} \cdot ( c_2 \cdot \mathbf{1}_{\mathcal{G}^s} + (-c_1 \alpha) (1 - \mathbf{1}_{\mathcal{G}^s}) \right)  \cdot \exp\left( \frac{c_2 \alpha}{n} \cdot G_{t_0}^{s - 1} \right) \cdot \exp\left( - \frac{c_1 \alpha^2}{n} \cdot B_{t_0}^{s - 1} \right) \\
 &\qquad \mbox{} - c_3 - c_3 \cdot (s-t_0).
\end{align*}
We see that it suffices to prove that 
\begin{align*}
\ex{ \Phi^{s+1} \mid \mathfrak{F}^s} \cdot \exp\left(\frac{\alpha}{n} \cdot ( c_2 \cdot \mathbf{1}_{\mathcal{G}^s} + (-c_1 \alpha) (1 - \mathbf{1}_{\mathcal{G}^s}) \right) &\leq \Phi^{s} + c_3 \cdot \exp\left(- \frac{c_2 \alpha}{n} G_{t_0}^{s-1} \right),
\end{align*}
which in turn is implied by
\begin{align}
\ex{ \Phi^{s+1} \mid \mathfrak{F}^s} \cdot \exp\left(\frac{\alpha}{n} \cdot ( c_2 \cdot \mathbf{1}_{\mathcal{G}^s} + (-c_1 \alpha) (1 - \mathbf{1}_{\mathcal{G}^s}) \right) &\leq \Phi^{s} + c_3 \cdot \exp\left(- c_2 \alpha  \right),
\label{eq:sufficient}
\end{align}
as $G_{t_0}^{s-1} \leq s-t_0 \leq n$.
To show \eqref{eq:sufficient}, we consider two cases based on whether $\mathcal{G}^s$ holds. 

\textbf{Case 1} [$\mathcal{G}^s$ holds]. By \cref{lem:filling_good_quantile},
\[
\ex{\Phi^{s+1} \mid \mathfrak{F}^s, \mathcal{G}^s} \leq \Phi^s \cdot \left(1 - \frac{c_2 \alpha}{n} \right) + e^{3\alpha} \leq \Phi^s \cdot \exp\left(- \frac{c_2 \alpha}{n} \right) +  e^{3\alpha} .
\]
Hence, \NOTE{T}{John can you check this part after Hence? It is a bit informal I think, we use the conditioning inside the expectation to manipulate the indicator which is outside.}
\begin{align*}
\lefteqn{ \ex{\Phi^{s+1} \mid \mathfrak{F}^s, \mathcal{G}^s} \,\cdot\, \exp\left(\frac{\alpha}{n} \cdot ( c_2 \cdot \mathbf{1}_{\mathcal{G}^s} + (-c_1 \alpha) (1 - \mathbf{1}_{\mathcal{G}^s}) \right) } \\ &=
\ex{\Phi^{s+1} \mid \mathfrak{F}^s, \mathcal{G}^s } \cdot \exp\left( \frac{\alpha}{n} \cdot c_3 \right) \\
&\leq \left( \Phi^s \cdot \exp\left(- \frac{c_3 \alpha}{n} \right) + e^{3 \alpha} \right) \cdot \exp\left( \frac{\alpha}{n} \cdot c_3 \right) \\
&\leq \Phi^s + 2 e^{3 \alpha} \\
&= \Phi^s + c_3 \cdot \exp\left( -c_2 \alpha \right), 
\end{align*}
where the last line holds by definition of $c_3 = 2 e^{3 \alpha} \exp( c_2 \alpha)$.

\textbf{Case 2} [$\mathcal{G}^s$ does not hold]. Then \cref{lem:filling} implies,
\[
\ex{\Phi^{s+1} \mid \mathfrak{F}^s} \leq \Phi^s \cdot \left(1 + \frac{c_1 \alpha^2}{n} \right) + e^{3\alpha} \leq \Phi^s \cdot \exp\left(\frac{c_1 \alpha^2}{n} \right) + e^{3\alpha}.
\]
Hence,
\[
\ex{ \Phi^{s+1} \mid \mathfrak{F}^s } \cdot \exp\left( \frac{\alpha}{n} \cdot (-c_1 \alpha) \right) \leq \Phi^s + e^{3 \alpha} \leq \Phi^s + c_3 \cdot \exp\left( -c_2 \alpha \right),
\] 
which establishes \eqref{eq:sufficient} and the proof is complete.
\end{proof}}

Combining \cref{cor:enough_good_quantiles}, which shows that a constant fraction of rounds satisfy $\mathcal{G}^t$, with \cref{lem:filling_supermartingale} establishes a multiplicative drop of $\Ex{ \Phi^{s} }$ (unless it is already linear). Thus $\Ex{ \Phi^{m}}=\Oh(n)$, which implies $\Gap(m)=\Oh(\log n)$. This is formalized in the proof below. %

{\renewcommand{\thethm}{\ref{thm:filling_key}}
	\begin{thm}[restated]
\fillingkey 
	\end{thm}}
	\addtocounter{thm}{-1}
	
\begin{proof}%
For any integer $t_0 \geq 1$, first consider rounds $[t_0,t_0+n]$. We will now fix the constant $\alpha := \min \{ 1/101, 1/20 \cdot c_2/c_1 \}$ in the exponential potential function $\Phi^t$ (thus this is also fixed in $\tilde{\Phi}^t)$. By \cref{lem:filling_supermartingale}, $\tilde{\Phi}$ forms a super-martingale over $[t_0,t_0+n]$, and thus
\[
 \Ex { \tilde{\Phi}^{t_0+n} \; \Big| \; \mathfrak{F}^{t_0} } \leq \tilde{\Phi}^{t_0} = \Phi^{t_0},
 \]
which implies
 \[
 \Ex{ \Phi^{t_0+n} \cdot \exp\left( +\frac{c_2 \alpha}{n} \cdot G_{t_0}^{t_0+n-1} \right) \cdot \exp\left( -\frac{c_1 \alpha^2}{n} \cdot B_{t_0}^{t_0+n-1} \right) - c_3 \cdot n
 \; \,\bigg| \,\; \mathfrak{F}^{t_0} } \leq \Phi^{t_0},
\]
Rearranging this, and using that by~\cref{cor:enough_good_quantiles}, $G_{t_0}^{t_0+n-1} \geq n/40$ holds deterministically, we obtain for any $t_0 \geq 1$
\begin{align*}
 \Ex{ \Phi^{t_0+n} \, \mid \, \mathfrak{F}^{t_0} } &\leq \left( \Phi^{t_0} + c_3 \cdot n \right )
 \cdot \exp\left(  -c_2 \alpha \cdot \frac{1}{40} + c_1 \alpha^2 \cdot \frac{39}{40}\right), \\
 \intertext{and now using $\alpha = \min \{ 1/101, (1/40) \cdot c_2/c_1 \}$ and defining $c_4 := c_2/40^2$  yields} 
\Ex{ \Phi^{t_0+n} \, \mid \, \mathfrak{F}^{t_0} } &\leq \left( \Phi^{t_0} + c_3 \cdot n \right ) \cdot \exp\left( -c_4 \alpha \right) \\
 &= \Phi^{t_0} \cdot \exp\left( -c_4 \alpha \right)
 + c_3 \exp(-c_4 \alpha) \cdot n.
\end{align*}
It now follows by the second statement in \cref{lem:geometric_arithmetic} with $a:=\exp(-c_4\alpha) < 1$ and $b:= c_3 \exp(-c_4 \alpha) \cdot n$ that for any integer $k \geq 1$,
\begin{align} \label{eq:decrease_over_k_batches}
 \Ex{ \Phi^{n \cdot k}} &\leq 
 \Phi^{0} \cdot \exp\left( - c_4 \alpha \cdot k \right) + \frac{c_3 \exp(-c_4 \alpha) \cdot n}{1 - \exp(-c_4 \alpha)} \leq
 c_5 \cdot n,
 \end{align}
 for some constant $c_5 > 0$ as $\Phi^{0} \leq n$ holds deterministically.
 
 Hence for any number of rounds $t = k \cdot n + r$, where $k \geq 0$ and $1 \leq r < n$, we use~\cref{lem:filling_good_quantile} (first statement) iteratively to conclude that
 \begin{align}
 \Ex{ \Phi^{n \cdot k + r}  } & = \Ex{\Ex{ \Phi^{n \cdot k + r} \, \mid \, \mathfrak{F}^{n \cdot k } }}\notag \\
 & \leq \Ex{ \Phi^{n \cdot k}} \cdot \left(1 + \frac{c_1 \alpha^2}{n} \right)^{r} 
 + n \cdot \left(1 + \frac{c_1 \alpha^2}{n} \right)^{n} \cdot e^{3\alpha} \notag \\
 &\leq c_5 \cdot n \cdot \exp( c_1 \alpha^2) + n \cdot \exp( c_1 \alpha^2) \cdot e^{3 \alpha} \leq c_6 \cdot n, \label{eq:potential_small}
 \end{align}
 for some constant $c_6 > 0$.
  Hence for any $m \geq 1$, by Markov's inequality, 
 \[
  \Pro{ \Phi^{m} \leq c_6 \cdot n^{3} } \geq 1-n^{-2}.
 \]
Since $\Phi^{m} \leq c_6 \cdot n^{3}$ implies $\Gap(m)=\Oh(\log n)$, the proof is complete.
\end{proof}

\subsection{Unfolding General Filling Processes}\label{sec:unfolding_general}

Recall the definition of the unfolding of a process from Page \pageref{folding}. First, we show that our notion of ``unfolding'' can be applied to capture the \Caching process:

{\renewcommand{\thelem}{\ref{lem:memory_compressed}}
	\begin{lem}[restated]
\memorycompressed 
	\end{lem}}
	\addtocounter{lem}{-1}

\begin{proof}
As in the definition of unfolding, we denote the load vector of $P$ after round $t$ by $x^t$, and the load vector of a suitable unfolding $U(P)$ after the $s$-th atomic allocation by $\widehat{x}^s$. We also denote the corresponding normalized load vectors by $y^t$ and $\widehat{y}^s$ respectively. 

We will construct by induction, a coupling between a suitable filling process $P$, satisfying \WOne and \POne at each round, and an unfolding $U(P)$ which follows the distribution of \Caching. That is for every round $t \geq 0$ of $P$, there exists a (unique) atomic allocation $s=s(t) \geq t$ in $U(P)$, such that $x^t=\widehat{x}^{s(t)}$, and $U(P)$ is an instance of \Caching. 

Assume that for a suitable unfolding of the process $P$, the load configuration of $P$ after round $t$ equals the load configuration of $U(P)$ after atomic allocation $s=s(t)$, i.e., $x^t = \widehat{x}^{s(t)}$. In case the cache is empty (which happens only at the first round $s=0$), \Caching will sample a uniform bin $i$ (which satisfies \POne). If the cache is not empty, \Caching will take as bin $i$ the least loaded of the bin in the cache and a uniformly chosen bin. This produces a distribution vector that is majorized by one-choice (thus satisfies \POne, again). Thus we may couple the two instances such that process $P$ samples the same bin $i$ in round $t$ and atomic allocation $s(t)$, respectively. We continue with a case distinction concerning the load of bin $i$ at round $t$:

\textbf{Case 1} [The bin $i$ is overloaded, i.e., $y_i^{t} = \widehat{y}_i^{s(t)} \geq 0$]. Then \Caching and $P$ both place one ball into bin $i$, satisfying \WOne. Further, $P$ proceeds to the next round and $U(P)$ proceeds to the next atomic allocation, which means that the coupling is extended.
\smallskip

\textbf{Case 2} [The bin $i$ is underloaded, i.e., $y_i^{t} = \widehat{y}_i^{s(t)} < 0$]. Then we can deduce by definition of \Caching that it will place the next $\lceil -y_i^{s(t)} \rceil +1$ balls in some way that it deterministically satisfies the following conditions:
$(i)$ the first $\lceil -y_i^{s(t)}  \rceil $ balls are placed into bins which have a normalized load $<0$ at the atomic allocation $s(t)$, $(ii)$ one ball is placed into a bin with normalized load $<1$ at the atomic allocation $s(t)$. This follows since bin $i$ gets stored in the cache and at each atomic allocation $ j= W^{t-1} + 1, \dots W^t $ the process has access to a cached bin with normalized load at most $y_i^{s(t)} +j-1$. This satisfies \WOne so we can continue the coupling. 
\smallskip

We have thus constructed a process $P$, such that some unfolding $U=U(P)$ of $P$ is an instance \Caching. \end{proof}

We now prove the general gap bound for unfolded processes. 

{\renewcommand{\thelem}{\ref{lem:slowed-down}}
\begin{lem}[restated]
\sloweddown
	\end{lem}}
	\addtocounter{lem}{-1}

\begin{proof}%

We will re-use the constants $\alpha \in (0,1)$ and $c_6 > 0$ defined in the proof 
of \cref{thm:filling_key}.
We now define
\[
\mathcal{B}:= \left| \left\{ t \in [1,m] \colon \Phi^t \geq c_6 \cdot n^{6+c}  \right\}  \right|,
\] 
which is the number of ``bad'' rounds of the filling process $P$. We continue with a case distinction for each round $t$ whether $t \in \mathcal{B}$ holds.

\begin{itemize}
\item \textbf{Case 1} [$t \not\in \mathcal{B}$].
By definition, for a round $t \not\in B$ we have $\Phi^{t} < c_6 n^{6+c}$.
Further, $\Phi^t < c_6 n^{6+c}$ implies $\Gap_{P}(t)< \frac{10}{\alpha} \cdot \log n$, for sufficiently large $n$. Now let $s(t),s(t)+1,\ldots,s(t)+w^t-1$, $w^t:=\lceil -y_i^t \rceil+1$ be the atomic allocations in $U(P)$ corresponding to round $t$ in $P$. Since all allocations of $U(P)$ are to bins with normalized load at most $1$ before the allocation, we conclude $\Gap_{U(P)}(s) \leq \max\{\Gap_{P}(t),2\} < \frac{10}{\alpha} \cdot \log n$ for all $s \in [s(t),s(t)+w^t-1]$.

\item \textbf{Case 2} [$t \in \mathcal{B}$]. %
We will use that for any $0 < \alpha < 1$ and $t \geq 0$ we have 
	\[
	\Delta^t \leq   2n \cdot \left( \frac{1}{\alpha}\log \Phi^t + 1\right).
	\]
To see this, observe that $\sum_{i\in B_+^t}y_i^t = -
\sum_{i\in B_-^t}y_i^t$ and thus \begin{equation}\label{eq:delbdd}\Delta^t\leq 2\sum_{i\in B_+^t}y_i^t\leq 2\sum_{i\in B_+^t}y_i^t\mathbf{1}_{y_i^t\geq 2} + 2n.\end{equation} Now, note that since $ \Phi^t =  \sum_{i\in[n] : y_{i}^{t}\geq 2} \exp\left( \alpha \cdot  y_i^t \right)$, if $ \Phi^t  \leq \lambda$ then  $y_i^t \leq \frac{1}{\alpha} \cdot \log \lambda$ for all $i\in[n]$ with $y_{i}^t\geq 2$. Thus by \eqref{eq:delbdd} we have $\Delta^t\leq 2n\cdot (1/\alpha)\log\Phi^t + 2n$ as claimed.

So for a round $t \in \mathcal{B}$, we have \NOTE{T}{I think the first inequality is correct, but maybe $\leq 2 \Delta^t + 2$ would be clearer?}\NOTE{J}{Why? It is not clear to me where the $2$ comes from in the weaker bound?}
 \begin{align}
 w^{t} \leq \Delta^{t} + 1 \leq \frac{2n}{\alpha}\log \Phi^t + 2n + 1 \leq c_\alpha n\log \Phi^t, \label{eq:weight_bound}
 \end{align}
 which holds deterministically for some constant $c_\alpha$. Again, every such round $t \in \mathcal{B}$ of $P$ corresponds to the atomic allocations $s(t),s(t)+1,\ldots,s(t)+w^{t}-1$ in $U(P)$, and we will (pessimistically) assume that the gap in all those rounds is large, i.e., at least $\frac{10}{\alpha} \cdot \log n$.
\end{itemize}
 
By the above case distinction and \eqref{eq:weight_bound}, we can upper bound the number of rounds $s$ in $[1,m]$ of $U(P)$ where $\Gap_{U(P)}(s) < \frac{10}{\alpha} \cdot \log n$ does not hold as follows:
 \begin{align}
  \left|\left\{ s \in [1,m] \colon \Gap_{U(P)}(s) \geq \frac{10}{\alpha} \cdot \log n \right\}\right| \leq \sum_{t=1}^m \mathbf{1}_{t \in \mathcal{B}} \cdot w^{t} \leq  c_{\alpha} n  \cdot \sum_{t=1}^m \mathbf{1}_{t \in \mathcal{B}} \cdot \log \Phi^t  . \label{eq:bad_times}
 \end{align}
Next define the sum of the  exponential potential function over rounds $1$ to $m$ as
\[
 \Phi := \sum_{t=1}^{m} \Phi^t.
 \]
 
 Then from the proof of \cref{thm:filling_key}, equation \eqref{eq:potential_small}, there is a constant $c_6 > 0$ such that $\Ex{ \Phi^t } \leq c_6 \cdot n$, and hence 
 \[
  \Ex{ \Phi } = \sum_{t=1}^m \Ex{ \Phi^t } \leq m \cdot c_6 \cdot n.
 \]
 By Markov's inequality,
 \[
  \Pro{  \Phi \leq m \cdot c_6 \cdot n^{3} } \geq 1 - n^{-2}.\label{eq:markov_inequality}
 \]
 Note that conditional on the above event occurring, the following bound holds deterministically:
 \[
  | \mathcal{B} | \leq \frac{ m c_6 n^{3} }{ c_6 n^{6+c} } \leq m \cdot n^{-3-c}.
 \]

 If $\mathcal{B}$ is empty we are done by \eqref{eq:bad_times}. So we can assume that $|\mathcal{B}|\geq 1 $ and apply the Log-sum inequality \cref{lem:logsum}, with $a_t=1$ and $b_t=\Phi^t$, to the sum over $\mathcal{B}$ in \eqref{eq:bad_times}, which gives us 
 \begin{equation}\label{eq:logsumbd}\sum_{t=1}^m \mathbf{1}_{t \in \mathcal{B}} \cdot  \log \Phi^t =   \sum_{t \in \mathcal{B}} \log \Phi^t \leq  \left( \sum_{t \in \mathcal{B}} 1 \right) \log\left(  \frac{\sum_{t \in \mathcal{B}}\Phi^t}{\sum_{t \in \mathcal{B}} 1}\right) \leq    |\mathcal{B}|   \log\left(  \frac{\Phi}{|\mathcal{B}| }\right).  \end{equation}Now, if the event $\Phi \leq m \cdot c_6 \cdot n^3$ occurs, then by \eqref{eq:bad_times} and \eqref{eq:logsumbd} we have %
 \begin{align*}
    \Bigl|\left\{ s \in [1,m] \colon \Gap_{U(P)}(s) \geq \frac{10}{\alpha} \cdot \log n \right\}\Bigr| &\leq c_{\alpha}n \cdot  |\mathcal{B}| \cdot   \log\left(  \frac{\Phi}{|\mathcal{B}| }\right)\\   
    &\leq c_{\alpha}n \cdot  (m \cdot n^{-3-c}) \cdot   \log\left(  m \cdot c_6 \cdot n^3\right)\\
      &\leq c_{\alpha}' m \cdot n^{-2-c}  \cdot \left( \log  m  + \log n \right)\\
      &\leq  m \cdot n^{-c}  \cdot \log  m ,
 \end{align*}
where the third inequality is for some constant $c'_\alpha$ depending on $\alpha$ and the last inequality holds since  $\alpha>0$ is a small but fixed constant thus $c_\alpha'$ is constant. Since the inequality holds for \emph{any} unfolding of $P$, once the event in \eqref{eq:markov_inequality} occurs, we obtain the corollary.
 \end{proof}
 
 \subsection{Improved Analysis of the \Caching Process} \label{sec:ballbyball_caching}

In this section, we will prove for the \Caching process that \Whp\footnote{In general, with high probability refers to probability of at least $1 - n^{-c}$ for some constant $c > 0$.}~$\Gap(m) = \Oh(\log n)$, where $m \geq 1$ is arbitrary. 
This improves the guarantee we obtained in \cref{sec:unfolding_general}, which was based on regarding \Caching as an ``unfolded'' version of a filling process.

We define an exponential potential $\Psi^t $ which is similar to $\Phi^t$, but with $\tilde{\alpha} := \frac{1}{12n}$, and is given by 
\begin{equation}\label{eq:modpot}
\Psi^t 
:= \sum_{i: y_{i}^t\geq 2} \exp\left( \tilde{\alpha} \cdot  y_i^t \right) =  \sum_{i=1}^n \exp\left( \tilde{\alpha} \cdot  y_i^t \right)\cdot  \mathbf{1}_{\{ y_{i}^{t}\geq 2\}}.
\end{equation}
Our first goal will be to prove that $\ex{\Psi^{t_0}} \leq 144 n^3$ for an arbitrary round $t_0$, which in turn implies, using Markov's inequality, that \Whp~$\Delta^{t_0} = \Oh(n^2 \log n)$. Note that this is weaker than the $\Oh(\log n)$ gap we are aiming to prove. %

We will then use this gap as a starting point at $t_0 := m - n^7$ for a filling process satisfying \POne and \WOne, whose unfolding is \Caching. Such a process exists by \cref{lem:memory_compressed} and we will call it \FoldedCaching in the rest of the section. For the potential function $\Phi$ defined in \cref{sec:filling_potential}, we will show using, \cref{thm:filling_key} and \cref{lem:slowed-down}, that $\ex{\Phi^{t_0 + s}} = \Oh(n)$ for any $s \geq t_0 + n^3 \log^2 n$ rounds (which may unfold to multiple atomic allocations as the rounds of the \Caching process). Finally, by taking Markov's inequality and the union bound over the next $n^7$ rounds which include the allocation of the $m$-th ball, we conclude that the exponential potential is small \Whp~and hence $\Gap(m) = \Oh(\log n)$.

The most involved part of this proof is to prove that if $\Psi^t$ is large, then it decreases in expectation. To this end, we will analyze two consecutive rounds of the \Caching process. In these two rounds we show that even for the worst-case choice of $b$ at the beginning of the two rounds, the \Caching process has at least a $1/n^2$ bias to place the two balls away from the heaviest bin. This bias is enough to obtain a $\poly(n)$ gap. Note that over one step, the potential might increase in expectation even when large, e.g., when $b$ stores the heaviest bin, in which case \Caching allocates the next ball using \OneChoice. 

\begin{lem} \label{lem:caching_mult_drop}
For any round $t \geq 0$, the \Caching process satisfies
\[
\ex{\Psi^{t+2} \mid \mathfrak{F}^{t}} \leq \Psi^t \cdot \Big( 1 - \frac{1}{24n^3} \Big) + 6.
\]
\end{lem}
\begin{proof}
Recall that the filtration $\mathfrak{F}^t$ reveals the load vector $x^t$. Throughout this proof, we consider any labeling of the $n$ bins such that $x_1^{t} \geq x_2^{t} \geq \cdots \geq x_n^{t}$ for round $t$. 
Now, consider a modified version of \Caching, which we call \WeakCaching, which ``forgets'' the cache at time $t$ and makes its decision for rounds $t$ and $t+1$ using the loads $x^t$:
\begin{itemize}[itemsep=0pt]
  \item In round $t$, place the ball in a bin $i_1$ sampled using \OneChoice.
  \item In round $t+1$, sample another bin $i_2$ using \OneChoice and place the ball in the least loaded of the $i_1$ and $i_2$, using the load information of round $t$. We break ties using the bin indices at round $t$. 
\end{itemize}
Let $\tilde{\Psi}$ be the potential \eqref{eq:modpot} for \WeakCaching. The fact that \WeakCaching resets the cache implies that $\ex{\Psi^{t+2}  \mid x^t} \leq \ex{\tilde{\Psi}^{t+2}  \mid x^t} + 3$ (as we will show in \cref{clm:modified_process_majorises_caching}), while the fact that \WeakCaching uses outdated information from $x^t$ allows us to analyze $\ex{\tilde{\Psi}^{t+2} \mid x^t}$. 
\begin{clm}
\label{clm:modified_process_majorises_caching}
For any $x^t$ and any choice of the cache $b$ at round $t$,%
\[
\Ex{\Psi^{t+2}  \mid x^t} \leq \Ex{\tilde{\Psi}^{t+2}  \mid x^t} + 3.
\]
\end{clm}\begin{poc}
Fix an arbitrary load vector $x^{t}$ and consider the allocation of \WeakCaching and \Caching in rounds $t$ and $t+1$.
We consider a coupling, where $i_1$ and $i_2$ are the bin samples of the two processes at rounds $t$ and $t+1$. Let $i_1^A$, $i_2^A$ and $i_1^B$, $i_2^B$ be the pairs of bin indices where the \Caching and \WeakCaching processes respectively allocate the two balls corresponding to $t$\textsuperscript{th} and $(t+1)$\textsuperscript{th} allocation, and let $x_{A}^{t+1}, x_{A}^{t+2}$ and $x_{B}^{t+1}, x_{B}^{t+2}$ be the resulting load distributions. %
We will prove that $x_{B}^{t+2}$ majorizes $x_{A}^{t+2}$ by showing that $\langle x_{A, i_1^A}^{t+1}, x_{A, i_2^A}^{t+2}\rangle \leq \langle x_{B, i_1^B}^{t+1}, x_{B, i_2^B}^{t+2} \rangle$, where $\langle a, b \rangle := (\max\{a, b\}, \min\{a, b\})$. The claim will follow by \cref{clm:schur_varphi}, since $2e^{3\alpha} < 3$.

The claim holds because \WeakCaching can allocate according to the best of two choices only for the second ball using outdated information, while \Caching has this ability for the allocation of \emph{both} balls. More formally, \WeakCaching will allocate to loads that are at least as large (because of the decisions being made on outdated information) as the two smallest loads in the multiset $M_B := \{ x_{i_1}, x_{i_1} + 1, x_{i_2} \}$. We consider the following cases (where each process allocates two balls) based on the cache $b$ of \Caching and the two choices $i_1$ and $i_2$.
\begin{itemize}
  \item \textbf{Case A} [$b, i_1, i_2$ all different]. \Caching allocates to the two smallest loads from the multiset $M_A = \{ x_b, x_b + 1, x_{i_1}, x_{i_1} + 1, x_{i_2} \}$, while $M_B = \{ x_{i_1}, x_{i_1} + 1, x_{i_2} \}$. So $M_B \subseteq M_A$. 
  \item \textbf{Case B} [$b \neq i_1 = i_2$]. \Caching allocates to the two smallest loads from the multiset $M_A = \{ x_b, x_b + 1, x_{i_1}, x_{i_1} + 1 \}$, while $M_B = \{ x_{i_1}, x_{i_1} + 1 \}$. So $M_B \subseteq M_A$.
  \item \textbf{Case C} [$b = i_1 \neq i_2$]. \Caching allocates to the two smallest loads from the multiset $M_A = \{ x_{i_1}, x_{i_1} + 1, x_{i_2} \}$, while $M_B = \{ x_{i_1}, x_{i_1} + 1, x_{i_2} \}$. So $M_B \subseteq M_A$. 
  \item \textbf{Case D} [$b = i_2 \neq i_1$]. \Caching allocates to the two smallest loads from the multiset $M_A = \{ x_{i_1}, x_{i_1} + 1, x_{i_2}, x_{i_2} + 1 \}$, while $M_B = \{ x_{i_1}, x_{i_1} + 1, x_{i_2} \}$. So $M_B \subseteq M_A$.
  \item \textbf{Case E} [$b = i_2 = i_1$]. \Caching and \WeakCaching will allocate both balls to $i_1$.
\end{itemize} In all these cases, $M_B \subseteq M_A$, so $\langle x_{A, i_1^A}^{t+1}, x_{A, i_2^A}^{t+2}\rangle \leq \langle x_{B, i_1^B}^{t+1}, x_{B, i_2^B}^{t+2} \rangle$, establishing the claim.
\end{poc}

Next we define $\tilde{p}_{1,i}^t$ as the probability of allocating exactly one ball and $\tilde{p}_{2,i}^t$ the probability of allocating two balls in bin $i$ over the next two rounds $t$ and $t+1$ using \WeakCaching. For any bin $i \in [n]$ with $y_i^t \geq 4 + 2/n$, we know that $y_i^{t+2} \geq 2$ so,%
\begin{align*}
\ex{\tilde{\Psi}_i^{t+2} \mid x^t} 
 & = \tilde{\Psi}_i^t \cdot e^{-2\tilde{\alpha}/n} \cdot (1 - \tilde{p}_{1,i}^t - \tilde{p}_{2,i}^t) + \tilde{\Psi}_i^t \cdot e^{\tilde{\alpha} - 2\tilde{\alpha}/n} \cdot \tilde{p}_{1,i}^t + \tilde{\Psi}_i^t \cdot e^{2\tilde{\alpha} - 2\tilde{\alpha}/n} \cdot \tilde{p}_{2,i}^t \\
 & = \tilde{\Psi}_i^t \cdot e^{-2\tilde{\alpha}/n} \cdot (1 + \tilde{p}_{1,i}^t \cdot (e^{\tilde{\alpha}} - 1) + \tilde{p}_{2,i}^t \cdot (e^{2\tilde{\alpha}} - 1)) \\
 & = \Psi_i^t \cdot e^{-2\tilde{\alpha}/n} \cdot (1 + \tilde{p}_{1,i}^t \cdot (e^{\tilde{\alpha}} - 1) + \tilde{p}_{2,i}^t \cdot (e^{2\tilde{\alpha}} - 1)).
\end{align*}
Since for any $i$ with $y_i^{t+2} < 2$, we have $\tilde{\Psi}_i^{t+2} = 0$, for any bin $i \in [n]$ with $y_i^t \in [2,4+2/n]$
\[
\ex{\tilde{\Psi}_i^{t+2} \mid x^t} \leq \Psi_i^t \cdot e^{-2\tilde{\alpha}/n} \cdot (1 + \tilde{p}_{1,i}^t \cdot (e^{\tilde{\alpha}} - 1) + \tilde{p}_{2,i}^t \cdot (e^{2\tilde{\alpha}} - 1)).
\]
In rounds $t, t+1$ at most two bins can be change from $y_i^t < 2$ to $y_i^{t+2} > 2$. In this case, we have $y_i^{t+2} < 4$. So, their total contribution will be at most $2e^{\tilde{\alpha}(3 + 2/n)} < 2e^{4\tilde{\alpha}}$, so
\[
\ex{\tilde{\Psi}^{t+2} \mid x^{t}} \leq \sum_{i : y_i^t \geq 2} \Psi_i^t \cdot e^{-2\tilde{\alpha}/n} \cdot (1 + \tilde{p}_{1,i}^t \cdot (e^{\tilde{\alpha}} - 1) + \tilde{p}_{2,i}^t \cdot (e^{2\tilde{\alpha}} - 1)) + 2e^{4\tilde{\alpha}}.
\]
Using the Taylor estimate $e^z \leq 1 + z + z^2$ for $z < 1.75$ and that $2e^{4\tilde{\alpha}} < 3$,
\[
\ex{\tilde{\Psi}^{t+2} \mid x^t}
\leq \sum_{i : y_i^t \geq 2} \Psi_i^t \cdot \Big(1 - \frac{2\tilde{\alpha}}{n} + \frac{4\tilde{\alpha}^2}{n^2} \Big) \cdot (1 + \tilde{p}_{1,i}^t \cdot (\tilde{\alpha} + \tilde{\alpha}^2) + \tilde{p}_{2,i}^t \cdot (2\tilde{\alpha} +4\tilde{\alpha}^2)) + 3.
\]

For \WeakCaching, the probability of allocating exactly one ball to bin $i \in [n]$ (with respect to the load ordering at round $t$) over the next two rounds is,
\begin{align*}
\tilde{p}_{1,i}^t & = \Pro{\left\{i_1 = i\right\} \cap \left\{ i_2 > i\right\}} + \Pro{\left\{i_1 < i\right\} \cap \left\{i_2 = i\right\}}  = \frac{1}{n} \cdot \Big(1 - \frac{i}{n}\Big) + \frac{i-1}{n} \cdot \frac{1}{n} = \frac{1}{n} - \frac{1}{n^2},
\end{align*}
where we have used the fact that \WeakCaching allocates the two balls using the information of the load vector $x^t$ only, and break ties according to the ball index of the (sorted) load vector $x^t$. Similarly, the probability of allocating two balls is, 
\[
\tilde{p}_{2,i}^t = \Pro{ \left\{ i_1 = i \right\} \cap \left\{ i_2 \leq i \right\} } =  \frac{1}{n} \cdot \frac{i}{n}.
\]
Let $\overline{p}_2^t := \frac{1}{|\{ i : y_i^t \geq 2 \}|} \cdot \sum_{i : y_i^t \geq 2} \tilde{p}_{2,i}^t$ be the average probability of allocating two balls to bin $i$ with $y_i^t \geq 2$. Note that since $\tilde{p}_{2,i}^t$ is increasing in $i$, the larger $|\{ i : y_i^t \geq 2 \}|$ is, the larger $\overline{p}_2^t$ will be. Since there can be at most $n - 1$ bins in with $y_i^t \geq 2$, we have that $|\{ i : y_i^t \geq 2 \}| \leq n - 1$, hence,  
\[
\overline{p}_2^t \leq \frac{1}{n^2} \cdot \frac{1}{n-1} \cdot  \sum_{i = 1}^{n-1} i
=  \frac{n(n-1)}{2 n^2 (n-1)}
=  \frac{1}{2 n}.
\]
Now, the expected change for the $\Psi$ potential over the next two rounds is at most,
\begin{align*}
\ex{\Psi^{t+2} \mid \mathcal{F}^t}
 & = \ex{\Psi^{t+2} \mid \mathcal{F}^t, x^t} \\
 & \stackrel{(a)}{\leq} \ex{\tilde{\Psi}^{t+2} \mid x^t} + 3 \\
 & \leq \sum_{i : y_i^t \geq 2}\Psi_i^t \cdot \Big(1 - \frac{2\tilde{\alpha}}{n} + \frac{4\tilde{\alpha}^2}{n^2} \Big) \cdot (1 + \tilde{p}_{1,i}^t \cdot (\tilde{\alpha} + \tilde{\alpha}^2) + \tilde{p}_{2,i}^t \cdot (2\tilde{\alpha} +4\tilde{\alpha}^2)) + 6 \\
  & \stackrel{(b)}{\leq} \sum_{i : y_i^t \geq 2}\Psi_i^t \cdot \Big(1 - \frac{2\tilde{\alpha}}{n} + \frac{4\tilde{\alpha}^2}{n^2} \Big) \cdot (1 + \tilde{p}_{1,i}^t \cdot (\tilde{\alpha} + \tilde{\alpha}^2) + \overline{p}_{2}^t \cdot (2\tilde{\alpha} +4\tilde{\alpha}^2)) + 6 \\
 & \leq \sum_{i : y_i^t \geq 2} \Psi_i^t \cdot \Big(1 - \frac{2\tilde{\alpha}}{n} + \frac{4\tilde{\alpha}^2}{n^2} \Big) \cdot \Big(1 + \Big(\frac{1}{n} - \frac{1}{n^2}\Big) \cdot (\tilde{\alpha} + \tilde{\alpha}^2)  + \frac{1}{2 n} \cdot (2\tilde{\alpha} +4\tilde{\alpha}^2)\Big) + 6 \\
  &
  = \sum_{i : y_i^t \geq 2} \Psi_i^t \cdot \Big( 1 - \frac{2\tilde{\alpha}}{n} + \frac{4\tilde{\alpha}^2}{n^2} \Big) \cdot 
  \Big(1 + \frac{2\tilde{\alpha}}{n} - \frac{\tilde{\alpha}}{n^2} + \frac{3\tilde{\alpha}^2}{n} - \frac{\tilde{\alpha}^2}{n^2} \Big) + 6 \\
 & \stackrel{(c)}{=} \sum_{i : y_i^t \geq 2} \Psi_i^t \cdot \Big(1 - \frac{2\tilde{\alpha}}{n} + \frac{2\tilde{\alpha}}{n} - \frac{\tilde{\alpha}}{n^2} + \frac{3\tilde{\alpha}^2}{n} + o(n^{-3}) \Big) + 6 \\
 & = \sum_{i : y_i^t \geq 2} \Psi_i^t \cdot \Big(1 - \frac{\tilde{\alpha}}{n^2} + \frac{3\tilde{\alpha}^2}{n} + o(n^{-3}) \Big) + 6 \\
 & \stackrel{(d)}{\leq} \sum_{i : y_i^t \geq 2} \Psi_i^t \cdot \Big(1 - \frac{1}{24 n^3} \Big) + 6,
\end{align*}
where $(a)$ uses \cref{clm:modified_process_majorises_caching}, $(b)$ uses majorization (\cref{lem:quasilem}), since $\Psi_i^t$ is non-decreasing and $\tilde{p}_{2,i}^t$ is increasing and $(c)$ as well as $(d)$ use that $\tilde{\alpha} = 1/(12n)$.
\end{proof}

\begin{lem} \label{lem:caching_poly_n_expectation}
For any round $t \geq 0$ of the \Caching process we have
\[
\ex{\Psi^t } \leq 144 n^3.
\]
\end{lem}
\begin{proof}
We will prove the statement by strong induction. For $t=0$ and $t=1$, we have $\Psi^0 = 0$ and $\Psi^1 = 0$ deterministically, so the statement holds. Assume $\ex{\Psi^{s}} \leq 144 n^3$ for all $s\leq t+1$, then  using \cref{lem:caching_mult_drop},
\[
\ex{\Psi^{t+2}} = \ex{\ex{\Psi^{t+2} \mid \mathfrak{F}^{t}}} \leq \ex{\Psi^t} \cdot \Big(1 - \frac{1}{24 n^3} \Big) + 6 \leq 144 n^3 - \frac{144n^3}{24 n^3} + 6 = 144 n^3.
\qedhere\]
\end{proof}

\begin{lem} \label{lem:caching_base_case}
For the \Caching process, for any round $t \geq 0$,
\[
\Pro{\max_{i \in [n]} y_i^t \leq 14 \cdot n \log n} \geq 1 - n^{-9}.
\]
\end{lem}
\begin{proof}
Using \cref{lem:caching_poly_n_expectation} and Markov's inequality, we get,
\[
\Pro{\Psi^t \leq 144 n^{13}} \geq 1 - n^{-9}.
\]
Since $\Psi^t \leq 144 n^{13}$ implies $\min_{i \in [n]} y_i^t \leq 13 \cdot n \log n + \log 144 \leq 14 \cdot n \log n$, we conclude that
\[
\Pro{\max_{i \in [n]} y_i^t \leq 14 \cdot n \log n} \geq 1 - n^{-9}.\qedhere
\]
\end{proof}

The next claim holds not only for \FoldedCaching, but more generally for any allocation process satisfying \WOne, and we will make use of the more general claim in Section~\ref{sec:lower}.
\begin{clm}\label{clm:bound_wone}  
For any allocation process satisfying \WOne, it holds for any $t \geq 0$ that $\Delta^{t+1} \leq \Delta^t + 4$.
\end{clm}
\begin{proof}
We have to relate $\Delta^{t}$ to $\Delta^{t+1}$ where we recall that, for any $t\geq 0$,
\[
 \Delta^{t} := \sum_{i \in [n]} \left| x_i^{t} - \frac{W^{t}}{n} \right|.
\]
To do this, let $i \in [n]$ be the chosen bin for the allocation in round $t$.

\medskip

\noindent\textbf{Case 1} [$y_i^{t} \geq 0$]. If $i \in [n]$ satisfies $y_i^{t} \geq 0$, i.e., the chosen bin is overloaded, then we place exactly one ball into $i$. This increases the average load by $\frac{1}{n}$ and increments the load of exactly one bin by $1$. By using the triangle inequality this implies
\[
 \Delta^{t+1} \leq \Delta^{t} + 2.
\]

\medskip 

\noindent\textbf{Case 2} [$y_i^{t} < 0$]. If bin $i$ is underloaded, then we allocate exactly $w^{t}:=\lceil -y_i^{t} \rceil + 1$ many balls in round $t$. We now split the allocation of these $w^t$ balls within round $t$ into $w^{t}$ atomic allocations labeled $1,2,\ldots,w^{t}$.
Correspondingly, we define the absolute value potential for atomic allocations between $t$ and $t+1$, by setting for any $\ell \in [0,w^{t}]$
\[
 \Delta^{t,\ell} := \sum_{i \in [n]} \left| x_i^{t,\ell} - \frac{W^{t}+\ell}{n} \right|,
\] 
where $x_{t,\ell}$ is the load vector after $\ell$-th atomic allocation of round $t$. Since the order of the placement of the balls does not matter for $\Delta^{t,w^{t}}=\Delta^{t+1}$, we may assume that the bins $k_1$ and $k_2$ in the definition of \WOne (if they exist), are used in the first and second atomic allocation of round $t$, respectively. Using the same as argument as in Case 1, we conclude that
$
 \Delta^{t,1} \leq \Delta^{t,0} + 2,
$
and
$
 \Delta^{t,2} \leq \Delta^{t,1} + 2,
$
so that
\[
 \Delta^{t,2} \leq \Delta^{t,0} + 4.
\]
Further, note that any bin $j \in [n]$ apart from $k_1$ and $k_2$ receives at most $\lceil -y_j^{t} \rceil - 1$ many balls in round $t$. Hence if we allocate to a bin $j$ in the $\ell$-th atomic allocation, where $\ell \geq 3$, it can have received at most $\bigl( \lceil -y_j^{t} \rceil - 1\bigr) - 1
=\lceil -x_j^{t} + \frac{W^{t}}{n} \rceil - 2$ many balls in previous atomic allocations. Hence
\begin{align*}
 x_{j}^{t,\ell} - \frac{W^{t}+\ell}{n} &\leq  x_{j}^{t,\ell} - \frac{W^{t}}{n} \leq x_{j}^{t} + \left\lceil -x_j^{t} + \frac{W^{t}}{n} \right\rceil - 2 - \frac{W^{t}}{n} \leq -1.
\end{align*} 

Therefore,
\begin{align}
 \left| x_{j}^{t,\ell+1} - \frac{W^{t}+(\ell+1)}{n} \right| &=
  \left| \underbrace{ x_{j}^{t,\ell} - \frac{W^{t}+\ell}{n}}_{\leq -1}  + 1 - \frac{1}{n} \right|
 =
 \left| x_{j}^{t,\ell}  - \frac{W^{t}+\ell}{n} \right| - \left( 1 - \frac{1}{n} \right). \label{eq:une}
\end{align}
On the other hand, for any $j \neq i$, we have by the triangle inequality
\begin{align}
    \sum_{j \in [n],j \neq i} \left| x_{j}^{t,\ell+1} - \frac{W^{t}+(\ell+1)}{n}
    \right| &=  \sum_{j \in [n],j \neq i} \left| x_{j}^{t,\ell} - \frac{W^{t}+\ell}{n} + \frac{1}{n}
    \right| \notag \\
    &\leq \sum_{j \in [n],j \neq i} \left| x_{j}^{t,\ell} - \frac{W^{t}+\ell}{n} \right| + (n-1) \cdot \frac{1}{n}. \label{eq:deux}
\end{align}
If we combine \eqref{eq:une} and \eqref{eq:deux} then the $-(1-1/n )$ and $+(n-1 )/n$ terms cancel, so we have  
\[
 \Delta^{t,\ell+1} \leq \Delta^{t,\ell}.
\]
Aggregating the changes over $\ell =0,1,2,\ldots,w^{t}$, we conclude
\[
  \Delta^{t+1} = \Delta^{t,w^{t}} \leq \Delta^{t,0} + 4 = \Delta^{t} + 4.\qedhere
\]
\end{proof}

We now (re-)state and prove the main result of this section: the gap bound for \Caching.
{\renewcommand{\thelem}{\ref{lem:cachingfixedtimestep}}
\begin{lem}[restated]
\cachingfixedtimestep
	\end{lem}}
	\addtocounter{lem}{-1}
\begin{proof}
First note that if $m < n^7$, then the statement of the lemma is also covered by \cref{lem:slowed-down} with $c=8$. Thus we will assume $m \geq n^7$ in the following. 

In order to analyze the state of \Caching at round $m$, we will analyze \Caching in two different phases. The first phase lasts from round $1$ to $m-n^7$, and in this phase we consider the (original) \Caching process. We use the above analysis to prove a (polynomial) bound on the potential $\Psi$ for round $m-n^7$, which implies that $\Delta^t \leq 14 n^2 \log n$ and an initial bound on $\Phi^{m - n^7}$. The second phase starts with the load configuration at round $m-n^7$, and then considers $n^7$ rounds of \FoldedCaching (as a process which satisfies \POne and \WOne). The folded process allocates at least one ball (and possibly more) at each round, and hence the allocation of the $m$-th ball comes before round $m$ in the second phase. Also, as we will prove below, the allocation of the $m$-th ball does not happen before round $m-\frac{1}{2} n^{3.5}$, which gives enough time for the exponential potential $\Phi$ (used in the analysis in \cref{sec:analysis_filling}) to decrease.

Define $t_0:=m-n^7$, which is the last round of phase one when we switch from \Caching to \FoldedCaching.
Following \cref{lem:caching_base_case}, we
define an event
\[
 \mathcal{G}^{t_0} := \left\{ \max_{i \in [n]} y_i^{t_0} \leq 14 \cdot n \log n
 \right\}
\]
and thus at at round $t_0$, we have
\begin{align}
\Pro{ \mathcal{G}^{t_0} } \geq 1-n^{-9}. \label{eq:lemma_58}
\end{align}
If $\mathcal{G}^{t_0}$ occurs, then the following two bounds on potential functions hold:
\begin{align}
\Phi^{t_0} \leq n \cdot \exp(\alpha \cdot 14 \cdot n \log n), \label{eq:phi_t_0}
\end{align}
and
\begin{align}
\Delta^{t_0} \leq 14 \cdot n^2 \log n. \label{eq:Delta_t_0}
\end{align}

Fix now any filtration $\mathcal{F}^{t_0}$, including the load vector $x^{t_0}$, such that $\mathcal{G}^{t_0}$ holds.
We will now look at the process after round $t_0$, as mentioned before, we will consider \FoldedCaching process (i.e., multiple balls may be allocated in each round).
For convenience, the time-indices larger than $t_0$ will refer to the rounds of \FoldedCaching the instead of atomic allocations. Thus, for any integer $s \geq 0$, $x^{t_0+s}$ refers to the resulting load vector after $s$ allocations of \FoldedCaching have been completed, starting with $x^{t_0}$.

Next consider $s \geq n^3 \log^2 n > \log(n \cdot \exp(14 \cdot n \log n))$ rounds after $t_0$ and write $s = \kappa_1 \cdot n + \kappa_2$ for $\kappa_1, \kappa_2 \in \mathbb{N}$ and $0 \leq \kappa_2 < n$. Then,
\begin{align*}
\Ex{ \Phi^{t_0 + n \cdot \kappa_1} \,|\, \mathfrak{F}^{t_0} },\mathcal{G}^{t_0} & \stackrel{(a)}{\leq}
 \Phi^{t_0} \cdot e^{- c_4 \alpha \cdot \kappa_1} +  \sum_{i=0}^{\kappa_1-1} c_3 e^{-c_4 \alpha} \cdot e^{- c_4 \alpha \cdot i} \\ 
 & \stackrel{(b)}{\leq} 
 n \cdot e^{14 \cdot n \log n} \cdot e^{- c_4 \alpha \cdot \lfloor n^2 \log^2 n \rfloor} +  \sum_{i=0}^{\kappa_1-1} c_3 e^{-c_4 \alpha} \cdot e^{- c_4 \alpha \cdot i } \leq c_5 n,
\end{align*}
for some constant $c_5 > 0$, where $(a)$ used \eqref{eq:decrease_over_k_batches} and $(b)$ used \eqref{eq:phi_t_0}.

Similarly to \eqref{eq:potential_small},
\begin{align*}
 \Ex{ \Phi^{t_0 + \kappa_1 \cdot n + \kappa_2} \, \mid \, \mathfrak{F}^{t_0},\mathcal{G}^{t_0} } &\leq \Ex{\Phi^{\kappa_1 \cdot n} \, \mid \, \mathfrak{F}^{t_0}, \mathcal{G}^{t_0}} \cdot \left(1 + \frac{c_1 \alpha^2}{n} \right)^{\kappa_2} 
 + n \cdot \left(1 + \frac{c_1 \alpha^2}{n} \right)^{n} \cdot e^{3\alpha} \\
 &\leq c_5 \cdot n \cdot \exp( c_1 \alpha^2) + n \cdot \exp( c_1 \alpha^2) \cdot e^{3 \alpha} \leq c_6 \cdot n,
 \end{align*}
for some constant $c_6 > 0$. Note that in any round $t$, \FoldedCaching allocates at most $\Delta^{t}+4$ balls. Hence during the $s$ rounds after $t_0$ of \FoldedCaching, the total number of balls allocated is at most 
\begin{align*}
(\Delta^{t_0}+4) + (\Delta^{t_0+1}+4) + \ldots + (\Delta^{t_0+s}+4)
& \leq (\Delta^{t_0}+4) + (\Delta^{t_0} + 8) + \ldots + (\Delta^{t_0} + 4\cdot (s+1)) \\
 & \leq \Delta^{t_0} \cdot (s +1) + 2 (s+1)(s+2),
\end{align*}
where we used that for any round $t \geq t_0$, $\Delta^{t+1} \leq \Delta^{t} + 4$ (\cref{clm:bound_wone}).

So, for any $s < \frac{1}{2} n^{3.5}$, the total number of balls allocated in the $s$ rounds after $t_0$ is at most 
\[
\Delta^{t_0} \cdot (s+1) + 2 (s+1)(s+2) < (14 \cdot n^2 \log n) \cdot \left(\frac{1}{2} n^{3.5}+1\right) + \frac{1}{2} n^{7}\left(1+\mathcal{O}\left(\frac{1}{n}\right)\right) < n^7,
\]
where we used \eqref{eq:Delta_t_0} in the first inequality. Since the first $t_0$ rounds allocate $m-n^7$ balls in total, this means that we do not allocate the $m$-th ball before the $s$-th round after $t_0$. Also since \FoldedCaching allocates at each round at least one ball, we allocate the $m$-th ball at some round $r$ in $\mathcal{R} := [\frac{1}{2} n^{3.5}, n^7]$ after $t_0$.

Now applying Markov's inequality we have, for any fixed $r \in \mathcal{R}$,
\[
\Pro{\Phi^r \leq c_6 n^{11} \, \mid \, \mathcal{G}^{t_0} } \geq 1 - n^{-{10}}.
\]
By taking the union bound over the rounds $\mathcal{R}$ (with $|R| < n^7$),
\[
\Pro{\bigcap_{s \in \mathcal{R}} \left\{ \Phi^{t_0 + s} \leq c_6 n^{11} \right\} \mid \,  \mathcal{G}^{t_0} } \geq 1 - n^7 \cdot n^{-{10}} \geq 1 - n^{-3}.
\]
Note that $\Phi^{t_0 + s} \leq c_6 n^{11}$ implies that for some constant $C > 0$, $x_{\max}^{t_0 + s} \leq C \log n$. Using this,
\[
\Pro{\bigcap_{s \in \mathcal{R}}  \left \{ x_{\max}^{t_0 + s} \leq C \log n \right\} \mid \,  \mathcal{G}^{t_0} } \geq 1 - n^{-3}.
\]
Finally, since if $\mathcal{G}^{t_0}$ holds, the round when we throw the $m$-th ball is in $\mathcal{R}$, by the relationship between the \FoldedCaching and \Caching in \cref{lem:memory_compressed} for the last $n^7$ rounds, we conclude from the above and \eqref{eq:lemma_58}, 
\begin{align*}
\Pro{\Gap(m) \leq C \log n } &=
\Pro{\Gap(m) \leq C \log n \, \mid \,  \mathcal{G}^{t_0} } \cdot \Pro{ \mathcal{G}^{t_0} } \\
&\geq (1-n^{-3}) \cdot (1-n^{-9})
\geq 1 - n^{-2}.\qedhere
\end{align*}
\end{proof}

\section{Overview of the Analysis of Non-Filling Processes}\label{sec:analysis_thinning}

The analysis for non-filling processes re-uses some ideas from \cref{sec:analysis_filling}, but is substantially more involved (a diagram summarizing all key steps in the analysis is given in \cref{fig:proof_outline_tikz}). The reason for this complexity comes from the inability to fill ``big holes'', i.e., even if a bin with extremely small load is sampled, only a constant number (e.g.\ two for \Twinning) of balls will be allocated (due to \WThree). This is not only a challenge for the analysis, but it also leads to a completely different behavior of the process, in comparison to filling processes studied in \cref{sec:analysis_filling}.

All of the following analysis in this section (and  \cref{sec:non_filling_analysis_tools,sec:non_filling_potential_functions,sec:non_filling_analysis}) will be for processes which satisfy \PTwo and \WTwo. If the stronger condition (which is the precondition of \cref{thm:main_technical}) that additionally at least one of \PThree or \WThree hold, then we state this explicitly by referring to a $\PThree \cap \WTwo$-process, or, $\PTwo \cap \WThree$-process.

Our analysis will study the interplay between the following potential functions. %
\begin{itemize}\itemsep0pt
    \item The \emph{absolute value potential} (this is also known as the \textit{number of holes} in~\cite{BCSV06}): 
    \[
    \Delta^{t}:=\sum_{i=1}^{n} |y_i^t|.
    \]
    In \cref{lem:good_quantile}, we prove that when $\Delta^t = \Oh(n)$, then \Whp, among next $\Theta(n)$ rounds a constant fraction of them are rounds $t$ with $\delta^t \in (\epsilon, 1 - \epsilon)$, for some constant $\epsilon > 0$.
    
    \item The \emph{quadratic potential}: 
    \[
    \Upsilon^t:=\sum_{i=1}^{n} \left( y_i^t \right)^2.
    \]
    In \cref{lem:quadratic_absolute_relation_for_w_plus_w_minus}, we prove that for processes satisfying \WThree or \PThree, if for some round $t$, $\Delta^t = \Omega(n)$ holds, then $\Upsilon^t$ decreases in expectation. 

    \item The exponential potential function for a constant $\alpha > 0$ to be specified later: 
    \[
    \Lambda^t:= \sum_{i=1}^{n} \exp\left( \alpha \cdot |y_i^t | \right) =
    \sum_{i \in B_+^t} \exp\left( \alpha y_i^{t} \right) + \sum_{i \in B_{-}^t} \exp\left( -\alpha (-y_i^{t}) \right).
    \]
    In \cref{cor:change_for_large_lambda}, we prove that for $\Lambda^t = \Omega(n)$, if $\delta^t \in (\epsilon, 1 - \epsilon)$, then $\Lambda^t$ decreases in expectation, otherwise it increases by a smaller factor (\cref{cor:change_for_large_lambda}). Note that unlike the analysis of the $(1+\beta)$-process in~\cite{PTW15} and unlike one batch of $n$ balls in filling processes, the potential can increase in expectation over a single round, even when it is large (\cref{clm:bad_configuration_lambda}).
    
    \item A ``weaker'' instance of the $\Lambda$ potential function where $\tilde{\alpha} = \Theta(1/n)$: 
    \[
    V^t := \sum_{i=1}^{n} \exp\left( \tilde{\alpha} \cdot |y_i^t | \right) = \sum_{i \in B_+^t} \exp\left( \tilde{\alpha} y_i^t  \right) + \sum_{i \in B_{-}^t} \exp\left( -\tilde{\alpha} (-y_i^t) \right).
    \]
    For this potential, we can directly establish that $\ex{V^t} = \poly(n)$ at an arbitrary round $t$, without considering the quantile $\delta^t$. Then, using Markov's inequality we establish \Whp~that $\Gap(t) = \Oh(n \log n)$. Note that this is similar to the case $\beta = \Theta(1/n)$ in the $(1+\beta)$ process~\cite{PTW15}. We use this as the base case of the heavily-loaded case.
\end{itemize}

Having defined the potential functions, let us now describe in more detail why the exponential potential function $\Lambda$ may increase in expectation in some of the rounds.

While for filling processes, from {\em any} load configuration, the exponential potential function goes down every $\Oh(n)$ rounds (unless it is already small), this is not true for, e.g., the \Thinning process. There exist bad configurations where $(i)$ $\Gap(t)$ is large and $(ii)$ for all $s \in [t,t+\omega(n)]$ rounds $\delta^{s}=1-o(1)$ (or $\delta^{s}=o(1)$), where $\delta^s$ is the quantile of the mean load at round $t$. Note that if $\delta^s$ is too close to $1$ (or $0$), then the bias away from any \emph{fixed} overloaded (or towards any \emph{fixed} underloaded) bin is too small, and the process allocates balls almost uniformly, similarly to \OneChoice. As a result, the exponential potential may increase for several rounds, until $\delta^s$ is bounded away from $0$ and $1$ (see \cref{fig:experiments} for an illustration showing experimental results and \cref{clm:bad_configuration_lambda} for a concrete example).

So, instead we start with the weaker exponential potential function $V$ where $\tilde{\alpha} = \Theta(1/n)$. Because non-filling processes have a small bias to place away from the maximum load, we are able to prove in \cref{sec:v_potential} that when $V$ is sufficiently large it decreases in expectation. This allows us to prove that $\ex{V^t} = \poly(n)$ at an arbitrary round $t$ and infer, using Markov's inequality, that the gap is \Whp $\Gap(t) = \Oh(n \log n)$ (\cref{lem:initial_gap_nlogn}). Then, our next goal is to show that starting with $\Gap(t_0) = \Oh(n \log n)$ we reach $s \in [t_0, t_0 + \Theta(n^3 \log^4 n)]$ where $\Lambda^s = \Oh(n)$ (the \textit{recovery phase} -- \cref{sec:recovery}) and finally show that the gap remains $\Oh(\log n)$ for steps $[s, m]$ (the \textit{stabilization phase} -- \cref{sec:stabilization}).

\begin{figure}[h]
    \centering
    \includegraphics[scale=0.7]{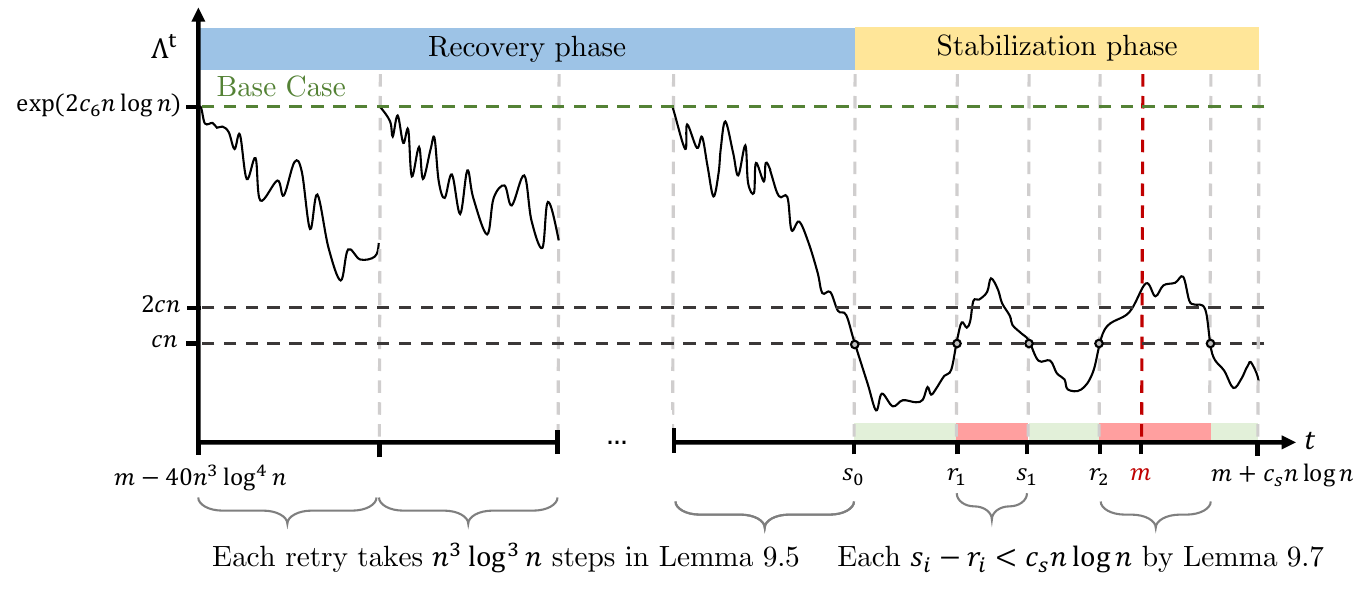}
    \caption{Overview of the recovery phase and stabilization phase. In the recovery phase starting with $V^t = \poly(n)$ (and so $\Gap(t) = \Oh(n \log n)$), we will perform several retries, each consisting of $n^3 \log^3 n$ consecutive steps until we find a round $s_0$ with $\Lambda^{s_0} < c n$. Each retry is successful with constant probability $>0$, such after at most $\Oh(\log n)$ retries we will have \Whp one success. We then switch to the stabilization phase, where in \cref{lem:stabilization} we prove that every $\C{stab_time} n \log n$ rounds $\Lambda^{s} < cn$ is satisfied, which allows us to infer that $\Gap(m)=\Oh(\log n)$.}
    \label{fig:recovery_stabilisation}
\end{figure}

In both the recovery and the stabilization phase, we will study the evolution of $\delta^t$ and prove that from any load configuration, the process eventually reaches a value $\delta^t$ in $(\epsilon,1-\epsilon)$ for some constant $\epsilon > 0$, sufficiently often. The next lemma provides a sufficient condition for the process at round $t_1$ to reach a phase $[t_1,t_1+\Theta(n)$] where $\delta^t \in (\epsilon,1-\epsilon)$ occurs often: 

\def\goodquantile{
Consider any allocation process satisfying \PTwo and \WTwo. Then for any integer constant $\C{small_delta} \geq 1$, there exists some $\epsilon:=\epsilon( \C{small_delta}) > 0$ such that for any integers $t_0\geq 0$ and $t_3:=t_0+\left\lceil \frac{2\C{small_delta} n }{w_{+}} \right\rceil + \left\lceil \frac{n}{w_{+}}  \right\rceil + \left\lceil \frac{n}{10w_{-}}\right\rceil$ we have  
\[
 \Pro{ \left| \left\{ t \in [t_0, t_3] \colon 
 \delta^{t} \in (\epsilon,1-\epsilon) \right\} \right| \geq \eps \cdot n ~ \Big| ~ \mathfrak{F}^{t_0}, \Delta^{t_0} \leq C \cdot n } \geq 1-e^{-\eps \cdot n}.
\]  }

\begin{lem}[Mean Quantile Stabilization]\label{lem:good_quantile}
\goodquantile
\end{lem}

Thus, whenever the absolute value potential $\Delta^{t_1}$ is at most linear, the quantile $\delta^t$ is `good' (bounded away from $0$ or $1$) in a constant fraction of the next $\Theta(n)$ rounds.
The next step in the proof is to establish that the sufficient condition on $\Delta^{t_1}$ will be satisfied. To this end, we use a relation between the absolute value potential $\Delta^{t}$ and the quadratic potential $\Upsilon^t$, showing that $\Upsilon^{t}$ drops in expectation as long as $\Delta^t=\Omega(n)$. Thus, $\Delta^t$ must eventually become linear.
\def\quadraticabsoluterelation{Consider any allocation process satisfying \WTwo and \PTwo. Then for any $t \geq 0$, the quadratic potential satisfies
\[
\Ex{\Upsilon^{t+1} \mid \mathfrak{F}^t} \leq \Upsilon^t  - (p_-^t \cdot w_- - p_+^t \cdot w_+) \cdot \Delta^t + 4 \cdot (w_-)^2.
\]
Hence for any $\PThree \cap \WTwo$-process or $\PTwo \cap \WThree$-process, this implies by \cref{lem:additive_drift} that there exist  constants $\C{quad_delta_drop}, \C{quad_const_add} > 0$ such that for any $t \geq 0$,
\[
\Ex{\Upsilon^{t+1} \mid \mathfrak{F}^t} \leq \Upsilon^t  - \frac{\C{quad_delta_drop}}{n} \cdot \Delta^t + \C{quad_const_add}.
\]}

\begin{lem}\label{lem:quadratic_absolute_relation_for_w_plus_w_minus}
\quadraticabsoluterelation
\end{lem}

Combining the above two lemmas, we prove that for a constant fraction rounds the mean quantile $\delta^t$ is good. In particular, for the recovery phase, we show this guarantee holds with constant probability for an interval of length $\Theta(n \log^3 n)$ (\cref{lem:average_expected_delta_is_small}) and for the stabilization phase, we prove that it holds \Whp\ for an interval of length $\Theta(n \log n)$ given that we start with $\Lambda^s = \Oh(n)$  (\cref{lem:stabilisation_many_good_quantiles_whp}). %
In these good rounds, similarly to \cref{lem:filling_good_quantile}, we prove that the exponential potential $\Lambda^t$ decreases by $(1-\C{good_quantile_mult}\alpha / n)$ (see \cref{lem:good_quantile_good_decrease}). In other rounds, the potential $\Lambda^t$ increases by $(1+c_5 \alpha^2 /n)$ (see \cref{lem:bad_quantile_increase_bound}). Combining these, we obtain that the exponential potential function eventually becomes $\Oh(n)$, which implies a logarithmic gap. For the recovery phase, we use a ``retry'' argument to amplify the probability from constant to $1 - n^{-\Omega(1)}$. We refer to \cref{fig:recovery_stabilisation} for a high-level overview of recovery and stabilization, as well as \cref{fig:proof_outline_tikz} for a diagram summarizing most of the crucial lemmas used in the analysis and outlining their relationship.

 \newcommand{\wcp}{w.c.p.\xspace}

\begin{figure}[H]
\begin{center}
\makebox[\textwidth][c]{
\scalebox{0.91}{
\begin{tikzpicture}[
  txt/.style={anchor=west,inner sep=0pt},
  Dot/.style={circle,fill,inner sep=1.25pt},
  implies/.style={-latex,dashed}]

\def\betaO{0}
\def\betaA{1}
\def\betaB{5}
\def\betaC{12.5}
\def\End{14.5}
\def\yO{-9.5}

\draw[dashed] (\betaA, 1.0) -- (\betaA, \yO);
\node[anchor=north] at (\betaA, \yO) {$m - 40 \cdot n^3 \log^4 n$};
\node[anchor=north] at (\betaB, \yO) {};
\draw[dashed] (\betaC, -8) -- (\betaC, \yO);
\node[anchor=north] at (\betaC, \yO) {$\phantom{\log n}m\phantom{\log n}$};
\draw[dashed] (\End, 1.0) -- (\End, \yO);
\node[anchor=north] at (\End, \yO) {$m + c_s n \log n$};
\draw[dashed] (\betaC, -8) -- (\betaC, \yO);
\node[anchor=north] at (10.5, \yO) {$m - c_s n \log n$};
\draw[dashed] (10.5, -8) -- (10.5, \yO);
\draw[|-|, thick] (\betaA, 1.0) -- (\betaC, 1.0) ;

\node (step1) at (\betaA,1.5) {};
\node[txt] at (step1) {$\Ex{V^t} = \Oh(n^4)$ for any $t \in [m -40 \cdot n^3 \log^4 n, m]$ (\cref{clm:v_expected_value})};
\draw[|-|, thick] (\betaA, 1.0) -- (\betaC, 1.0) ;
\draw[-, dashed] (\betaA, 1.0) -- (\End, 1.0) ;

\node (step2) at (\betaA,0.5) {};
\node[txt] at (step2) {$\Lambda^t \leq \exp(2 c_6 n \log n)$ \Whp for any $t \in [m -40 \cdot n^3 \log^4 n, m]$ (\cref{lem:initial_gap_nlogn})};
\draw[|-|, thick] (\betaA, 0) -- (\betaC, 0) ;
\draw[-, dashed] (\betaA, 0) -- (\End, 0) ;

\node (step3) at (\betaA + 0.5,-0.5) {};
\node[txt] at (step3) {For any interval $[s, s + n^3 \log^3 n]$ for $s \in [m - 40 \cdot n^3 \log^4 n, m]$ with $\Lambda^s \leq \exp(2 c_6 n \log n)$};
\node[txt] at (\betaA + 1.5,-1.0) {$\tilde{G}_{s}^{s + n^3 \log^3 n} \geq \tilde{r} \cdot n^3 \log^3 n$ \wcp, i.e., const fraction of rounds $t$ with $\Delta^t \leq Cn$};
\node (step3H1) at (\betaA + 1.2,-1.0) {};
\node[txt] at (\betaA + 1.5,-1.5) {$G_{s}^{s + n^3 \log^3 n} \geq r \cdot n^3 \log^3 n$ \wcp, i.e., const fraction of rounds $t$ with $\delta^t \in (\eps, 1 - \eps)$};
\node[txt] at (\betaA + 13.2, -2.0) {(\cref{lem:newcorrespondence})};
\node (step3H2) at (\betaA + 1.4,-1.65) {};
\draw[-, dashed] (\betaA, -2.0) -- (\End, -2.0) ;
\draw[|-|, thick] (\betaA + 1.5, -2.0) -- (\betaA + 5, -2.0) ;

\node (step4) at (\betaA + 1.5,-2.5) {};
\node[txt] at (\betaA + 1.5,-2.5) {$\exists \tilde{s} \in [s, s + n^3 \log^3 n]: \Lambda^{\tilde{s}} < cn$, \wcp (\cref{clm:final_claim})};
\node[Dot] at (\betaA + 2.5,-3.0){};
\draw[-, dashed] (\betaA, -3.0) -- (\End, -3.0) ;
\draw[|-|, thick] (\betaA + 1.5, -3.0) -- (\betaA + 5, -3.0) ;

\node (step5) at (\betaB - 1,-3.5) {};
\node[txt] at (\betaB - 1, -3.5) {$\exists s \in [m - 40 \cdot n^3 \log^4 n, m] : \Lambda^s < cn$ \Whp ({\small \textbf{Recovery \cref{lem:recovery}}})};
\node[Dot] at (\betaB,-4.0){};
\draw[|-|, thick] (\betaA, -4.0) -- (\betaC, -4.0) ;

\node (step6) at (\betaB + 0.5,-4.25) {};
\node[txt] at (step6) {For any interval $[t, t + c_s n \log n]$ for $t \in [s, m]$ with $\Lambda^t < cn$};
\node (step6H1) at (\betaB + 1,-4.8) {};
\node[txt] at (\betaB + 1,-5.0) {$\tilde{G}_{t}^{t + c_s n \log n} \geq \tilde{r} \cdot c_s n \log n$ \Whp};
\node (step6H2) at (\betaB + 1,-5.5) {};
\node[txt] at (\betaB + 1,-5.5) {$G_{t}^{t + c_s n \log n} \geq r \cdot c_s n \log n$ \Whp (\cref{lem:stabilisation_many_good_quantiles_whp})};
\draw[-, dashed] (\betaA, -6.0) -- (\End, -6.0) ;
\draw[|-|, thick] (\betaB + 1, -6.0) -- (\betaB + 3.5, -6.0) ;

\node (step7) at (\betaB + 1,-6.5) {};
\node[txt] at (\betaB + 1, -6.5) {$\exists \tilde{s} \in [s, s + c_s n \log n]: \Lambda^{\tilde{s}} < cn$ \Whp ({\small \textbf{Stabilization Lemma~\ref{lem:stabilization}}})};
\node[Dot] at (\betaB + 3,-7.0){};
\draw[-, dashed] (\betaA, -7.0) -- (\End, -7.0) ;
\draw[|-|, thick] (\betaB + 1, -7.0) -- (\betaB + 3.5, -7.0) ;

\node (step8) at (\betaB + 1,-7.5) {};
\node[txt] at (\betaB + 1, -7.5) {$\forall s: \exists \tilde{s} \in [s, s + c_s n \log n]: \Lambda^{\tilde{s}} < cn$ \Whp ({\small\cref{lem:nonfilling_good_gap_after_good_lambda}})};
\draw[-, dashed] (\betaA, -8.0) -- (\End, -8.0) ;
\draw[|-|, thick] (\betaB + 1, -8.0) -- (\End, -8.0) ;

\node (step9) at (\betaC - 3,-8.5) {};
\node[txt] at (\betaC - 3, -8.5) {$\forall i \in [n] : |y_i^m| < \kappa \log n$ \Whp ({\small \cref{thm:main_technical}})};
\node[Dot] at (\betaC,-9.0){};

\draw[implies] (step1) edge[bend right=70] (step2);
\node[anchor=east, black!50!blue] at (\betaA - 0.5, 1) {\small Markov};

\draw[implies] (step2) edge[bend right=70] (step3H1);
\node[anchor=east, black!50!blue] at (\betaA+0.4, -0.6) {\small \cref{lem:quadratic_absolute_relation_for_w_plus_w_minus} \&};
\node[anchor=east, black!50!blue] at (\betaA + 0.2, -0.95) {\small Markov};

\draw[implies] (step3H1) edge[bend right=70] (step3H2);
\node[anchor=east, black!50!blue] at (\betaA + 0.9, -1.5) {\small Lemma~\ref{lem:good_quantile}};

\draw[implies] (step3H2) edge[bend right=70] (step4);
\node[anchor=east, black!50!blue] at (\betaA + 1.1, -2.3) {\small Lemma~\ref{lem:gamma_tilde_is_supermartingale}};

\draw[implies] (step4) edge[bend right=70] (step5);
\node[anchor=east, black!50!blue] at (\betaA + 2.0, -3.5) {\small Using $40 \log n$ retries};

\draw[implies] (step6) edge[bend right=70] (step6H1);

\draw[implies] (step6H1) edge[bend right=70] (step6H2);
\node[anchor=east, black!50!blue] at (\betaA + 4.25, -4.55) {\small \cref{lem:quadratic_absolute_relation_for_w_plus_w_minus} \& Azuma};

\draw[implies] (step6H2) edge[bend right=70] (step7);
\node[anchor=east, black!50!blue] at (\betaA + 4.45, -5.25) {\small Lemma~\ref{lem:good_quantile}};

\node[anchor=east, black!50!blue] at (\betaA + 4.45, -6.25) {\small Lemma~\ref{lem:gamma_tilde_is_supermartingale}};

\draw[implies] (step7) edge[bend right=70] (step8);
\node[anchor=east, black!50!blue] at (\betaA + 4.45, -7.3) {\small U.~Bound};

\draw[implies] (step8) edge[bend right=70] (step9);
\draw[-latex, thick] (\betaA - 0.5, \yO) -- (\End + .8, \yO);
\end{tikzpicture}}}
\end{center}
\caption{Summary of the key steps in the proof for \cref{thm:main_technical}.}
\label{fig:proof_outline_tikz}
\end{figure}
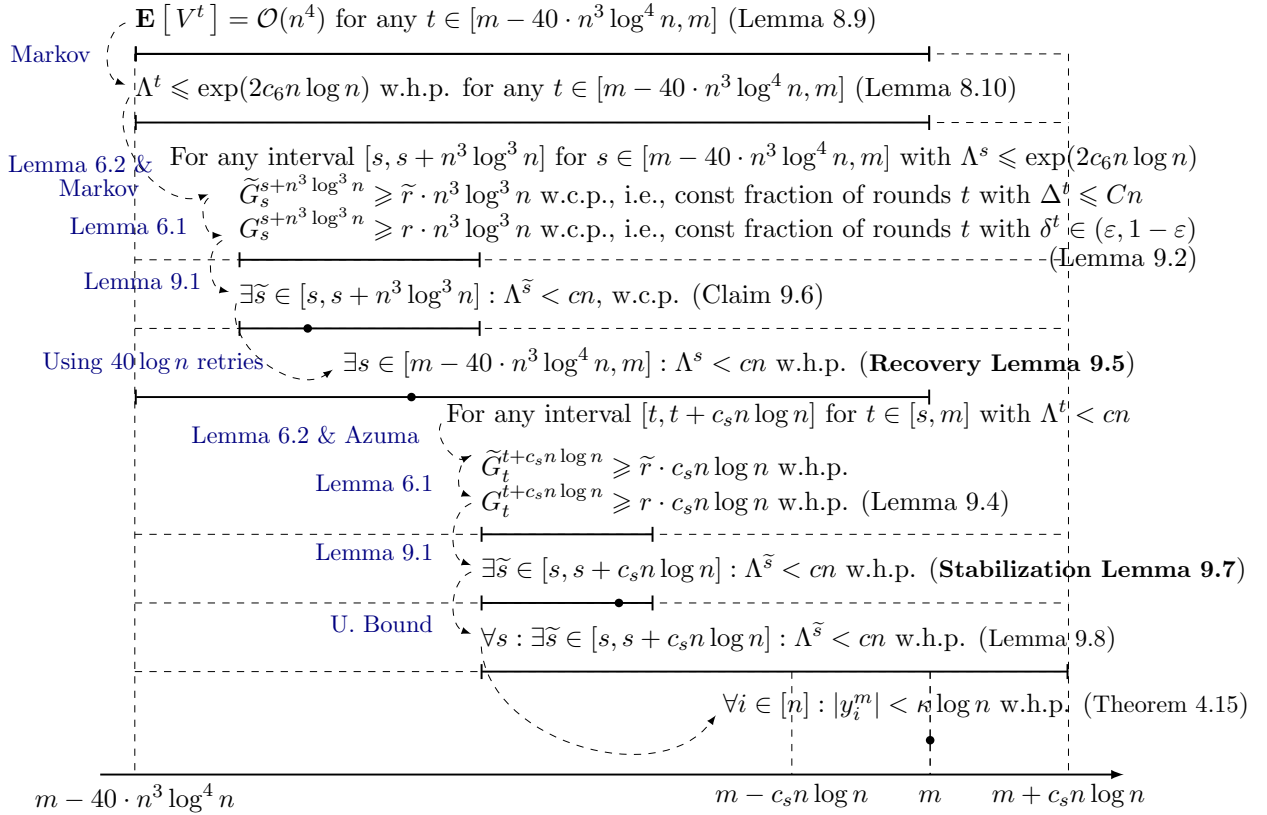
 
\begin{figure}

 \begin{center}
\includegraphics[scale=0.75]{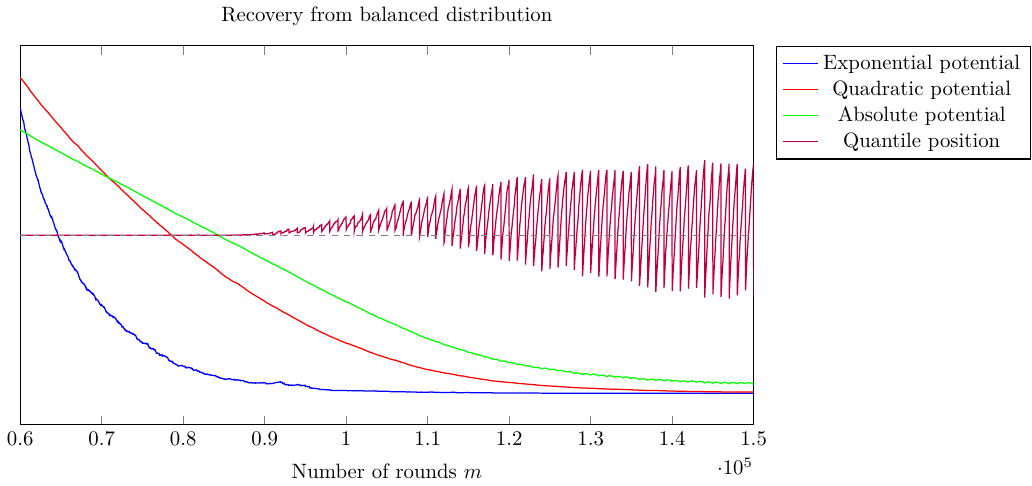}%

~\vspace{.7em}~

\includegraphics[scale=0.75]{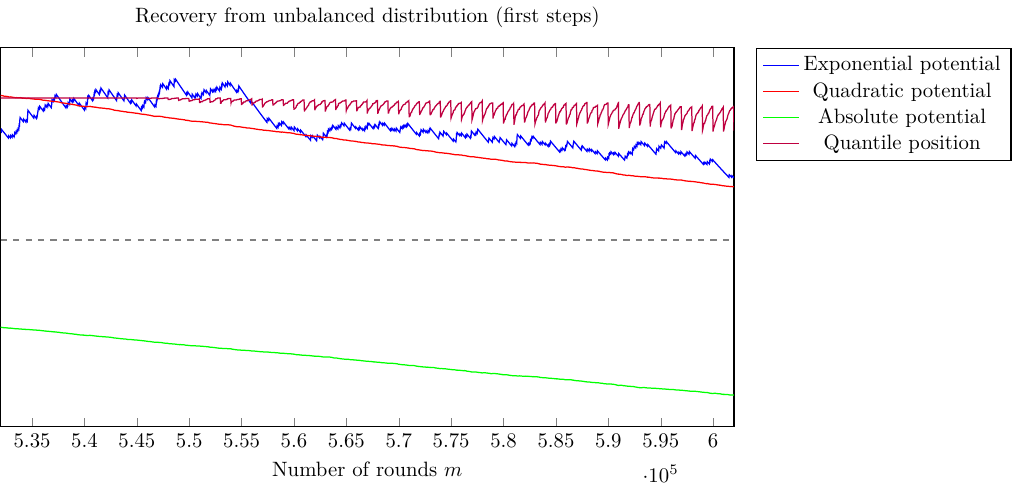}

~\vspace{.em}~

\includegraphics[scale=0.75]{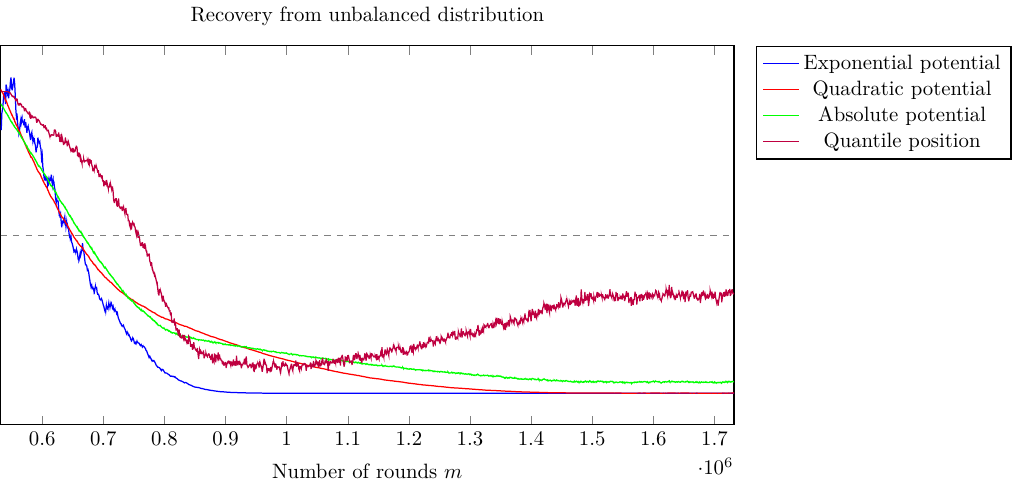}

\caption{Scaled versions of the potential functions for \MeanThinning with different initial load distributions for $n = 1000$. The first example above starts from a load distribution where half of the bins have (normalized) load $+\log n$ and the other half have load $-\log n$. 
The second example starts from a load distribution such that the quantile of the average load remains very close to $1$ for  $\omega(n \log n)$ many rounds. As it can be seen, the exponential potential increases a bit at the beginning, but both the absolute and quadratic potential all improve immediately. Once the quantile is sufficiently bounded away from $0$ and $1$, the exponential potential also decreases, eventually stabilizing at $\Oh(n)$, as shown in the third figure. Note that the dashed gray line corresponds a perfectly balanced mean quantile, i.e., $\delta^t=1/2$.
} \label{fig:experiments}
\end{center}
\end{figure}

There are a considerable number of constants in this paper and their interdependence can be quite complex. We hope the following remark will shed some light on their respective roles. 

\begin{rem}[Relationship Between the Constants] All the constants used in the analysis depend on the constants $w_-, w_+$ (and \C{p2k1}, \C{p2k2} if \PThree holds) of the process being analyzed. The relation between the absolute value and quadratic potential functions (and specifically \C{quad_delta_drop} and \C{quad_const_add} in \cref{lem:quadratic_absolute_relation_for_w_plus_w_minus}) defines what it means for $\Delta^t$ to be small, i.e., $\Delta^t \leq \C{small_delta} \cdot n$. This $\C{small_delta}$ in turn specifies the constant $\eps > 0$, given by \cref{lem:good_quantile}, which defines a good quantile to be $\delta^t \in (\eps, 1 - \eps)$. Then, $\eps$ specifies the fraction $r$ of rounds with $\delta^t \in (\eps, 1 - \eps)$ (in \cref{lem:average_expected_delta_is_small,lem:stabilisation_many_good_quantiles_whp}) and the constants $\C{good_quantile_mult}, \C{bad_quantile_mult}$ in the inequalities for the expected change of $\Lambda^t$ for sufficiently small constant $\alpha$ (in Lemmas~\ref{lem:good_quantile_good_decrease} and~\ref{lem:bad_quantile_increase_bound}). Then, using the constants $\eps$, $r$ and $\C{good_quantile_mult}, \C{bad_quantile_mult}$ we finally set the value of $\alpha$ in $\Lambda$. Finally, $\eps$ and $\alpha$ specify $\C{lambda_bound}$ in \cref{cor:change_for_large_lambda} which we use in proving $\Lambda^t < 2cn$.%
\end{rem}

\bigskip

\noindent \textbf{Organization of the Remaining Part of the Proof.}

\begin{itemize}[itemsep=-1pt, topsep=-1pt]
    \item In \cref{sec:non_filling_analysis_tools}, we prove the mean quantile stabilization lemma (\cref{lem:good_quantile}) and the relation between the absolute value and the quadratic potential (\cref{lem:quadratic_absolute_relation_for_w_plus_w_minus}).
    \item In \cref{sec:non_filling_potential_functions}, we prove simple relations between the quadratic and the exponential potential (\cref{sec:quadratic_and_exp_potentials}), analyze the expected change of the exponential potential $\Lambda$ (\cref{sec:non_filling_lambda}) and prove that $\Ex{V^t} = \poly(n)$ (\cref{sec:v_potential}), deducing an $\Oh(n \log n)$ gap. 
    \item In \cref{sec:non_filling_analysis}, we complete the proof for the $\Oh(\log n)$ gap for non-filling processes by analyzing the recovery (\cref{sec:recovery}) and  stabilization (\cref{sec:stabilization}) phases.
\end{itemize}

\section{Mean Quantile Stabilization} \label{sec:non_filling_analysis_tools}

In \cref{sec:mean_quantile_stab}, we prove the stabilization of the mean quantile $\delta^t$. In \cref{sec:quadratic_absolute_potential}, we relate the change of the quadratic potential $\Upsilon^t$ to the absolute value potential $\Delta^t$.

\subsection{Mean Quantile Stabilization based on Absolute Value Potential} \label{sec:mean_quantile_stab}

The next lemma proves that once the absolute value potential is $\Oh(n)$, then the allocation process will satisfy $\delta^{t} \in (\epsilon,1-\epsilon)$ for ``many'' of the following rounds.
The condition on $\delta^{t}$ means that a constant fraction of bins are overloaded and a constant fraction are underloaded. Interestingly, this lemma does not need the stronger property \PThree or \WThree.

{\renewcommand{\thelem}{\ref{lem:good_quantile}}
\begin{lem}[Mean Quantile Stabilization, restated]
\goodquantile
	\end{lem}}
	\addtocounter{lem}{-1}

\begin{proof}

Define $ B_*:=\left\{ i \in [n] \colon |y_i^{t_0}| < 2\C{small_delta} \right\}$ to be the bins whose load deviates from the average $W^{t}/n$ by less than $2\C{small_delta}$. Then, 
\[
 \C{small_delta} \cdot n \geq \Delta^{t_0} = \sum_{i \in [n]} |y_i^{t_0}| \geq  \sum_{i \colon |y_i^{t_0}| \geq 2\C{small_delta}} 2\C{small_delta} = (n-|B_*|) \cdot 2\C{small_delta},
\]
rearranging this gives that conditional on the event $\mathcal{C}:=\{\Delta^{t_0} \leq \C{small_delta}\cdot  n \} $ we have
\begin{equation}\label{eq:sizeofB_*}
 |B_*| \geq \frac{n \cdot \C{small_delta}}{2\C{small_delta}} = \frac{n}{2}.
\end{equation}Note that the bins in $B_*$ may be underloaded or overloaded. 

We now proceed with two claims. Using the fact that $|B_*|$ is large, the first claim (\cref{clm:underloaded}) proves that there exists an ``early'' round $t \in [t_0,t_1]$ such that a constant fraction of bins are underloaded. Similarly, the second claim (\cref{clm:overloaded})  proves that there exists a round $t \in [t_2,t_3]$ with $t_2 \geq t_1$ such that a constant fraction of bins are overloaded. Since the set of overloaded bins can only increase by $1$ per round, we then finally conclude that for $\Omega(n)$ rounds $t$, both conditions hold, i.e., $\delta^t \in (\epsilon,1-\epsilon)$.

\arxive{Interestingly, we seem to run into problems here if the loads were arbitrary fractions. I guess we maybe then would need to use property that the process is not able to allocate to an underloaded bin with probability larger than $2/n$.}
\begin{clm}\label{clm:underloaded}
For any integer constant $\C{small_delta} \geq 1$, there exists a constant $\kappa_1:=\kappa_1(C)>0$ such that for $t_1:=t_0+\left\lceil\frac{2\C{small_delta}}{w_{+}} \cdot n\right\rceil$ we have 
\[
 \Pro{  \bigcup_{ t \in [t_0,t_1] } \left\{\delta^t \geq \kappa_1  \right\}  \Big| ~ \mathfrak{F}^{t_0}, \Delta^{t_0} \leq C \cdot n } \geq 1-e^{-\kappa_1 \cdot n }.
\]
 
\end{clm}
\begin{poc}
If $|B_{-}^{t_0} \cap B_*| \geq 1/2 \cdot |B_*|$, then the statement of the claim follows immediately for $t=t_0$ and $\kappa_1=1/4$.
Otherwise, assume $|B_{+}^{t_0} \cap B_*| \geq 1/2 \cdot |B_*|$, i.e., at least half of the bins in $B_*$ are overloaded at round $t_0$. Note that the bins in $B_{+}^{t_0} \cap B_*$ all have loads in the range $[W^{t_0}/n,\;W^{t_0}/n + 2C)$. Thus, since loads are integers (as $w_{+}$ and $w_{-}$ are integers), there can be at most $2C$ different load levels within $B_{+}^{t_0} \cap B_*$. Hence, by the pigeonhole principle and \eqref{eq:sizeofB_*}, there exists a subset $\tilde{B_*} \subseteq B_{+}^{t_0} \cap B_*$, with $|\tilde{B_*}| \geq |B_{+}^{t_0} \cap B_*| \cdot \frac{1}{2\C{small_delta}} \geq  \frac{1}{8 \C{small_delta}} \cdot n$ such that all bins $i \in \tilde{B_*}$ have the same (non-negative) load at round $t_0$.%

Note that if a bin $i \in  \tilde{B_*}$ was never chosen for an allocation during rounds $[t_0,t_1]$, then its normalized load would satisfy
\[
 y_i^{t_1} \leq y_i^{t_0} - \frac{w_+}{n} \cdot (t_1-t_0)\leq y_i^{t_0} - 2\C{small_delta} < 2\C{small_delta} - 2\C{small_delta}= 0,
\]
and hence the bin would become underloaded at least once before round $t_1 =t_0+\left\lceil 2\C{small_delta} n /w_{+} \right\rceil$.

Since for any round, a fixed overloaded bin is chosen for an allocation with probability at most $p_{+}^t \leq \frac{1}{n}$, we conclude for any $i \in  \tilde{B_*}$, 
\[
 \Pro{ i \in \bigcup_{t \in [t_0,t_1]} B_{-}^t ~ \Big| ~ \mathfrak{F}^{t_0} } \geq \Pro{ \mathsf{Bin}\left( \frac{2\C{small_delta}}{w_{+}} \cdot n , \frac{1}{n} \right) = 0 } =\left(1-\frac{1}{n}\right)^{n\cdot 2C/w_+} \geq 4^{-2C/w_+}  ,
\]
where in the last inequality we used that $\left( 1 - \frac{1}{n} \right)^{n}$ is non-decreasing in $n \geq 2$. Define the constant $\kappa_2 := 4^{-2C/w_+}>0$. Further, let
\[
 Y:= \left| \bigcup_{t \in [t_0,t_1]} B_{-}^t   \cap \tilde{B_*} \right|,
\]
be the number of bins in $\tilde{B_*}$ that become underloaded at least once during the interval $[t_0,t_1]$.

Let us now consider a modified process in which the probability for each overloaded bin to be chosen is not only at most $\frac{1}{n}$, but is instead equal to $\frac{1}{n}$ in each round. Let $\tilde{Y}$ be the number of bins in $\tilde{B_*}$ that are never chosen in the modified process during the interval $[t_0,t_1] =\left[t_0,t_0+\left\lceil 2\C{small_delta} n /w_{+} \right\rceil \right]$. Note that $\tilde{Y}$ is stochastically smaller than $Y$. 

Recall that $\mathcal{C}=\{\Delta^{t_0} \leq \C{small_delta} \cdot n \}$. By linearity of expectation and \eqref{eq:sizeofB_*}, we have
\[
 \ex{\tilde{Y}\mid \mathfrak{F}^{t_0} ,\; \mathcal{C}} \geq \kappa_2 \cdot \Ex{| \tilde{B_*} |\;\big| \; \mathfrak{F}^{t_0} ,\; \mathcal{C} } \geq \kappa_2 \cdot \frac{1}{4\C{small_delta}} \Ex{|  B_*|\;\big| \; \mathfrak{F}^{t_0} ,\;\mathcal{C} } \geq \kappa_2 \cdot \frac{1}{4\C{small_delta}} \cdot \frac{ n }{2 }= \kappa_2 \cdot \frac{n}{8\C{small_delta}}.
\]
In the modified process, changing the bin sample in one round can change $\tilde{Y}$ only by at most one, and hence applying the Method of Bounded Independent Differences~(\cref{mobd}) yields,
\begin{align*}
 \Pro{  \tilde{Y}  \leq \frac{1}{2}  \ex{\tilde{Y}\mid \mathfrak{F}^{t_0} ,\; \mathcal{C}}  ~ \Big| ~ \mathfrak{F}^{t_0} ,\; \mathcal{C} } \leq   \exp\left(- \frac{ \frac{1}{4} \left(\ex{\tilde{Y}\mid \mathfrak{F}^{t_0} ,\; \mathcal{C}}\right)^2 }{\frac{2\C{small_delta}}{w_+} \cdot n \cdot 1^2} \right) \leq  \exp\left(-\frac{\kappa_2^2w_+}{8^3C^3} \cdot n \right),
\end{align*}
and so by stochastic domination and the two equations above we have 
\[
 \Pro{ Y \geq \frac{\kappa_2 n}{ 16C} \;\Big|\; \mathfrak{F}^{t_0},\;  \mathcal{C} } \geq 1-   \exp\left(-\frac{\kappa_2^2w_+}{8^3C^3} \cdot n \right) .
\]
Since all load levels in $\tilde{B_*}$ before round $t_0$ are identical, we conclude that they all become underloaded the first time in the same round, so
\[
 \Pro{ | B_{-}^{t} \cap \tilde{B_*} | \geq \kappa_1 \cdot n \mid \mathfrak{F}^{t_0}, \; \Delta^{t_0} \leq C \cdot n } \geq 1-e^{-\kappa_1 \cdot n },
\]
where $\kappa_1:=\min\left\{\frac{\kappa_2}{16C} ,  \frac{\kappa_2^2w_+}{8^3C^3}\right\}  >0$ and $\kappa_2 = 4^{-2C/w_+}$, as defined above.
\end{poc}

\begin{clm} \label{clm:overloaded}
For any integer constant $\C{small_delta} \geq 1$ there exists a constant $\kappa_3:=\kappa_3(C)>0$ such that for $t_2 := t_1 + \left\lceil  \frac{n}{w^{+}}\right\rceil $, and
$t_3 := t_2 + \left\lceil \frac{n}{10w_{-}}\right\rceil $ we have
\[
\Pro{  \bigcup_{ t \in [t_2,t_3] } \left\{\delta^t \leq 1- \kappa_3 \right\}  \Big| ~ \mathfrak{F}^{t_0}, \Delta^{t_0} \leq C \cdot n } \geq 1-e^{-\kappa_3 \cdot n }.
\]
\end{clm}

\begin{poc}%
	If  $|B_* \cap B_{+}^{t_2}| \geq \frac{1}{2} \cdot |B_*| \geq \frac{n}{4}$, the claim follows immediately, so assume $|B_* \cap B_{-}^{t_2}| \geq \frac{1}{2} \cdot |B_*|$. 
	Let $\tilde{B_*}:= B_* \cap B_{-}^{t_2}$ and consider any bin $i$ in $\tilde{B_*}$. After round $t_2$, the normalized load satisfies 
	\[
	y_i^{t_2} \geq y_{i}^{t_0} - (t_2-t_0) \cdot \frac{w_{-}}{n} \geq
	-2\C{small_delta} - \left(\frac{2C+1}{w_{+}} +\frac{2}{n}\right)\cdot w_{-}  \geq -\left(5 C+\frac{2}{n}\right)\cdot w_{-},
	\]
	where we used $y_i^{t_0} > -2C$ as $i \in \tilde{B_*} \subseteq B_*$, $C\geq 1$, and $w_-\geq w_+\geq 1$ by \WTwo.

		As long as a bin $i$ is underloaded, we choose it for allocation in each round with probability at least $p_{-}^t \geq 1/n$, independently of previous rounds. If the bin has been chosen $6C$ times (while being underloaded), it must become overloaded before round $t_3$ at least once since
	\begin{align*}
	y_i^{t_3} &\geq y_i^{t_2}  - (t_3-t_2)\frac{w_-}{n} + 6C \cdot w_{-} \geq 	-\left(5 C+\frac{2}{n}\right)\cdot w_{-}  - \left(\frac{1}{10}+\frac{w_-}{n}\right) + 6C \cdot w_{-}
	> 0.
	\end{align*} 
 
	Note that the lower bound of $1/n$ on the probability for allocating to any underloaded bin in a single round $t$ holds in any round, regardless of the ball configuration. Hence, for any $i \in  \tilde{B_*}$,
	\begin{align*}
	\Pro{ i \in \bigcup_{t \in [t_2,t_3]} B_{+}^t  ~ \,\Bigg|\, \mathfrak{F}^{t_0}  }
	&\geq \Pro{ \mathsf{Bin}\left( \frac{n}{10w_{-}}, \frac{1}{n}\right) \geq 6C  }\\ &\geq  \binom{\frac{n}{10w_{-}}}{6C} \left(\frac{1}{n} \right)^{6C}\cdot\left(1-\frac{1}{n}  \right)^{\frac{n}{10w_{-}} - 6C}  \\
	&\geq \left(\frac{1}{60Cw_-} \right)^{6C} \cdot e^{-\frac{1}{10w_{-}}}/2,
	\end{align*}
where we used \cref{lem:cheatsheet} and $\binom{n}{k}\geq (\frac{n}{k})^k$ in the last inequality. Define the constant $\kappa_4:=\left( 60Cw_-\right)^{-6C} \cdot e^{-\frac{1}{10w_{-}}}/2>0$ and let  %
	\[
	Z:= \left| \bigcup_{t \in [t_2,t_3]} B_{+}^t \cap \tilde{ B_*} \right|,
	\]
	that is, the number of bins in $\tilde{B_*}$ that become overloaded at least once during the interval $[t_2,t_3]$. Similar to the proof of the previous claim, consider a modified process where each underloaded bin has a probability of exactly $\frac{1}{n}$ to be incremented in each round. Let $\tilde{Z}$ be the number of bins in $\tilde{B_*}$, where $|\tilde{B_*}|\geq n/4$ by \eqref{eq:sizeofB_*} if we condition on $\mathcal{C}=\Delta^{t_0} \leq C \cdot n$, that are chosen at least $6C$ times in this modified process during rounds $[t_2,t_3]$. Note that for any $\mathfrak{F}^{t_0}$, $\tilde{Z}$ is stochastically smaller than $Z$. 
	By linearity of expectations we have  
	\[
 \ex{\tilde{Z}\;\big|\;\mathfrak{F}^{t_0} ,\; \mathcal{C}} \geq \kappa_4 \cdot \Ex{| \tilde{B_*} |\;\Big| \; \mathfrak{F}^{t_0} ,\; \mathcal{C} } \geq \kappa_2 \cdot \frac{1}{2} \cdot \Ex{|  B_*|\;\big| \; \mathfrak{F}^{t_0} ,\;\mathcal{C} } \geq \kappa_2 \cdot \frac{1}{2} \cdot \frac{ n }{2 }= \kappa_2 \cdot \frac{n}{4}.
\]
	In the modified process, changing the bin sample in one round can change $\tilde{Z}$ only by at most one, and hence applying the Method of Bounded Independent Differences~(\cref{mobd}) yields,
	\begin{equation*}
	\Pro{   \tilde{Z} < \frac{1}{2} \cdot  \ex{\tilde{Z}\mid \mathfrak{F}^{t_0} ,\; \mathcal{C} } \;\Bigg| \; \mathfrak{F}^{t_0} ,\; \mathcal{C}} \leq    \exp\left(- \frac{ \frac{1}{4} \left(\ex{\tilde{Z} \mid \mathfrak{F}^{t_0} ,\; \mathcal{C} } \right)^2 }{2\cdot (t_3-t_2)  \cdot 1^2} \right) = \exp\left(-\frac{5\kappa_4^2\cdot w_-  }{2^4} \cdot n  \right).	\end{equation*}
	and by the stochastic domination, we have
	\[
	\Pro{ Z \geq \frac{\kappa_4}{4} \cdot n \,\Big|\, \mathfrak{F}^{t_0},\; \mathcal{C}} \geq 1-\exp\left(-\frac{5\kappa_4^2\cdot w_-  }{2^4} \cdot n  \right).
	\]
We now fix the constant $\kappa_3 := \min\left\{\frac{1}{2}\cdot\frac{\kappa_4}{4}, \;    \frac{5\kappa_4^2\cdot w_-  }{2^4} \right\}$, where $\kappa_4>0$ is given above.

 At this point we see that with probability at least $1-e^{-\kappa_3 n} $ at least $2\kappa_3 \cdot n $ bins become overloaded at least once at some round in $[t_2,t_3]$. However, for each bin, there is the possibility that it is overloaded only for one round before becoming underloaded again. If each bin is overloaded for only one round, and all at different rounds, then it is possible that at no single round in $[t_2,t_3]$ do we have $\Omega(n)$ overloaded bins. We now explain why this is not the case. 
	
	Note that during the interval $[t_2,t_3]$, as $t_3-t_2 = \left\lceil \frac{n}{10 w_-}\right\rceil$, there is at most one round $t$ where the integer parts of the average load changes, that is, there is only at most one round $t\in [t_2,t_3)$ such that $\lfloor W^{t}/  n \rfloor < \lfloor W^{t+1}/n \rfloor$. Hence only at the transition from $t$ to $t+1$ could an overloaded bin become underloaded during the interval $[t_2,t_3]$. Therefore, if $|Z|\geq 2\kappa_3 \cdot n $, we must have
	\[
	| B_{+}^{t} \cap B_* | \geq  \kappa_3\cdot n \qquad \mbox{ or } \qquad
	| B_{+}^{t_3} \cap B_* | \geq \kappa_3 \cdot n.
	\]
	Hence $|Z| \geq 2\kappa_3 \cdot n$, implies that there exists $t \in [t_2,t_3]$ such that $|B_{+}^{t}| \geq \kappa_3 \cdot n$.
\end{poc}

The first claim implies that, w.p.\ $ 1- e^{-\kappa_1 n}$, that the process reaches a round $s_1 \in [t_0,t_1]$ with $\kappa_1 \cdot n$ underloaded bins. The second claim shows, w.p.\ $ 1-e^{-\kappa_3 n}$, that the process reaches a round $s_2 \in [t_1+\lceil \frac{n}{w_{+}}\rceil ,t_3]$ with $ \kappa_3  \cdot n$ overloaded bins. By the union bound, both events occur with probability $ 1- e^{-\kappa_1 n}-e^{-\kappa_3 n}$.
Since $\delta^{t}=|B_{+}^{t}|/n$, we have $\delta^{t+1} \leq \delta^{t} + \frac{1}{n}$,
and in this case applying \cref{lem:smoothness} with $r_0=s_1$, $r_1=s_2$, $f(t)=\delta^{t}$, $\epsilon:= \min \{ \kappa_3, \kappa_1,2/3 \}$ and $\xi =1/n$ gives,
\begin{align*}
\left| \left\{ t \in [t_0,t_3] \colon \delta^t \in (\epsilon/2,1-\epsilon/2) \right\} \right| &\geq
 \left| \left\{ t \in [s_1,s_2] \colon \delta^t \in (\epsilon/2,1-\epsilon/2) \right\} \right| \\ &\geq \min\{ \epsilon/2 \cdot n, s_2 - s_1  \} \\
 &\geq \min\{ \epsilon/2 \cdot n, \frac{n}{w^+} \} .
\end{align*}
 Thus, taking $\eps$ in the statement to be $\min\{ \min\{\kappa_3, \kappa_1,2/3 \}/2, 1/w_+\} $ gives the result.   
\end{proof}

\subsection{Quadratic and Absolute Value Potential Functions} \label{sec:quadratic_absolute_potential}
{\renewcommand{\thelem}{\ref{lem:quadratic_absolute_relation_for_w_plus_w_minus}}
\begin{lem}[restated]

\quadraticabsoluterelation
	\end{lem}}
	\addtocounter{lem}{-1}
\begin{proof}

We begin by decomposing the potential $\Upsilon$ over the $n$ bins, using $\Upsilon_i^{t+1}:=(y_i^{t+1})^2$:
\[
 \Upsilon^{t+1} = \sum_{i=1}^n \Upsilon_i^{t+1} = \sum_{i=1}^n  (y_i^{t+1})^2.
\]
We shall analyze each term $\Upsilon_{i}^{t+1}$ separately in two cases depending on the load of the bin $i$. By the \WTwo assumption, if the chosen bin is overloaded, we place a ball of weight $w_+$, otherwise we place a ball of weight $w_-$. Also, recall that $P_+^t := \sum_{i \in B_+^t} p_i^t$ and $P_-^t := \sum_{i \in B_-^t} p_i^t$.

\medskip 

\noindent\textbf{Case 1} [$i \in B_{+}^t$]. If $i$ is assigned to an overloaded bin,
\begin{align*}
\Ex{\Upsilon_{i}^{t+1} \mid \mathfrak{F}^t } 
 & = \underbrace{p_i^t \cdot \Big(y_i^t + w_+ - \frac{w_+}{n}\Big)^2}_{\text{placing a ball in $i$}} + \underbrace{(P_+^t - p_i^t) \cdot \Big(y_i^t - \frac{w_+}{n}\Big)^2}_{\text{placing a ball in $B_+^t \setminus \lbrace i \rbrace$}} + \underbrace{P_-^t \cdot \Big(y_i^t - \frac{w_-}{n}\Big)^2}_{\text{placing a ball in $B_-^t$}}.
\end{align*}
Expanding out the squares,
\begin{align*}
\Ex{\Upsilon_{i}^{t+1} \mid \mathfrak{F}^t }
 & = (y_i^t)^2 + 2y_i^t \cdot \left(p_i^t \cdot \Big(w_+ - \frac{w_+}{n}\Big) - (P_+^t - p_i^t) \cdot \frac{w_+}{n} - P_-^t \cdot \frac{w_-}{n} \right) \\ 
 & \qquad + p_i^t \cdot \Big(w_+ - \frac{w_+}{n} \Big)^2 + (P_+^t - p_i^t) \cdot \frac{(w_+)^2}{n^2} +  P_-^t \cdot \frac{(w_-)^2}{n^2} \\
 & = (y_i^t)^2 + 2y_i^t \cdot \Big(p_i^t \cdot w_+ - P_+^t \cdot \frac{w_+}{n} - P_-^t \cdot \frac{w_-}{n}\Big) \\ 
 & \qquad + p_i^{t} \cdot (w_{+})^2 -2 \cdot p_i^t \cdot \frac{(w_+)^2}{n} + P_{+}^t \cdot \frac{(w_+)^2}{n^2} + P_{-}^t \cdot \frac{(w_-)^2}{n^2} \\
 & \leq (y_i^t)^2 + 2y_i^t \cdot \Big(p_i^t \cdot w_+ - P_+^t \cdot \frac{w_+}{n} - P_-^t \cdot \frac{w_-}{n}\Big) + 2 \cdot \frac{(w_{-})^2}{n},
\end{align*}
where in the last step we used that $p_{i}^t \leq p_+^t \leq 1/n$ by \PTwo, which implies $p_i^{t} \cdot (w_{+})^2 + P_{+}^t \cdot \frac{(w_+)^2}{n^2} + P_{-}^t \cdot \frac{(w_-)^2}{n^2} \leq 2 \cdot \frac{(w_{-})^2}{n}$.

\medskip

\noindent\textbf{Case 2} [$i \in B_{-}^t$]. If $i$ is assigned to an underloaded bin then 
\begin{align*}
\Ex{\Upsilon_{i}^{t+1} \mid \mathfrak{F}^t } 
 & = \underbrace{p_i^t \cdot \Big(y_i^t + w_- - \frac{w_-}{n}\Big)^2}_{\text{placing a ball in $i$}} + \underbrace{(P_-^t - p_i^t) \cdot \Big(y_i^t - \frac{w_-}{n}\Big)^2}_{\text{placing a ball in $i \in B_-^t \setminus \lbrace i \rbrace$}} + \underbrace{P_+^t \cdot \Big(y_i^t - \frac{w_+}{n}\Big)^2}_{\text{placing a ball in $i \in B_+^t$}}.
\end{align*}
Expanding out the squares,
\begin{align*}
\Ex{\Upsilon_{i}^{t+1} \mid \mathfrak{F}^t } 
 & = (y_i^t)^2 + 2y_i^t \cdot \Big(p_i^t \cdot \Big(w_- - \frac{w_-}{n}\Big) - (P_-^t - p_i^t) \cdot \frac{w_-}{n} - P_+^t \cdot \frac{w_+}{n}\Big) \\
 & \qquad +p_i^t \cdot \Big(w_- - \frac{w_-}{n} \Big)^2 + (P_-^t - p_i^t) \cdot \frac{(w_-)^2}{n^2} + P_+^t \cdot \frac{(w_+)^2}{n^2} \\
 & = (y_i^t)^2 + 2y_i^t \cdot \Big(p_i^t \cdot w_- - P_-^t \cdot \frac{w_-}{n} - P_+^t \cdot \frac{w_+}{n} \Big) \\
 & \qquad +p_i^t \cdot (w_-)^2 - 2 \cdot p_i^t \cdot \frac{(w_-)^2}{n} + P_-^t \cdot \frac{(w_-)^2}{n^2} + P_+^t \cdot \frac{(w_+)^2}{n^2} \\
 & \leq (y_i^t)^2 + 2y_i^t \cdot \Big(p_i^t \cdot w_- - P_-^t \cdot \frac{w_-}{n} - P_+^t \cdot \frac{w_+}{n} \Big) + 2\cdot p_i^t \cdot (w_-)^2,
\end{align*}
where in the last step we used that $p_i^{t} \geq p_{-}^t \geq 1/n$ by \PTwo, which implies
$p_i^{t} \cdot (w_{-})^2 + P_{-}^t \cdot \frac{(w_-)^2}{n^2} + P_{+}^t \cdot \frac{(w_+)^2}{n^2} \leq 2 \cdot p_i^t \cdot (w_{-})^2$.

\medskip 

Combining the two inequalities for the two cases and since $\sum_{i \in B_+^t} y_i^t = - \sum_{i \in B_-^t} y_i^t = \frac{1}{2} \cdot \Delta^t$ we get,
\begin{align*}
\Ex{\Upsilon^{t+1} \mid \mathfrak{F}^t} 
&= \sum_{i \in B_+^t} \Ex{\Upsilon_{i}^{t+1} \mid \mathfrak{F}^t} + \sum_{i \in B_-^t} \Ex{\Upsilon_{i}^{t+1} \mid \mathfrak{F}^t} \\
 & \leq 
  \Upsilon^t  + \sum_{i \in B_+^t} 2y_i^t p_i^t \cdot w_+ + \sum_{i \in B_-^t} 2y_i^t p_i^t \cdot w_- + \sum_{i \in B_+^t} 2 \cdot \frac{(w_-)^2}{n} + \sum_{i \in B_-^t} 2\cdot p_i^t  \cdot (w_-)^2\\
  &\leq \Upsilon^t  - ( p_-^t \cdot w_- - p_+^t \cdot w_+) \cdot \Delta^t + 4 \cdot (w_-)^2 , 
\end{align*}
where the last inequality follows since $p_i^t \leq p_+^t$ for $i \in B_+^t$, and $p_i^t \geq p_-^t$ for $i \in B_-^t$.
\end{proof}

\section{Potential Function Inequalities} \label{sec:non_filling_potential_functions}

In this section we derive several inequalities involving potential functions. Most of the effort goes into establishing a drop of the exponential potential function for some suitable choices of $\alpha$, which in turn depends on the constants defined by the process. One of the main insights is \cref{cor:change_for_large_lambda}, which establishes: $(i)$ a significant drop of the exponential potential function if the quantile satisfies $\delta^t \in (\epsilon,1-\epsilon)$, and $(ii)$ a not too large increase of the exponential potential function for any quantile.

On a high level the analysis follows relatively standard estimates and bears resemblance to the one in \cite{PTW15}. Even though some extra care is needed due to the more general allocation process, a reader may wish to skip this section (or the proofs) and continue with the proof in \cref{sec:non_filling_analysis}. 

\subsection{Quadratic and Exponential Potential Functions} \label{sec:quadratic_and_exp_potentials}

The next lemma bounds the quadratic potential in terms of the exponential potential.%
\begin{lem} \label{clm:bound_on_gamma_implies_bound_on_upsilon}\label{clm:bound_on_gamma_implies_bound_on_upsilon_2}
Consider any $\PThree \cap \WTwo$-process or $\PTwo \cap \WThree$-process. For any $0 < \alpha < 1$ . Then for any $t \geq 0$,
\[
\Upsilon^t \leq \alpha^{-2} \cdot n \cdot ( \log \Lambda^t )^2, \qquad \text{and}\qquad \Upsilon^t \leq \left(\frac{4}{\alpha} \cdot \log\left(\frac{4}{\alpha}\right)\right)^2\cdot \Lambda^t.
\]
\end{lem}
\begin{proof}
Note that by the definition of $\Lambda^t $, if $\Lambda^t \leq \lambda$ then  $y_i^t \leq \frac{1}{\alpha} \cdot \log \lambda$ and $-y_i^t \leq \frac{1}{\alpha} \cdot \log \lambda$ for all $i\in[n]$. Hence, $(y_i^t)^2 \leq \frac{1}{\alpha^2} \cdot (\log \lambda)^2$ for all $i\in[n]$, which proves the first statement by aggregating over all bins.

For the second statement let $\kappa := ((4/\alpha) \cdot \log(4/\alpha))^2$. Note that $e^y \geq y$ (for any $y \geq 0$) and hence for any $y \geq (4/\alpha) \cdot \log(4/\alpha)$,
\[
e^{\alpha y/2} = e^{\alpha y/4} \cdot e^{\alpha y/4} \geq \frac{\alpha y}{4} \cdot \frac{4}{\alpha} \geq y.
\]
Hence for $y \geq (4/ \alpha) \cdot \log(4/\alpha)$,
\[
e^{\alpha y} = e^{\alpha y/2} \cdot e^{\alpha y/2} \geq y \cdot y = y^2.
\]
Thus we conclude
\begin{align*}
 \Upsilon^t &= \sum_{i=1}^n ( y_i^t )^2 \leq \sum_{i=1}^n \max \left\{ \Lambda_i^{t},  \kappa \right\} \leq \sum_{i=1}^n \max \left\{ \Lambda_i^{t},  \Lambda_i^{t} \kappa \right\} \leq \Lambda^t \cdot  \kappa,
\end{align*} 
where the third inequality used that $\Lambda_i^{t} \geq 1$ for any $i \in [n]$ and the fourth used $\kappa \geq 1$.
\end{proof}

The next lemma is very basic but is used in \cref{lem:stabilisation_many_good_quantiles_whp} so we prove it for completeness. 
\begin{lem}\label{lem:basic}
Consider any $\PThree \cap \WTwo$-process or $\PTwo \cap \WThree$-process. For any $t\geq 0$ we have \[|\Upsilon^{t+1}-\Upsilon^{t}|\leq 4 w_-\cdot \max_{i \in [n]} |y_i^t| + 2w_-^2.\]
	\end{lem}
\begin{proof}
Recall that $w^{t}=W^{t+1}-W^{t}$ is the total number of balls allocated in the round $t$. To begin since only one bin, we shall call this $i$, is updated at each round we have
	\begin{align*}
	|\Upsilon^{t-1}-\Upsilon^{t}| &= \left|\sum_{j=1}^n ( y_j^{t+1} )^2 - \sum_{j=1}^n ( y_j^{t} )^2\right|\\&\leq \left|\left(y_{i}^t +   w^t  \cdot \Bigl( 1 - \frac{1}{n}\Bigr) \right)^2 -  \left( y_i^t \right)^2\right| + \sum_{j\in[n], j\neq i }\left|  \left( y_j^{t} -\frac{w^t}{n}   \right)^2 -  \left( y_j^{t}  \right)^2\right|.%
	\end{align*}Now for any $j\in [n]$ we have \begin{align*}\left|  \left( y_j^{t} -\frac{w^t}{n}   \right)^2 -  \left( y_j^{t}  \right)^2\right| &\leq 2 \left|   \frac{w^t}{n}\right|\cdot |y_{i}^t| + \left|\frac{w^t}{n}\right|^2  \leq \frac{2w_- }{n }\cdot  \max_{i \in [n]} |y_i^t| + \frac{w_-^2}{n^2}, 
\end{align*}since  $W^{t}\leq W^{t+1}\leq W^{t} + w_-$ and $n\geq 1$. Similarly for any $i\in [n]$ we have 
\begin{align*}
&  \left|\left(y_{i}^t + w^t \cdot \Bigl( 1 - \frac{1}{n} \Bigr) \right)^2 -  \left( y_i^t \right)^2\right|\\
 & \qquad \leq 2 \cdot |y_i^t| \cdot \left|w^t \cdot \Bigl( 1 - \frac{1}{n}\Bigr) \right| + \left| w^t \cdot \Bigl( 1 - \frac{1}{n}\Bigr) \right|^2 \\
 & \qquad \leq 2\max_{i \in [n]} |y_i^t|\cdot w_- + w_-^2. \qedhere \end{align*}
\end{proof}

\subsection{Exponential Potential \texorpdfstring{$\Lambda$}{Lambda}} \label{sec:non_filling_lambda}

In this section, we consider the exponential potential $\Lambda$. Let $\mathcal{G}^t$ be the event that $\delta^t \in (\epsilon, 1 - \epsilon)$ holds. We will prove in \cref{lem:good_quantile_good_decrease} that the potential drops in expectation when the quantile is good after round $t$, i.e.,
\[
\ex{\Lambda^{t+1} \mid \mathfrak{F}^t, \mathcal{G}^t} \leq \Lambda^t \cdot \Big(1 - \frac{2 \C{good_quantile_mult} \alpha}{n} \Big) + \C{good_quantile_mult}',
\]
and in \cref{lem:bad_quantile_increase_bound} that it has a bounded increase at a round when the quantile is not good,
\[
\ex{\Lambda^{t+1} \mid \mathfrak{F}^t, \neg \mathcal{G}^t} \leq  \Lambda^t \cdot \Big( 1 + \frac{\C{bad_quantile_mult} \alpha^2}{2n}\Big) + \C{bad_quantile_mult}.
\]
Note that the decrease factor can be made arbitrarily larger than the increase factor, by choosing $\alpha > 0 $ smaller. So, once we prove that there is a constant fraction of rounds with a good quantile (\cref{lem:stabilisation_many_good_quantiles_whp}), we can deduce that there is overall an expected decrease in the exponential potential (\cref{lem:gamma_tilde_is_supermartingale} and \cref{lem:stabilization}), when the potential is sufficiently large. Note that the exponential potential cannot decrease in expectation in every round (\cref{clm:bad_configuration_lambda}). 

For the analysis of $\Lambda$ as well as $V$ (in \cref{sec:v_potential}), we consider the labeling of the bins $i \in [n]$ used by the allocation process in round $t$ so that $x_i^{t}$ is non-decreasing in $i \in [n]$. We write 
\[
\Lambda^t =: \sum_{i=1}^n \Lambda_i^t = \sum_{i=1}^n e^{\alpha |y_i^t|} \quad  \Big(\text{and} \quad V^t =: \sum_{i=1}^n V_i^t = \sum_{i=1}^n e^{\tilde{\alpha} |y_i^t|} \Big),
\]
and handle separately the following three cases of bins based on their load: %
\begin{itemize}
  \item \textbf{Case 1} [Robustly Overloaded Bins]. The set of bins $B_{++}^t$ with load $y_i^t \geq \frac{w_-}{n}$. These are bins in $B_+^t$ that are guaranteed to be in $B_+^{t+1}$ (that is, overloaded), since the average load can increase by at most $w_-/n$.
  
  For the exponential potential $\Lambda^t$ or (and $V^t$ respectively), the change of a single bin $i \in B_{++}^t$ is given by,
\begin{equation*} 
\Ex{ \Lambda_i^{t+1} \mid \mathfrak{F}^t}
 = \Lambda_i^t \cdot \Big(\underbrace{p_i^t \cdot e^{\alpha w_+ -\alpha w_+/n}}_{\text{placing a ball in $i$}} + \underbrace{(P_+^t - p_i^t) \cdot e^{-\alpha w_+/n}}_{\text{placing a ball in $B_{+}^t\setminus \{ i \}$}} + \underbrace{P_-^t \cdot e^{-\alpha w_-/n} }_{\text{placing a ball in $B_{-}^t$}}\Big).
\end{equation*}
  
  \item \textbf{Case 2} [Robustly Underloaded Bins]. The set of bins $B_{--}^t$ with load $y_i^t \leq -w_-$. These are bins in $B_-^t$ that are guaranteed to be in $B_-^{t+1}$ (that is, underloaded), since a bin can receive a weight of at most $w_-$ in one round.
  
  For the exponential potential $\Lambda^t$ or (and $V^t$ respectively), the change of a single bin $i \in B_{--}^t$ is given by,
\begin{align*}
\Ex{ \Lambda_i^{t+1} \mid \mathfrak{F}^t}
  = \Lambda_i^t \cdot \Big( \underbrace{p_i^t \cdot e^{- \alpha w_- + \alpha w_-/n}}_{\text{placing a ball in $i$}} + \underbrace{(P_-^t - p_i^t) \cdot e^{\alpha w_-/n}}_{\text{placing a ball in $B_-^t \setminus \{i \}$}} + \underbrace{P_+^t \cdot e^{\alpha w_+/n}}_{\text{placing a ball in $B_+^t$}} \Big).
\end{align*}
  \item \textbf{Case 3} [Swinging Bins]. The set of bins $B_{+/-}^t$ with load $y_i^t \in (-w_-, \frac{w_-}{n})$.
\end{itemize}

We begin by showing that the aggregated contribution of the swinging bins to the change of the potential $\Lambda$ is at most a constant. This will be used in the proofs of Lemmas \ref{lem:good_quantile_good_decrease} and \ref{lem:bad_quantile_increase_bound}.

\begin{lem}\label{lem:bins_close_to_mean}
For any constant $\alpha \in (0, 1]$, for any constant $\kappa_1 \geq 0$ and any $t \geq 0$, we have 
\[
\sum_{i \in B_{+/-}^t} \ex{\Lambda_i^{t+1} \mid \mathfrak{F}^t} \leq \sum_{i \in B_{+/-}^t} \Lambda_i^t \cdot \Big( 1 - \frac{2 \kappa_1 \alpha}{n} \Big) + 3(\kappa_1 + w_-) \cdot e^{2 w_-}.
\]
\end{lem}
\begin{proof}

For each bin $i \in B_{+/-}^t$, the two events affecting the contribution of the bin to $\Lambda^{t+1}$ are $(i)$ ``internal'' due to a ball being placed $i$ and $(ii)$ ``external'' due to the change in the average. 

For $(i)$, the chosen bin $i \in B_{+/-}^t$ can increase by at most $w_-$, so the potential value satisfies $\Lambda_i^{t+1} \leq e^{2\alpha w_-}$. For $(ii)$, the maximum change in the average load is at most $w_-/n$, and this leads to $\Lambda_i^{t+1} \leq \Lambda_i^t \cdot e^{\alpha w_- /n}$.  Combining the two contributions
\begin{align*}
\sum_{i \in B_{+/-}^t} \ex{\Lambda_i^{t+1} \mid \mathfrak{F}^t} 
 & \leq 
\sum_{i \in B_{+/-}^t} \Lambda_i^t \cdot e^{\alpha w_-/n} \cdot (1 - p_i^t) + \sum_{i \in B_{+/-}^t} e^{2\alpha w_-} \cdot p_i^t \\
 & \leq \sum_{i \in B_{+/-}^t} \Lambda_i^t \cdot e^{\alpha w_-/n} + e^{2\alpha w_-} \\
 & \stackrel{(a)}{\leq} \sum_{i \in B_{+/-}^t} \Lambda_i^t \cdot \Big(1 + \frac{2 w_- \alpha}{n}\Big) + e^{2\alpha w_-} 
 \\ &\stackrel{(b)}{\leq} \sum_{i \in B_{+/-}^t} \Lambda_i^t + 2\alpha w_- \cdot e^{\alpha w_-} +  e^{2\alpha w_-},
\end{align*}
where we used the Taylor estimate $e^z \leq 1 + 2z$ for any $z < 1.2$, for sufficiently large $n$ in $(a)$, and $\Lambda_i^{t} \leq e^{\alpha w_{-}}$ for any $i \in B_{+/-}^t$ in $(b)$.  By adding and subtracting $\sum_{i \in B_{+/-}^t} \Lambda_i^t \cdot \frac{2 \alpha \kappa_1}{n}$,
\begin{align*}
\sum_{i \in B_{+/-}^t} \ex{\Lambda_i^{t+1} \mid \mathfrak{F}^t} 
 & \leq \sum_{i \in B_{+/-}^t} \Lambda_i^t\cdot \Big(1 - \frac{2 \alpha \kappa_1}{n}\Big) + \sum_{i \in B_{+/-}^t} \Lambda_i^t \cdot \frac{2 \alpha \kappa_1}{n} + 2\alpha w_- \cdot e^{\alpha w_-} +  e^{2\alpha w_-} \\
 & \leq \sum_{i \in B_{+/-}^t} \Lambda_i^t\cdot\Big(1 - \frac{2 \alpha \kappa_1}{n}\Big) + (2\alpha \kappa_1 + 2\alpha w_-) \cdot e^{\alpha w_-} +  e^{2\alpha w_-}, %
\end{align*}
where we used that $\Lambda_i^t \leq e^{\alpha w_-}$ and $|B_{+/-}^t| \leq n$. Finally we have $(2\alpha \kappa_1 + 2\alpha w_-) \cdot e^{\alpha w_-} +  e^{2\alpha w_-}\leq (2\kappa_1 + 2 w_-) \cdot e^{w_-} +  e^{2 w_-} \leq 3(\kappa_1 + w_-) \cdot e^{2 w_-}$, as $w_-\geq 1$ and  $\alpha\leq 1$.
\end{proof}

Now, we show that if the quantile $\delta^t$ of the mean is in $(\epsilon, 1 - \epsilon)$, then the potential function exhibits a multiplicative drop.

\begin{lem}\label{lem:good_quantile_good_decrease} For any $\WThree \cap \PThree$-process and any constant $\epsilon \in (0, 1)$, 
choose a constant $\alpha:=\alpha(\eps)$ such that
\begin{equation}\label{eq:c_3alphacond1}
0 < \alpha \leq \min\left\lbrace\frac{1}{w_-}, \frac{\C{p2k2} \epsilon}{2 w_- (1 + \C{p2k2}\epsilon)}, \frac{\C{p2k1} \epsilon}{2 w_+(1-\C{p2k1} \epsilon)} \right\rbrace.
\end{equation}
For any $\PTwo \cap \WTwo$-process and any constant $\epsilon \in (0, 1)$,
choose a constant $\alpha:=\alpha(\eps)$ such that
\begin{equation}\label{eq:c_3alphacond2}
0 <  \alpha \leq \min\left\lbrace\frac{1}{w_-}, \frac{\epsilon( w_- - w_+)}{4 w_-^2}, \frac{\epsilon}{2 \cdot w_- \cdot (2+\epsilon)} \right\rbrace.
\end{equation}
Then there exists a constant $\C{good_quantile_mult}:=\C{good_quantile_mult}(\epsilon) > 0$ such that for $\C{good_quantile_mult}':=3(\C{good_quantile_mult} + w_-) \cdot e^{2 w_-}$ and any $t\geq 0$ we have
\[
\Ex{ \Lambda^{t+1} \mid \mathfrak{F}^t, \mathcal{G}^t } \leq \Lambda^t \cdot \Big(1 - \frac{2 \C{good_quantile_mult} \alpha}{n} \Big) + \C{good_quantile_mult}'. %
\]
\end{lem}
\begin{proof}This proof will work with the labeling of the bins $i \in [n]$ used by the allocation process in round $t$ so that $x_i^{t}$ is non-decreasing in $i \in [n]$.

\noindent\textbf{Case 1} [Robustly Overloaded Bins]. For $i \in B_{++}^t$ in this case
\begin{align*}
\Ex{ \Lambda_i^{t+1} \mid \mathfrak{F}^t}
 & = \Lambda_i^t \cdot \Big(p_i^t \cdot e^{-\alpha w_+/n + \alpha w_+} + (P_+^t - p_i^t) \cdot e^{-\alpha w_+/n} + P_-^t \cdot e^{-\alpha w_-/n} \Big) . \\
 \intertext{Applying the Taylor estimate $e^{z} \leq 1+z+z^2$, which holds for any $z \leq 1.75$, since $\alpha w_+ \leq 1$ (and $\alpha w_- \leq 1$),}
 \Ex{ \Lambda_i^{t+1} \mid \mathfrak{F}^t} & \leq \Lambda_i^t \cdot \Big(1 + p_i^t \cdot \Big(-\frac{\alpha w_+}{n} + \alpha w_+ + \Bigl(-\frac{\alpha w_+}{n} + \alpha w_+ 
 \Bigr)^2 \Big)   \\ 
 & \qquad + (P_+^t - p_i^t) \cdot \Big(-\frac{\alpha w_+}{n} + \Big(\frac{\alpha w_+}{n}\Big)^2 \Big) + P_-^t \cdot \Big(- \frac{\alpha w_-}{n} + \Big(\frac{\alpha w_-}{n} \Big)^2 \Big) \Big). \\
 \intertext{By gathering $o(n^{-1})$ terms and rearranging terms we obtain,}
 \Ex{ \Lambda_i^{t+1} \mid \mathfrak{F}^t} & \leq \Lambda_i^t \cdot \Big(1 + p_i^t \cdot \Big(-\frac{\alpha w_+}{n} + \alpha w_+ + \Bigl(-\frac{\alpha w_+}{n} + \alpha w_+ 
 \Bigr)^2 \Big)   \\ 
 & \qquad - (P_+^t - p_i^t) \cdot \frac{\alpha w_+}{n} - P_-^t \cdot \frac{\alpha w_-}{n} + o(n^{-1}) \Big) \\
 & = \Lambda_i^t \cdot \Big(1 + p_i^t \cdot \Bigl(\alpha w_+ + \Bigl(-\frac{\alpha w_+}{n} + \alpha w_+ 
 \Bigr)^2\Bigr) - P_+^t \cdot \frac{\alpha w_+}{n} - P_-^t \cdot \frac{\alpha w_-}{n} + o(n^{-1}) \Big).
\end{align*}
Applying $(-\frac{\alpha w_{+}}{n} + \alpha w_{+})^2 \leq (\alpha w_{+})^2$ and then replacing $P_-^t$ by $1 - P_+^t$, gives,
\begin{align}
 \Ex{ \Lambda_i^{t+1} \mid \mathfrak{F}^t} & \leq \Lambda_i^t \cdot \Big(1 + p_i^t \cdot (\alpha w_+ + (\alpha w_+)^2) - P_+^t \cdot \frac{\alpha w_+}{n} - P_-^t \cdot \frac{\alpha w_-}{n} + o(n^{-1}) \Big) \nonumber \\
 & = \Lambda_i^t \cdot \Big(1 + p_i^t \cdot (\alpha w_+ + (\alpha w_+)^2) + P_+^t \cdot \Big(\frac{\alpha w_-}{n} - \frac{\alpha w_+}{n}\Big) - \frac{\alpha w_-}{n} + o(n^{-1}) \Big).\label{eq:lambda_case_1}
\end{align}
By condition \PTwo, we have $p_i^t \leq 1/n$ for any overloaded bin $i \in B_{+}^t$. Since by assumption, $\delta^t \leq 1-\epsilon$ we also have that $P_+^t \leq (\delta^t \cdot n) \cdot 1/n \leq 1- \epsilon$. Hence,
\begin{align*}
 \Ex{ \Lambda_i^{t+1} \mid \mathfrak{F}^t} & \leq \Lambda_i^t \cdot \Big(1 + p_i^t  \cdot (\alpha w_+ + (\alpha w_+)^2) + (1 - \epsilon) \cdot \Bigl(\frac{\alpha w_-}{n} - \frac{\alpha w_+}{n} \Bigr) - \frac{\alpha w_-}{n} + o(n^{-1}) \Big) \nonumber\\
 & = \Lambda_i^t \cdot \Big(1 - \alpha w_+ \cdot \Bigl(\frac{1}{n} - p_i^t \cdot (1 + \alpha w_+ )\Bigr) - \frac{\alpha \epsilon}{n} \cdot (w_- - w_+) + o(n^{-1}) \Big).
\end{align*}
\noindent \textbf{Case 1.A} [\PThree holds]. Then, as for any overloaded bin $p_i^t \leq \frac{1 - \C{p2k1} \epsilon}{n}$ and $w_- \geq w_-$,
\begin{align*}
\Ex{ \Lambda_i^{t+1} \mid \mathfrak{F}^t}
 & \leq \Lambda_i^t \cdot \Big(1 - \frac{\alpha w_+}{n} \cdot \Bigl(1 - (1 - \C{p2k1} \epsilon) \cdot (1 + \alpha w_+ )\Bigr) - \frac{\alpha \epsilon}{n} \cdot (w_- - w_+) + o(n^{-1}) \Big) \\
 & \leq \Lambda_i^t \cdot \Big(1 - \frac{\alpha w_+}{n} \cdot \Bigl(1 - (1 - \C{p2k1} \epsilon) \cdot (1 + \alpha w_+ )\Bigr) + o(n^{-1}) \Big) \\
 & = \Lambda_i^t \cdot \Big(1 - \frac{\alpha w_+}{n} \cdot \Bigl(\C{p2k1} \epsilon - \alpha w_+ \cdot (1 - \C{p2k1} \epsilon )\Bigr) + o(n^{-1}) \Big) \\
 & \stackrel{(a)}{\leq} \Lambda_i^t \cdot \Big(1 - \frac{\alpha w_+ }{n} \cdot \frac{\C{p2k1} \epsilon}{2} + o(n^{-1}) \Big) \\ 
 & \stackrel{(b)}{\leq} \Lambda_i^t \cdot \Big(1 - \frac{\alpha w_+ \C{p2k1} \epsilon}{4n} \Big),
\end{align*}
where in inequality $(a)$ we used  that $\frac{\C{p2k1} \epsilon}{2}\geq \alpha w_+(1 - \C{p2k1} \epsilon )$ (as implied by $\alpha \leq \frac{\C{p2k1} \epsilon}{2 w_+(1-\C{p2k1} \epsilon)}$) and in inequality $(b)$ that $\alpha, w_+, \epsilon, \C{p2k1}$ are constants.

\noindent \textbf{Case 1.B} [\WThree holds]. Applying the bound $p_i^t \leq p_+^t \leq 1/n$,
\begin{align*}
\Ex{ \Lambda_i^{t+1} \mid \mathfrak{F}^t}
 & \leq \Lambda_i^t \cdot \Big(1 - \alpha w_+ \cdot \Bigl(\frac{1}{n} - \frac{1}{n} \cdot (1 + \alpha w_+ )\Bigr) - \frac{\alpha \epsilon}{n} \cdot (w_- - w_+) + o(n^{-1}) \Big) \\
 & = \Lambda_i^t \cdot \Big(1 + \frac{(\alpha w_+)^2}{n} - \frac{\alpha \epsilon}{n} \cdot (w_- - w_+) + o(n^{-1}) \Big) \\
 & = \Lambda_i^t \cdot \Big(1 - \frac{\alpha}{n} \cdot (\epsilon \cdot (w_- - w_+) - \alpha w_+^2) + o(n^{-1}) \Big) \\
 & \stackrel{(a)}{\leq} \Lambda_i^t \cdot \Big(1 - \frac{\alpha }{n} \cdot \frac{\epsilon}{2} \cdot (w_- - w_+) + o(n^{-1}) \Big) \\
 & \stackrel{(b)}{\leq} \Lambda_i^t \cdot \Big(1 - \frac{\alpha \epsilon}{4n} \cdot(w_- - w_+) \Big),
\end{align*}
where we used in $(a)$ that $\frac{\epsilon}{2} \cdot (w_- - w_+) \geq \alpha w_+^2$ (as implied by $\alpha \leq \frac{\epsilon (w_- - w_+)}{4w_-^2} \leq \frac{\epsilon (w_- - w_+)}{2w_+^2}$) and in $(b)$ that $\alpha, \epsilon, (w_- - w_+)$ are constants.

So, in both \textbf{Case 1.A} and \textbf{Case 1.B}, we conclude that there exists a constant $\C{good_quantile_mult} >0$ such that,
\begin{align*}
\sum_{i \in B_{++}^t} \Ex{ \Lambda_i^{t+1} \mid \mathfrak{F}^t} 
\leq \sum_{i \in B_{++}^t} \Lambda_i^t \cdot \Big(1 - \frac{2 \C{good_quantile_mult} \alpha}{n} \Big). 
\end{align*}

\noindent \textbf{Case 2} [Robustly Underloaded Bins]. For $i \in B_{--}^t$,
\begin{align*}
\Ex{ \Lambda_i^{t+1} \mid \mathfrak{F}^t}
 & = \Lambda_i^t \cdot \Big(p_i^t \cdot e^{\alpha w_-/n - \alpha w_-} + (P_-^t - p_i^t) \cdot e^{\alpha w_-/n} + P_+^t \cdot e^{\alpha w_+/n} \Big). 
 \intertext{Applying the Taylor estimate $e^{z} \leq 1+z+z^2$, which holds for any $z \leq 1.75$, since $\alpha w_+/n \leq 1$ (and $\alpha w_-/n \leq 1$),}
 \Ex{ \Lambda_i^{t+1} \mid \mathfrak{F}^t} & \leq \Lambda_i^t \cdot \Big(1 + p_i^t \cdot \Big(\frac{\alpha w_-}{n} - \alpha w_- + \Big(\frac{\alpha w_-}{n} - \alpha w_- \Big)^2 \Big)  \\ 
 & \qquad + (P_-^t - p_i^t) \cdot \Big(\frac{\alpha w_-}{n} + \Big(\frac{\alpha w_-}{n} \Big)^2 \Big) + P_+^t \cdot \Big(\frac{\alpha w_+}{n} + \Bigl(\frac{\alpha w_+}{n} \Big)^2 \Big) \Big).
 \intertext{By gathering $o(n^{-1})$ terms and rearranging terms we obtain,}
 \Ex{ \Lambda_i^{t+1} \mid \mathfrak{F}^t} & \leq \Lambda_i^t \cdot \Big(1 + p_i^t \cdot \Big(\frac{\alpha w_-}{n} - \alpha w_- + \Big(\frac{\alpha w_-}{n} - \alpha w_- \Big)^2 \Big) \\ 
 & \qquad + (P_-^t - p_i^t) \cdot \frac{\alpha w_-}{n} + P_+^t \cdot \frac{\alpha w_+}{n} + o(n^{-1}) \Big) \\
 & = \Lambda_i^t \cdot \Big(1 - p_i^t \cdot \Big(\alpha w_- - \Big(\frac{\alpha w_-}{n} - \alpha w_- \Big)^2 \Big) + P_-^t \cdot \frac{\alpha w_-}{n}  + P_+^t \cdot \frac{\alpha w_+}{n} + o(n^{-1}) \Big).
\end{align*}
Applying $(\frac{\alpha w_-}{n} - \alpha w_-)^2 \leq (\alpha w_-)^2$ and then replacing $P_+^t$ by $1 - P_-^t$,
\begin{align}
 \Ex{ \Lambda_i^{t+1} \mid \mathfrak{F}^t} & \leq \Lambda_i^t \cdot \Big(1 - p_i^t \cdot (\alpha w_- - (\alpha w_-)^2) + P_-^t \cdot \frac{\alpha w_-}{n}  + P_+^t \cdot \frac{\alpha w_+}{n} + o(n^{-1}) \Big) \nonumber \\
 & = \Lambda_i^t \cdot \Big(1 - p_i^t \cdot (\alpha w_- - (\alpha w_-)^2) + P_-^t \cdot \left(\frac{\alpha w_-}{n} - \frac{\alpha w_+}{n}\right) + \frac{\alpha w_+}{n} + o(n^{-1}) \Big). \label{eq:lambda_case_2} %
\end{align}

\noindent\textbf{Case 2.A} [\PThree holds]. Using that $p_-^t \geq \frac{1 + \C{p2k2}\epsilon}{n}$, $P_-^t \leq 1$, and applying this to \eqref{eq:lambda_case_2} yields $(a)$ below (since $(\alpha w_-)^2 \leq \alpha w_-$), and further rearranging gives,
\begin{align*}
\Ex{ \Lambda_i^{t+1} \mid \mathfrak{F}^t}
 & \stackrel{(a)}{\leq} \Lambda_i^t \cdot \Big(1 - \frac{1 + \C{p2k2}\epsilon}{n} \cdot (\alpha w_- - (\alpha w_-)^2) + 1 \cdot \left(\frac{\alpha w_-}{n} - \frac{\alpha w_+}{n}\right) + \frac{\alpha w_+}{n} + o(n^{-1}) \Big) \\
 & = \Lambda_i^t \cdot \Big(1 - \alpha w_- \cdot \Big(\frac{1 + \C{p2k2}\epsilon}{n} \cdot (1 - \alpha w_-) - \frac{1}{n} \Big)  + o(n^{-1}) \Big) \\
 & = \Lambda_i^t \cdot \Big(1 - \frac{\alpha w_-}{n} \cdot \Big((1 + \C{p2k2}\epsilon ) \cdot (1 - \alpha w_-) - 1 \Big) + o(n^{-1})\Big) \\
 & = \Lambda_i^t \cdot \Big(1 - \frac{\alpha w_-}{n} \cdot \Big(\C{p2k2}\epsilon - \alpha w_- \cdot (1 + \C{p2k2} \epsilon)\Big) + o(n^{-1})\Big) \\
 & \stackrel{(b)}{\leq} \Lambda_i^t \cdot \Big(1 - \frac{\alpha w_- }{n} \cdot \frac{\C{p2k2} \epsilon}{2} + o(n^{-1}) \Big) \\ 
 & \stackrel{(c)}{\leq} \Lambda_i^t \cdot \Big(1 - \frac{\alpha w_- \C{p2k2} \epsilon}{4n}\Big),
\end{align*}
where we used in $(b)$ that $\frac{\C{p2k2} \epsilon}{2} \geq \alpha w_- \cdot (1 + \C{p2k2} \epsilon)$ (as implied by $\alpha \leq \frac{\C{p2k2} \epsilon}{2 w_- (1 + \C{p2k2}\epsilon)}$) and in $(c)$ that $\alpha, w_-, \C{p2k2}, \epsilon$ are constants.

\noindent \textbf{Case 2.B} [\WThree and $P_-^t \leq 1 - \frac{\epsilon}{2}$ holds]. Then since $p_i^t \geq 1/n$ and $(\alpha w_-)^2 \leq \alpha w_-$, \eqref{eq:lambda_case_2} implies,
\begin{align*}
\Ex{ \Lambda_i^{t+1} \mid \mathfrak{F}^t} 
 & \leq \Lambda_i^t \cdot \Big(1 - \frac{1}{n} \cdot (\alpha w_- - (\alpha w_-)^2) + (1 - \frac{\epsilon}{2} ) \cdot \Big(\frac{\alpha w_-}{n} - \frac{\alpha w_+}{n}\Big) + \frac{\alpha w_+}{n} + o(n^{-1}) \Big) \\
 & = \Lambda_i^t \cdot \Big(1 + \frac{(\alpha w_-)^2}{n} - \frac{\epsilon}{2} \cdot \Big(\frac{\alpha w_-}{n} - \frac{\alpha w_+}{n}\Big) + o(n^{-1}) \Big) \\
 & = \Lambda_i^t \cdot \Big(1 - \frac{\alpha}{n} \cdot \Big( \frac{\epsilon}{2} \cdot (w_- - w_+) - \alpha w_-^2 \Big) + o(n^{-1}) \Big) \\
 & \stackrel{(a)}{\leq} \Lambda_i^t \cdot \Big(1 - \frac{\alpha}{n} \cdot \frac{\epsilon}{4} \cdot (w_- - w_+) + o(n^{-1}) \Big) \\
 & \stackrel{(b)}{\leq} \Lambda_i^t \cdot \Big(1 - \frac{\alpha \epsilon}{8n} \cdot (w_- - w_+) \Big).
\end{align*}
where we used in $(a)$ that $\frac{\epsilon}{4} \cdot (w_- - w_+) \geq \alpha w_-^2$ (as implied by $\alpha \leq \frac{\epsilon( w_- - w_+)}{4 w_-^2}$) and in $(b)$ that $\alpha, \epsilon, (w_- - w_+)$ are constants.

\medskip 

\noindent\textbf{Case 2.C} [\WThree and $P_-^t > 1 - \frac{\epsilon}{2}$ holds]. We will derive a similar inequality by a majorization argument. Using \eqref{eq:lambda_case_2}, for the bins in $B_{--}^t$,
\[
\sum_{i \in B_{--}^t} \ex{\Lambda_i^{t+1} \mid \mathfrak{F}^t} \leq \sum_{i \in B_{--}^t} \Lambda_i^t \cdot \Big(1 - p_i^t \cdot (\alpha w_- - (\alpha w_-)^2) + P_-^t \cdot \left(\frac{\alpha w_-}{n} - \frac{\alpha w_+}{n}\right) + \frac{\alpha w_+}{n} + o(n^{-1}) \Big).
\]
Since $\Lambda_i^t$ is non-decreasing for $i \in B_{--}^t$ (assuming bins are sorted decreasingly according to their load) and $p_i^t$ is non-decreasing, applying \cref{lem:quasilem}, we can upper bound the above expression by replacing each $p_i^t$ by the average probability $\overline{p}_i^t$ of robustly underloaded bins. Note that because of monotonicity of the allocation distribution $p^t$, the average probability over $B_{--}^t$ is at least as large as the average probability over $B_-^t$, so it satisfies
\[
\overline{p}_i^t \geq  \frac{P_-^t}{|B_-^t|} 
> \frac{1 - \frac{\epsilon}{2}}{(1 - \epsilon)n} > \frac{1 + \frac{\epsilon}{2}}{n},
\]
where we have used $P_-^t > 1 - \frac{\epsilon}{2}$, $|B_-^t| \geq (1 - \epsilon) \cdot n$ and \cref{clm:eps_ineq}. Hence, using that $(\alpha w_-)^2 \leq \alpha w_-$
\begin{align*}
\sum_{i \in B_{--}^t} \ex{\Lambda_i^{t+1} \mid \mathfrak{F}^t} & \leq \sum_{i \in B_{--}^t} \Lambda_i^t \cdot \Big(1 - \overline{p}_i^t \cdot (\alpha w_- - (\alpha w_-)^2) + 1 \cdot \Big(\frac{\alpha w_-}{n} - \frac{\alpha w_+}{n} \Big) + \frac{\alpha w_+}{n} + o(n^{-1}) \Big) \\
 & = \sum_{i \in B_{--}^t} \Lambda_i^t \cdot \Big(1 - \overline{p}_i^t \cdot (\alpha w_- - (\alpha w_-)^2) + \frac{\alpha w_-}{n} + o(n^{-1}) \Big) \\
 & = \sum_{i \in B_{--}^t} \Lambda_i^t \cdot \Big(1 - \alpha w_- \cdot \Big(\overline{p}_i^t \cdot (1 - \alpha w_-) - \frac{1}{n}\Big) + o(n^{-1}) \Big) \\
 & \leq \sum_{i \in B_{--}^t} \Lambda_i^t \cdot \Big(1 - \frac{\alpha w_-}{n} \cdot \Big( \Big(1 + \frac{\epsilon}{2}\Big) \cdot (1 - \alpha w_-) - 1 \Big) + o(n^{-1})\Big) \\ 
 & = \sum_{i \in B_{--}^t} \Lambda_i^t \cdot \Big(1 - \frac{\alpha w_-}{n} \cdot \Big(\frac{\epsilon}{2} - \alpha w_- \cdot \Big(1 + \frac{\epsilon}{2} \Big)\Big) + o(n^{-1})\Big) \\ 
 & \stackrel{(a)}{\leq} \sum_{i \in B_{--}^t} \Lambda_i^t \cdot \left(1 - \frac{\alpha w_-}{n} \cdot \frac{\epsilon}{4} + o(n^{-1})\right) \\
 & \stackrel{(b)}{\leq} \sum_{i \in B_{--}^t} \Lambda_i^t \cdot \left(1 - \frac{\alpha w_-}{4} \cdot \frac{\epsilon}{2n}\right),
\end{align*}
where we used in $(a)$ that $\frac{\epsilon}{4} \geq \alpha w_- \cdot (1 + \frac{\epsilon}{2})$ (as implied by $\alpha \leq \frac{\epsilon}{2 \cdot w_- \cdot (2+\epsilon)}$) and in $(b)$ that $\alpha, w_-, \epsilon$ are constants. Hence, in all subcases of \textbf{Case 2}, we can find a constant $\C{good_quantile_mult} > 0$,
\[
\sum_{i \in B_{--}^t} \Ex{ \Lambda_i^{t+1} \mid \mathfrak{F}^t } \leq \sum_{i \in B_{--}^t} \Lambda^t \cdot \Big(1 - \frac{2 \C{good_quantile_mult} \alpha}{n} \Big).
\]

\noindent\textbf{Case 3} [Swinging Bins]. For the bins in $B_{+/-}^t$, using \cref{lem:bins_close_to_mean} with $\kappa_1 := \C{good_quantile_mult}$ gives 
\[
\sum_{i \in B_{+/-}^t} \ex{\Lambda_i^{t+1} \mid \mathfrak{F}^t} \leq \sum_{i \in B_{+/-}^t} \Lambda_i^t \cdot \Big( 1 - \frac{2 \C{good_quantile_mult} \alpha}{n} \Big) +  3(\C{good_quantile_mult} + w_-) \cdot e^{2 w_-} .
\]
In conclusion, by choosing a sufficiently small constant $\C{good_quantile_mult}=\C{good_quantile_mult}(w_+, w_-,$ $\C{p2k1}, \C{p2k2}, \epsilon) > 0$ if \PThree holds and $\C{good_quantile_mult}=\C{good_quantile_mult}(w_+, w_-, \epsilon) > 0$ if \WThree holds, which satisfies all the above cases, the claim follows.
\end{proof}

Now, we will prove that the expected increase of $\Lambda$ is bounded by a factor of $\alpha^2 \C{bad_quantile_mult}/2n$ at an arbitrary round. 

\begin{lem} \label{lem:bad_quantile_increase_bound}
Consider any $\PThree \cap \WTwo$-process or $\PTwo \cap 
\WThree$-process. Then for the constant $\C{bad_quantile_mult} := 3w_- \cdot e^{2 w_-} > 0$ and any constant $\alpha$ satisfying the preconditions of \cref{lem:good_quantile_good_decrease}, and any $t \geq 0$, we have 
\[
\ex{\Lambda^{t+1} \mid \mathfrak{F}^t} \leq \Lambda^t \cdot \Big(1 + \frac{\alpha^2 \C{bad_quantile_mult}}{2n} \Big) + \C{bad_quantile_mult}.
\]
\end{lem}
\begin{proof}

\noindent \textbf{Case 1} [Robustly Overloaded Bins]. Using \eqref{eq:lambda_case_1} in \textbf{Case 1} of \cref{lem:good_quantile_good_decrease},
\begin{align*}
\Ex{ \Lambda_i^{t+1} \mid \mathfrak{F}^t}
 & \leq \Lambda_i^t \cdot \Big(1 + p_i^t \cdot \Big(\alpha w_+ + (\alpha w_+)^2\Big) + P_+^t \cdot \Big(\frac{\alpha w_-}{n} - \frac{\alpha w_+}{n} \Big) - \frac{\alpha w_-}{n} + o(n^{-1}) \Big). \\
 \intertext{Using that $p_+^t \leq 1/n$ and $P_+^t \leq 1$,}
 \Ex{ \Lambda_i^{t+1} \mid \mathfrak{F}^t} & \leq \Lambda_i^t \cdot \Big(1 + \frac{1}{n} \cdot \Big(\alpha w_+ + (\alpha w_+)^2\Big) + \Big(\frac{\alpha w_-}{n} - \frac{\alpha w_+}{n} \Big) - \frac{\alpha w_-}{n} + o(n^{-1}) \Big) \\
 & = \Lambda_i^t \cdot \Big(1 + \frac{1}{n} \cdot (\alpha w_+)^2 + o(n^{-1}) \Big) \leq \Lambda_i^t \cdot \Big(1 + \frac{2\alpha^2 w_+^2}{n}\Big),
\end{align*}
where in the last step we used that $\alpha$ and $w_+$ are constants.

\medskip

\noindent \textbf{Case 2} [Robustly Underloaded Bins]. Using \eqref{eq:lambda_case_2} in \textbf{Case 2} of \cref{lem:good_quantile_good_decrease}, we get
\begin{align*}
\Ex{ \Lambda_i^{t+1} \mid \mathfrak{F}^t}
 & \leq \Lambda_i^t \cdot \Big(1 + p_i^t \cdot (-\alpha w_- + (\alpha w_-)^2) + P_-^t \cdot \Big(\frac{\alpha w_-}{n} - \frac{\alpha w_+}{n} \Big) + \frac{\alpha w_+}{n} + o(n^{-1})\Big) \\
\intertext{Using that $p_-^t \geq \frac{1}{n}$, $\alpha w_- \geq (\alpha w_-)^2$ and $P_-^t \leq 1$,}
\Ex{ \Lambda_i^{t+1} \mid \mathfrak{F}^t} & \leq \Lambda_i^t \cdot \Big(1 + \frac{1}{n} \cdot (-\alpha w_- + (\alpha w_-)^2) + \Big(\frac{\alpha w_-}{n} - \frac{\alpha w_+}{n}\Big) + \frac{\alpha w_+}{n} + o(n^{-1})\Big) \\
 & = \Lambda_i^t \cdot \Big(1 + \frac{1}{n} \cdot (\alpha w_-)^2 + o(n^{-1})\Big) \leq \Lambda_i^t \cdot \Big(1 + \frac{2\alpha^2 w_-^2}{n}\Big),
\end{align*}
where in the last step we used that $\alpha$ and $w_-$ are constants.
\medskip

\noindent \textbf{Case 3} [Swinging Bins]. Using \cref{lem:bins_close_to_mean} with $\kappa_1 = 0$ gives 
\[
\sum_{i \in B_{+/-}^t} \ex{\Lambda_i^{t+1} \mid \mathfrak{F}^t} \leq \sum_{i \in B_{+/-}^t} \Lambda_i^t + 3w_- \cdot e^{2 w_-} \leq \sum_{i \in B_{+/-}^t} \Lambda_i^t \cdot \Big( 1 + \frac{\alpha^2w_-^2}{2n}\Big) + 3w_- \cdot e^{2 w_-}.
\]
Thus, aggregating over the three cases and choosing $\C{bad_quantile_mult} := 3w_- \cdot e^{2 w_-} \geq 4w_-^2$, gives the result.
\end{proof}

We now combine the statements (and constants) from Lemmas \ref{lem:good_quantile_good_decrease} and \ref{lem:bad_quantile_increase_bound} into a single corollary. 

\begin{cor}
\label{cor:change_for_large_lambda}
Consider any $\PThree \cap \WTwo$-process or $\PTwo \cap \WThree$-process, and let $\eps \in (0, 1)$ be any constant. Choose $\C{lambda_bound} := \max\Big(\frac{3(\C{good_quantile_mult} + w_-) \cdot e^{2 w_-}}{\alpha\C{good_quantile_mult}},\frac{2}{\alpha^2}\Big) > 1$, for the constants $\C{good_quantile_mult}(\eps)$ and $\alpha(\eps)$ as defined in \cref{lem:good_quantile_good_decrease}.
Then for any $t \geq 0$,
\[
\ex{\Lambda^{t+1} \mid \mathfrak{F}^t, \mathcal{G}^t, \Lambda^t \geq \C{lambda_bound} \cdot n} \leq \Lambda^t \cdot \Big(1 - \frac{\C{good_quantile_mult} \alpha}{n}\Big).
\]
More generally, for any $t \geq 0$, and $\C{bad_quantile_mult} := 3w_- \cdot e^{2 w_-} > 0$, we have 
\[
\ex{\Lambda^{t+1} \mid  \mathfrak{F}^t, \Lambda^t \geq \C{lambda_bound} \cdot n} \leq \Lambda^t \cdot \Big(1 + \frac{\alpha^2 \C{bad_quantile_mult}}{n}\Big).
\]
\end{cor}

\begin{proof}
For the first statement, using \cref{lem:good_quantile_good_decrease}, where $\C{good_quantile_mult}':=3(\C{good_quantile_mult} + w_-) \cdot e^{2 w_-}$,  gives\[\ex{\Lambda^{t+1} \mid \mathfrak{F}^t, \mathcal{G}^t, \Lambda^t \geq \C{lambda_bound} \cdot n} \leq \Lambda^t \cdot \Big(1 - \frac{2\C{good_quantile_mult} \alpha}{n}\Big) + \C{good_quantile_mult}'  = \Lambda^t \cdot \Big(1 - \frac{\C{good_quantile_mult} \alpha}{n}\Big) + \Big(\C{good_quantile_mult}' - \Lambda^t \cdot \frac{\C{good_quantile_mult} \alpha}{n}\Big)  \leq \Lambda^t \cdot \Big(1 - \frac{\C{good_quantile_mult} \alpha}{n}\Big).
\]
For the second statement, by applying \cref{lem:bad_quantile_increase_bound} we have  
\[
\ex{\Lambda^{t+1} \mid \mathfrak{F}^t, \Lambda^t \geq \C{lambda_bound} \cdot n} 
\leq \Lambda^t \cdot \Big(1 + \frac{\alpha^2 \C{bad_quantile_mult}}{2n}\Big) + \C{bad_quantile_mult}
\leq \Lambda^t \cdot \Big(1 + \frac{\alpha^2 \C{bad_quantile_mult}}{2n}\Big) + \Lambda^t \cdot \frac{\alpha^2 \C{bad_quantile_mult}}{2n}
\leq \Lambda^t \cdot \Big(1 + \frac{\alpha^2 \C{bad_quantile_mult}}{n}\Big),
\]as claimed.
\end{proof}

The next lemma shows that loads cannot deviate too wildly within a $\Theta(n\log n)$-length interval following a round with small exponential potential. 
\begin{lem} \label{clm:small_change_for_linear_lambda}
Consider any $\PThree \cap \WTwo$-process or $\PTwo \cap \WThree$-process. For any constant $\kappa >0$, and any rounds $t_0$, $t_1$ such that $t_0 \leq t_1 \leq t_0 + \kappa \cdot n \log n$,
\[
\Pro{\max_{t \in [t_0,t_1]}\max_{i \in [n]} | y_i^{t} | \leq \log^2 n \, \Big| \, \mathfrak{F}^{t_0},\; \Lambda^{t_0} \leq n^2} \geq 1 - n^{-12}.
\]
\end{lem}
\begin{proof}
Consider the sequence $(\ex{\Lambda^t \mid \mathfrak{F}^{t_0}, \Lambda^{t_0} \leq n^2 })_{t = t_0}^{t_1}$, then for every $t \in (t_0, t_1]$, using \cref{lem:bad_quantile_increase_bound},
\[
\ex{\Lambda^{t+1} \mid \mathfrak{F}^{t}} \leq  \Lambda^t \cdot \Big( 1 + \frac{\C{bad_quantile_mult} \alpha^2}{2n}\Big) + \C{bad_quantile_mult}.
\]
Using the tower law of expectation, we have that
\[
\ex{\Lambda^{t+1} \mid \mathfrak{F}^{t_0}, \Lambda^{t_0} \leq n^2} \leq  \Ex{\Lambda^t \mid \mathfrak{F}^{t_0}, \Lambda^{t_0} \leq n^2} \cdot \Big( 1 + \frac{\C{bad_quantile_mult} \alpha^2}{2n}\Big) + \C{bad_quantile_mult}.
\]
Hence, applying \cref{lem:geometric_arithmetic} with $a := 1 + \frac{\C{bad_quantile_mult} \alpha^2}{2n} > 1$ and $b := \C{bad_quantile_mult} > 0$, we get that 
\begin{align*}
\ex{\Lambda^{t} \mid \mathfrak{F}^{t_0}, \Lambda^{t_0} \leq n^2} 
 & \leq \Lambda^{t_0} \cdot a^{t - t_0} + b \cdot \sum_{s = t_0}^{t-1} a^{s-t_0} \\
 & \leq n^2 \cdot \Big(1 + \frac{\alpha^2 \C{bad_quantile_mult}}{n}\Big)^{t_1-t_0} + \C{bad_quantile_mult} \cdot (t_1 - t_0) \cdot \Big(1 + \frac{\alpha^2 \C{bad_quantile_mult}}{n}\Big)^{t_1-t_0} \\
 & \stackrel{(a)}{\leq} n^2 \cdot n^{\kappa \cdot \alpha^2 \cdot \C{bad_quantile_mult}} + \C{bad_quantile_mult} \cdot (\kappa n \log n) \cdot n^{\kappa \cdot \alpha^2 \cdot \C{bad_quantile_mult}}
 \leq n^{3 + \kappa \cdot \alpha^2 \cdot \C{bad_quantile_mult}}.
\end{align*}
where in $(a)$ we used that $1 + z \leq e^z$ for any $z$.

Using Markov's inequality, $\Pro{\Lambda^{t} \leq n^{3 + \kappa \cdot \alpha^2 \cdot \C{bad_quantile_mult} + 14} ~\big|~ \mathfrak{F}^{t_0}, \Lambda^{t_0} \leq n^{2}} \geq 1 - n^{-14}$ for any $t\in [t_0, t_1]$, which implies
\[
 \Pro{ \max_{i \in [n]} |y_i^{t}| \leq \frac{1}{\alpha} \cdot (\kappa \cdot \alpha^2 \cdot \C{bad_quantile_mult} + 17) \cdot \log n  ~\Big|~ \mathfrak{F}^{t_0}, \;\Lambda^{t_0} \leq n^2 } \geq 1 - n^{-14}.
\]
Since $\frac{1}{\alpha} \cdot (\kappa \cdot \alpha^2 \cdot \C{bad_quantile_mult} + 17) \cdot \log n < \log^2 n$ for sufficiently large $n$, by taking a union bound over all rounds $t \in [t_0, t_1]$ we get the claim.
\end{proof}

\subsection{Weak Exponential Potential \texorpdfstring{$V$}{V}} \label{sec:v_potential}

For the massively-loaded case, we will use the exponential potential $V$ with a parameter $\tilde{\alpha} = \Theta(1/n)$ (to be fixed below) and show that its expectation is $\Oh(n^4)$. In \cref{clm:v_expected_value}, using Markov's inequality, we will deduce that \Whp~the gap at an arbitrary round is $\Oh(n \log n)$.

As in \cref{sec:non_filling_lambda}, the analysis will be done over the three possible bins: robustly underloaded bins $B_{++}^t$ with load $y_i^t \geq \frac{w_-}{n}$, robustly underloaded bins $B_{--}^t$ with load $y_i^t \leq -w_-$ and swinging bins $B_{+/-}^t$ with load $y_i^t \in (-w_-, \frac{w_-}{n})$. %

\begin{lem} \label{lem:v_multiplicative_drop}
For any \WThree-process set 
\[
\tilde{\alpha} := \min\left\lbrace\frac{\C{p2k1}}{2 \cdot w_- \cdot (n - \C{p2k1})}, \frac{\C{p2k2}}{2 \cdot w_- \cdot (n + \C{p2k2})} \right\rbrace,
\]
and for any \PThree-process set
\[
 \tilde{\alpha} := \min\left\lbrace \frac{w_- - w_+}{4 \cdot w_-^2 \cdot n}, \frac{1}{w_- \cdot (4n + 2)} \right\rbrace.
\]
Then there exists a constant $\C{v_mult_factor} > 0$ such that for any $t \geq 0$,
\[
\Ex{ V^{t+1} \mid \mathfrak{F}^t} 
  \leq V^t \cdot \left(1 - \frac{\C{v_mult_factor}}{n^3} \right) + 2n.
\]
\end{lem}

\begin{proof}
Following the text before the lemma, we divide $\Ex{ V^{t+1} \mid \mathfrak{F}^t} $ into three cases of bins as follows,
\begin{align}
    \Ex{ V^{t+1} \mid \mathfrak{F}^t} 
    &= \sum_{i \in B_{++}^t} \Ex{ V_i^{t+1} \mid \mathfrak{F}^t} 
      +
    \sum_{i \in B_{--}^t} \Ex{ V_i^{t+1} \mid \mathfrak{F}^t} + \sum_{i \in B_{+/-}^t} \Ex{ V_i^{t+1} \mid \mathfrak{F}^t}, \label{eq:split}
\end{align}
and upper bound each of the three sums separately. Note that if $\delta^t = 1$, then all bins have $y^t=0$, so all bins belong to \textbf{Case 3}. So, in \textbf{Case 1} and \textbf{Case 2}, we can assume that $1/n \leq \delta^t \leq 1 - 1/n$.

\noindent \textbf{Case 1} [Robustly Overloaded Bins]. For $i \in B_{++}^t$,
\begin{align}
\Ex{ V_i^{t+1} \mid \mathfrak{F}^t}
 & = V_i^t \cdot \Big(p_i^t \cdot e^{-\tilde{\alpha} w_+/n + \tilde{\alpha} w_+} + (P_+^t - p_i^t) \cdot e^{-\tilde{\alpha} w_+/n} + P_-^t \cdot e^{-\tilde{\alpha} w_-/n} \Big) \notag \\
 \intertext{Applying the Taylor estimate $e^{z} \leq 1+z+z^2$, which holds for any $z \leq 1.75$, since $\tilde{\alpha} \leq \frac{1}{n}$,}
\Ex{ V_i^{t+1} \mid \mathfrak{F}^t}  &\leq V_i^t \cdot \Big(1 + p_i^t \cdot \Big(-\frac{\tilde{\alpha} w_+}{n} + \tilde{\alpha} w_+ + \Big(-\frac{\tilde{\alpha} w_+}{n} + \tilde{\alpha} w_+\Big)^2\Big)  \notag \\
 & \qquad + (P_+^t - p_i^t) \cdot \Big(-\frac{\tilde{\alpha} w_+}{n} + \Big(\frac{\tilde{\alpha} w_+}{n}\Big)^2\Big) + P_-^t \cdot \Big(- \frac{\tilde{\alpha} w_-}{n} + \Big(\frac{\tilde{\alpha} w_-}{n}\Big)^2 \Big) \Big).\notag \\
 \intertext{Using the fact that $\frac{\tilde{\alpha}^2}{n^2} = \Oh(n^{-4})$ and then subsequently rearranging terms we obtain}
 \Ex{ V_i^{t+1} \mid \mathfrak{F}^t} &\leq V_i^t \cdot \Big(1 + p_i^t \cdot \Big(-\frac{\tilde{\alpha} w_+}{n} + \tilde{\alpha} w_+ + \Big(-\frac{\tilde{\alpha} w_+}{n} + \tilde{\alpha} w_+\Big)^2\Big) \notag \\ 
 & \qquad - (P_+^t - p_i^t) \cdot \frac{\tilde{\alpha} w_+}{n} - P_-^t \cdot \frac{\tilde{\alpha} w_-}{n} +o(n^{-3})\Big) \notag \\
 & = V_i^t \cdot \Big(1 + p_i^t \cdot \Big(\tilde{\alpha} w_+ + \Big(-\frac{\tilde{\alpha} w_+}{n} + \tilde{\alpha} w_+\Big)^2\Big) - P_+^t \cdot \frac{\tilde{\alpha} w_+}{n} - P_-^t \cdot \frac{\tilde{\alpha} w_-}{n} + o(n^{-3}) \Big). \notag \\
 \intertext{Applying $(-\frac{\tilde{\alpha} w_{+}}{n} + \tilde{\alpha} w_{+})^2 \leq (\tilde{\alpha} w_{+})^2$ then gives}
 \Ex{ V_i^{t+1} \mid \mathfrak{F}^t} & \leq V_i^t \cdot \Big(1 + p_i^t \cdot (\tilde{\alpha} w_+ + (\tilde{\alpha} w_+)^2) - P_+^t \cdot \frac{\tilde{\alpha} w_+}{n} - P_-^t \cdot \frac{\tilde{\alpha} w_-}{n} + o(n^{-3}) \Big) \notag \\
 & = V_i^t \cdot \Big(1 + p_i^t \cdot (\tilde{\alpha} w_+ + (\tilde{\alpha} w_+)^2) + P_+^t \cdot \Big(\frac{\tilde{\alpha} w_-}{n} - \frac{\tilde{\alpha} w_+}{n} \Big)- \frac{\tilde{\alpha} w_-}{n} + o(n^{-3}) \Big) \notag ,\\
 \intertext{where in the first equality we replaced $P_{-}^t$ by $1- P_{+}^t$ and rearranged. Using $P_+^t \leq \delta^t \leq 1 - \frac{1}{n}$,}
\Ex{ V_i^{t+1} \mid \mathfrak{F}^t}  & \leq  V_i^t \cdot \Big(1 + p_i^t \cdot (\tilde{\alpha} w_+ + (\tilde{\alpha} w_+)^2) + \Big(1 - \frac{1}{n}\Big) \cdot \Big(\frac{\tilde{\alpha} w_-}{n} - \frac{\tilde{\alpha} w_+}{n} \Big)- \frac{\tilde{\alpha} w_-}{n} + o(n^{-3}) \Big) \notag \\
 & =  V_i^t \cdot \Big(1 - \tilde{\alpha} w_+ \cdot \Big(\frac{1}{n} - p_i^t \cdot (1 + \tilde{\alpha} w_+)  \Big)  -\frac{\tilde{\alpha}}{n^2} \cdot (w_- - w_+) + o(n^{-3}) \Big). \label{eq:overloaded_sum}
 \end{align}
\textbf{Case 1.A} [\PThree holds]. The bounds $p_+^t \leq \frac{1}{n} - \frac{\C{p2k1} \cdot (1 - \delta^t)}{n} \leq \frac{1}{n} - \frac{\C{p2k1}}{n^2}$ and $w_- \geq w_+$ applied to \eqref{eq:overloaded_sum} above give
\begin{align*}
\Ex{ V_i^{t+1} \mid \mathfrak{F}^t}
 & \leq V_i^t \cdot \Big(1 - \tilde{\alpha} w_+ \cdot \Big(\frac{1}{n} - \Big( \frac{1}{n} - \frac{\C{p2k1}}{n^2} \Big)\cdot (1 + \tilde{\alpha} w_+) \Big) + o(n^{-3}) \Big) \\
 & = V_i^t \cdot \Big(1 - \frac{\tilde{\alpha} w_+}{n} \cdot \Big(\frac{\C{p2k1}}{n} - \Big(1 - \frac{\C{p2k1}}{n} \Big)\cdot (\tilde{\alpha} w_+) \Big) + o(n^{-3}) \Big) \\
 & \stackrel{(a)}{\leq} V_i^t \cdot \Big(1 - \frac{\tilde{\alpha} w_+}{n} \cdot \frac{\C{p2k1}}{2n} + o(n^{-3}) \Big) \\
 &\stackrel{(b)}{\leq} V_i^t \cdot \Big(1 - \frac{\tilde{\alpha} w_+ \C{p2k1}}{4n^2}\Big),
\end{align*}
where we used in $(a)$ that $\frac{\C{p2k1}}{2n} \leq (1 - \frac{\C{p2k1}}{n} )\cdot \tilde{\alpha} w_+$ (as implied by $\tilde{\alpha} \leq \frac{\C{p2k1}}{2 \cdot w_- \cdot (n - \C{p2k1})}$) and in $(b)$ that $\tilde{\alpha} = \Omega(n^{-1})$.

\noindent \textbf{Case 1.B} [\WThree holds]. Applying the bound $p_+^t \leq \frac{1}{n}$ to \eqref{eq:overloaded_sum} yields
\begin{align*}
\Ex{ V_i^{t+1} \mid \mathfrak{F}^t}
 & \leq V_i^t \cdot \Big(1 - \tilde{\alpha} w_+ \cdot \Big(\frac{1}{n} - \frac{1}{n}\cdot (1 + \tilde{\alpha} w_+) \Big) -\frac{\tilde{\alpha}}{n^2} \cdot (w_- - w_+) + o(n^{-3}) \Big) \\
 & = V_i^t \cdot \Big(1 +\frac{(\tilde{\alpha} w_+)^2}{n} -\frac{\tilde{\alpha}}{n^2} \cdot (w_- - w_+) + o(n^{-3}) \Big) \\
 & \stackrel{(a)}{\leq} V_i^t \cdot \Big(1 -\frac{\tilde{\alpha}}{2n^2} \cdot (w_- - w_+) + o(n^{-3}) \Big)\\
 &\stackrel{(b)}{\leq} V_i^t \cdot \Big(1 -\frac{\tilde{\alpha}}{4n^2} \cdot (w_- - w_+)\Big),
\end{align*}
where we used in $(a)$ that $\frac{(\tilde{\alpha} w_+)^2}{2n} \leq \frac{\tilde{\alpha}}{n^2} \cdot (w_- - w_+)$ (as implied by $\tilde{\alpha} \leq \frac{w_- - w_+}{4 \cdot w_-^2 \cdot n} \leq \frac{w_- - w_+}{2 \cdot w_+^2 \cdot n}$) and in $(b)$ that $\tilde{\alpha} = \Omega(n^{-1})$.

So, in both \textbf{Case 1.A} and \textbf{Case 1.B}, we get for some constant $\C{v_mult_factor} > 0$,
\[
\sum_{i \in B_{++}^t} \Ex{ V_i^{t+1} \mid \mathfrak{F}^t} 
  \leq \sum_{i \in B_{++}^t} V_i^t \cdot \Big(1 - \frac{\C{v_mult_factor}}{n^3} \Big).
\]

\noindent \textbf{Case 2} [Robustly Underloaded Bins]. For any $i \in B_{--}^t$, using again the Taylor estimate $e^z \leq 1 + z + z^2$ (since $\tilde{\alpha} = \Theta(1/n)$) yields
\begin{align}
\Ex{ V_i^{t+1} \mid \mathfrak{F}^t}
 & = V_i^t \cdot (p_i^t \cdot e^{\tilde{\alpha} w_-/n - \tilde{\alpha} w_-} + (P_-^t - p_i^t) \cdot e^{\tilde{\alpha} w_-/n} + P_+^t \cdot e^{\tilde{\alpha} w_+/n}) \notag \\
 & \leq V_i^t \cdot \Big(1 + p_i^t \cdot \Big(\frac{\tilde{\alpha} w_-}{n} - \tilde{\alpha} w_- + \Big(\frac{\tilde{\alpha} w_-}{n} - \tilde{\alpha} w_-\Big)^2 \Big) \notag \\ 
 & \qquad + (P_-^t - p_i^t) \cdot \Big(\frac{\tilde{\alpha} w_-}{n} + \Big(\frac{\tilde{\alpha} w_-}{n}\Big)^2\Big) + P_+^t \cdot \Big(\frac{\tilde{\alpha} w_+}{n} + \Big(\frac{\tilde{\alpha} w_+}{n}\Big)^2 \Big) \Big) \notag \\
 \intertext{Using the fact that $\frac{\tilde{\alpha}^2}{n^2} = \Oh(n^{-4})$ and then substituting $P_{+}^t=1-P_{-}^t$ yields}
 & \leq V_i^t \cdot \Big(1 + p_i^t \cdot \Big(\frac{\tilde{\alpha} w_-}{n} - \tilde{\alpha} w_- + \Big(\frac{\tilde{\alpha} w_-}{n} - \tilde{\alpha} w_-\Big)^2 \Big)  \notag \\ 
 & \qquad + (P_-^t - p_i^t) \cdot \frac{\tilde{\alpha} w_-}{n} + P_+^t \cdot \frac{\tilde{\alpha} w_+}{n} + o(n^{-3}) \Big) \notag \\
 & = V_i^t \cdot \Big(1 - p_i^t \cdot (\tilde{\alpha} w_- - \Big(-\frac{\tilde{\alpha} w_-}{n} + \tilde{\alpha} w_-\Big)^2\Big) + P_-^t \cdot \Big(\frac{\tilde{\alpha} w_-}{n} - \frac{\tilde{\alpha} w_+}{n}\Big) + \frac{\tilde{\alpha} w_+}{n} + o(n^{-3}) \Big). \notag \\
 \intertext{Now, $(-\frac{\tilde{\alpha} w_-}{n} + \tilde{\alpha} w_-)^2 \leq (\tilde{\alpha} w_-)^2$ implies that}
  \Ex{ V_i^{t+1} \mid \mathfrak{F}^t}& \leq V_i^t \!\cdot\! \Big(1 - p_i^t \cdot (\tilde{\alpha} w_- - (\tilde{\alpha} w_-)^2) + P_-^t \cdot \Big(\frac{\tilde{\alpha} w_-}{n} - \frac{\tilde{\alpha} w_+}{n}\Big) + \frac{\tilde{\alpha} w_+}{n} + o(n^{-3}) \Big). \label{eq:underloaded_bins}
\end{align}

\textbf{Case 2.A} [\PThree holds]. Using that $p_i^t \geq \frac{1 + \C{p2k2} \cdot \delta^t}{n} \geq \frac{1 + \C{p2k2} \cdot \frac{1}{n}}{n} \geq \frac{1}{n} + \frac{\C{p2k2}}{n^2}$, $P_-^t \leq 1$, and applying this to the previous inequality yields $(a)$ below (since $(\tilde{\alpha}w_-)^2 \leq \tilde{\alpha}w_-$ the factor after $p_i^t$ is positive), and further rearranging gives
\begin{align*}
 \Ex{ V_i^{t+1} \mid \mathfrak{F}^t}
  & \stackrel{(a)}{\leq} V_i^t \cdot \Big(1 - \Big(\frac{1}{n} + \frac{\C{p2k2}}{n^2}\Big) \cdot (\tilde{\alpha} w_- - (\tilde{\alpha} w_-)^2) + 1 \cdot \Big(\frac{\tilde{\alpha} w_-}{n} - \frac{\tilde{\alpha} w_+}{n}\Big) + \frac{\tilde{\alpha} w_+}{n} + o(n^{-3}) \Big) \\
 & = V_i^t \cdot \Big(1 - \tilde{\alpha} w_- \cdot \Big(\Big( \frac{1}{n} + \frac{\C{p2k2}}{n^2}\Big) \cdot (1 - \tilde{\alpha} w_-) - \frac{1}{n}\Big) + o(n^{-3})\Big)  \\
 & = V_i^t \cdot \Big(1 - \frac{\tilde{\alpha} w_-}{n} \cdot \Big(\frac{\C{p2k2}}{n} -\Big(1 + \frac{\C{p2k2}}{n} \Big) \cdot \tilde{\alpha} w_-\Big) + o(n^{-3}) \Big) \\
 & \stackrel{(b)}{\leq} V_i^t \cdot \Big(1 - \frac{\tilde{\alpha} w_-}{n} \cdot \frac{\C{p2k2}}{2n} + o(n^{-3}) \Big) \\
 &\stackrel{(c)}{\leq} V_i^t \cdot \Big(1 - \frac{\tilde{\alpha} w_- \C{p2k2}}{4n^2}\Big).
\end{align*}
where we used in $(b)$ that
$\frac{\C{p2k2}}{2n} \geq (1 + \frac{\C{p2k2}}{n}) \cdot \tilde{\alpha} w_-$ (equivalent to $\tilde{\alpha} \leq \frac{\C{p2k2}}{2 \cdot w_- \cdot (n + \C{p2k2})}$) and in $(c)$ that $\alpha = \Omega(n^{-1})$.

\noindent \textbf{Case 2.B} [\WThree holds and $P_-^t \leq 1 - \frac{1}{2n}$]. Applying $p_i^t \geq 1/n$ and $P_-^t \leq 1 - \frac{1}{2n}$ to \eqref{eq:underloaded_bins} yields $(a)$ below, and then further rearranging gives
\begin{align*}
 \Ex{ V_i^{t+1} \mid \mathfrak{F}^t}
 & \stackrel{(a)}{\leq} V_i^t \cdot \Big(1 - \frac{1}{n} \cdot (\tilde{\alpha} w_- - (\tilde{\alpha} w_-)^2) + \Big( 1 - \frac{1}{2n}\Big) \cdot \Big(\frac{\tilde{\alpha} w_-}{n} - \frac{\tilde{\alpha} w_+}{n}\Big) + \frac{\tilde{\alpha} w_+}{n} + o(n^{-3}) \Big) \\
 & = V_i^t \cdot \Big(1 - \Big(\frac{1}{n} \cdot (1 - \tilde{\alpha} w_-) - \frac{1}{n}\Big) \cdot \tilde{\alpha} w_- - \frac{1}{2n} \cdot \Big(\frac{\tilde{\alpha} w_-}{n} - \frac{\tilde{\alpha} w_+}{n}\Big) + o(n^{-3}) \Big) \\
 & = V_i^t \cdot \Big(1 + \frac{(\tilde{\alpha} w_-)^2}{n} - \frac{\tilde{\alpha}}{2n^2} \cdot (w_- - w_+) + o(n^{-3}) \Big) \\
 & \stackrel{(b)}{\leq} V_i^t \cdot \Big(1 - \frac{\tilde{\alpha}}{4n^2} \cdot (w_- - w_+) + o(n^{-3}) \Big)\\
& \stackrel{(c)}{\leq} V_i^t \cdot \Big(1 - \frac{\tilde{\alpha}}{8n^2} \cdot (w_- - w_+) \Big),
\end{align*}
where we used in $(b)$ that $\frac{\tilde{\alpha}}{4n^2} \cdot (w_- - w_+)\geq \frac{(\tilde{\alpha} w_-)^2}{n}$ (equivalent to $\tilde{\alpha} \leq \frac{w_- - w_+}{4 \cdot w_-^2 \cdot n}$) and in $(c)$ that $\tilde{\alpha} = \Omega(n^{-1})$ and $w_- - w_+$ is a positive constant.

\noindent \textbf{Case 2.C} [\WThree holds and $P_-^t > 1 - \frac{1}{2n}$]. We will derive a similar inequality as in the previous cases, but we will take the sum over all $i \in B_{--}^t$ first to apply a majorization argument.  %
Aggregating the contributions over all bins in $B_{--}^t$ yields (taking the sum of \eqref{eq:underloaded_bins}),
\begin{align*}
\sum_{i \in B_{--}^t} \ex{V_i^{t+1} \mid \mathfrak{F}^t} 
 & \leq \sum_{i \in B_{--}^t} V_i^t \cdot \Big(1 - p_i^t \cdot (\tilde{\alpha} w_- - (\tilde{\alpha} w_-)^2) + P_-^t \cdot \Big(\frac{\tilde{\alpha} w_-}{n} - \frac{\tilde{\alpha} w_+}{n}\Big) + \frac{\tilde{\alpha} w_+}{n} + o(n^{-3}) \Big) \\
 & \leq  \sum_{i \in B_{--}^t} V_i^t \cdot \Big(1 - p_i^t \cdot (\tilde{\alpha} w_- - (\tilde{\alpha} w_-)^2) + 1 \cdot \Big(\frac{\tilde{\alpha} w_-}{n} - \frac{\tilde{\alpha} w_+}{n} \Big) + \frac{\tilde{\alpha} w_+}{n} + o(n^{-3}) \Big) \\
 & = \sum_{i \in B_{--}^t} V_i^t \cdot \Big(1 - p_i^t \cdot (\tilde{\alpha} w_- - (\tilde{\alpha} w_-)^2) + \frac{\tilde{\alpha} w_-}{n} + o(n^{-3}) \Big) \\
 & = \sum_{i \in B_{--}^t} V_i^t \cdot \Big(1 - \tilde{\alpha} w_- \cdot \Big(p_i^t \cdot (1 - \tilde{\alpha} w_-) - \frac{1}{n} \Big) + o(n^{-3})\Big),
\end{align*}
Since $V_i^t$ is non-decreasing for $i \in B_{--}^t$ (recall the bins are labeled in non-increasing order) and $p_i^t$ is non-decreasing, applying \cref{lem:quasilem}, we can upper bound the above expression by replacing each $p_i^t$ by the average probability $\overline{p}_i^t$ of robustly underloaded bins. Note that because of monotonicity of $p_i^t$, the average probability over $B_{--}^t$ is at least as large as the average probability over $B_-^t$, so it satisfies
\[
\overline{p}_i^t \geq  \frac{P_-^t}{|B_-^t|} 
\geq \frac{1 - \frac{1}{2n}}{n-1} = \frac{1 - \frac{1}{2n}}{1 - \frac{1}{n}} \cdot \frac{1}{n} \geq \frac{1}{n} + \frac{1}{2n^2},
\]
where we have used the case assumption $P_-^t > 1 - \frac{1}{2n}$ and the invariant $|B_-^t| \leq n - 1$ and \cref{clm:eps_ineq} for $\epsilon := \frac{1}{n}$ (and $n > 1$). Hence, since the factor after $p_i^t$ is positive as $\tilde{\alpha} w_- \leq 1$, we get 
\begin{align*}
 \sum_{i \in B_{--}^t} \ex{V_i^{t+1} \mid \mathfrak{F}^t}
 & \leq \sum_{i \in B_{--}^t} V_i^t \cdot \Big(1 - \tilde{\alpha} w_- \cdot \Big( \Big(\frac{1}{n} + \frac{1}{2n^2}\Big) \cdot (1 - \tilde{\alpha} w_-) - \frac{1}{n} \Big) + o(n^{-3})\Big) \\ 
 & = \sum_{i \in B_{--}^t} V_i^t \cdot \Big(1 - \frac{\tilde{\alpha} w_-}{n} \cdot \Big( \frac{1}{2n} - \tilde{\alpha} w_- \cdot \Big( 1 + \frac{1}{2n}\Big) \Big) + o(n^{-3})\Big). \end{align*} 
 Now, since $\frac{1}{4n} \geq \tilde{\alpha} w_- \cdot ( 1 + \frac{1}{2n})$ (as implied by $\tilde{\alpha} \leq \frac{1}{w_- \cdot (4n + 2)}$)  we have 
 \begin{align*} 
 \sum_{i \in B_{--}^t} \ex{V_i^{t+1} \mid \mathfrak{F}^t} &  \leq \sum_{i \in B_{--}^t} V_i^t \cdot \Big(1 - \frac{\tilde{\alpha} w_-}{n} \cdot \frac{1}{4n} + o(n^{-3})\Big).%
\end{align*}
Finally as  $\tilde{\alpha} = \Omega(n^{-1})$ for the contribution of the robustly underloaded bins we have
\[  \sum_{i \in B_{--}^t} \ex{V_i^{t+1} \mid \mathfrak{F}^t} \leq \sum_{i \in B_{--}^t} V_i^t \cdot \left(1 - \frac{\tilde{\alpha} w_-}{16n^2}\right).\]

So, in all \textbf{Case 2.A}, \textbf{Case 2.B} and \textbf{Case 2.C}, we get for some constant $\C{v_mult_factor} > 0$,
\[
\sum_{i \in B_{--}^t} \Ex{ V_i^{t+1} \mid \mathfrak{F}^t} 
  \leq \sum_{i \in B_{--}^t} V_i^t \cdot \Big(1 - \frac{\C{v_mult_factor}}{n^3} \Big).
\]

\noindent\textbf{Case 3} [Swinging Bins]. For bins $B_{+/-}^t$ with load $y_i^t \in (-w_-, \frac{w_-}{n})$, we have,
\[
\sum_{i \in B_{+/-}^t} \ex{V_i^{t+1} \mid \mathfrak{F}^t} \leq
\sum_{i \in B_{+/-}^t} e^{2 \tilde{\alpha} w_-} \leq 2n \leq 2n + \sum_{i \in B_{+/-}^t} V_i^t \cdot \Big(1 - \frac{\C{v_mult_factor}}{n^3}\Big).
\]
Applying the above inequalities for the six cases to \eqref{eq:split}, we get the claim of the lemma.
\end{proof}

\begin{lem} \label{clm:v_expected_value}
For any $\PThree \cap \WTwo$-process or $\PTwo \cap \WThree$-process, it holds that for any $t \geq 0$,
\[
\Ex{ V^{t}} \leq \frac{2}{\C{v_mult_factor}} \cdot n^4.
\]
\end{lem}
\begin{proof}[Proof of \cref{clm:v_expected_value}]
We will prove this by induction. The base case follows trivially, since $V^0 = n$. Assume that $\Ex{V^{t}} \leq n^4 \cdot \frac{2}{\C{v_mult_factor}}$ for some $t \geq 0$, then by \cref{lem:v_multiplicative_drop} and the tower property of expectation,
\[
\Ex{V^{t+1}} %
\leq \Ex{V^t} \cdot \Big(1 - \frac{\C{v_mult_factor}}{n^3}\Big) + 2n \leq \frac{2}{\C{v_mult_factor}} \cdot n^4 - \frac{2n^4}{\C{v_mult_factor}} \cdot \frac{\C{v_mult_factor}}{n^3} + 2n = \frac{2}{\C{v_mult_factor}} \cdot n^4.
\qedhere \]
\end{proof}
We can now apply the above results to deduce that there is a $\poly(n)$ gap at an arbitrary round. This will be used in the proof of \cref{thm:main_technical} as the starting point of the ``recovery phase'' (\cref{lem:recovery}).  

\begin{lem} \label{lem:initial_gap_nlogn}
For any $\PThree \cap \WTwo$-process or $\PTwo \cap \WThree$-process, there is some constant $\C{poly_n_gap} > 0$ such that for any $m \geq 1$,
\[
\Pro{\Gap(m) \leq \C{poly_n_gap} \cdot n \log n} \geq 1 - n^{-12},
\]
and so, for any constant $\alpha < 1$,
\[
\Pro{\Lambda^m \leq \exp(2\C{poly_n_gap} n \log n)} \geq 1 - n^{-12}.
\]
\end{lem}
\begin{proof}
Using \cref{clm:v_expected_value}, we have that
\[
\Ex{ V^{m}} \leq \frac{2}{\C{v_mult_factor}} \cdot n^4.
\]
By Markov's inequality, we have $\Pro{V^m \leq \frac{2}{\C{v_mult_factor}} \cdot n^{16}}  \geq 1 - n^{-12}$. By taking the log of the potential function we obtain 
\begin{align*}
 \Pro{\max_{i\in [n]} |y_i^m| \leq  \log\left(\frac{2}{\C{v_mult_factor}} \right) + \frac{16}{\tilde{\alpha}} \cdot \log n } \geq 1 - n^{-12}, 
\end{align*} thus, as $\tilde{\alpha}=\Theta(1/n)$, we have $ \Pro{\Gap(m) \leq \C{poly_n_gap} n \log n } \geq 1 - n^{-12}$ for some constant $\C{poly_n_gap} > 0$.
Finally, since $\alpha < 1$, when $\max_{i \in [n]} |y_i^m| \leq \C{poly_n_gap} n \log n$, we get $\Lambda^m \leq n \cdot \exp(\C{poly_n_gap} n \log n) \leq \exp(2\C{poly_n_gap} n \log n)$ and deduce the second statement.
\end{proof}

\section{Analysis of Non-Filling Processes} \label{sec:non_filling_analysis}
 
In this section, we complete the proof that any $\PThree \cap \WTwo$ or $\PTwo \cap \WThree$-process satisfies $\Gap(m) = \Oh(\log n)$ \Whp~at an arbitrary round $m$. As outlined in \cref{sec:analysis_thinning}, the proof consists of a recovery phase and a stabilization phase.  \cref{fig:recovery_stabilisation} depicts these two phases. 

In the recovery phase (\cref{sec:recovery})\NOTE{T}{maybe briefly mention te other sections here or later?}\NOTE{J}{Added the stuff in red below}, we prove that starting at $t_0 = m - \Theta(n^3 \log^4 n)$ with $V^{t_0} = \poly(n)$ (i.e., $\Lambda^{t_0} \leq \exp(2 c_6 n \log n)$), \Whp~there exists a round $s \in [t_0, m]$ with $\Lambda^s < cn$. We do this by proving that with constant probability if $\Lambda^{t_0} \leq \exp(2 c_6 n \log n)$, then there is a constant fraction of rounds $r \in [t_0, t_0 + n^3 \log^3 n]$ with $\delta^r \in (\eps, 1 - \eps)$. By analyzing an ``adjusted version'' $\tilde{\Lambda}$ of the exponential potential $\Lambda$, taking advantage of the fact that $\Lambda^t$ decreases in expectation when $\delta^t \in (\eps, 1- \eps)$ and increases at most by a smaller factor otherwise (\cref{cor:change_for_large_lambda}),  we show that there exists $s \in [t_0, t_0 +  n^3 \log^3 n]$ such that $\Lambda^s < cn$. By repeating this argument $\Theta(\log n)$ times we amplify the probability to get that \Whp~there exists $s \in [t_0, m]$ with $\Lambda^s < cn$. 

Then, in the stabilization phase (\cref{sec:stabilization}), we first show that if $\Lambda^s < 2cn$ for any round $s$, then $\Lambda^{t} < cn$ for some round $t\in (s, s + \Theta(n \log n)]$. We do this by proving that \Whp~if $\Lambda^s < cn$, then there is a constant fraction of rounds $r \in [s, s + \Theta(n \log n)]$ with $\delta^r \in (\eps, 1 - \eps)$. 
Again, by analyzing an ``adjusted version'' $\tilde{\Lambda}$ of the exponential potential $\Lambda$, we show that \Whp~there exists $t \in (s, s + \Theta(n \log n)]$ such that $\Lambda^t < cn$. Next, we take the union bound over the remaining $\Oh(n^3 \log^4 n)$ rounds, which gives $\Lambda^{r_1} \leq cn$ \Whp~at some $r_1 \in  [m, m + \Theta(n \log n)]$ and $\Lambda^{r_2} \leq cn$ \Whp~at some $r_2 \in [m - \Theta(n \log n), m]$, this in turn implies that $\max_{i \in [n]} |y_i^m| = \Oh(\log n)$. 
 
\red{In order to complete the analysis in Sections \ref{sec:recovery} and \ref{sec:stabilization} described above, we must first give some definitions and establish some technical tools. In particular, in \cref{sec:adj_martin} we shall prove that the new adjusted exponential potential function $\tilde{\Lambda}^t$ is a super-martingale. In \cref{sec:taming} we then prove bounds on the random variables involved in this new potential so that we can utilize the super-martingale property to establish a drop in potential later in the proof.}

In the following, we consider an arbitrary round $t_0 \geq 0$ which will be starting point of our analysis.
For integers $s\geq t_0\geq 0$ and some arbitrary $\eps\in (0,1)$, to be fixed later in \cref{lem:stabilization}, we let $\C{lambda_bound}(\eps)>1$ be the constant from \cref{cor:change_for_large_lambda}, and with it we define the event \begin{equation*}%
\mathcal{E}_{t_0}^{s} := \bigcap_{r \in [t_0, s]} \{\Lambda^r \geq \C{lambda_bound} n\}.  \end{equation*} 
Showing that this event does not hold for suitable $t_0,s$ is a large part of proving the gap bound. For rounds $s \geq t_0$, we let $G_{t_0}^{s}:=G_{t_0}^{s}(\eps)$ be the number of rounds $r \in [t_0, s]$ with $\delta^r \in (\epsilon, 1 - \epsilon)$, and similarly we define $B_{t_0}^{s} := (s - t_0+1) - G_{t_0}^s$. 

We now introduce the \textit{adjusted exponential potential function} $\tilde{\Lambda}_{t_0}^s$ which involves the random variables and event above. Let $\C{good_quantile_mult}(\eps)$ be the constant in Lemma \ref{lem:good_quantile_good_decrease} and $0<\gamma \leq 1$ be an arbitrary constant (fixed in \cref{lem:recovery}). Then we set $\tilde{\Lambda}_{t_0}^{t_0} := {\Lambda}^{t_0} $ and, for any $s > t_0$, we define the sequence
\begin{equation}\label{eq:newpot}
\tilde{\Lambda}_{t_0}^s := \Lambda^s \cdot \mathbf{1}_{\mathcal{E}_{t_0}^{s-1}} \cdot \exp\left( - \frac{\C{good_quantile_mult} \alpha \gamma}{n} \cdot B_{t_0}^{s-1} \right) \cdot \exp\left( + \frac{\C{good_quantile_mult} \alpha}{n} \cdot G_{t_0}^{s-1} \right). 
\end{equation}

In this section we shall fix many of the constants such as $\alpha$. We begin by setting \[C := \left\lceil\frac{8\C{quad_const_add}}{\C{quad_delta_drop}}\right\rceil+1,\] where $\C{quad_delta_drop}, \C{quad_const_add}$ are the constants from \cref{lem:quadratic_absolute_relation_for_w_plus_w_minus}. 

We also define $\tilde{G}_{t_0}^{s} := \tilde{G}_{t_0}^{s}(C)$ to be the number of rounds $r \in [t_0, s]$ with $\Delta^r \leq C \cdot n$. 

\subsection{The Adjusted Exponential Potential is a Super-Martingale} \label{sec:adj_martin}
We now show that, for a suitably choice of parameters, the sequence defined by \eqref{eq:newpot} forms a super-martingale.

\begin{lem}\label{lem:gamma_tilde_is_supermartingale}
Consider any $\PThree \cap \WTwo$-process or $\PTwo \cap \WThree$-process. Let $\epsilon\in (0,1)$ and $0 < \gamma \leq 1$ be arbitrary constants, and let
$\C{good_quantile_mult}:=\C{good_quantile_mult}(\eps)$ be the constant in Lemma \ref{lem:good_quantile_good_decrease}. Further, let $0 < \alpha \leq \C{good_quantile_mult} \gamma/(3w_- \cdot e^{2 w_-})$ be a constant which additionally meets the conditions of \cref{lem:good_quantile_good_decrease}. Then for any $t_0 \geq 0$, we have for any $s \geq t_0$ that
\[
 \ex{\tilde{\Lambda}_{t_0}^{s+1} \mid \mathfrak{F}^s} \leq \tilde{\Lambda}_{t_0}^{s}.
\]
\end{lem}

\begin{proof}
We see that $\ex{ \tilde{\Lambda}_{t_0}^{s+1}   \mid \mathfrak{F}^s}$ is given by 
\begin{align*}
   &\ex{ \Lambda^{s+1} \cdot \mathbf{1}_{\mathcal{E}_{t_0}^{s}} \mid \mathfrak{F}^s} \cdot \exp\left( - \frac{\C{good_quantile_mult} \alpha \gamma}{n} \cdot B_{t_0}^{s} \right) \cdot \exp\left( \frac{\C{good_quantile_mult} \alpha}{n} \cdot G_{t_0}^{s} \right) \\ 
 & = \ex{ \Lambda^{s+1} \cdot \mathbf{1}_{\mathcal{E}_{t_0}^{s}} \mid \mathfrak{F}^s} \cdot \exp\left(\frac{\C{good_quantile_mult} \alpha \gamma}{n} \cdot \Big(\frac{\gamma + 1}{\gamma} \cdot \mathbf{1}_{\mathcal{G}^s} - 1\Big) \right) \cdot \exp\left( - \frac{\C{good_quantile_mult} \alpha \gamma}{n} \cdot B_{t_0}^{s - 1} \right) \cdot \exp\left( \frac{\C{good_quantile_mult} \alpha}{n} \cdot G_{t_0}^{s - 1} \right).
\end{align*}
Thus, we see that it suffices to prove that \begin{equation}\label{eq:sufficient2}\ex{\Lambda^{s+1} \cdot \mathbf{1}_{\mathcal{E}_{t_0}^{s}} \mid \mathfrak{F}^s} \cdot \exp\left(\frac{\C{good_quantile_mult} \alpha \gamma}{n} \cdot \Big(\frac{\gamma + 1}{\gamma} \cdot \mathbf{1}_{\mathcal{G}^s} - 1\Big) \right) \leq \Lambda^{s} \cdot \mathbf{1}_{\mathcal{E}_{t_0}^{s-1}}.\end{equation}To show \eqref{eq:sufficient2}, we consider two cases based on whether $\mathcal{G}^s$ holds.

\medskip 

\noindent\textbf{Case 1} [$\mathcal{G}^s$ holds]. Recall that the event $\mathcal{G}^s$ means $\delta^s\in(\eps,1-\eps)$ holds. Further, we are additionally conditioning on the event $\mathcal{E}_{t_0}^{s}$, via the indicator, and so $\Lambda^t>cn$ holds for any round $t \in [t_0,s]$. Thus we can use the upper bound from \cref{cor:change_for_large_lambda} to give  
\[
\ex{ \Lambda^{s+1} \cdot \mathbf{1}_{\mathcal{E}_{t_0}^{s}} \mid \mathfrak{F}^s, \mathcal{G}^s} \leq \Lambda^{s} \cdot \mathbf{1}_{\mathcal{E}_{t_0}^{s-1}}  \cdot \left(1 - \frac{\C{good_quantile_mult} \alpha}{n} \right) \leq \Lambda^{s} \cdot \mathbf{1}_{\mathcal{E}_{t_0}^{s-1}} \cdot \exp\left(- \frac{\C{good_quantile_mult} \alpha}{n} \right).
\]
Hence, since in this case $\mathbf{1}_{\mathcal{G}^s}=1$, the left hand side of \eqref{eq:sufficient2} is equal to
\begin{align*}
\ex{\Lambda^{s+1} \cdot \mathbf{1}_{\mathcal{E}_{t_0}^{s}} \mid \mathfrak{F}^s, \mathcal{G}^s} \cdot \exp\left(\frac{\C{good_quantile_mult} \alpha \gamma}{n} \cdot \left(\frac{\gamma + 1}{\gamma} - 1\right) \right) &\leq \left(\Lambda^{s} \cdot \mathbf{1}_{\mathcal{E}_{t_0}^{s-1}} \cdot \exp\left(- \frac{\C{good_quantile_mult} \alpha}{n} \right)\right) \cdot \exp\left( \frac{\C{good_quantile_mult} \alpha}{n} \right) \\ & = \Lambda^{s} \cdot \mathbf{1}_{\mathcal{E}_{t_0}^{s-1}} .
\end{align*}
\noindent\textbf{Case 2} [$\mathcal{G}^s$ does not hold]. Recall from \cref{cor:change_for_large_lambda} that $\C{bad_quantile_mult} = 3w_- \cdot e^{2 w_-}$, thus our condition on $\alpha$ can be expressed as $\alpha \leq \frac{\C{good_quantile_mult}\gamma }{\C{bad_quantile_mult}}$. The second inequality of~\cref{cor:change_for_large_lambda} then  implies \[
\ex{\Lambda^{s+1} \cdot \mathbf{1}_{\mathcal{E}_{t_0}^{s}} \mid \mathfrak{F}^s, \neg \mathcal{G}^s} \leq \Lambda^{s} \cdot \mathbf{1}_{\mathcal{E}_{t_0}^{s-1}} \cdot \left(1 + \frac{\C{bad_quantile_mult} \alpha^2 }{n}\right)  \leq \Lambda^{s} \cdot \mathbf{1}_{\mathcal{E}_{t_0}^{s-1}} \cdot \left(1 + \frac{\C{good_quantile_mult} \alpha \gamma}{n}\right) \leq \Lambda^{s} \cdot \mathbf{1}_{\mathcal{E}_{t_0}^{s-1}}\cdot \exp\left(\frac{\C{good_quantile_mult} \alpha \gamma}{n} \right).
\]
Hence,  since in this case $\mathbf{1}_{\mathcal{G}^s}=0$, the left hand side of \eqref{eq:sufficient2} is equal to
\begin{align*}
\ex{\Lambda^{s+1} \cdot \mathbf{1}_{\mathcal{E}_{t_0}^{s}} \mid \mathfrak{F}^s, \neg \mathcal{G}^s} \cdot \exp\left(\frac{\C{good_quantile_mult} \alpha \gamma}{n} \cdot (  - 1) \right) &\leq \left( \Lambda^{s} \cdot \mathbf{1}_{\mathcal{E}_{t_0}^{s-1}}\cdot \exp\left(\frac{\C{good_quantile_mult} \alpha \gamma}{n} \right)\right)   \cdot \exp\left( -\frac{\C{good_quantile_mult} \alpha \gamma}{n}   \right) \\  &= \Lambda^{s} \cdot \mathbf{1}_{\mathcal{E}_{t_0}^{s-1}}.
\end{align*} Since  \eqref{eq:sufficient2} holds in either case, we deduce that $(\tilde{\Lambda}_{t_0}^s)_{s \geq t_0}$ forms a super-martingale. 
\end{proof}

\subsection{Taming the Mean Quantile and Absolute Value Potential} \label{sec:taming}
 
We have seen in the previous section that if we augment the exponential potential $\Lambda^t$ with some terms involving the random variable $G_{t_0}^s$, then the resulting potential $\tilde{\Lambda}^t$ is a super-martingale. Thus it will be useful to control $G_{t_0}^{s}:=G_{t_0}^{s}(\eps)$, which we recall is the number of rounds $r \in [t_0, s]$ with $\delta^r \in (\epsilon, 1 - \epsilon)$. This is done in part by controlling $\tilde{G}_{t_0}^{s} := \tilde{G}_{t_0}^{s}(C)$, which we recall is the number of rounds $r \in [t_0, s]$ with $\Delta^r \leq C \cdot n$. In particular, our first lemma shows that for an interval with length at least $5Cn$, if we have a constant fraction of steps $t$ with $\Delta^t \leq Cn$, then \Whp we also have a constant fraction of steps with $\delta^t \in (\eps, 1 - \eps)$.

\begin{lem}\label{lem:newcorrespondence}
Consider any integer constant $C \geq 1$ and any two rounds $t_0$ and $t_1$ such that $t_0 + 5Cn \leq t_1$. Then there exists a constant $\eps:=\eps(C)>0$ such that %
\[\Pro{ G_{t_0}^{t_1}> \frac{ \epsilon}{5C } \cdot \tilde{G}_{t_0}^{t_1-5Cn} \;\Big|\; \mathfrak{F}^{t_0}} \geq 1- (t_1-t_0) \cdot e^{-\eps n}. \]
\end{lem}

\begin{proof}  Let $\ell:=\frac{n}{10}+ n\cdot (2\C{small_delta}+1)/w_{+} $. For any round $t \in [t_0,t_1]$, define
\[
\mathcal{Q}^t := \Bigl\{ \Delta^{t} > C \cdot n  \Bigr\} \cup \Bigl\{ \bigl| \left\{ s \in [t,t+\ell] \colon 
 \delta^{s} \in (\epsilon,1-\epsilon) \right\} \bigr| \geq \eps \cdot n \Bigr\},
\]
where  $ \eps:=\eps(C)$ is the constant from \cref{lem:good_quantile}.
Note that event $\mathcal{Q}^t$ is logically equivalent to the statement: $\Delta^{t} \leq C \cdot n$ implies 
$ \left| \left\{ s \in [t,t+\ell] \colon 
 \delta^{s} \in (\epsilon,1-\epsilon) \right\} \right| \geq \eps \cdot n$. 
Then by \cref{lem:good_quantile},
\[
 \Pro{ \mathcal{Q}^t \;\Big|\; \mathfrak{F}^{t_0}} \geq 1- e^{-\epsilon n},
\]
and so the union bound gives
\[
  \Pro{ \bigcap_{t=t_0}^{t_1} \mathcal{Q}^{t} \;\Big|\; \mathfrak{F}^{t_0}} \geq 1 - (t_1-t_0) \cdot e^{- \epsilon n}.
\]
In the following, we will condition on the event $\bigcap_{t=t_0}^{t_1} \mathcal{Q}^{t} $. Next define for any $t \in [t_0,t_1-\ell]$,
\[
  g(t) := \left\{ s \in [t,t+\ell] \colon 
 \delta^{s} \in (\epsilon,1-\epsilon) \right\}.
\]Then,
\begin{align*}
 G_{t_0}^{t_1} &\geq \left|
 \bigcup_{t=t_0}^{t_1-\ell} g(t) \right|  \geq \frac{\sum_{t=t_0}^{t_1-\ell} |g(t)|}{ \max_{r \in [t,t+\ell]} \Bigl| \bigr\{t \in [t_0,t_1-\ell] \colon r \in g(t)  \bigr\}\Bigr| }  \geq \frac{ \tilde{G}_{t_0}^{t_1-\ell} \cdot \epsilon \cdot n}{ \ell } \geq \frac{ \tilde{G}_{t_0}^{t_1-5Cn} \cdot \epsilon}{5C },
\end{align*}
where the last step used $\ell \leq 5 Cn$. This completes the proof.\end{proof}

We will now prove using \cref{lem:quadratic_absolute_relation_for_w_plus_w_minus} that when $\Lambda^{t_0} \leq \exp(2\C{poly_n_gap} n \log n)$, half of the rounds $t \in [t_0, t_0 + n^3 \log^3 n]$ satisfy $\Delta^t \leq Cn$ with constant probability. We will use this lemma in the recovery phase (\cref{lem:recovery}).

\begin{lem} \label{lem:average_expected_delta_is_small}
Consider any $\PThree \cap \WTwo$-process or $\PTwo \cap \WThree$-process. Then for $\eps := \eps(C)$ from \cref{lem:good_quantile}, $r= \min\{\frac{\eps}{20C}, \frac{1}{2}\}$, and for any rounds $t_0$ and $t_1$ with $t_1 \geq t_0 + n^3 \log^3 n$, we have
\[
\Pro{ G_{t_0 }^{t_1} > r \cdot (t_1 - t_0) \;\Big|\; \mathfrak{F}^{t_0}, \Lambda^{t_0} \leq \exp(2 \C{poly_n_gap} n \log n)} \geq \frac{1}{2} - (t_1 - t_0)\cdot e^{-\eps n},
\]
where $\C{poly_n_gap} > 0$ is the constant from \cref{lem:initial_gap_nlogn}.
\end{lem}

\begin{proof}

Using \cref{lem:quadratic_absolute_relation_for_w_plus_w_minus}, we have for any $t \geq 0$,
\begin{align*}
\ex{\Upsilon^{t+1} \mid \mathfrak{F}^t} &\leq \Upsilon^{t} - \frac{\C{quad_delta_drop}}{n} \cdot \Delta^t  + \C{quad_const_add}.
\end{align*}
By taking the expectations on both sides, we get
\[
\ex{\Upsilon^{t+1} \mid \mathfrak{F}^{t_0}} = \ex{\ex{\Upsilon^{t+1} \mid \mathfrak{F}^t} \mid \mathfrak{F}^{t_0}} \leq \ex{\Upsilon^{t} \mid \mathfrak{F}^{t_0}} - \frac{\C{quad_delta_drop}}{n} \cdot \ex{\Delta^t \mid \mathfrak{F}^{t_0}} + \C{quad_const_add}.
\]
Applying this to rounds $t_0, t_0+1, \ldots , t_1 + 1$, we get
\begin{align*}
\ex{\Upsilon^{t_0+1} \mid \mathfrak{F}^{t_0}} & \leq  \ex{\Upsilon^{t_0} \mid \mathfrak{F}^{t_0}} - \frac{\C{quad_delta_drop}}{n} \cdot \ex{\Delta^{t_0} \mid \mathfrak{F}^{t_0}} + \C{quad_const_add}, \\
\ex{\Upsilon^{t_0+2} \mid \mathfrak{F}^{t_0}} & \leq  \ex{\Upsilon^{t_0+1} \mid \mathfrak{F}^{t_0}} - \frac{\C{quad_delta_drop}}{n} \cdot \ex{\Delta^{t_0+1} \mid \mathfrak{F}^{t_0}} + \C{quad_const_add}, \\
 & \vdots \\
\ex{\Upsilon^{t_1+1} \mid \mathfrak{F}^{t_0}} & \leq  \ex{\Upsilon^{t_1} \mid \mathfrak{F}^{t_0}} - \frac{\C{quad_delta_drop}}{n} \cdot \ex{\Delta^{t_1} \mid \mathfrak{F}^{t_0}} + \C{quad_const_add}.
\end{align*}
Hence by induction, and using $\ex{\Upsilon^{t_0} \mid \mathfrak{F}^{t_0}} = \Upsilon^{t_0}$, we get
\[
\ex{\Upsilon^{t_1+1} \mid \mathfrak{F}^{t_0}} \leq \Upsilon^{t_0} - \frac{\C{quad_delta_drop}}{n} \cdot \sum_{r = t_0}^{t_1} \ex{\Delta^{r} \mid \mathfrak{F}^{t_0}} + \C{quad_const_add} \cdot (t_1 - t_0 +1).
\]
Since the left hand side is at least $0$, rearranging the above inequality yields
\[
 \sum_{r = t_0}^{t_1} \ex{\Delta^{r} \mid \mathfrak{F}^{t_0}} \leq \Upsilon^{t_0} \cdot \frac{n}{\C{quad_delta_drop}} + \C{quad_const_add} \cdot (t_1 - t_0+1) \cdot \frac{n}{\C{quad_delta_drop}}.
\]
We will now make use of the condition $\Lambda^{t_0} \leq \exp(2c_6 n \log n)$, and conclude by the first statement of \cref{clm:bound_on_gamma_implies_bound_on_upsilon} that
\[
\Upsilon^{t_0} \leq \alpha^{-2}\cdot n \cdot (\log \Lambda^{t_0})^2 \leq 4 \cdot \alpha^{-2} \cdot \C{poly_n_gap}^2 \cdot n^3 \cdot \log^2 n.
\] So, we obtain that
\begin{align*}
 \sum_{r = t_0}^{t_1 } \ex{\Delta^{r} \mid \mathfrak{F}^{t_0}, \Lambda^{t_0} \leq \exp(2 \C{poly_n_gap} n \log n)} &\leq
4 \cdot \alpha^{-2} \cdot \C{poly_n_gap}^2 \cdot n^3 \cdot \log^2 n \cdot \frac{n}{\C{quad_delta_drop}} +\C{quad_const_add} \cdot (t_1 - t_0+1)\cdot \frac{n}{\C{quad_delta_drop}} \\
 & \leq 2 \cdot \frac{\C{quad_const_add}}{\C{quad_delta_drop}} \cdot (t_1 - t_0) \cdot n \leq C \cdot (t_1 - t_0) \cdot \frac{n}{\C{quad_delta_drop}},
\end{align*}
where we have used that $t_1 - t_0  \geq n^3 \log^3 n  > 4 \cdot \alpha^{-2} \cdot \C{poly_n_gap}^2 \cdot n^3 \cdot \log^2 n + \frac{\C{quad_const_add}}{\C{quad_delta_drop}} \cdot n$ for constant $\alpha > 0$.

Finally, by applying Markov's inequality, we obtain
\begin{equation*}
\Pro{\sum_{r = t_{0}}^{t_1} \Delta^r \leq  \C{small_delta} \cdot (t_1 - t_0) \cdot n/2  ~\Bigg|~ \mathfrak{F}^{t_{0}}, \Lambda^{t_0} \leq \exp(2 \C{poly_n_gap} n \log n)} \leq \frac{1}{2},
\end{equation*}
and so at least half of the rounds $t \in [t_{0},t_{1}]$ satisfy $\Delta^t \leq \C{small_delta} n$ w.p.\ at least $1/2$, i.e., 
\begin{equation} \label{eq:weak_bound_on_g_tilde}
\Pro{\tilde{G}_{t_{0} }^{t_1} > \frac{1}{2} \cdot (t_1 - t_0) ~\Bigg|~ \mathfrak{F}^{t_{0}}, \Lambda^{t_0} \leq \exp(2 \C{poly_n_gap} n \log n)} \geq \frac{1}{2},
\end{equation}
By \cref{lem:newcorrespondence}, there exists a constant $\eps:=\eps(C)>0$
\begin{align}\label{eq:8_1_restatement}
\Pro{G_{t_0 }^{t_1} > \frac{\eps}{5C} \tilde{G}_{t_0 }^{t_1 - 5Cn} \mid \mathfrak{F}^{t_0}, \Lambda^{t_0} \leq \exp(2 \C{poly_n_gap} n \log n)} \geq 1 - (t_1 - t_0) \cdot e^{-\eps n}.
\end{align}
Then, noting that
$\tilde{G}_{t_0 }^{t_1 - 5Cn} \geq \tilde{G}_{t_0 }^{t_1} - 5Cn$ and by taking the union bound of \eqref{eq:weak_bound_on_g_tilde} and \eqref{eq:8_1_restatement} we have,
\begin{equation*}
\Pro{ G_{t_0}^{t_1}> \frac{ \epsilon}{5C } \cdot \Big(\frac{1}{2} \cdot (t_1 - t_0) - 5Cn\Big) \;\Big|\; \mathfrak{F}^{t_{0}}, \Lambda^{t_0} \leq \exp(2 \C{poly_n_gap} n \log n)} \geq \frac{1}{2} - (t_1 - t_0)\cdot e^{-\eps n}.
\end{equation*}
Finally, since $ \frac{t_1 - t_0}{4}  > 5Cn$, we can deduce for $r= \min\{\frac{\eps}{20C}, \frac{1}{2}\}$,
\[
\Pro{ G_{t_0 }^{t_1} > r \cdot (t_1 - t_0) \;\Big|\; \mathfrak{F}^{t_0}, \Lambda^{t_0} \leq \exp(2 \C{poly_n_gap} n \log n)} \geq \frac{1}{2} - (t_1 - t_0)\cdot e^{-\eps n}.\qedhere 
\]
\end{proof}

The next lemma establishes the key fact that \Whp~when $\Lambda^{t_0} < \kappa_1 n$ there are many rounds close to $t_0$ with small absolute value potential. In contrast to \cref{lem:average_expected_delta_is_small}, this claim is proven \Whp and will be used in the stabilization phase (\cref{lem:stabilization}).

 \begin{lem} \label{lem:stabilisation_many_good_quantiles_whp}
Consider any $\PThree \cap \WTwo$-process or $\PTwo \cap \WThree$-process. Let $\eps := \eps(C)$ be as in \cref{lem:good_quantile}, $r:= \min\{\frac{\eps}{20C}, \frac{1}{2}\}$ and $\alpha>0$ be any constant. Then, for any constants $\kappa_1,\kappa_2 > 0$ and for any rounds $t_0$ and $t_1$ satisfying $t_1 := t_0 + \kappa_2 \cdot n \log n$, we have 
\[
\Pro{ G_{t_0}^{t_1}> r \cdot (t_1 - t_0) \;\Big|\; \mathfrak{F}^{t_0}, \Lambda^{t_0} \leq \kappa_1 \cdot n} \geq 1 - 3 \cdot n^{-12}.
\] 
\end{lem}  
\begin{proof}
For the constants $\C{quad_delta_drop}, \C{quad_const_add}$ given in \cref{lem:quadratic_absolute_relation_for_w_plus_w_minus}, we define for any $t \geq t_0$ the sequence
\[
 Z^{t} := \Upsilon^{t} - \C{quad_const_add} \cdot (t-t_0) + \frac{\C{quad_delta_drop}}{n}\sum_{i=t_0}^{t-1} \Delta^{i}.
\]
This sequence forms a super-martingale since by \cref{lem:quadratic_absolute_relation_for_w_plus_w_minus},
\begin{align*}
\Ex{Z^{t+1} \mid \mathfrak{F}^{t}} 
  & = \Ex{\Upsilon^{t+1} - \C{quad_const_add} \cdot (t-t_0+1) + \frac{\C{quad_delta_drop}}{n}\sum_{i=t_0}^{t} \Delta^{i} ~\Big|~ \mathfrak{F}^{t}} \\
  & = \Ex{\Upsilon^{t+1}\mid \mathfrak{F}^t} - \C{quad_const_add} \cdot (t-t_0+1) + \frac{\C{quad_delta_drop}}{n}\sum_{i=t_0}^{t} \Delta^{i} \\
  & \leq \Upsilon^{t} +\C{quad_const_add} - \frac{\C{quad_delta_drop}}{n} \cdot \Delta^t - \C{quad_const_add} \cdot (t-t_0+1) + \frac{\C{quad_delta_drop}}{n}\sum_{i=t_0}^{t} \Delta^{i} \\
  & = \Upsilon^t - \C{quad_const_add} \cdot (t-t_0 ) + \frac{\C{quad_delta_drop}}{n}\sum_{i=t_0}^{t-1} \Delta^{i} \\
  &= Z^{t}.
\end{align*}

Further, let $\tau:=\min\{ t \geq t_0 \colon \max_{i \in [n]} |y_i^{t}| >\log^2 n \}$ and consider the stopped random variable
\[
 \tilde{Z}^{t} := Z^{t \wedge \tau},
\]
which is then also a super-martingale. \cref{clm:small_change_for_linear_lambda} implies that %
\begin{equation}\label{eq:tau}
 \Pro{ \tau \leq t_1 } \leq n^{-12},
\end{equation}
i.e., the gap does not increase above $\log^2 n$ in any of the rounds $[t_0,t_1]$ \Whp

 To prove concentration of $\tilde{Z}^{t+1}$, we will now derive an upper bound on $\left| \tilde{Z}^{t+1} - \tilde{Z}^{t} \right|$ conditional on $\mathfrak{F}^{t}$ and $\Lambda^{t_0} \leq \kappa_1 \cdot n$.
\medskip 

\noindent\textbf{Case 1} [$t \geq \tau$].
In this case, $\tilde{Z}^{t+1} = Z^{(t+1) \wedge \tau} = Z^{\tau}$, and similarly, $\tilde{Z}^{t} = Z^{t \wedge \tau} = Z^{\tau}$, so $| \tilde{Z}^{t+1} - \tilde{Z}^{t}|=0$.
\medskip

\noindent\textbf{Case 2} [$t < \tau$]. Hence for $t$ we have $\max_{i \in [n]} |y_i^t| < \log^2 n$ and thus \cref{lem:basic} implies that the biggest change in the quadratic potential is bounded by $4w_-\cdot \log^2 n + 2w_-^2\leq 10 w_-\cdot \log^2 n $.
\medskip
Combining the two cases above, we conclude
\[
 | \tilde{Z}^{t+1} - \tilde{Z}^{t}| \leq 10 w_-\cdot \log^2 n.
\]

Using Azuma's inequality (\cref{lem:azuma}) for super-martingales,
\begin{align*}
 \Pro{ \tilde{Z}^{t_1+1} - \tilde{Z}^{t_0} > \lambda \, \mid \, \mathfrak{F}^{t_0}, \Lambda^{t_0} \leq \kappa_1 \cdot n } &\leq \exp\left( - \frac{ \lambda^2 }{ 2 \cdot \sum_{t=t_0}^{t_1} (10 w_-\cdot \log^2 n)^2  } \right) \\ &= \exp\left( - \frac{ \lambda^2 }{ (t_1-t_0)\cdot 200w_-^2\cdot \log^4 n  } \right),
\end{align*}
which means that for $\lambda = n$ we conclude $\Pro{ \tilde{Z}^{t_1+1}  > \tilde{Z}^{t_0} + n \, \mid \, \mathfrak{F}^{t_0}, \Lambda^{t_0} \geq \kappa_1 \cdot n } \leq n^{-\omega(1)}$. Thus by \eqref{eq:tau} and the union bound we have  \[\Pro{ Z^{t_1+1}  \leq Z^{t_0} + n \, \mid \, \mathfrak{F}^{t_0}, \Lambda^{t_0} \geq \kappa_1 \cdot n } \geq 1 - 2 \cdot n^{-12}.
\]
For the sake of a contradiction, assume now that at least half of the rounds $t \in [t_0,t_1]$ satisfy $\Delta^t \geq Cn$, which implies
\[
  \sum_{t=t_0}^{t_1} \Delta^t \geq \frac{t_1-t_0+1}{2} \cdot C \cdot n.
\]
If $Z^{t_1+1} \leq Z^{t_0} + n$ holds, then we have
\[
\Upsilon^{t_1+1}  - \C{quad_const_add} \cdot (t_1-t_0+1) + \frac{\C{quad_delta_drop}}{n}\sum_{t=t_0}^{t_1} \Delta^{t} \leq \Upsilon^{t_0} + n .\]Rearranging the inequality above gives
\begin{align} \Upsilon^{t_1+1} & \leq \Upsilon^{t_0} + n + \C{quad_const_add} \cdot (t_1-t_0+1) - \frac{\C{quad_delta_drop}}{n}\sum_{t=t_0}^{t_1} \Delta^{t} \notag  \\
  & \leq \Upsilon^{t_0} + n + \C{quad_const_add} \cdot (t_1-t_0+1) - \frac{\C{quad_delta_drop}}{n} \cdot \frac{t_1 - t_0+1}{2} \cdot C \cdot n \notag \\
  & \leq \Upsilon^{t_0} + n + (t_1 - t_0 +1) \cdot (\C{quad_const_add} - \frac{\C{quad_delta_drop}}{2} \cdot C) .\label{eq:upsilon_contradiction}
\end{align} Recall that we start from a round $t_0$ where $\Lambda^{t_0} \leq \kappa_1 \cdot n$, and therefore also $\Upsilon^{t_0} \leq \big( \frac{4}{\alpha} \cdot \log \frac{4}{\alpha} \big)^2 \cdot \kappa_1 \cdot n$ by \cref{clm:bound_on_gamma_implies_bound_on_upsilon_2}. Thus, recalling that $C > \frac{2c_2}{c_1}$,  by \eqref{eq:upsilon_contradiction} we have   
\[\Upsilon^{t_1+1} \leq  \left( \frac{4}{\alpha} \cdot\log \frac{4}{\alpha} \right)^2 \cdot \kappa_1 \cdot n  + n + (\kappa_2 n\log n +1) \cdot (\C{quad_const_add} - \frac{\C{quad_delta_drop}}{2} \cdot C) < 0\]
which is a contradiction for large $n$. We conclude that if $Z^{t_1+1} \leq Z^{t_0} + n$, then half of the rounds $t \in [t_0,t_1]$ satisfy $\Delta^{t} \leq C n$, thus this event holds w.p. $1-2 \cdot n^{-12}$, i.e.
\begin{align}
\Pro{\tilde{G}_{t_0}^{t_1} \geq \frac{1}{2} \cdot (t_1 - t_0) ~\Big|~ \mathfrak{F}^{t_0}, \Lambda^{t_0} \leq \kappa_1 \cdot n} \geq 1 - 2\cdot n^{-12}. \label{eq:first_union}
\end{align}
Applying \cref{lem:newcorrespondence} gives
\begin{equation}%
\Pro{ G_{t_0}^{t_1}> \frac{\eps}{5C} \cdot \tilde{G}_{t_0}^{t_1 - 5Cn} \;\Big|\; \mathfrak{F}^{t_0}, \Lambda^{t_0} \leq \kappa_1 \cdot n} \geq 1- (t_1-t_0) \cdot e^{-\eps n}. \label{eq:second_union}
\end{equation}
Then, noting that
$\tilde{G}_{t_0}^{t_1 - 5Cn} \geq \tilde{G}_{t_0}^{t_1} - 5Cn$ and
by taking the union bound of \eqref{eq:first_union} and \eqref{eq:second_union} we have,
\begin{equation*}
\Pro{ G_{t_0}^{t_1}> \frac{ \epsilon}{5C } \cdot \Big(\frac{1}{2} \cdot (t_1 - t_0) - 5Cn\Big) \;\Big|\; \mathfrak{F}^{t_0}, \Lambda^{t_0} \leq \kappa_1 \cdot n} \geq 1- (t_1-t_0) \cdot e^{-\eps n} - 2n^{-12}.
\end{equation*}
Since, $\frac{1}{4} \cdot (t_1 - t_0) = \frac{c_s}{4} \cdot n \log n > 5 Cn$, we can deduce that for $r = \min\{\frac{ \epsilon}{20C}, \frac{1}{2}\}$,
\begin{equation*}
\Pro{ G_{t_0}^{t_1}> r \cdot (t_1 - t_0) \;\Big|\; \mathfrak{F}^{t_0}, \Lambda^{t_0} \leq \kappa_1 \cdot n} \geq 1 - 3 \cdot n^{-12}.\qedhere
\end{equation*}
\end{proof}

\subsection{Recovery of the Process} \label{sec:recovery}

In the following, we start the analysis at round $t_0 = m - \Theta(n^3 \log^4 n)$ and using \cref{lem:initial_gap_nlogn}, we obtain that \Whp~$\Gap(t_0) = \Oh(n \log n)$ and $\Lambda^t \leq \exp(2 \C{poly_n_gap} n \log n)$. Using \cref{lem:average_expected_delta_is_small} we conclude that with constant probability for a constant fraction of the rounds $t \in [t_0,t_0+\Theta(n^3 \log^3 n)]$, it holds that $\delta^t \in (\epsilon, 1 - \epsilon)$.
Then we exploit the drop of the exponential potential function in those rounds to infer that with constant probability there is a round $s \in [t_0,t_0+\Theta(n^3 \log^3 n)]$ with $\Lambda^{s}=\Oh(n)$. Then, we use a retry argument $\Theta(\log n)$ times, to amplify the probability of finding such a round.
For each failure, we can still restart from a round with $\Lambda^t \leq \exp(2 \C{poly_n_gap} n \log n)$. Hence after these $\Theta(\log n)$ repetitions, we deduce that \Whp, there exists a round $s \in [t_0, m]$, where $\Lambda^s \leq cn$, which yields \cref{lem:recovery}. From that point onwards, we will use the stabilization lemma (\cref{lem:nonfilling_good_gap_after_good_lambda}) for the remaining $\Oh(n^3 \log^4 n)$ rounds (see \cref{sec:stabilization}).

\begin{lem}[Recovery] \label{lem:recovery}
Consider any $\PThree \cap \WTwo$-process or $\PTwo \cap \WThree$-process. Then, for some constant $\alpha > 0$, and for $\C{lambda_bound}>1$ being the constant in \cref{cor:change_for_large_lambda}, for any $m \geq 40 n^3 \log^4 n$,
\[
\Pro{\bigcup_{s \in [m - 40 n^3 \log^4 n, m]}\{ \Lambda^s < \C{lambda_bound} n\} } \geq 1 - n^{-10}.
\]
\end{lem}
\begin{proof}
Let $k := 40\log n$ and $T := n^3 \log^3 n$. 
Consider $t_0 := m - k \cdot T$, and define $k$ rounds $t_1 := t_0 + T, t_2 := t_0 + 2T, \ldots, t_k := m$. We now recall for any $0\leq i\leq k$ the event
\[
\mathcal{E}_{t_0}^{t_i} :=\bigcap_{s \in [t_0, t_i]} \left\{\Lambda^{s} \geq \C{lambda_bound} n \right\}. \]
Further, we define for any integer $0 \leq i \leq k$ the event 
\[
\mathcal{H}^{t_i} := \left\{\Lambda^{t_i} \leq \exp(2\C{poly_n_gap} n \log n) \right\}.
\] 
We will prove that, for any integer $0 \leq i\leq k$, we have
\[
\Pro{ \mathcal{E}_{t_0}^{t_i} } \leq (3/4)^{i} + 4\cdot n^{-12}.
\]
Thus for $i = k$, we obtain the statement of the lemma. We next establish the following claim:

\begin{clm}\label{clm:final_claim}
For any integer $1 \leq i \leq k$, we have $ 
\Pro{\mathcal{E}_{t_0}^{t_i} \mid \mathfrak{F}^{t_{i-1}}, \mathcal{H}^{t_{i-1}} } \leq 3/4$.
\end{clm}
\begin{poc}
Using \cref{lem:average_expected_delta_is_small} for $t_{i-1}$ and $t_i$ (since $t_i - t_{i-1} = T$), we get that 
for $r= \min\{\frac{\eps}{20C}, \frac{1}{2}\}$,
\begin{equation}\label{eq:many_good_quantiles_wcp}
\Pro{ G_{t_{i-1} }^{t_i} > r \cdot T \;\Big|\; \mathfrak{F}^{t_{i-1}}, \mathcal{H}^{t_{i-1}}} \geq 1/2 - T\cdot e^{-\eps n}.
\end{equation}

Let $\gamma := \gamma(\eps) = \frac{r}{2(1-r)} < 1$ as $r \leq 1/2$, and choose the constant $\alpha:=\alpha(\eps)$ so as to satisfy $\alpha \leq \C{good_quantile_mult} \gamma/(3w_- \cdot e^{2 w_-})$, where $\C{good_quantile_mult}:=\C{good_quantile_mult}(\eps)>0$ is given by \cref{lem:good_quantile_good_decrease}, and to satisfy either \eqref{eq:c_3alphacond1} or \eqref{eq:c_3alphacond2} depending on whether the process satisfies \WThree or \PThree respectively. 
Since all the previous steps (and constants determined by them) hold for any $c$ (and the constants $C, \eps, r, \gamma, \alpha$ and \cref{lem:good_quantile_good_decrease} do not depend on $c$) we can now set \[\C{stab_time} := \frac{2 \cdot 9}{\C{good_quantile_mult} \cdot \alpha \cdot r}>0 \qquad\text{and} \qquad  \C{lambda_bound}  := \max\Big(\frac{3(\C{good_quantile_mult} + w_-) \cdot e^{2 w_-}}{\alpha\C{good_quantile_mult}},\frac{2}{\alpha^2}\Big) > 1,\] so $\C{lambda_bound}$ satisfies the assumptions of \cref{cor:change_for_large_lambda}. As $\alpha$ satisfies \cref{lem:gamma_tilde_is_supermartingale}, $(\tilde{\Lambda}_{t_{i-1}}^{t})_{t \geq t_{i-1}}$ is a super-martingale. Thus, $\ex{\tilde{\Lambda}_{t_{i-1}}^{t_i+1}\mid \mathfrak{F}_{t_{i-1}}} \leq \tilde{\Lambda}_{t_{i-1}}^{t_{i-1}}$.
Applying Markov's inequality gives \begin{equation}\label{eq:bddontildelambda}\Pro{\tilde{\Lambda}_{t_{i-1}}^{t_i+1} > 5 \cdot \tilde{\Lambda}_{t_{i-1}}^{t_{i-1}} \mid \mathfrak{F}^{t_{i-1}}, \mathcal{H}^{t_{i-1}}} \leq 1/5.\end{equation} 
Thus by the definition of $\tilde{\Lambda}_{t_{i-1}}^{t}$ and taking the union bound of \eqref{eq:bddontildelambda} and \eqref{eq:many_good_quantiles_wcp}, we have
\begin{align}
 &~ \Pro{ \left\{ \Lambda^{t_i+1}\cdot \mathbf{1}_{\mathcal{E}_{t_0}^{t_i}} \leq 5 \cdot \Lambda^{t_{i-1}} \cdot e^{ \frac{\C{good_quantile_mult} \alpha \gamma}{n} \cdot B_{t_{i-1} }^{t_i}  - \frac{\C{good_quantile_mult} \alpha}{n} \cdot G_{t_{i-1} }^{t_i} } \right\} \bigcap \left\{ G_{t_{i-1} }^{t_i} \geq r \cdot T \right\} \, \bigg| \, \mathfrak{F}^{t_{i-1}}, \mathcal{H}^{t_{i-1}}}\notag  \\ &\qquad\qquad\qquad\geq 1 - 1/2 - o(1) - 1/5 \geq 1/4. \label{eq:lastmin}
\end{align}
However, if the event $G_{t_{i-1} }^{t_i} \geq r \cdot T $ occurs, then
\begin{align*}
\Lambda^{t_i+1}\cdot \mathbf{1}_{\mathcal{E}_{t_0}^{t_i}}  &\leq 5 \cdot \Lambda^{t_{i-1}} \cdot \exp\left(\frac{\C{good_quantile_mult} \alpha \gamma}{n} \cdot B_{t_{i-1} }^{t_i} - \frac{\C{good_quantile_mult} \alpha}{n} \cdot G_{t_{i-1} }^{t_i} \right) \leq 5 \cdot \Lambda^{t_{i}} \cdot \exp\left( - \frac{\C{good_quantile_mult} \alpha}{n} \cdot \frac{r}{2} \cdot T \right),
\end{align*}
where the second inequality used $\gamma=\frac{r}{2(1-r)}$.
 Further, since we condition on $\mathcal{H}^{t_i}$, we obtain
\begin{align*}
\Lambda^{t_i+1}\cdot \mathbf{1}_{\mathcal{E}_{t_0}^{t_i}}  &\leq 5\cdot \exp\left(2 \C{poly_n_gap} n \log n \right) \cdot \exp\left( - \frac{\C{good_quantile_mult} \alpha}{n} \cdot \frac{r}{2} \cdot \left(n^3 \log^3 n\right) \right) < 1,
\end{align*} as $ T=  n^3 \log^3 n$ for any $i\leq k$. Recall that $\Lambda^{t} \geq n$ holds deterministically for any $t \geq 0$, thus we have a contradiction. 
We conclude that the event $\neg \mathcal{E}_{t_0}^{t_i} = \{\bigcup_{r \in [t_0, t_i]} \Lambda^r < \C{lambda_bound} n\}$ is implied whenever the events $\{G_{t_{i-1} }^{t_i} \geq r \cdot T\}$, $\{\tilde{\Lambda}_{t_{i-1}}^{t_i+1} \leq 5 \cdot \tilde{\Lambda}_{t_{i-1}}^{t_{i-1}} \}$, and $\mathcal{H}^{t_{i-1}}$ all occur. The bound on $\Pro{\mathcal{E}_{t_0}^{t_i} \mid \mathfrak{F}^{t_{i-1}}, \mathcal{H}^{t_{i-1}} } = 1- \Pro{\neg \mathcal{E}_{t_0}^{t_i} \mid \mathcal{H}^{t_{i-1}}, \mathfrak{F}^{t_{i-1}} } $ then follows from \eqref{eq:lastmin}.
\end{poc}

By the second statement in \cref{lem:initial_gap_nlogn}, we get that for any $0 \leq i\leq k$ 
\[
\Pro{\neg \mathcal{H}^{t_i}} \leq n^{-12}.
\]

Combining this bound with \cref{clm:final_claim}, for any $1 \leq i\leq k$ we have
\begin{align*}
\Pro{\mathcal{E}_{t_0}^{t_i}}
 & = \Pro{\mathcal{H}^{t_{i-1}} \cap \mathcal{E}_{t_0}^{t_i}} + \Pro{\neg \mathcal{H}^{t_{i-1}} \cap \mathcal{E}_{t_0}^{t_i}} \\
 & \leq \Pro{\mathcal{E}_{t_0}^{t_i} \mid \mathcal{H}^{t_{i-1}}, \mathcal{E}_{t_0}^{t_{i-1}}} \cdot \Pro{ \mathcal{H}^{t_{i-1}} \cap \mathcal{E}_{t_0}^{t_{i-1}}} + \Pro{ \neg\mathcal{H}^{t_{i-1}}}\\ %
 & \leq \Pro{\mathcal{E}_{t_0}^{t_i} \mid \mathcal{H}^{t_{i-1}}, \mathcal{E}_{t_0}^{t_{i-1}}} \cdot \Pro{ \mathcal{E}_{t_0}^{t_{i-1}}} + n^{-12} \\ 
 & \leq (3/4) \cdot \Pro{ \mathcal{E}_{t_0}^{t_{i-1}}} + n^{-12}.
\end{align*}
By the second statement of \cref{lem:geometric_arithmetic} for $(\Pro{\mathcal{E}_{t_0}^{t_i}})_{i \geq 0}$, $a := 3/4 < 1$ and $b := n^{-12}$, we get
\begin{align*}
\Pro{\mathcal{E}_{t_0}^{t_i}}
\leq  \Pro{\mathcal{E}_{t_0}^{t_0}} \cdot (3/4)^{i} + \frac{n^{-12}}{1 - \frac{3}{4}} \leq (3/4)^{i} + 4\cdot n^{-12},
\end{align*}as $\Pro{\mathcal{E}_{t_0}^{t_0}} \leq 1$. Setting $i := k = 40\log n$, gives the result since $ (3/4)^{40\log n} + 4\cdot n^{-12}\leq n^{-10} $.
\end{proof}

\subsection{Stabilization of the Process} \label{sec:stabilization}
 
 The next lemma establishes that a small value of the exponential potential function is preserved for some longer time period. We will call this property of the exponential potential function of being ``trapped'' in some region \emph{stabilization}. More precisely, we prove in the lemma below that if for some round $t_0$ the exponential potential is not too small, i.e., at most $2cn$, then within the next $\Oh(n \log n)$ rounds, the exponential potential will be smaller than $cn$ at least once \Whp.

\begin{lem}[Stabilization]\label{lem:stabilization}
For any $\PThree \cap \WTwo$-process or $\PTwo \cap \WThree$-process there exists constants $\alpha \in (0, 1/w_-]$, $c>1 $ and $\C{stab_time} > 0$, such that for any round $t_0 \geq 0$,
\[
 \Pro{ \bigcup_{t \in [t_0,t_0 + \C{stab_time} n \log n-1]} \left\{\Lambda^{t} < \C{lambda_bound} n \right\} ~\Big|~ \mathfrak{F}^{t_0}, \Lambda^{t_0} \in [\C{lambda_bound}n, 2\C{lambda_bound}n] } \geq 1 - \frac{1}{2} \cdot n^{-7}.
\] 
\end{lem}

\begin{proof} 
Let $t_1 :=t_0+\C{stab_time} n \log n$, for \C{stab_time} a constant to be defined below. By \cref{lem:stabilisation_many_good_quantiles_whp} with $\kappa_1=2c$ and $\kappa_2=c_s$ we have for $r := \min\{\frac{ \epsilon}{20C}, \frac{1}{2}\}$ (the same $r$ as defined in the proof of \cref{lem:recovery}),

\begin{equation}\label{eq:many_good_quantiles_whp_new}
\Pro{ G_{t_0}^{t_1}> r \cdot (t_1 - t_0) \;\Big|\; \mathfrak{F}^{t_0}, \Lambda^{t_0} \in [\C{lambda_bound}n, 2\C{lambda_bound}n]} \geq 1 - 3 \cdot n^{-12}.
\end{equation}
Again, we use $\gamma = \frac{r}{2(1-r)} \leq 1$ and choose $\alpha$ and $c$ as in \cref{lem:recovery}. As $\alpha$ satisfies \cref{lem:gamma_tilde_is_supermartingale}, $(\tilde{\Lambda}_{t_{i-1}}^{t})_{t \geq t_{i-1}}$ is a super-martingale, so $\ex{\tilde{\Lambda}_{t_0}^{t_1}\mid \mathfrak{F}^{t_0}} \leq \tilde{\Lambda}_{t_0}^{t_0} =  \Lambda^{t_0}$. Hence, using Markov's inequality we get $\Pro{\tilde{\Lambda}_{t_0}^{t_1} >\Lambda^{t_0} \cdot n^8\mid \mathfrak{F}^{t_0}, \Lambda^{t_0} \in [\C{lambda_bound}n, 2\C{lambda_bound}n] } \leq  n^{-8}$. Thus, by the definition of $\tilde{\Lambda}_{t_0}^t$ given in \cref{lem:gamma_tilde_is_supermartingale}, 
we have 
\begin{equation} \label{eq:supermartingale_markov_new}
\Pro{\Lambda^{t_1} \cdot \mathbf{1}_{\mathcal{E}_{t_0}^{t_1-1}} \leq \Lambda^{t_0} \cdot n^8 \cdot \exp\left(  \frac{\C{good_quantile_mult} \alpha \gamma}{n} \cdot B_{t_0}^{t_1-1}  -\frac{\C{good_quantile_mult} \alpha}{n} \cdot G_{t_0}^{t_1-1} \right) \, \Bigg| \, \mathfrak{F}^{t_0}, \Lambda^{t_0} \in [\C{lambda_bound}n, 2\C{lambda_bound}n] } \geq  1 - n^{-8}.    
\end{equation}
Further, if in addition to the two events $\{\tilde{\Lambda}_{t_0}^{t_1} \leq \Lambda^{t_0} \cdot n^8\}$ and $\{\Lambda^{t_0} \in [\C{lambda_bound}n, 2\C{lambda_bound}n]\}$, also the event $\{G_{t_0}^{t_1-1} \geq r \cdot (t_1 - t_0)\}$ holds, then
\begin{align*}
\Lambda^{t_1} \cdot \mathbf{1}_{\mathcal{E}_{t_0}^{t_1-1}} & \leq \Lambda^{t_0} \cdot n^8 \cdot \exp\left( \frac{\C{good_quantile_mult} \alpha \gamma}{n} \cdot B_{t_0}^{t_1-1} - \frac{\C{good_quantile_mult} \alpha}{n} \cdot G_{t_0}^{t_1-1} \right) \\
 & \leq 2cn^9 \cdot \exp\left( \frac{\C{good_quantile_mult} \alpha \gamma}{n} \cdot (1 - r) \cdot (t_1 - t_0) - \frac{\C{good_quantile_mult} \alpha}{n} \cdot r \cdot (t_1 - t_0) \right) \\
 & \stackrel{(a)}{=} 2cn^9 \cdot \exp\left( - \frac{\C{good_quantile_mult} \alpha}{n} \cdot \frac{r}{2} \cdot (t_1 - t_0) \right) \\
 & = 2cn^9 \cdot \exp\left( - \frac{\C{good_quantile_mult} \alpha}{n} \cdot \frac{r}{2} \cdot  \frac{2 \cdot 9}{\C{good_quantile_mult} \cdot \alpha \cdot r} \cdot n \log n \right)  = 2\C{lambda_bound},
\end{align*}
where $(a)$ holds since $\gamma(1-r) - r = -\frac{r}{2}$ due to the definition of $\gamma$.
Observe that $\Lambda^{t_1} \geq n$ holds deterministically, so we deduce from the above inequality that $\mathbf{1}_{\mathcal{E}_{t_0}^{t_1-1}}=0$, that is,
\[\Pro{ \neg \mathcal{E}_{t_0}^{t_1 -1} \; \Bigg|\; \mathfrak{F}^{t_0},  \;\; \; \tilde{\Lambda}_{t_0}^{t_1} \leq \Lambda^{t_0} \cdot n^8, \;\;\; \Lambda^{t_0} \in [\C{lambda_bound}n, 2\C{lambda_bound}n],  \;\;\; G_{t_0}^{t_1-1} \geq r \cdot (t_1 - t_0)} =1.\] 
Recalling the definition of $\mathcal{E}_{t_0}^{t_1-1} = \bigcap_{r \in [t_0, t_1-1]} \{ \Lambda^r \geq cn \}$, and taking the union bound over \eqref{eq:many_good_quantiles_whp_new} and \eqref{eq:supermartingale_markov_new} yields 
\[
\Pro{ \bigcup_{r \in [t_0, t_0 + \C{stab_time} n \log n-1]} \{ \Lambda^r < cn \} \; \big| \; \mathfrak{F}^{t_0}, \Lambda^{t_0} \in [\C{lambda_bound}n, 2\C{lambda_bound}n]  }\geq 1 - 3 \cdot n^{-12} -n^{-8}\geq 1 - \frac{1}{2} \cdot n^{-7},
\]
as claimed. 
\end{proof} 

The next lemma shows that if there is a round with linear (exponential) potential then the gap is at most logarithmic for the next $n^4$ rounds. The idea is to repeatedly apply \cref{lem:stabilization} to find a interval of length $n^4$ where any contiguous sub-interval of length $\Theta(n\log n)$ contains a round with linear potential. The result then follows since a linear potential implies a logarithmic gap, and the gap can grow by at most $\Theta(\log n)$ in $\Theta(n\log n)$ rounds.

\begin{lem} \label{lem:nonfilling_good_gap_after_good_lambda}
For any $\PThree \cap \WTwo$-process or $\PTwo \cap \WThree$-process there exists a constant $\kappa > 0$ such that, for any rounds $t_0$, $t_1$ with $t_0 < t_1 \leq t_0 + n^4$,
\[
\Pro{\max_{i \in [n] } \lvert y_i^{t_1} \rvert \leq \kappa \cdot \log n \mid \mathfrak{F}^{t_0}, \Lambda^{t_0} < cn} \geq 1 - \frac{1}{2} \cdot n^{-3}.
\]
\end{lem}

\begin{proof} 
Choose the constants $\alpha \in (0, 1/w_-]$, $\C{lambda_bound}>1$ and
 $\C{stab_time} > 0$ so as to satisfy \cref{lem:stabilization}. Next define the event
\[ \mathcal{M}_{t_0}^{t_1} = \left\{\text{for all }t\in [t_0, t_1 ]\text{ there exists } s\in [t, t+\C{stab_time} n\log n  ] \text{ such that }\Lambda^s < cn\right\},\]	
that is, if $\mathcal{M}_{t_0}^{t_1}$ holds then we have $\Lambda^s < cn $ at least once every $\C{stab_time} n\log n$ rounds. 

Assume now that $\mathcal{M}_{t_0}^{t_1}$ holds. Choosing $t=t_1$, there exists $s \in [t_1,t_1+\C{stab_time} n \log n]$ such that $\Lambda^s < cn$, which in turn implies by definition of $\Lambda$ that $\max_{i \in [n]} |y_i^s| \leq \frac{1}{\alpha} \cdot \log (cn) < \frac{2}{\alpha} \cdot \log n$. Clearly, any $y_i^t$ can decrease by at most $w_-/n$ in each round, and from this it follows that if $\mathcal{M}_{t_0}^{t_1}$ holds, then 
$
\max_{i \in [n]} y_i^{t_1} \leq \max_{i \in [n]} |y_i^s| + w_-\cdot \C{stab_time} \log n \leq \kappa \cdot \log n, $
for $\kappa := \frac{2}{\alpha} +  w_- \cdot \C{stab_time}$.

If $t_1 \geq t_0 + \C{stab_time} n \log n$ and $\mathcal{M}_{t_0}^{t_1}$ holds, then choosing $s = t_1 - \C{stab_time} n \log n$, there exists $s \in [t_1 - \C{stab_time} n \log n, t_1]$ such that $\Lambda^s < cn$. (in case $t_1 -  \C{stab_time} n \log n < t_0 $, then we arrive at the same conclusion by choosing $s = t_0$). %
This in turn implies
$
 \max_{i \in [n]} | y_i^s  | < \frac{2}{\alpha} \cdot \log n.
$
Hence 
\[
\min_{i \in [n]} y_i^{t_1} \geq \min_{i \in [n]} y_i^{s} - w_-\cdot \C{stab_time} \log n 
\geq -\max_{i \in [n]} |y_i^{s}| -  w_-\cdot \C{stab_time} \log n 
\geq -\kappa \cdot \log n.
\]
 Hence, $\mathcal{M}_{t_0}^{t_1}$ (conditioned on $\Lambda^{t_0} < cn$) implies that $\max_{i \in [n]} \lvert y_i^{t_1} \rvert \leq \kappa \cdot \log n$. It remains to bound $\Pro{\neg \mathcal{M}_{t_0}^{t_1} \mid \mathfrak{F}^{t_0}, \Lambda^{t_0} < cn}$.

Note that since we start with $\Lambda^{t_0} < \C{lambda_bound} n$, if for some round $j$, $\Lambda^j> 2 \C{lambda_bound} n$,  there must exist some $s \in [t_0,j)$, such that $\Lambda^s \in [\C{lambda_bound} n, 2\C{lambda_bound}n]$, 
since for every $t \geq 0$ it holds $\Lambda^{t+1} \leq \Lambda^t \cdot e^{\alpha w_-} \leq 2 \Lambda^t$ and $\alpha \leq 1/w_-$. Let $t_0 < r_1 <r_2<\cdots $ and $t_0 =: s_0<s_1<\cdots $ be two interlaced sequences defined recursively for $i\geq 1$ by \[r_i := \inf\{r >  s_{i-1}: \Lambda^r \in [\C{lambda_bound}n, 2\C{lambda_bound}n] \}\qquad\text{and} \qquad s_i := \inf\{s> r_i : \Lambda^t < cn\}. \] 
  Thus we have
  \[
  t_0 = s_0 < r_1 < s_1 < r_2 <s_2 < \cdots, 
  \]
  and since $r_i>r_{i-1}$ we have $ r_{t_1 - t_0}\geq t_1 - t_0$. Observe that if the event $\cap_{i=1}^{t_1 - t_0}\{s_i-r_i\leq \C{stab_time}n \log n \} $ holds, then also $ \mathcal{M}_{t_0}^{t_1}$ holds.

   Recall that by \cref{lem:stabilization} we have for any $i=1,2,\ldots, t_1 - t_0$ and any $r = t_0 + 1, \ldots , t_1$ \[
 \Pro{ \bigcup_{t \in [r_i,r_i + \C{stab_time} n \log n-1]} \left\{\Lambda^{t} < \C{lambda_bound} n \right\} ~\Big|~ \mathfrak{F}^{r} , \; \Lambda^{r} \in [\C{lambda_bound}n, 2\C{lambda_bound}n], r_i = r } \geq  1 - \frac{1}{2} \cdot n^{-7},
  \] and by negating and the definition of $s_i$,
  \[
  \Pro{s_i-r_i> \C{stab_time}n \log n \, \Big| \, \mathfrak{F}^{r}, \Lambda^{r} \in [\C{lambda_bound}n, 2\C{lambda_bound}n], r_i = r } \leq \frac{1}{2} n^{-7}.
  \]
  Since the above bound holds for any $i$ and $\mathfrak{F}^{r}$, with $r_i=r$, it follows by the union bound over all $i=1,2,\ldots,t_1 - t_0$, as $t_1 - t_0 \leq n^4$,
  \[
  \Pro{\neg \mathcal{M}_{t_0}^{t_1} \mid \mathfrak{F}^{t_0}, \Lambda^{t_0} < cn }\leq (t_1 - t_0)\cdot \frac{1}{2} n^{-7} \leq \frac{1}{2} n^{-3}. \qedhere\] \end{proof}

\subsection{Completing the Proof of Theorem~\ref{thm:main_technical}}

We are now ready to complete the proof.%

{\renewcommand{\thethm}{\ref{thm:main_technical} }
\begin{thm}[restated]
\maintechnical
	\end{thm}}
	\addtocounter{thm}{-1}

\begin{proof}
Consider an arbitrary round $m \geq 1$. If $m < 40 n^3 \log^4 n$, then the claim follows by \cref{lem:nonfilling_good_gap_after_good_lambda} since $\Lambda^{t_0} = n < cn$.

Otherwise, let $t_0 := m - 40 n^3 \log^4 n$. %
Firstly, by \cref{lem:recovery}, we get
\[
\Pro{\bigcup_{s \in [t_0, m] } \left\{ \Lambda^{s} \leq cn \right\} } \geq 1 - n^{-10}.
\]
Hence for $\tau:=\inf\{ s \geq t_0 \colon \Lambda^s < c n \}$ we have $\Pro{ \tau \leq m } \geq 1-n^{-10}$.

Secondly, using \cref{lem:nonfilling_good_gap_after_good_lambda}, there is some constant $\kappa := \kappa(\alpha)$ such that for any round $s \in [t_0, m]$
\[
\Pro{\max_{i \in [n]} \lvert y_i^m \rvert \leq \kappa \cdot \log n ~\Big|~ \mathfrak{F}^s , \Lambda^s < cn} \geq 1 - \frac{1}{2} \cdot n^{-3}.
\]
Combining the two inequalities from above, 
\begin{align*}
    \Pro { \max_{i \in [n]} \lvert y_i^m \rvert \leq \kappa \cdot \log n} &\geq \sum_{s=t_0}^{m} \Pro{ \tau = s} \cdot \Pro{ \max_{i \in [n]} \lvert y_i^m \rvert \leq \kappa \cdot \log n  ~\Big|~ \tau = s } \\
    &\geq \sum_{s=t_0}^{m} \Pro{ \tau = s} \cdot \Pro{ \max_{i \in [n]} \lvert y_i^m \rvert \leq \kappa \cdot \log n  ~\Big|~ \mathfrak{F}^{s}, \Lambda^{s} < cn } \\
    &\geq \left( 1-\frac{1}{2}n^{-3} \right) \cdot \Pro{ \tau \leq m} \\
    &\geq \left(1 - \frac{1}{2} n^{-3} \right) \cdot \left(1 - n^{-10} \right) \geq 1 - n^{-3}.\qedhere
\end{align*}

\end{proof}

\section{Lower Bounds}\label{sec:lower}

In this section we shall prove several lower bounds for both filling and non-filling processes. See Table \ref{tab:overview} for an concise overview of our lower bounds and together with previously known lower bounds.
One interesting aspect is that our lower bound for non-filling processes makes use of the quantile stabilization result from the previous section (\cref{lem:stabilisation_many_good_quantiles_whp}).

\subsection{Lower Bounds for Uniform Processes}
We first prove a general lower bound which holds for any process that picks bins for allocation uniformly and can then increment the load of that chosen bin arbitrarily. Since \Twinning and \Packing choose a uniform bin for allocation at each round, the result below immediately yields a gap bound of $\Omega(\frac{\log n}{\log \log n})$ for these processes for $m=\Oh(n)$ rounds.

\NewConstant{adaptive_filling_lower_bound_gap}{d}
\begin{lem}
Consider any allocation process, which at each round $t \geq 0$, picks a bin $i^t$ uniformly. Furthermore, assume that at any round $t \geq 0$ the allocation process increments the load of bin $i^t$ by some function $f^t \geq 1$, which may depend on $\mathfrak{F}^{t}$ and $i^t$. Then, there is a constant $\C{adaptive_filling_lower_bound_gap} > 0$  such that
\[
 \Pro{ \Gap\left(\frac{n}{2}\right) \geq \C{adaptive_filling_lower_bound_gap} \cdot \frac{\log n}{\log \log n} } \geq 1-o(1).
\]
\end{lem}
\begin{proof}

\NewConstantWithName{one_choice_lower_bound}{\ensuremath{\kappa}}
Recall the fact (see, e.g.,~\cite{RS98}) that in a \OneChoice process with $n/2$ balls into $n$ bins, with probability at least $1-o(1)$ there is a bin $i \in [n]$ which will be chosen at least $\C{one_choice_lower_bound} \frac{\log n}{\log \log n}$ times during the first $n/2$ allocations, where $\C{one_choice_lower_bound}>0$ is some constant. Hence in our process
\[
 x_i^{n/2} \geq  \left(\C{one_choice_lower_bound} \frac{\log n}{\log \log n}\right) \cdot 1,
\]for some $i\in[n]$ w.h.p.\ as at least one ball is allocated at each round. 

Let us define $f^{*}:=\max_{1 \leq t \leq n/2} f^t$ to be the largest number of balls allocated in one round. Then, clearly, $W^{n/2} \leq (n/2) \cdot f^{*}$. Furthermore, there must be at least one bin $j \in [n]$ which receives $f^*$ balls in one of the first $n/2$ rounds. Thus for any such bin, 
\[
 x_j^{n/2} \geq f^*,
\]
and therefore the gap is lower bounded by
\begin{align*}
 \Gap(n/2) &\geq \max \left\{x_i^{n/2}, x_j^{n/2} \right\} - \frac{W^{n/2}}{n} \geq \max \left\{ \C{one_choice_lower_bound} \frac{\log n}{\log \log n}, f^{*} \right\} - \frac{f^{*}}{2}.
\end{align*}
If $f^{*} \geq \C{one_choice_lower_bound} \frac{\log n}{\log \log n}$, then the lower bound is $\frac{1}{2} f^{*} \geq \frac{\C{one_choice_lower_bound}}{2} \frac{\log n}{\log \log n}$. Otherwise, $f^{*} < \C{one_choice_lower_bound} \frac{\log n}{\log \log n}$, and the lower bound is at least
$ \C{one_choice_lower_bound} \frac{\log n}{\log \log n} - \frac{\C{one_choice_lower_bound}}{2} \frac{\log n}{\log \log n} = \frac{ \C{one_choice_lower_bound}}{2} \frac{\log n}{\log \log n}$.
\end{proof}

The technique used in the proof above can be extended to the special case of \Packing, yielding a more general lower bound which holds for any round $m \geq n$. 
\NewConstant{packing_all_m}{d}

\begin{lem}\label{lem:packing_all_m}
For \Packing, there is a constant $\C{packing_all_m} > 0$ such that for any round $m \geq n$ we have
\[
 \Pro{ \Gap(m) \geq \C{packing_all_m} \cdot \frac{\log n}{\log \log n}} \geq 5/8.
\]
\end{lem}
\begin{proof}
Consider the round $t_0:=m-n$ and the interval $[t_0,m]$. For the \Packing process, it then follows from \eqref{eq:potential_small} that there is a constant $c_6 > 0$ such that
\[
 \Ex{ \Phi^{t_0}  } \leq c_6 \cdot n,
\]
where the constant $\alpha$ in $\Phi^{t_0} = \sum_{i \colon y_i^{t} \geq 2} \exp\left( \alpha y_i^{t_0} \right)$ is given by $\alpha = \min \{ 1/101, (1/40) \cdot c_2/c_1 \}$ for the constants $c_1,c_2$ are given by \cref{lem:filling_good_quantile}.
Hence by Markov's inequality,
\begin{align}
 \Pro{ \Phi^{t_0} \geq 8 c_6 n } \leq \frac{1}{8}. \label{eq:eins}
\end{align}
Conditioning on $\Phi^{t_0} \geq 8 c_6 n$ implies
\[
 \Delta^{t_0} = 2 \cdot \sum_{i \colon y_i^{t_0} \geq 0} y_i^{t_0} \leq 8 \cdot n 
 + \frac{1}{\alpha} \sum_{i \colon y_i^{t_0} \geq 2} \exp\left( \alpha y_i^{t} \right) \leq \kappa_1 \cdot n,
\]
for some constant $\kappa_{1} >0$, where we used the fact that $z \leq \exp(z)$ for any $z \geq 0$.

Next define $\mathcal{B}:=\left\{ i \in [n] \colon y_i^{t} \geq -2 \kappa_1 \right\}$. Then, as $\Delta^{t_0} \leq \kappa_1 \cdot n$, it follows that $|\mathcal{B}| \geq n/2$.
Further, using \cref{clm:bound_wone}, we conclude that for any round $s \geq 0$, $\Delta^{s+1} \leq \Delta^{s} + 4$ and thus
\[
\Delta^{t} \leq \Delta^{t_0} +4 (t-t_0),
\]
Hence
\[
 \sum_{t=t_0}^m \Delta^t \leq \kappa_2 \cdot n^2,
\]
holds deterministically for some constant $\kappa_2$, where we can assume $\kappa_2 > 1$. 
Next note that the expected total number of balls allocated in a round $t \in [t_0,m]$, denoted by $w^t$, satisfies
\[
w^{t} =
\sum_{i \colon y_i^t \geq 0} \frac{1}{n} \cdot 1 + \sum_{i \colon y_i^t < 0}
\frac{1}{n} \cdot \left( \lceil -y_i^{t} \rceil + 1 \right)
\leq \frac{1}{n} \cdot \left( \Delta^{t} + n \right),
\]
as $\sum_{i \colon y_i^t < 0} -y_i^t = \frac{1}{2} \Delta^t$.
Thus the expected number of balls allocated overall in all rounds $t \in [t_0,m]$ together satisfies
\[
  \Ex{ \sum_{t=t_0}^m w^{t} \, \Big| \, \mathfrak{F}^{t_0}, \Phi^{t_0} < 8 c_6 n } \leq 2 \kappa_2 \cdot n.
\]
Using Markov's inequality, 
\begin{align}
 \Pro{ \sum_{t=t_0}^m w^{t} \geq 16 \kappa_2 \cdot n \, \Big| \, \mathfrak{F}^{t_0}, \Phi^{t_0} < 8 c_6 n } \leq \frac{1}{8}. \label{eq:zwei}
\end{align}
In the following, let us also condition on the event $\sum_{t=t_0}^m W^{t} < 16 \kappa_2 \cdot n $, which implies that the average load increases by at most $16 \kappa_2$ between rounds $t_0$ and $m=t_0+n$. 

Next consider the allocation to the bins in $\mathcal{B}$ during the time-interval $[t_0,m]$. The maximum number of times a bin in $\mathcal{B}$ is chosen corresponds to the maximum load of a \OneChoice with (at least) $n/2$ balls into $n/2$ bins, which is $\Omega( \frac{\log n}{\log \log n})$ with probability at least $1-o(1)$ (see, e.g.,~Raab and Steger~\cite{RS98}). Every time such a bin in $\mathcal{B}$ is chosen, its load is incremented by at least one, so
\begin{align}
 \Pro{ \max_{i \in \mathcal{B}} (x_i^{t} - x_i^{t_0}) \geq \kappa_3 \cdot \frac{\log n}{\log \log n} ~\Big|~ \mathfrak{F}^{t_0}, \Phi^{t_0} < 8 c_6 n } \geq 1-o(1), \label{eq:drei}
\end{align}
for some suitable constant $\kappa_{3} > 0$.
Finally, taking the union bound over \eqref{eq:eins}, \eqref{eq:zwei} and \eqref{eq:drei}, we conclude that with probability at least $1-1/8-1/8-o(1) \geq 5/8$ it holds that 
\begin{align*}
  \max_{i \in [n]} y_i^{m}
  &\geq \max_{i \in \mathcal{B}} \left( x_i^{m} - \frac{W^{m}}{n} \right) \\
  &= \max_{i \in \mathcal{B}} \left( x_i^{m} - x_i^{t_0} + x_i^{t_0} - \frac{W^{t_0}}{n} + \frac{W^{t_0}}{n} - \frac{W^{m}}{n} \right) \\
  &\geq \max_{i \in \mathcal{B}}
  \left( x_i^{m} - x_i^{t_0} \right) + \min_{i \in \mathcal{B}} \left( x_i^{t_0} - \frac{W^{t_0}}{n} \right) + \left(\frac{W^{t_0}}{n} - \frac{W^{m}}{n}  \right) \\
  &\geq \kappa_3 \cdot \frac{\log n}{\log \log n} - 2 \kappa_1 - 16 \kappa_2,
\end{align*}
and the claim follows for some constant $\C{packing_all_m} >0$ since $\kappa_1, \kappa_2$ and $\kappa_3 $ are all positive constants.
\end{proof}

\subsection{Lower Bounds for Non-Filling Processes}
We shall now define a new condition for allocation processes which is satisfied for many natural processes, including \Twinning,  \Thinning, and $(1+\beta)$ with constant $\beta$. 
\NewConstant{p4k}{k}
\begin{itemize}
	\item \textbf{Condition \hypertarget{p4}{$\mathcal{P}_4$}}: for any $\eps>0$ there exists a constant $0<\C{p4k} \leq 1$ such that for all $t \geq 0$ with $\delta^t \in (\epsilon, 1 - \epsilon)$ and all bins $i \in [n]$ we have \[p_i^t \geq \frac{\C{p4k}}{ n}.\] 
\end{itemize}
So essentially this condition implies that in any round $t$ where the mean quantile is in $(\epsilon,1-\epsilon)$, there is at least a $\Omega(1/n)$-probability of allocating to each bin.

We shall now observe that 
\begin{itemize}
	\item For the \MeanThinning process, the probability of allocating to an overloaded bin $i \in [n]$ is $p_i^t = \frac{\delta^t}{n}$. Thus condition \PFour is satisfied with $\C{p4k} := \epsilon$. 
	\item For the \Twinning process, we have $p_i^t = \frac{1}{n}$. So \PFour is satisfied with $\C{p4k} := 1$. 
	\item The $(1+\beta)$-process has $ p_i^t = \frac{1-\beta}{n}+ \frac{\beta(2i-1)}{n^2}> \frac{1-\beta}{n}$, satisfying \PFour with $\C{p4k} := 1-\beta$. 
\end{itemize}

The next claim shows that for many rounds the mean quantile is not at the extremes. 
\NewConstant{good_steps_fraction_lower_bound}{d}
\begin{clm}
\label{clm:many_good_quantiles_in_lightly_loaded}
	Consider any process satisfying the conditions \PTwo and \WThree, or, \WTwo and \PThree. For any $m = \Theta(n \log n)$ and $\eps>0$, let $G_1^m:=G_1^m(\eps)$ be the number of rounds $s \in [1, m]$ with $\delta^s \in (\epsilon, 1 - \epsilon)$. Then there exists some $\eps>0$ and a constant $\C{good_steps_fraction_lower_bound} := \C{good_steps_fraction_lower_bound}(\epsilon)>0$ such that 
	\[
	\Pro{G_1^m \geq \C{good_steps_fraction_lower_bound} \cdot m} \geq 1 - 2\cdot n^{-12}.
	\]
\end{clm}
\begin{proof}By applying \cref{lem:stabilisation_many_good_quantiles_whp} for $t_0 = 0$, (where $\Lambda^0 = n$), we get that there exists a constant $C > 0$ such that $\Delta^t \leq C \cdot n$ for a constant fraction of the rounds $t$ in the range $[1,m]$ with probability at least $1 - n^{-12}$. The claim then follows by applying \cref{lem:good_quantile} at each of these rounds.
\end{proof}

\NewConstant{one_choice_prob_lower_bound}{d}
\begin{lem}
Consider any process satisfying \PFour and either the conditions \PTwo and \WThree, or, \WTwo and \PThree. Let $k= \C{p4k}\C{good_steps_fraction_lower_bound}/1000 $ where $\C{p4k}$ is specified by \PFour and $\C{good_steps_fraction_lower_bound}$ is from \cref{clm:many_good_quantiles_in_lightly_loaded}. Then  
\[
\Pro{\Gap(k \cdot n \log n) \geq k \cdot \log n} \geq 1 - o(1).
\]
\end{lem}
\begin{proof}Let $m=k\cdot n\log n$. By \cref{clm:many_good_quantiles_in_lightly_loaded} w.h.p.\ there exists some constants $\eps>0$ and $\C{good_steps_fraction_lower_bound}>0$ such that there are at least $\C{good_steps_fraction_lower_bound}m$ rounds $s\in[0,m]$ with $\delta^t\in(\eps,1-\eps)$. Denote this set of rounds by $S$ and observe that by condition \PFour, for any $s\in S$ and $i\in [n]$ we have $p_{i}^s \geq \C{p4k}/n $. 
	
	Observe that we can couple the location of the balls allocated as rounds in $S$ to locations under a one choice process as follows: before each round $s\in S$ we sample an independent Bernoulli random variable $X_s\sim \mathsf{Ber}(\C{p4k})$ with success probability $\C{p4k}$. If $X_s=1$ then we allocate the ball(s) to a uniformly random bin. Otherwise, if $X_s=0$ we allocate the ball(s) to the $i$\textsubscript{th} loaded bin with probability $(p_{i}^s -\C{p4k}/n)/(1-\C{p4k})$. If we let $X= \sum_{s\in S} X_s$ then it follows that, conditional on $|S|\geq \C{good_steps_fraction_lower_bound}m$ we have $X>\C{p4k}\C{good_steps_fraction_lower_bound}m/2$ w.h.p.\ by the Chernoff bound. It follows that w.h.p.\ at least $(\C{p4k}\C{good_steps_fraction_lower_bound}k/2)\cdot n\log n $ balls are allocated according to the one choice protocol.

	Raab and Steger~\cite{RS98} (see also~\cite[Section 4]{PTW15}) show that if $cn\log n$ balls are allocated according to the one choice protocol, for any constant $c>0$,  then w.h.p.\ the max load is at least  $(c+\sqrt{c}/10)\log n$. Thus if we choose $k= \C{p4k}\C{good_steps_fraction_lower_bound}/1000 $ then we have 
	\[ \frac{\Gap(m)}{\log n} \geq \frac{\C{p4k}\C{good_steps_fraction_lower_bound}k}{2}+ \frac{1}{10}\sqrt{\frac{\C{p4k}\C{good_steps_fraction_lower_bound}k}{2}}   - k = \frac{\C{p4k}^2\C{good_steps_fraction_lower_bound}^2}{2000}+ \frac{1}{10}\cdot\frac{\C{p4k}\C{good_steps_fraction_lower_bound}}{20\sqrt{5}}   - \frac{\C{p4k}\C{good_steps_fraction_lower_bound}}{1000}>\frac{\C{p4k}\C{good_steps_fraction_lower_bound}}{1000},\] with probability $1 - o(1)$ by taking the union bound of these three events. 
\end{proof}

 Hence, we can deduce from the lemma above by recalling that $\MeanThinning$, $\Twinning$ and $(1+\beta)$ all satisfy \PFour and either the conditions \PTwo and \WThree, or, \WTwo and \PThree:
\begin{cor}
For either the \MeanThinning, \Twinning or $(1+\beta)$-processes,  there exist a  constant $k>0$ (different for each process) such that 
\[
\Pro{\Gap(k\cdot n\log n ) \geq   k \cdot \log n} \geq 1 - o(1).
\]
\end{cor}

\section{Conclusions}
In this work we introduced two general frameworks that imply a $\Oh(\log n)$ gap bound via conditions on the probability bias or weight/filling bias. This framework can be applied to a variety of known but also new allocation processes to deduce a tight or nearly-tight gap bound. One important novelty of our proof is to relate the exponential potential to the quadratic (and absolute value) potential in order to prove stabilization.

One natural open problem is to refine some of the bounds, e.g., formulate stronger conditions which imply a gap bound of $o(\log n)$. A different direction is to study other ranges of our conditions, e.g., in \WTwo, we may even consider $w_{+}=0$ (overloaded bins are avoided at the cost of more bin samples), or $w_{+}=-1$ (balls from overloaded bins are processed and removed).

A different direction is to investigate processes satisfying relaxed versions of the conditions, i.e., processes where the probability bias (or the weight bias) is close to $0$ (as $n \rightarrow \infty$). In this case, one would expect the gap to grow by a function of the bias, and determining this exact relationship seems to be a challenging open problem.

\section*{Acknowledgments} 
We thank David Croydon and Martin Krejca for some helpful discussions.

\renewcommand\bibname{}

\phantomsection

\addcontentsline{toc}{section}{References}

\bibliographystyle{abbrv}
\bibliography{bibliography}

\clearpage
\appendix

\section{Analysis and Concentration Inequalities}

\begin{lem}[Log-Sum Inequality {\cite[Theorem 2.7.1]{CovTho}}]\label{lem:logsum}
For any natural number $n\geq 1$, let $a_{1},\ldots ,a_{n}$ and $b_{1},\ldots ,b_{n}$ be nonnegative real numbers and $a:= \sum_{i=1}^n a_i$ and $b:=\sum_{i=1}^n b_i$. Then  
\begin{equation*}\sum_{i=1}^n a_i \log \left(\frac{a_i}{b_i}  \right)\geq a\log \left(\frac{a}{b}  \right) \quad\text{or equivalently} \quad \sum_{i=1}^n a_i \log \left(\frac{b_i}{a_i}  \right)\leq a\log \left(\frac{b}{a}  \right) .\end{equation*}
\end{lem}

The following lemma is similar to~\cite[Lemma A.1.]{FGS12}

\begin{lem}\label{lem:quasilem}Let $(a_k)_{k=1}^n , (b_k)_{k=1}^n $ be non-negative and $(c_k)_{k=1}^n$ be non-negative and non-increasing. If $\sum_{k=1}^i a_k \leq \sum_{k=1}^i b_k$ holds for all $1\leq i\leq n$ then, \begin{equation} \label{eq:toprove}\sum_{k=1}^n a_k\cdot c_k \leq \sum_{k=1}^n b_k\cdot c_k.\end{equation}
\end{lem}
\begin{proof}
	We shall prove \eqref{eq:toprove} holds by induction on $n\geq 1$. The base case $n=1$ follows immediately from the fact that $a_1\leq b_1$ and $c_1\geq 0$. Thus we assume $\sum_{k=1}^{n-1} a_k\cdot c_k \leq \sum_{k=1}^{n-1} b_k\cdot c_k$ holds for all sequences $(a_k)_{k=1}^{n-1}, (b_k)_{k=1}^{n-1}$ and $ (c_k)_{k=1}^{n-1}$ satisfying the conditions of the lemma.
	
	For the inductive step, suppose we are given sequences $(a_k)_{k=1}^{n}$,  $(b_k)_{k=1}^{n}$ and $(c_k)_{k=1}^{n}$ satisfying the conditions of the lemma. If $c_2=0$ then, since $(c_k)_{k=1}^n$ is non-increasing and non-negative, $c_k=0$ for all $k\geq 2$. Thus as $a_1\leq b_1$ and $c_1\geq 0$ by the precondition of the lemma, we conclude 
\[ \sum_{k=1}^n a_k\cdot c_k = a_1 \cdot c_1 \leq b_1\cdot c_1 =\sum_{k=1}^n b_k\cdot c_k. \] We now treat the case $c_2>0 $. Define the non-negative sequences $(a_k')_{k=1}^{n-1}$ and $(b_k')_{k=1}^{n-1}$ 
	as follows:
	\begin{itemize}
		\item $a'_{1} = \frac{c_1}{c_{2}} \cdot a_{1} + a_2$ and $a'_{k} = a_{k+1}$ for $2\leq k \leq n-1$ ,
		 \item $b'_{1} = \frac{c_1}{c_{2}} \cdot b_{1} + b_2$ and $b'_{k} = b_{k+1}$ for $2\leq k \leq n-1$,
	\end{itemize} Then as the inequalities $c_1\geq c_2 $, $a_1 \leq b_1$ and $\sum_{i=1}^n a_k \leq \sum_{i=1}^n b_k$ hold by assumption, we have  
\[ \sum_{k=1}^{n-1}a_k' = \left(\frac{c_{1}}{c_{2}} - 1  \right)a_1 + \sum_{k=1}^{n}a_k  \leq \left(\frac{c_{1}}{c_{2}} - 1  \right)b_1 +  \sum_{k=1}^{n}b_k   = \sum_{k=1}^{n-1}b_k'. \] Thus if we also let $(c_k')_{k=1}^{n-1} = (c_{k+1})_{k=1}^{n-1} $, which is positive and non-increasing, then 
\[\sum_{k=1}^{n-1} a_k'\cdot c_k'  \leq \sum_{k=1}^{n-1} b_k'\cdot c_k',\]by the inductive hypothesis. However \[\sum_{k=1}^{n-1} a_k'\cdot c_k' = \left(\frac{c_1}{c_{2}} \cdot a_{1} + a_2 \right)c_2 +  \sum_{k=2}^{n-1} a_{k+1}\cdot c_{k+1} = \sum_{k=1}^{n} a_k\cdot c_k, \] and likewise $\sum_{k=1}^{n-1} b_k'\cdot c_k' = \sum_{k=1}^{n} b_k\cdot c_k$. The result follows.\end{proof}

For completeness, we define Schur-convexity (see \cite{MRBook}) and state two basic results:

\begin{defi}
A function $f : \mathbb{R}^n \to \mathbb{R}$ is Schur-convex if for any non-decreasing $x, y \in \mathbb{R}^n$, if $x$ majorizes $y$ then $f(x) \geq f(y)$. A function $f$ is Schur-concave if $-f$ is Schur-convex.
\end{defi}

\begin{lem} \label{lem:sum_of_convex_is_schur_convex}
Let $g : \mathbb{R} \to \mathbb{R}$ be a convex (resp. concave) function. Then $g(x_1, \ldots, x_n) := \sum_{i = 1}^n g(x_i)$ is Schur-convex (resp. Schur-concave). 
\end{lem}

\begin{lem}\label{clm:schur}
For any $\alpha > 0$, for any $\beta \in \mathbb{R}$ and any $\Delta \in \mathbb{R}$, consider $f(x_1,x_2,\ldots,x_k) = \sum_{j=1}^{k} \exp\left(- \alpha x_j \right)$, where $\sum_{j=1}^{k} x_j \geq  \Delta$ and $x_{j} \geq \beta $ for all $1 \leq j \leq k$. Then,
\[
  f(x_1,x_2,\ldots,x_k) \leq (k-1) \cdot \exp\left(- \alpha \cdot \beta \right) + 1 \cdot \exp\left(- \alpha \cdot \left( \Delta - (k-1) \cdot \beta \right) \right).
\]
\end{lem}
\begin{proof}
Note that by \cref{lem:sum_of_convex_is_schur_convex}, it follows that $f(x_1, \ldots, x_n)$ is Schur-concave, as $f(x_1, \ldots, x_n) = \sum_{i=1}^n g(x_i)$ and $g(z) = e^{-\alpha z}$ is concave for $\alpha > 0$. As a consequence, the function attains its maximum if the values $(x_1,x_2,\ldots,x_k)$ are as ``spread out'' as possible, i.e., if any prefix sum of the values ordered non-increasingly is as large as possible.
\end{proof}

\begin{lem} \label{clm:schur_varphi}
Consider a normalized load vector $y$ and vectors $y_1$, $y_2$ obtained by adding two balls to $y$, such that $y_2$ majorizes $y_1$. Then, for any $\alpha > 0$,
\[
\sum_{i : y_{1, i} \geq 2} e^{\alpha y_{1, i}} \leq \sum_{i : y_{2, i} \geq 2} e^{\alpha y_{2, i}} + 2 e^{3\alpha}.
\]
\end{lem}
\begin{proof}
Let $f(z) = \sum_{i : z_{i} \geq 2} e^{\alpha z_{i}}$. Note that for any $k \in \mathbb{N}$ and $\mathcal{R} = \{ v \in \mathbb{R}^k : v_i \geq 2\}$, the function $\tilde{f} : \mathcal{R} \to \mathbb{R}$ with $\tilde{f}(z) = \sum_{i = 1}^k e^{\alpha z_{i}}$ is Schur-convex, as it is the sum of convex functions $g(u) = e^{\alpha u}$ for $\alpha > 0$ (\cref{lem:sum_of_convex_is_schur_convex}).

We consider cases based on the (at most) two positions $i_1$ and $i_2$ (with $y_{i_1} \geq y_{i_2}$) where $y_{1, i_1} > y_{2, i_1}$ and $y_{1, i_2} > y_{2, i_2}$:
\begin{itemize}
    \item \textbf{Case 1} [$y_{1, i_2} \geq 2$ (and $y_{1, i_1} \geq 2$)]. In this case, because $y_2$ majorizes $y_1$, vectors $y_1$ and $y_2$ disagree in positions $j$ with $y_j \geq 2$. So by considering $\mathcal{J} := \{ j : y_{1, j} \geq 2\} = \{ j : y_{2, j} \geq 2\} =: j_1, \ldots , j_k$, for $k = |\mathcal{J}|$, the function $\tilde{f} : \mathcal{R} \to \mathbb{R}$ is Schur-convex. Hence, 
    \[f(y_1) = \tilde{f}(y_{1, j_1}, \ldots , y_{1, j_k}) \leq \tilde{f}(y_{2, j_1}, \ldots , y_{2, j_k}) = f(y_2).
    \]
    
    \item \textbf{Case 2} [$y_{1, i_1} \geq 2$ and $y_{1, i_2} < 2$]. The contribution of $y_{1, i_2} < 2$ is at most $e^{3\alpha}$. For $y_{1, i_1}$, because of $y_2$ majorizes $y_1$, the second ball must be placed in $j$, such that $y_{2, j} \geq y_{1, i_1}$. So, $f(y_1) \leq f(y_2) + e^{3\alpha}$.
    \item \textbf{Case 3} [$y_{1, i_1} < 2$]. Both can have a maximum contribution of $2 e^{3\alpha}$, so $f(y_1) \leq f(y_2) + 2 e^{3\alpha}$.
\end{itemize}
Combining the three cases, the claim follows.
\end{proof}

 \begin{lem}\label{lem:smoothness}
For integers $0\leq r_0 \leq r_1$, and real numbers $\epsilon \in (0,2/3)$ and $\xi> 0 $, we let $f: [r_0,r_1] \cap \mathbb{N} \rightarrow [0,1]$ be a function satisfying 
\begin{enumerate}
 \item $f(r_0) \leq 1-\epsilon$,
 \item $f(r_1) \geq \epsilon$,
 \item and for all $t \in [r_0,r_1-1]$ we have $f(t+1) \leq f(t) + \xi$.
\end{enumerate}
Then,
\[
 \left| \left\{ t \in [r_0,r_1] \colon f(t) \in \Bigr(\frac{\epsilon}{2},1- \frac{\epsilon}{2} \Bigl) \right\} \right| \geq \min \left\{ \frac{\epsilon}{2 \xi} , r_1-r_0  \right \}.
\]
\end{lem}
\begin{proof}
First consider the case where there exists a $t \in [r_0,r_1]$ with $f(t) \geq 1-\epsilon/2$. Further, let $t$ be the first round with that property, hence for any $s \in [r_0,t]$, $f(s) < 1 - \epsilon/2$. Further, thanks to the third property, $f(t-x) \geq f(t) - x \xi \geq  1 - \epsilon/2 - x \xi$ for any $x \geq 0$ (which also implies $r_0-t \geq \epsilon / (2 \xi) $), and thus as long as $0  \leq x  \leq \epsilon/(2 \xi)$,
\[
 f(t - x) \geq f(t) - x \xi \geq 1 - \epsilon/2 - \epsilon /2  \geq \epsilon /2,
\]
since $\epsilon \leq 2/3$.
Hence for any $s \in [t-\epsilon/(2 \xi),t]$,
\[
 f(s) \in [ \epsilon/2, 1-\epsilon/2].
\]
Now consider the case where for all rounds $t \in [r_0,r_1]$ we have $f(t) \leq 1 - \epsilon/2$. Since $f(r_1) \geq \epsilon$, we conclude for any $x \leq \epsilon/(2 \xi)$,
\[
 f(r_1 - x) \geq f(r_1) - x \xi \geq \epsilon -\epsilon / 2 \geq \epsilon / 2.
\]
Hence for any $s \in [\max\{r_0,r_1-\epsilon/(2 \xi)\},r_1]$ we have $ f(s) \in [ \epsilon/2, 1-\epsilon/2]$.
\end{proof}

 \begin{lem}[Method of Bounded Independent Differences {\cite[Corollary 5.2]{DubPan}}]\label{mobd} Let $f$ be a function of $n$ independent random variables $X_1 ,\dots , X_n$, where each $X_i$ takes values in a set $A_i$. Assume that for each $i \in [n]$ there exists a $c_i \geq 0 $ such that
	\[ \left|f(x_1 ,\dots,x_{i-1}, x_{i},x_{i+1},\dots, x_n) - f(x_1 ,\dots,x_{i-1}, x_{i}',x_{i+1},\dots, x_n) \right| \leq  c_i,\] for any $x_1\in A_1,\dots,x_{i-1}\in A_{i-1}, x_{i},x_{i}'\in A_{i} ,x_{i+1}\in A_{i+1},\dots, x_n \in A_n$. Then for any real $t>0$
	\[\Pro{f < \Ex{f} - t },\Pro{f > \Ex{f} + t }  \leq \exp\left(\frac{-t^2}{2 \sum_{i=1}^n c_i^2} \right).\]\end{lem}

\begin{lem}[Azuma's Inequality for Super-Martingales {\cite[Problem 6.5]{DubPan}}] \label{lem:azuma}
Let $X_0, \ldots, X_n$ be a super-martingale satisfying $|X_{i} - X_{i-1}| \leq c_i$ for any $i \in [n]$, then for any $\lambda > 0$,
\[
\Pro{X_n \geq X_0 + \lambda} \leq \exp\left(- \frac{\lambda^2}{2 \cdot \sum_{i=1}^n c_i^2} \right).
\]
\end{lem}

\begin{clm} \label{clm:eps_ineq}
For any $\epsilon \in (0,1)$, we have $(1 - \frac{\epsilon}{2}) /(1 -\epsilon) \geq 1 + \frac{\epsilon}{2}$.
\end{clm}
\begin{proof}
For $\epsilon \in (0,1)$ the following chain of double implications holds
\[
\frac{1 - \frac{\epsilon}{2}}{1 -\epsilon} \geq 1 + \frac{\epsilon}{2} \Leftrightarrow 1 - \frac{\epsilon}{2} \geq (1 -\epsilon) \cdot \Big(1 + \frac{\epsilon}{2}\Big) = 1 - \frac{\epsilon}{2} + \frac{\epsilon^2}{2} \Leftrightarrow \frac{\epsilon^2}{2} \geq 0.
\qedhere \]
\end{proof}

\begin{lem}[{\cite[Proposition B.3]{Motwani}}]\label{lem:cheatsheet} For any integer $n\geq 1$ and real $|x|\leq n$ we have $(1 +x/n)^n\geq e^x (1-x^2/n)$.  \end{lem}

\begin{lem} \label{lem:geometric_arithmetic}
Consider any sequence $(z_i)_{i \in \mathbb{N}}$ such that, for some $a > 0$ and $b > 0$, for every $i \geq 1$,
\[
z_i \leq z_{i-1} \cdot a + b.
\]
Then for every $i \in \mathbb{N}$, 
\[
z_i \leq z_0 \cdot a^i + b \cdot \sum_{j = 0}^{i-1} a^j.
\]
Further, if $a < 1$,
\[
z_i 
\leq z_0 \cdot a^i + \frac{b}{1 - a}.
\]
\end{lem}
\begin{proof}
We will prove the first claim by induction. For $i = 0$, $z_0 \leq z_0$. Assume the induction hypothesis holds for some $i \geq 0$, then since $a > 0$,
\[
z_{i+1} \leq z_{i} \cdot a + b \leq \Big(z_0 \cdot a^i + b \cdot \sum_{j = 0}^{i-1} a^j \Big) \cdot a + b = z_0 \cdot a^{i+1} +b \cdot \sum_{j = 0}^i a^j.
\]
Hence, the first claim follows. The second part of the claim is immediate, since for $a \in (0,1)$, $\sum_{j=0}^{\infty} a^j = \frac{1}{1-a}$.
\end{proof}

\section{Counterexamples for Exponential Potential Functions}

In this section, we present two configurations for which the potential functions $\Phi$ of \cref{sec:analysis_filling} and $\Lambda$ of \cref{sec:non_filling_lambda} may increase in expectation over one round, if we don't condition on ``good events''.

\begin{clm}\label{clm:counterexample}
For any constant $\alpha > 0$ and for sufficiently large $n$, consider the (normalized) load configuration,
\[
y^t = (\sqrt{n}, \underbrace{0, \ldots, 0}_{n-\sqrt{n}-1\text{ bins}}, \underbrace{-1, \ldots, -1}_{\sqrt{n}\text{ bins}}).
\]
Then, for the \Packing process, the potential function $\Phi^t := \sum_{i : y_i^t \geq 0} e^{\alpha y_i^t}$ will increase in expectation, i.e.,
\[
\ex{\Phi^{t+1} \mid \mathfrak{F}^t} \geq \Phi^t \cdot \Big(1 + 0.1 \cdot \frac{\alpha^2}{n}\Big).
\]
\end{clm}
\begin{proof}
Consider the contribution of bin $i=1$, with $y_1^t = \sqrt{n}$.
\begin{align*}
\ex{\Phi_1^{t+1} \mid \mathfrak{F}^t} & = e^{\alpha \sqrt{n}} \cdot \Big( 1 + \frac{1}{n} \cdot (e^{\alpha - \alpha/n} - 1) + \frac{n - \sqrt{n} - 1}{n} \cdot (e^{-\alpha/n} -1) + \frac{\sqrt{n}}{n} \cdot (e^{-2\alpha/n} -1) \Big).
\intertext{Now using a Taylor estimate $e^z \geq 1 + z + 0.3 z^2$ for $z \geq -1.5$,}
\ex{\Phi_1^{t+1} \mid \mathfrak{F}^t}
 & \geq e^{\alpha \sqrt{n}} \cdot \Big( 1 + \frac{1}{n} \cdot \Big(\alpha - \frac{\alpha}{n} + 0.3 \cdot \Big( \alpha - \frac{\alpha}{n}\Big)^2\Big) + \frac{n - \sqrt{n} - 1}{n} \cdot \Big(-\frac{\alpha}{n} + 0.3 \cdot \frac{\alpha^2}{n^2}\Big) \\
 & \qquad + \frac{\sqrt{n}}{n} \cdot \Big(-\frac{2\alpha}{n} + 1.2 \cdot \frac{\alpha^2}{n^2}\Big) \Big) \\
 & = e^{\alpha \sqrt{n}} \cdot \Big( 1 + \frac{\alpha + 0.3 \cdot \alpha^2}{n} -\frac{\alpha}{n} + o(n^{-1}) \Big) \\ &= e^{\alpha \sqrt{n}} \cdot \Big( 1 + 0.3 \cdot \frac{\alpha^2}{n} + o(n^{-1}) \Big) \\
 & \geq e^{\alpha \sqrt{n}} \cdot \Big( 1 + 0.2 \cdot \frac{\alpha^2}{n} \Big).
\end{align*}
At round $t$, the contribution of the rest of the bins is at most $n$. Note that for sufficiently large $n$, we have $n \cdot (1 + 0.1 \cdot \frac{\alpha^2}{n})  < 0.1 \cdot \alpha^2 \cdot e^{\alpha \sqrt{n}}$. Hence,
\begin{align*}
\ex{\Phi^{t+1} \mid \mathfrak{F}^t} 
 & \geq e^{\alpha \sqrt{n}} \cdot \Big( 1 + 0.2 \cdot \frac{\alpha^2}{n} \Big)
 \geq e^{\alpha \sqrt{n}} \cdot \Big( 1 + 0.1 \cdot \frac{\alpha^2}{n} \Big) + 0.1 \cdot \frac{\alpha^2}{n} \cdot e^{\alpha \sqrt{n}} \\ 
 & \geq \Phi_1^t \cdot \Big( 1 + 0.1 \cdot \frac{\alpha^2}{n} \Big) + n \cdot (1 + 0.1 \cdot \frac{\alpha^2}{n}) \\
 & \geq \Phi_1^t \cdot \Big( 1 + 0.1 \cdot \frac{\alpha^2}{n} \Big) + \left(\sum_{i > 1, y_i^t \geq 0} \Phi^t \right) \cdot \Big( 1 + 0.1 \cdot \frac{\alpha^2}{n} \Big) \\
 & = \Phi^t \cdot \Big( 1 + 0.1 \cdot \frac{\alpha^2}{n} \Big). \qedhere
\end{align*}
\end{proof}

In contrast to processes with constant probability bias in every round, such as those studied in~\cite{PTW15}, for the \MeanThinning process, there exists configurations where the potential $\Lambda$ for constant $\alpha$ will increase in expectation over a single round, even when it is $\omega(n)$. This is because in the worst-case, the \MeanThinning process may have a $\Theta(1/n)$ probability bias, corresponding to \textsc{Towards-Min} process in~\cite{PTW15}.

\begin{clm} \label{clm:bad_configuration_lambda}
For any constant $\alpha > 0$ and for sufficiently large $n$ consider the (normalized) load configuration,
\[
y^t = \Big(n^2, n, n, \ldots, n , - \frac{n\cdot (2n-3)}{2}, - \frac{n \cdot (2n-3)}{2}\Big).
\]
Then for the \MeanThinning process, the exponential potential $\Lambda$ will increase in expectation over one step, i.e.,
\[
\ex{\Lambda^{t+1} \mid \mathfrak{F}^t} \geq \Lambda^t \cdot \Big( 1 + 0.1 \cdot \frac{\alpha^2}{n}\Big),
\]
and $\delta^s \geq 1 - 2/n$ for $s \in [t, t + n^2)$.
\end{clm}
\begin{proof}
Consider the contribution of bin $i=1$, with $y_i^t = n^2$. The probability of allocating to that bin is $\frac{1 - \frac{2}{n}}{n}$, so using the Taylor estimate $e^z \geq 1 + z + 0.3 \cdot z^2$ for $z > 1.5$,
\begin{align*}
\ex{\Lambda_1^{t+1} \mid \mathfrak{F}^t} 
 & = e^{\alpha n^2} \cdot e^{-\alpha / n} \cdot (1 + (e^{\alpha} -1) \cdot \frac{1 - \frac{2}{n}}{n}) \\
 & \geq e^{\alpha n^2} \cdot (1 -\frac{\alpha}{n} + 0.3 \cdot  \frac{\alpha^2}{n^2}) \cdot \Big(1 + (\alpha + 0.3 \cdot \alpha^2) \cdot \frac{1 - \frac{2}{n}}{n} \Big) \\
 & = e^{\alpha n^2} \cdot \Big(1 -\frac{\alpha}{n} + (\alpha + 0.3\alpha^2) \cdot \frac{1}{n} + o(n^{-1}) \Big) \\
 & = e^{\alpha n^2} \cdot \Big(1 +0.2 \cdot \frac{\alpha^2}{n} \Big).
\end{align*}
Now note that for sufficiently large $n$,
\[
0.1 \cdot \frac{\alpha^2}{n} \cdot e^{\alpha n^2} 
\geq \Big(1 + 0.1 \cdot \frac{\alpha^2}{n} \Big) \cdot n \cdot e^{\alpha n^2} \cdot e^{-\alpha \cdot 3n/2} 
= \Big(1 + 0.1 \cdot \frac{\alpha^2}{n} \Big) \cdot n \cdot e^{\alpha n(2n-3)/2} 
\geq \Big(1 + 0.1 \cdot \frac{\alpha^2}{n} \Big) \cdot \left( \sum_{i > 1} \Lambda_i^t \right), 
\]
since $\Lambda_i^t \leq e^{\alpha n(2n-3)/2}$ for $i > 1$. Hence,
\begin{align*}
\ex{\Lambda^{t+1} \mid \mathfrak{F}^t} 
 & \geq \ex{\Lambda_1^{t+1} \mid \mathfrak{F}^t} \geq e^{\alpha n^2} \cdot \Big(1 + 0.2 \cdot \frac{\alpha^2}{n} \Big)
 \\ &= e^{\alpha n^2} \cdot \Big(1 + 0.1 \cdot \frac{\alpha^2}{n} \Big) + 0.1 \cdot \frac{\alpha^2}{n} \cdot e^{\alpha n^2} \\
 & \geq \Lambda_1^t \cdot \Big(1 + 0.1 \cdot \frac{\alpha^2}{n} \Big)+ \left( \sum_{i > 1} \Lambda_i^t \right) \cdot \Big(1 + 0.1 \cdot \frac{\alpha^2}{n} \Big)\\
 &= \Lambda^t \cdot \Big(1 + 0.1 \cdot \frac{\alpha^2}{n} \Big).
\end{align*}
Finally, since there are $n-2$ overloaded bins with overload at least $n$ and these can decrease by at most $1/n$ in each round, we have that $\delta^s \geq 1 - 2/n$ for any $s \in [t, t + n^2)$.
\end{proof}

\newpage

\section{Experimental Results} \label{sec:experiemental_results}

In this section, we present some empirical results for four processes: \MeanThinning, \Twinning, \Packing and \Caching (\cref{tab:gap_distribution} and \cref{fig:gap_vs_bins_alt}) and compare their load with that of a $(1 + \beta)$ process with $\beta = 0.5$ and with that of \TwoChoice process. 

\colorlet{GA}{black!40!white}
\colorlet{GB}{black!70!white}
\colorlet{GC}{black}
\newcommand{\CA}[1]{\textcolor{GA}{#1}} %
\newcommand{\CB}[1]{\textcolor{GB}{#1}} %
\newcommand{\CC}[1]{\textcolor{GC}{#1}} %
\newcommand{\CI}[2]{\FPeval{\result}{min(30 + 3 * #1, 100)} \colorlet{tmpC}{black!\result!white} {\textcolor{tmpC}{#2}}}

\begin{table}[ht]
    \centering
   \scriptsize{
    \begin{tabular}{|c|c|c|c|c|}
    \hline
$n$ & \MeanThinning & \Twinning & \Packing & \Caching \\ \hline
        $10^3$ &
\makecell{
\CI{2}{\textbf{\ 4} : \ 2\%} \\
\CI{38}{\textbf{\ 5} : 38\%} \\ 
\CI{35}{\textbf{\ 6} : 35\%} \\
\CI{15}{\textbf{\ 7} : 15\%} \\
\CI{8}{\textbf{\ 8} : \ 8\%} \\
\CI{1}{\textbf{\ 9} : \ 1\%} \\
\CI{1}{\textbf{12} : \ 1\%}} &
\makecell{
\CI{3}{\textbf{\ 8} : \ 3\%} \\
\CI{21}{\textbf{\ 9} : 21\%} \\
\CI{25}{\textbf{10} : 25\%} \\
\CI{18}{\textbf{11} : 18\%} \\
\CI{13}{\textbf{12} : 13\%} \\
\CI{8}{\textbf{13} : \ 8\%} \\
\CI{9}{\textbf{14} : \ 9\%} \\
\CI{3}{\textbf{17} : \ 3\%}} &
\makecell{
\CI{3}{\textbf{\ 6} : \ 3\%} \\
\CI{14}{\textbf{\ 7} : 14\%} \\
\CI{30}{\textbf{\ 8} : 30\%} \\
\CI{23}{\textbf{\ 9} : 23\%} \\
\CI{15}{\textbf{10} : 15\%} \\
\CI{8}{\textbf{11} : \ 8\%} \\
\CI{4}{\textbf{12} : \ 4\%} \\
\CI{1}{\textbf{13} : \ 1\%} \\
\CI{1}{\textbf{14} : \ 1\%} \\
\CI{1}{\textbf{15} : \ 1\%} } &
\makecell{
\CI{67}{\textbf{2} : 67\%} \\
\CI{33}{\textbf{3} : 33\%} } \\ \hline

$10^4$ & 
\makecell{
\CI{2}{\textbf{\ 6} : \ 2\%} \\ 
\CI{30}{\textbf{\ 7} : 30\%} \\
\CI{38}{\textbf{\ 8} : 38\%} \\
\CI{19}{\textbf{\ 9} : 19\%} \\
\CI{9}{\textbf{10} : \ 9\%} \\
\CI{1}{\textbf{11} : \ 1\%} \\
\CI{1}{\textbf{14} : \ 1\%}} &
\makecell{
\CI{1}{\textbf{11} : \ 1\%} \\
\CI{9}{\textbf{12} : \ 9\%} \\
\CI{24}{\textbf{13} : 24\%} \\
\CI{22}{\textbf{14} : 22\%} \\
\CI{13}{\textbf{15} : 13\%} \\
\CI{9}{\textbf{16} : \ 9\%} \\
\CI{8}{\textbf{17} : \ 8\%} \\
\CI{5}{\textbf{18} : \ 5\%} \\
\CI{6}{\textbf{19} : \ 6\%} \\
\CI{1}{\textbf{20} : \ 1\%} \\
\CI{1}{\textbf{21} : \ 1\%} \\
\CI{1}{\textbf{26} : \ 1\%} } &
\makecell{
\CI{2}{\textbf{\ 9} : \ 2\%} \\
\CI{17}{\textbf{10} : 17\%} \\
\CI{28}{\textbf{11} : 28\%} \\
\CI{14}{\textbf{12} : 14\%} \\
\CI{22}{\textbf{13} : 22\%} \\
\CI{11}{\textbf{14} : 11\%} \\
\CI{3}{\textbf{15} : \ 3\%} \\
\CI{2}{\textbf{16} : \ 2\%} \\
\CI{1}{\textbf{17} : \ 1\%} } &
\makecell{
\CI{5}{\textbf{2} : \ 5\%} \\
\CI{95}{\textbf{3} : 95\%} } \\ \hline

$10^5$ & 
\makecell{
\CI{3}{\textbf{\ 8} : \ 3\%} \\ 
\CI{32}{\textbf{\ 9} : 32\%} \\
\CI{38}{\textbf{10} : 38\%} \\
\CI{15}{\textbf{11} : 15\%} \\
\CI{6}{\textbf{12} : \ 6\%} \\
\CI{3}{\textbf{13} : \ 3\%} \\
\CI{3}{\textbf{14} : \ 3\%} } &
\makecell{
\CI{2}{\textbf{14} : \ 2\%} \\
\CI{5}{\textbf{15} : \ 5\%} \\
\CI{25}{\textbf{16} : 25\%} \\
\CI{28}{\textbf{17} : 28\%} \\
\CI{17}{\textbf{18} : 17\%} \\
\CI{10}{\textbf{19} : 10\%} \\
\CI{8}{\textbf{20} : \ 8\%} \\
\CI{1}{\textbf{21} : \ 1\%} \\
\CI{1}{\textbf{22} : \ 1\%} \\
\CI{3}{\textbf{23} : \ 3\%}} &
\makecell{
\CI{2}{\textbf{12} : \ 2\%} \\
\CI{16}{\textbf{13} : 16\%} \\
\CI{20}{\textbf{14} : 20\%} \\
\CI{28}{\textbf{15} : 28\%} \\
\CI{23}{\textbf{16} : 23\%} \\
\CI{5}{\textbf{17} : \ 5\%} \\
\CI{3}{\textbf{18} : \ 3\%} \\
\CI{1}{\textbf{19} : \ 1\%} \\
\CI{2}{\textbf{20} : \ 2\%} } &
\makecell{
\CI{100}{\textbf{3} : 100\%}
} \\ \hline
    \end{tabular}}
    \caption{Summary of observed gaps  for $n \in \{ 10^3, 10^4, 10^5\}$ bins and $m=1000 \cdot n$ number of balls, for $100$ repetitions. The observed gaps are in bold and next to that is the $\%$ of runs where this was observed.}
    \label{tab:gap_distribution}
\end{table}

\begin{figure}[ht]
    \centering
    \includegraphics[height=4.7cm]{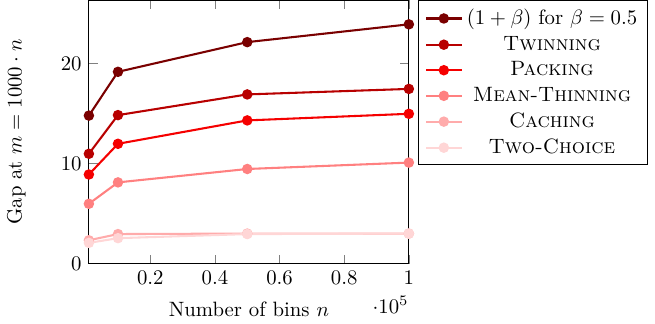}
    \caption{Average Gap vs.~$n \in \{ 10^3, 10^4, 5 \cdot 10^4, 10^5\}$ for the experimental setup of \cref{tab:gap_distribution}.}
    \label{fig:gap_vs_bins_alt}
\end{figure}

\end{document}